\documentclass[reqno,11pt,letterpaper]{amsart}

\usepackage{lipsum}
\usepackage{amsmath}
\usepackage{amssymb}
\usepackage{amsthm}
\usepackage{mathrsfs}
\usepackage{accents}
\usepackage{calc}
\usepackage{arydshln}
\usepackage{upgreek}
\usepackage{slashed}
\usepackage{xifthen}
\usepackage{graphicx}
\usepackage{longtable}
\usepackage[inline]{enumitem}

\usepackage{xr}

\newcommand{\citeI}[1]{\cite[#1]{HintzGlueLocI}}
\newcommand{\citeII}[1]{\cite[#1]{HintzGlueLocII}}

\usepackage{xcolor}
\definecolor{winered}{rgb}{0.6,0,0}
\definecolor{lessblue}{rgb}{0,0,0.7}

\usepackage[pdftex,colorlinks=true,linkcolor=winered,citecolor=lessblue,urlcolor=lessblue,breaklinks=true,bookmarksopen=true]{hyperref}

\makeatletter
\newcommand{\myitem}[2]{\item[\rm(#2)]\def\@currentlabel{#2}\label{#1}}
\makeatother

\usepackage{titletoc}

\makeatletter

\def\@tocline#1#2#3#4#5#6#7{
\begingroup
  \par
    \parindent\z@ \leftskip#3 \relax \advance\leftskip\@tempdima\relax
                  \rightskip\@pnumwidth plus 4em \parfillskip-\@pnumwidth
    \ifcase #1 % sections
       \vskip 0.6em \hskip 0em % add a little vspace before
       \or
       \or \hskip 0em % subsections
       \or \hskip 1em % subsubsections
    \fi%
    #6
    \nobreak\relax{\leavevmode\leaders\hbox{\,.}\hfill}
    \hbox to\@pnumwidth {\@tocpagenum{#7}}
  \par
\endgroup
}

 \def\l@section{\@tocline{0}{0pt}{0pc}{}{}}

\renewcommand{\tocsection}[3]{%
  \indentlabel{\@ifnotempty{#2}{ % for numbered sections
    \ignorespaces\bfseries{#2. #3}}}
  \indentlabel{\@ifempty{#2}{\ignorespaces\bfseries{#3}}{}} % for unnumbered sections
    \vspace{1.5pt}}

\renewcommand{\tocsubsection}[3]{%
  \indentlabel{\@ifnotempty{#2}{
    \ignorespaces#2. #3}}
  \indentlabel{\@ifempty{#2}{\ignorespaces #3}{}}
    \vspace{1.5pt}}

\renewcommand{\tocsubsubsection}[3]{%
  \indentlabel{\@ifnotempty{#2}{
    \ignorespaces#2. #3}}
  \indentlabel{\@ifempty{#2}{\ignorespaces #3}{}}
    \vspace{1.5pt}}

\makeatother

\makeatletter
\def\@nomenstarted{0}
\newlength{\@nomenoldtabcolsep}

\newcommand{\nomenstart}
  {%
    \def\@nomenstarted{1}%
    \setlength{\@nomenoldtabcolsep}{\tabcolsep}%
    \setlength{\tabcolsep}{3.5pt}%
    \begin{longtable}{p{0.11\textwidth} p{0.86\textwidth}}%found by hand
  }

\newcommand{\nomenitem}[2]{%
    \ifcase\@nomenstarted%
      \or % if nomenstarted=1, do nothing
      \or \\ % if nomenstarted=2, add newline to previous one
    \fi%
    #1\,{\leavevmode\leaders\hbox{\,.}\hfill} & #2%
    \def\@nomenstarted{2}%
  }%
\newcommand{\nomenend}
  {\\%
      \end{longtable}%
      \setlength{\tabcolsep}{\@nomenoldtabcolsep}%
      \def\@nomenstarted{0}%
  }
\makeatother

\makeatletter
\newcommand{\BIG}{\bBigg@{3.5}}
\newcommand{\vast}{\bBigg@{4}}
\newcommand{\Vast}{\bBigg@{5}}
\newcommand{\VAST}[1]{\bBigg@{#1}}
\makeatother

\allowdisplaybreaks

\numberwithin{equation}{section}
\numberwithin{figure}{section}
\newtheorem{thm}{Theorem}[section]

\newtheorem{prop}[thm]{Proposition}
\newtheorem{lemma}[thm]{Lemma}
\newtheorem{cor}[thm]{Corollary}

\newtheorem*{thm*}{Theorem}
\newtheorem*{prop*}{Proposition}
\newtheorem*{cor*}{Corollary}
\newtheorem*{conj*}{Conjecture}

\theoremstyle{definition}
\newtheorem{definition}[thm]{Definition}

\theoremstyle{remark}
\newtheorem{rmk}[thm]{Remark}

\makeatletter
\newcommand{\fakephantomsection}{%
  \Hy@MakeCurrentHref{\@currenvir.\the\Hy@linkcounter}
  \Hy@raisedlink{\hyper@anchorstart{\@currentHref}\hyper@anchorend}%
  \Hy@GlobalStepCount\Hy@linkcounter%
}
\makeatother

\newcommand{\mc}{\mathcal}
\newcommand{\cA}{\mc A}

\newcommand{\cC}{\mc C}
\newcommand{\cD}{\mc D}
\newcommand{\cE}{\mc E}
\newcommand{\cF}{\mc F}

\newcommand{\cH}{\mc H}

\newcommand{\cK}{\mc K}
\newcommand{\cL}{\mc L}
\newcommand{\cM}{\mc M}

\newcommand{\cO}{\mc O}
\newcommand{\cP}{\mc P}
\newcommand{\cQ}{\mc Q}

\newcommand{\cT}{\mc T}
\newcommand{\cU}{\mc U}
\newcommand{\cV}{\mc V}

\newcommand{\cX}{\mc X}

\newcommand{\ms}{\mathscr}

\newcommand{\sC}{\ms C}
\newcommand{\sD}{\ms D}

\newcommand{\sE}{\ms E}
\newcommand{\sF}{\ms F}

\newcommand{\scri}{\ms I}

\newcommand{\sO}{\ms O}

\newcommand{\sR}{\ms R}
\newcommand{\sS}{\ms S}

\newcommand{\C}{\mathbb{C}}
\newcommand{\N}{\mathbb{N}}
\newcommand{\R}{\mathbb{R}}
\newcommand{\Z}{\mathbb{Z}}

\newcommand{\Sph}{\mathbb{S}}

\newcommand{\sfb}{\mathsf{b}}

\newcommand{\sfr}{\mathsf{r}}
\newcommand{\sfs}{\mathsf{s}}

\newcommand{\sfG}{\mathsf{G}}

\newcommand{\bfB}{\mathbf{B}}

\newcommand{\fa}{\mathfrak{a}}
\newcommand{\fc}{\mathfrak{c}}

\newcommand{\fg}{\mathfrak{g}}
\newcommand{\fh}{\mathfrak{h}}

\newcommand{\fm}{\mathfrak{m}}
\newcommand{\fp}{\mathfrak{p}}
\newcommand{\fq}{\mathfrak{q}}

\newcommand{\ft}{\mathfrak{t}}

\newcommand{\fT}{\mathfrak{T}}

\newcommand{\sld}{\slashed{\dd}{}}
\newcommand{\slg}{\slashed{g}{}}

\newcommand{\slDelta}{\slashed{\Delta}{}}

\newcommand{\slstar}{\slashed{\star}}

\newcommand{\scal}{\mathsf{S}}
\newcommand{\scalspace}{\mathbf{S}}
\newcommand{\vect}{\mathsf{V}}

\newcommand{\vectspace}{\mathbf{V}}

\newcommand{\ran}{\operatorname{ran}}

\newcommand{\End}{\operatorname{End}}

\renewcommand{\Re}{\operatorname{Re}}
\renewcommand{\Im}{\operatorname{Im}}
\newcommand{\Id}{\operatorname{Id}}

\newcommand{\mathspan}{\operatorname{span}}
\newcommand{\supp}{\operatorname{supp}}

\newcommand{\tr}{\operatorname{tr}}

\newcommand{\halfopen}{[0,1)} % to avoid ViM's parenthesis matching error from messing up syntax highlighting in footnotes
\newcommand{\negreal}{(-\infty,0]}

\newcommand{\cd}{\fc}

\newcommand{\Ups}{\Upsilon}

\newcommand{\eps}{\epsilon}

\newcommand{\la}{\langle}

\newcommand{\ol}{\overline}
\newcommand{\pa}{\partial}
\newcommand{\dd}{{\mathrm d}}
\newcommand{\ra}{\rangle}

\newcommand{\specb}{\operatorname{spec}_\bop}
\newcommand{\Specb}{\operatorname{Spec}_\bop}

\newcommand{\ul}[1]{\underline{#1}{}}

\newcommand{\wh}{\widehat}
\newcommand{\wt}{\widetilde}
\newcommand{\xra}{\xrightarrow}

\newcommand{\ubar}[1]{\underaccent{\bar}#1}%{\mkern1.5mu\underline{\mkern-1.5mu #1\mkern-1.5mu}\mkern1.5mu}
\newcommand{\pfstep}[1]{$\bullet$\ \underline{\textit{#1}}}
\newcommand{\pfsubstep}[2]{{\bf#1}\ \textit{#2}}

\newcommand{\bop}{{\mathrm{b}}}

\newcommand{\sop}{{\mathrm{s}}}
\newcommand{\seop}{{\mathrm{se}}}

\newcommand{\scop}{{\mathrm{sc}}}

\newcommand{\scl}{{\mathrm{sc}}}

\newcommand{\tbop}{{3\mathrm{b}}}
\newcommand{\tscop}{{3\mathrm{sc}}}

\newcommand{\scbtop}{{\mathrm{sc}\text{-}\mathrm{b}}}
\newcommand{\schop}{{\mathrm{sc},\semi}}

\newcommand{\semi}{\hbar}

\newcommand{\ff}{\mathrm{ff}}

\newcommand{\scface}{{\mathrm{scf}}}

\newcommand{\tface}{{\mathrm{tf}}}

\newcommand{\zface}{{\mathrm{zf}}}

\newcommand{\rms}{{\mathrm{s}}}
\newcommand{\rmv}{{\mathrm{v}}}

\newcommand{\cp}{{\mathrm{c}}}

\newcommand{\Diff}{\mathrm{Diff}}

\newcommand{\Vb}{\cV_\bop}

\newcommand{\Vscbt}{\cV_\scbtop}

\newcommand{\Vs}{\cV_\sop}
\newcommand{\Vse}{\cV_\seop}

\newcommand{\Diffb}{\Diff_\bop}

\newcommand{\Diffscbt}{\Diff_\scbtop}

\newcommand{\Diffs}{\Diff_\sop}
\newcommand{\Diffse}{\Diff_\seop}

\newcommand{\Vtb}{\cV_\tbop}
\newcommand{\Vtsc}{\cV_{3\scl}}

\newcommand{\Difftb}{\Diff_\tbop}

\newcommand{\Tscbt}{{}^\scbtop T}

\newcommand{\Vsc}{\cV_\scop}
\newcommand{\Diffsc}{\Diff_\scop}

\newcommand{\Diffsch}{\Diff_\schop}

\newcommand{\WF}{\mathrm{WF}}

\newcommand{\Tse}{{}^\seop T}
\newcommand{\Tsc}{{}^{\scop}T}

\newcommand{\Ttb}{{}^{\tbop}T}
\newcommand{\Ttsc}{{}^{\tscop}T}

\newcommand{\Sse}{{}^\seop S}

\newcommand{\Stb}{{}^{\tbop}S}

\newcommand{\loc}{{\mathrm{loc}}}
\newcommand{\CI}{\cC^\infty}
\newcommand{\CIdot}{\dot\cC^\infty}

\newcommand{\CIc}{\cC^\infty_\cp}

\newcommand{\Hb}{H_{\bop}}

\newcommand{\Hbsupp}{\dot H_{\bop}}

\newcommand{\Htb}{H_\tbop}

\newcommand{\Hsc}{H_{\scop}}

\newcommand{\phg}{{\mathrm{phg}}}

\newcommand{\Riem}{\mathrm{Riem}}
\newcommand{\Ric}{\mathrm{Ric}}

\newcommand{\bhm}{\fm}
\newcommand{\bha}{\fa}

\newcommand{\openbigpmatrix}[1]
  {%
    \def\@bigpmatrixsize{#1}%
    \addtolength{\arraycolsep}{-#1}%
    \begin{pmatrix}%
  }
\newcommand{\closebigpmatrix}
  {%
    \end{pmatrix}%
    \addtolength{\arraycolsep}{\@bigpmatrixsize}%
  }

\newlength{\enummargin}

\newcommand{\usref}[1]{{\upshape\ref{#1}}}

\DeclareGraphicsExtensions{.mps}

\makeatletter
\newcommand*{\fwbw}[1]{\expandafter\@fwbw\csname c@#1\endcsname}
\newcommand*{\@fwbw}[1]{\ifcase #1 \or {\rm fw}\or {\rm bw}\fi}
\AddEnumerateCounter{\fwbw}{\@fwbw}
\makeatother

\begin{document}

%%%%%%%%%%%%%%%%%%%%%%%%%%%%%%%%%%%%%%%%%%%%%%%%%%%%%%%%%%%%%%%%%%%%%%
% title page
\title[Gluing small black holes along timelike geodesics III]{Gluing small black holes along timelike geodesics III: construction of true solutions and extreme mass ratio mergers}

\date{\today}

\begin{abstract}
  Given a smooth globally hyperbolic $(3+1)$-dimensional spacetime $(M,g)$ satisfying the Einstein vacuum equations (possibly with cosmological constant) and an inextendible timelike geodesic $\cC$, we construct, on any compact subset of $M$, solutions $g_\eps$ of the Einstein equations which describe a mass $\eps$ Kerr black hole traveling along $\cC$. More precisely, away from $\cC$ one has $g_\eps\to g$ as $\eps\to 0$, while the $\eps^{-1}$-rescaling of $g_\eps$ around every point of $\cC$ tends to a fixed subextremal Kerr metric. Our result applies on all spacetimes with noncompact Cauchy hypersurfaces, and also on spacetimes which do not admit nontrivial Killing vector fields in a neighborhood of a point on the geodesic. As an application, we construct spacetimes which model the merger of a very light subextremal Kerr black hole with a slowly rotating unit mass Kerr(--de~Sitter) black hole, followed by the relaxation of the resulting black hole to its final Kerr(--de~Sitter) state.

  In Part I \cite{HintzGlueLocI}, we constructed approximate solutions $g_{0,\eps}$ of the gluing problem which satisfy the Einstein equations only modulo $\cO(\eps^\infty)$ errors. Part II \cite{HintzGlueLocII} introduces a framework for obtaining uniform control of solutions of linear wave equations on $\eps$-independent precompact subsets of the original spacetime $(M,g)$. In this final part, we show how to correct $g_{0,\eps}$ to a true solution $g_\eps$ by adding a metric perturbation of size $\cO(\eps^\infty)$ which solves a carefully chosen gauge-fixed version of the Einstein vacuum equations. The main novel ingredient is the proof of suitable mapping properties for the linearized gauge-fixed Einstein equations on subextremal Kerr spacetimes.
\end{abstract}

% 83C57: black holes
% 83C05: Einstein field equations
% 35B25: singular perturbation in context of PDEs
% 35P25: scattering theory
% 58J47: propagation of singularities; initial value problems

\author{Peter Hintz}
\address{Department of Mathematics, ETH Z\"urich, R\"amistrasse 101, 8092 Z\"urich, Switzerland}
\email{peter.hintz@math.ethz.ch}

\maketitle

\setlength{\parskip}{0.00pt}
\tableofcontents
\setlength{\parskip}{0.05in}

%%%%%%%%%%%%%%%%%%%%%%%%%%%%%%%%%%%%%%%%%%%%%%%%%%%%%%%%%%%%%%%%%%%%%%
\section{Introduction}
\label{SI}

We complete the proof, begun in \cite{HintzGlueLocI,HintzGlueLocII}, of a general spacetime gluing result which allows us to glue a subextremal Kerr black hole \cite{KerrKerr} of small mass along a timelike geodesic inside of any spacetime $(M,g)$ which solves the Einstein vacuum equations $\Ric(g)-\Lambda g=0$ with $\Lambda\in\R$ (and has a noncompact Cauchy hypersurface or is subject to a genericity condition).

For parameters $b=(\bhm,\bha)$ where $\bhm>0$ and $\bha\in\R^3$, $|\bha|<\bhm$, we denote the Kerr metric on $\hat M_b^\circ:=\R_{\hat t}\times\hat X_b^\circ$, $\hat X_b^\circ:=\{\hat x\in\R^3\colon|\hat x|\geq\bhm\}$, by $\hat g_b=(\hat g_b)_{\hat\mu\hat\nu}(\hat x)\,\dd\hat z^\mu\,\dd\hat z^\nu$; here $\hat z=(\hat t,\hat x)$. This is a stationary solution of $\Ric(\hat g_b)=0$ which is asymptotically flat, i.e.
\[
  (\hat g_b)_{\hat\mu\hat\nu}(\hat x) \xra{|\hat x|\to\infty} \hat{\ubar g}_{\hat\mu\hat\nu},\qquad \hat{\ubar g}=-\dd\hat t^2+\dd\hat x^2.
\]
Recall that the event horizon is located at $|\hat x|=\bhm+\sqrt{\bhm^2-|\bha|^2}$; thus $\hat M_b^\circ$ includes a part lying inside the black hole. We define $\hat g_b$ in such a way that $\dd\hat t$ is everywhere timelike.

\begin{thm}[Main theorem, rough version]
\label{ThmI}
  Let $\Lambda\in\R$. Let $(M,g)$ denote a smooth, connected, globally hyperbolic spacetime solving $\Ric(g)-\Lambda g=0$. Let $\cC\subset M$ be an inextendible timelike geodesic, and let $X\subset M$ be a spacelike Cauchy hypersurface. Assume that
  \begin{enumerate}
  \myitem{ItINoncompact}{I} $X$ is noncompact; \emph{or}
  \myitem{ItIGeneric}{II} $(M,g)$ does not have nontrivial Killing vector fields near the point $\{\fp\}=X\cap\cC$.
  \end{enumerate}
  Write $(t,x)$ for Fermi normal coordinates around $\cC$. (In particular, $g|_{(0,x)}=\ubar g:=-\dd t^2+\dd x^2$.) Let $K^\circ\subset M$ be an open set with compact closure lying in the future of $X$. Then for all sufficiently small $\eps>0$, there exists a smooth solution $g_\eps$ of
  \begin{equation}
  \label{EqIEVE}
    \Ric(g_\eps)-\Lambda g_\eps=0
  \end{equation}
  on $K\cap M_\eps$ where $M_\eps:=M\setminus\{|x|\leq\eps\bhm\}$, so that $g_\eps\to g$ in $\CI(\bar V;S^2 T^*M)$ for all open $V\subset K$ with $\bar V\cap\cC=\emptyset$, while near $\cC$ we have
  \begin{equation}
  \label{EqIMetric}
    (g_\eps)_{\mu\nu}(t,x) = g_{\mu\nu}(t,x) - \ubar g_{\mu\nu} + (\hat g_{\eps\bhm,\eps\bha})_{\hat\mu\hat\nu}(x) + (h_\eps)_{\mu\nu}(t,x),
  \end{equation}
  where $(h_\eps)_{\mu\nu}\to 0$ as $\eps\to 0$ together with all derivatives along $\pa_t$ and $(\eps^2+|x|^2)^{\frac12}\pa_x$. In the setting~\eqref{ItIGeneric}, we can moreover ensure that $g_\eps=g$ outside of the domain of influence of any nonempty open neighborhood $\cU^\circ$ of $\fp$ with the property that there do not exist any nontrivial Killing vector fields in the domain of dependence of $\cU^\circ$ in $(M,g)$.
\end{thm}

The meaning of~\eqref{EqIMetric} is that $g_\eps$ arises from $g$ by inserting a Kerr black hole with parameters $\eps b=(\eps\bhm,\eps\bha)$ along $\cC$: the flat Minkowski metric $\ubar g$ (which is the local model of $g$ at $\cC=\{x=0\}$) is replaced by the Kerr metric $\hat g_{\eps b}$. (In particular, the set $\{|x|<\eps\bhm\}$ excised from $M$ lies deep inside of the small Kerr black hole.) The remainder term $h_\eps$ captures the gravitational perturbations created by this insertion. We present a more precise version of Theorem~\ref{ThmI} below (see Theorem~\ref{ThmIPrec}); the full result is Theorem~\ref{ThmTrSemi} (with Theorem~\ref{ThmTrGlueLocI}, which is the main theorem from \cite{HintzGlueLocI}, as an input).

If $K$ has spacelike boundary hypersurfaces and admits a Cauchy surface $X_K$, then small perturbations (with size bounded from above by some positive function of $\eps$ which goes to $0$ sufficiently quickly with $\eps$) of the initial data of $g_\eps$ at $X_K$ evolve into metrics on $K$ which are close to $g_\eps$ (when, say, using a generalized harmonic gauge condition relative to $g_\eps$); this follows simply from Cauchy stability. (Our proof shows that $\cO(\eps^\infty)$-perturbations are already sufficiently small in this sense.) Thus, Theorem~\ref{ThmI} probes an open set of the `moduli space' of solutions of the Einstein vacuum equations.

By choosing $(M,g)$ to be a slowly rotating unit mass Kerr(--de~Sitter) spacetime and $\cC$ a timelike geodesic which crosses the event horizon,
\[
  \parbox{0.8\textwidth}{the metric $g_\eps$ thus describes the merger of a very light (mass $\eps\bhm$) Kerr black hole with a unit mass Kerr(--de~Sitter) black hole.}
\]
One can then concatenate the description of $g_\eps$ on a compact subset of $M$ containing the crossing time with known black hole stability results \cite{HintzVasyKdSStability,FangKdS,KlainermanSzeftelKerr,GiorgiKlainermanSzeftelStability,ShenGCMKerr} (see also \cite{DafermosHolzegelRodnianskiTaylorSchwarzschild}) to deduce that the merged black hole settles down to a stationary Kerr(--de~Sitter) black hole at late times; see Theorem~\ref{ThmAppMerger} for details. We leave any further analysis of merger spacetimes (e.g.\ the construction of the event horizon, the structure of the black hole region---conjecturally connected---, the relationship between the initial and the final Bondi mass) to future work.

It is an interesting problem to determine how much the construction of $g_\eps$ can be localized. In the construction presented here, which involves the solution of (infinitely many) linear and nonlinear wave equations on $(M,g)$ and on perturbations of $(M_\eps,g_{0,\eps})$, we can only make simple domain-of-influence type statements. But in the context of black hole mergers, it would be interesting to see if, for example, one can ensure that the metric after the merger, at some finite retarded time, is exactly isometric to that of an extremal Kerr black hole, and thus to investigate whether it is possible to make a subextremal Kerr black hole extremal via a merger process.

The construction of $g_\eps$ on larger subsets of $M$ (which may thus grow as $\eps\to 0$) requires, at the very least, detailed control of perturbations of $(M,g)$ (as solutions of the Einstein equations) in the full causal future of $X$. It may be feasible to attempt gluing a small black hole along a future-complete geodesic when $(M,g)$ is de~Sitter space (or the cosmological region of Kerr--de~Sitter); see the discussion following Theorem~\ref{ThmAppMerger}. When $(M,g)$ is a Kerr black hole and $\cC$ is a bound orbit, the gluing construction (which would then conjecturally yield extreme mass ratio inspirals) is significantly more challenging.

\bigskip

For a detailed overview of the literature on extreme mass ratio mergers and gluing problems in relativity, as well as for a description of the senses in which Theorem~\ref{ThmI} is optimal and as general as one can hope for in the vacuum regime, we refer the reader to the introduction of \cite{HintzGlueLocI}. We only remark here that the statement that small massive bodies in general relativity must move (approximately) along timelike geodesics is known as the \emph{geodesic hypothesis}, which has received much attention especially in the physics literature \cite{EinsteinInfeldHoffmanGeodesics,ThomasGeodesicHypothesis,TaubGeodesicHypothesis,EhlersGerochMotion,GrallaWaldSelfForce}, and also in the more recent mathematics literature \cite{StuartGeodesicsEinstein,YangGeodesicHypo,HintzGlueLocI}.

In the remainder of this introduction, we shall only discuss the conceptual and technical challenges we face here given the prior works \cite{HintzGlueLocI,HintzGlueLocII} (the key results of which we briefly recall below).

%%%%%%%%%%%%%%%%%%%%%%%%%%%%%%%%%%%%%%%%%%%%%%%%%%
\subsection{Formal solutions; non-perturbative corrections to true solutions}
\label{SsIF}

The starting point of our proof of Theorem~\ref{ThmI} is the construction of \emph{formal} solutions $g_{0,\eps}$ of the gluing problem in \cite{HintzGlueLocI}. The precise statement we use is recalled in Theorem~\ref{ThmTrGlueLocI}; we present here an abbreviated version. \emph{For the remainder of this introduction, we work in Fermi normal coordinates $t,x$ near $\cC$, and with $\hat x=\frac{x}{\eps}$.}

\begin{thm}[Formal solutions, rough version]
\label{ThmIF}
  (See \citeI{Theorems~\ref*{ThmM},~\ref*{ThmIVP}, and Remark~\ref*{RmkXNc}}.) Under the assumptions of Theorem~\usref{ThmI}, there exists a family $g_{0,\eps}$ of smooth metrics on $M_\eps$ of the form~\eqref{EqIMetric} so that $\Ric(g_{0,\eps})-\Lambda g_{0,\eps}$ vanishes to infinite order as $\eps\searrow 0$ together with all derivatives along $\pa_t,\pa_x$, and $\pa_\eps$. Expressed in terms of
  \begin{equation}
  \label{EqIFCoord}
    t,\qquad \hat\rho := |x|,\qquad \rho_\circ := \frac{\eps}{|x|} = |\hat x|^{-1},\qquad \omega=\frac{x}{|x|},
  \end{equation}
  the $z=(t,x)$-metric coefficients $(g_{0,\eps})_{\mu\nu}=(g_{0,\eps})_{\mu\nu}(t,\hat\rho,\rho_\circ,\omega)$ satisfy
  \begin{align}
  \label{EqIFKerr}
    (g_{0,\eps})_{\mu\nu}(t,\hat\rho,\rho_\circ,\omega) &= (\hat g_b)_{\hat\mu\hat\nu}(\rho_\circ^{-1}\omega) + \cO(\hat\rho^2), \\
  \label{EqIFBg}
    (g_{0,\eps})_{\mu\nu}(t,\hat\rho,0,\omega) &= g_{\mu\nu}(\hat\rho\omega).
  \end{align}
  Furthermore, $(g_{0,\eps})_{\mu\nu}$ is polyhomogeneous on $\R_t\times[0,1)_{\hat\rho}\times[0,\bhm^{-1})_{\rho_\circ}\times\Sph^2$. Finally, we can arrange for $\Ric(g_{0,\eps})-\Lambda g_{0,\eps}$ to vanish to infinite order (together with all derivatives) also at the Cauchy hypersurface $X$ in setting~\eqref{ItIGeneric}, and at any fixed precompact open subset of $X$ in setting~\eqref{ItINoncompact}.
\end{thm}

Note that $\hat\rho=0$ defines the near field regime (described using $t$ and $\hat x=\rho_\circ^{-1}\omega$), while $\rho_\circ=0$ is the far field regime (described using $t$ and $x=\hat\rho\omega$). Observe that~\eqref{EqIFKerr} and \eqref{EqIFBg} are compatible since, by definition of Fermi normal coordinates, $g(\hat\rho\omega)$ is equal to the Minkowski metric plus terms of size $\cO(\hat\rho^2)$ which encode the Riemann curvature tensor of $g$ at $\cC=\{x=0\}$. In particular, the `Kerr-mod-$\cO(\hat\rho^2)$' behavior~\eqref{EqIFKerr} is optimal. We also recall that one can phrase Theorem~\ref{ThmIF} as the existence of a polyhomogeneous section of a suitable vector bundle of symmetric 2-tensors over a 5-dimensional manifold with corners $\wt M=[[0,1)_\eps\times M;\{0\}\times\sC]$, with~\eqref{EqIFCoord} being one local chart of $\wt M$.

\begin{thm}[Main theorem, more precise version]
\label{ThmIPrec}
  (See Theorem~\usref{ThmTrSemi}.) Given a formal solution $(g_{0,\eps})_{\eps\in(0,1)}$ of the black hole gluing problem produced by Theorem~\usref{ThmIF} and a compact set $K\subset M$ lying in the future of $X$, there exist $\eps_0>0$ and a family of correction terms $(h_\eps)_{\eps\in(0,\eps_0]}$ so that $g_\eps=g_{0,\eps}+h_\eps$ solves~\eqref{EqIEVE}, and so that
  \begin{equation}
  \label{EqIPrechBounds}
    |(h_\eps)_{\mu\nu}|\leq C_N\eps^N\quad\forall\,N,
  \end{equation}
  and the same bounds hold also for any finite number of derivatives along\footnote{Due to the $\eps^N$ decay, one can equivalently take the vector fields $\pa_t$ and $(\eps^2+|x|^2)^{\frac12}\pa_x$ here, which are the basic \emph{s-vector fields} introduced in \cite{HintzGlueLocII} and which arise naturally in the nonlinear analysis. We record regularity in $\eps$ only for aesthetic reasons, as this allows us to conclude that $(g_\eps)_{\eps\in(0,\eps_0]}$ is polyhomogeneous on $\wt M$; we point out that until we prove the $\eps$-regularity of $h_\eps$, the $\eps$-regularity of $g_{0,\eps}$ in Theorem~\ref{ThmIF} is not used. See the final part of Theorem~\ref{ThmTrLoc}.} $\pa_\eps,\pa_t,\pa_x$. In particular, $(g_\eps)_{\eps\in(0,\eps_0]}$ is polyhomogeneous on $\wt M$.
\end{thm}

This confirms the author's expectation expressed in \citeI{Conjecture~\ref*{ConjITrue}}.

\begin{rmk}[Relationship with initial data gluing results]
\label{RmkIID}
  The initial data (first and second fundamental form) of the family $(g_\eps)_{\eps\in(0,\eps_0]}$ have the same regularity and asymptotic behavior as those constructed in \cite{HintzGlueID}. Conversely, as already discussed in \cite[\S{1.4}]{HintzGlueID}, the evolution of general initial data as constructed in \cite{HintzGlueID} will \emph{not} be of the \emph{adiabatic} nature which the metrics $g_\eps$ constructed here possess (i.e.\ regularity in $t$, as opposed to merely regularity with respect to the fast time coordinate $\hat t=\frac{t}{\eps}$). It is thus important that the construction of formal solutions in \cite{HintzGlueLocI} was accomplished without reference to initial data, whereas in the construction of a correction to a true solution in the present paper we are free to use an initial value formulation, heuristically because the error term which we need to solve away vanishes to all orders in $\eps$, so there is no difference anymore between adiabatic and fast time regularity. --- There are only few results to date regarding details of the evolution of initial data arising from gluing constructions (e.g.\ \cite{FriedrichStability} given the initial data of \cite{CorvinoScalar,CorvinoSchoenAsymptotics}, or \cite{ShenMinkExtStab} given initial data such as those of \cite{CarlottoSchoenData}), and for the many-black-hole initial data of \cite{ChruscielDelayMapping} only \cite{ChruscielMazzeoManyBH}.
\end{rmk}

%%%%%%%%%%%%%%%%%%%%%%%%%%%%%%%%%%%%%%%%%%%%%%%%%%
\subsection{Main ideas of the proof}
\label{SsII}

A natural idea for proving Theorem~\ref{ThmIPrec} is as follows. Define the generalized harmonic gauge 1-form
\[
  \Ups(g_\eps;g_{0,\eps})_\mu := (g_\eps)_{\mu\nu}(g_\eps)^{\rho\lambda}\bigl(\Gamma(g_\eps)_{\rho\lambda}^\nu - \Gamma(g_{0,\eps})_{\rho\lambda}^\nu\bigr) = 0,\qquad g_\eps = g_{0,\eps} + h_\eps
\]
(see also \cite{GrahamLeeConformalEinstein} and Definition~\ref{DefEOp}). One then attempts to find $h_\eps$ as the solution, on a suitable compact domain $D$ in $M$ with spacelike boundary hypersurfaces which contains $K$, of the quasilinear tensorial wave equation
\begin{equation}
\label{EqIGaugeFixedEVE}
  \cP(g_\eps;g_{0,\eps}) := \Ric(g_\eps) - \Lambda g_\eps - \delta_{g_\eps}^*\Ups(g_\eps;g_{0,\eps}) = 0
\end{equation}
for which one imposes trivial Cauchy data for $h_\eps$ at the Cauchy hypersurface $X$.

If $g_{0,\eps}$ were a continuous family of smooth metrics on $M$ uniformly as $\eps\searrow 0$, this would be a straightforward small data existence result for a quasilinear wave equation; the smallness of the data here refers to the fact that $\cP(g_{0,\eps};g_{0,\eps})=\cO(\eps^\infty)$ is very small (in $H^k(D)$ for any fixed $k$).

However, $g_{0,\eps}$ is evidently singular near $x=0$ as $\eps\searrow 0$. Note moreover that one can rephrase~\eqref{EqIMetric} in regions $|x|\lesssim\eps$ by noting the rescaling property $(\hat g_{\bhm,\bha})_{\hat\mu\hat\nu}(x/\eps)=(\hat g_{\eps\bhm,\eps\bha})_{\hat\mu\hat\nu}(x)$: if near a point $(t_0,0)$ on $\cC$ we introduce rescaled coordinates
\begin{equation}
\label{EqIHatCoord}
  \hat t=\frac{t-t_0}{\eps},\qquad \hat x=\frac{x}{\eps},
\end{equation}
the metric coefficients $(g_\eps)_{\mu\nu}(t_0+\eps\hat t,\eps\hat x)$ of $g_\eps$ converge to those of $\hat g_b(\hat x)=\hat g_{\bhm,\bha}(\hat x)$. That is, constructing the required metric perturbations $h_\eps$ on a compact subset $K$ of $M$ requires controlling gravitational perturbations of size $\hat\rho^2=\eps^2|x|^2$ `perturbations' of the Kerr metric $\hat g_b$ for long $\hat t$-time scales, namely up to $\hat t\sim\eps^{-1}$.

In \cite{HintzGlueLocII}, we developed a framework for the uniform analysis of linear wave equations on \emph{glued spacetimes} such as $(M_\eps,g_{0,\eps})_{\eps\in(0,1)}$, and demonstrated in a scalar toy example in \citeII{\S\ref*{SssNToyT}} how one can use this framework to solve nonlinear wave equations on them over compact sets $K$ using a nonlinear (Nash--Moser) iteration scheme. Implementing an analogous strategy in the setting of the Einstein equations is considerably more difficult. For example, we mention already at this point that we expect uniform bounds~\eqref{EqIPrechBounds} for the solution $h_\eps$ of the specific version~\eqref{EqIGaugeFixedEVE} of the gauge-fixed Einstein equations to \emph{fail}, and in fact for $h_\eps$ to blow up after time $t\sim\eps$. (See the end of~\S\ref{SssIK1} below.) Thus, we will need to make a more careful choice of the gauge condition and the manner in which it is combined with the Ricci curvature operator in~\eqref{EqIGaugeFixedEVE} (cf.\ the operator $\delta_{g_\eps}^*$ and the idea of \emph{constraint damping} \cite{BrodbeckFrittelliHubnerReulaSCP,GundlachCalabreseHinderMartinConstraintDamping}).

In rough terms, the proof strategy is as follows.
\begin{enumerate}
\item\label{ItIStratNonlin}{\rm (Nonlinear iteration.)} Solve a suitable version of the gauge-fixed Einstein vacuum equations using an iteration scheme in which a linearized equation is solved at each step, with uniform (as $\eps\searrow 0$) estimates (on sufficiently small but fixed domains in $M$) in order to ensure convergence of the iteration and quantitative bounds for the solution for all sufficiently small $\eps>0$ at once.
\item\label{ItIStratLinI}{\rm (Linear analysis I: uniform estimates at low regularity.)} The linear analysis of wave-type operators $L_\eps$ such as the linearization of (a suitable modification of)~\eqref{EqIGaugeFixedEVE} around metrics like $g_\eps=g_{0,\eps}$ proceeds in two steps.
  \begin{enumerate}
  \item\label{ItIStratLinIReg}{\rm (Control of regularity.)} The first step is to control solutions of $L_\eps h=f$ at high frequencies on the appropriate scale of \emph{se-Sobolev spaces} (measuring regularity with respect to $(\eps^2+|x|^2)^{\frac12}\pa_t$ and $(\eps^2+|x|^2)^{\frac12}\pa_x$); this was already accomplished in great generality, and independently of any spectral assumptions on $L_\eps$, in \cite{HintzGlueLocII}. See~\S\ref{SssII2}.
  \item{\rm (Control near the small Kerr black hole.)} Controlling solutions of a suitable version $L$ of the linearized gauge-fixed Einstein equations on a fixed subextremal Kerr black hole allows one to improve the control from~\eqref{ItIStratLinIReg} to uniform control as $\eps\searrow 0$ on sufficiently small domains (on which the Kerr model is sufficiently accurate). Concretely, we need $L$ to have suitable spectral properties (mode stability, second order pole at zero frequency), which we arrange via (perturbative) modifications of the gauge condition and of the operator $\delta_{g_\eps}^*$ in~\eqref{EqIGaugeFixedEVE}. See~\S\S\ref{SssII2} and \ref{SssIKSpec}.
  \item\label{ItIStratLinIInter}{\rm (Caveat: interactions.)} The interaction of growing (in time) contributions of solutions of the Kerr model equation (encoding modulations of the black hole parameters, center of mass, and boost parameters) with the background curvature (cf.\ the $\cO(\hat\rho^2)$ term in~\eqref{EqIFKerr}) creates terms with poor spatial decay from the perspective of the small black hole. This forces us to invert the Kerr model operator $L$ on spaces with large spatial weights to absorb such terms. See~\S\ref{SssIKInter}.
  \end{enumerate}
\item\label{ItIStratTame}{\rm (Linear analysis II: higher regularity; tame estimates.)} Given uniform control on solutions of $L_\eps h=f$ on se-Sobolev spaces, we next need to control $h$ in spaces encoding \emph{s-regularity}, i.e.\ uniform $L^2$-type bounds upon application of a large number $k$ of $\pa_t$, $(\eps^2+|x|^2)^{\frac12}$. Moreover, we need estimates which are tame in the s-regularity order $k$. This is done exactly as in \citeII{\S\ref*{SN}}.
\item\label{ItIStratNM}{\rm (Implementing the nonlinear iteration.)} To close the iteration scheme, we need to have a rather precise description of $h$ (solving $L_\eps h=f$) which in particular captures the modulations of parameters mentioned in~\eqref{ItIStratLinIInter} above. Since this description induces a mild loss of control of the support of various pieces of $h$, we need to use a tailor-made version of Nash--Moser to close the argument. See~\S\ref{SssIHigh}.
\end{enumerate}

\begin{rmk}[Linear stability of subextremal Kerr]
\label{RmkILinStab}
  The linear stability of Kerr has thus far only been proved unconditionally in the slowly rotating case \cite{AnderssonBackdahlBlueMaKerr,HaefnerHintzVasyKerr}; see \cite{DafermosHolzegelRodnianskiSchwarzschildStability,HungKellerWangSchwarzschild,HungSchwarzschildOdd,HungSchwarzschildEven,JohnsonSchwarzschild} for the case of Schwarzschild, and \cite{DafermosHolzegelRodnianskiTeukolsky,MaZhangTeukolsky,MilletTeukolskyDecay,ShlapentokhRothmanTeixeiradCTeukolskyI,ShlapentokhRothmanTeixeiradCTeukolskyII} for results on the Teukolsky equation (including in the full subextremal range). The estimates for the linearized gauge-fixed Einstein equations $L h=f$ for the subextremal Kerr black hole $\hat g_{\bhm,\bha}$ that we need here are qualitatively different, and in any case do not, by themselves, imply linear stability. To wit, we need to control solutions sourced by forcing terms $f$ which are $L^2$ in $\hat t$; the solutions $h$ then turn out to be only in $\hat t^2 L^2$ for $\hat t>1$. (See~\eqref{EqILFwd} for a more precise statement.) When the forcing is better, e.g.\ compactly supported in time with suitable spatial decay---the setting relevant for linear stability---we do not prove stronger asymptotics and decay results here, though we expect that the methods of \cite{HaefnerHintzVasyKerr} can be easily adapted to yield such results. See Remark~\ref{RmkKEFwdLinStab} for further discussion.
\end{rmk}

We proceed to describe the steps of this strategy in more detail.

%%%%%%%%%%%%%%%%%%%%%%%%%%%%%%
\subsubsection{Estimates on se-Sobolev spaces; the role of the Kerr model}
\label{SssII2}

This part is an abbreviated discussion of the framework introduced in \cite{HintzGlueLocII}. For the sake of concreteness, our discussion from now on will take place on a domain of the form
\begin{equation}
\label{EqII2Dom}
  \Omega = \{ (t,x) \colon 0\leq t\leq\lambda,\ |x|\leq \lambda+2(\lambda-t) \},\qquad
  \Omega_\eps = \Omega \cap \{ |x|\geq\eps\bhm \},
\end{equation}
where the parameter $\lambda>0$ controls the size of the domain. (It will eventually be taken to be sufficiently small, but of course independent of $\eps$.) Vector fields interpolating between $\pa_{\hat t}=\eps\pa_t$, $\pa_{\hat x}=\eps\pa_x$ (see~\eqref{EqIHatCoord}) in the near field ($|x|\lesssim\eps$) and $\pa_t$, $\pa_x$ in the far field ($|x|\gtrsim 1$) regime are the \emph{se-vector fields}
\begin{equation}
\label{EqII2se}
  \hat\rho\pa_t,\quad \hat\rho\pa_x,\qquad \hat\rho := (\eps^2+|x|^2)^{\frac12}.
\end{equation}
Setting $\rho_\circ:=\la\hat x\ra^{-1}=\frac{\eps}{\hat\rho}$, we then work on \emph{weighted se-Sobolev spaces} $H_{\seop,\eps}^{s,\alpha_\circ,\hat\alpha}(\Omega_\eps)$ which are equal to $H^s(\Omega_\eps)$ as vector spaces, but whose norm depends on $\eps$ via
\begin{equation}
\label{EqII2seNorm}
  \|u\|_{H_{\seop,\eps}^{s,\alpha_\circ,\hat\alpha}(\Omega_\eps)} = \sum_{j+|\beta|\leq s} \| \rho_\circ^{-\alpha_\circ}\hat\rho^{-\hat\alpha} (\hat\rho\pa_t)^j (\hat\rho\pa_x)^\beta u \|_{L^2(|\dd t\,\dd x|)}.
\end{equation}
(The analysis in fact requires the usage of variable differential order functions $\sfs$, not merely constant integer $s$, though we will ignore this fact in this introduction.) We work with the subspace
\[
  H_{\seop,\eps}^{s,\alpha_\circ,\hat\alpha}(\Omega_\eps)^{\bullet,-}
\]
of functions which vanish in the past $t\leq 0$ of the initial boundary hypersurface of $\Omega_\eps$ (at $t=0$) (hence the `$\bullet$') but are unrestricted near the final boundary hypersurfaces of $\Omega$ (at $t=\lambda$, $|x|=\lambda+2(\lambda-t)$, and $|x|=\eps\bhm$) (hence the `$-$'); this is the natural functional analytic setting for the study of forward problems for wave equations on $(\Omega_\eps,g_\eps)$. \emph{Here and below we write $g_\eps$ for a family of metrics which has the structure~\eqref{EqIFKerr}.}

Corresponding to the vector fields~\eqref{EqII2se}, we have spaces $\Diffse^m$ of se-differential operators (and weighted versions, with weights $\rho_\circ^{-\ell_\circ}\hat\rho^{-\hat\ell}$). On $\Omega_\eps$ and recalling~\eqref{EqIFCoord}, such operators take the form $\sum_{j+|\beta|\leq m} a_{j\beta}(t,\hat\rho,\rho_\circ,\omega)(\hat\rho\pa_t)^j(\hat\rho\pa_x)^\beta$ where $a_{j\beta}$ is a smooth function of its arguments.\footnote{It is important to allow for less regular coefficients, but we shall ignore all those issues related to smoothness which were already discussed in \citeII{\S\ref*{SI}}.} Such operators are thus \emph{families}, parameterized by $\eps>0$, of differential operators on $M$. Examples of elements of $\hat\rho^{-2}\Diffse^2$ are families of (tensorial) wave operators $\Box_{g_\eps}$ acting on sections of tensor bundles over $M$, and, most importantly for present purposes, the linearization around $g_\eps$ of the Einstein vacuum equations in generalized harmonic gauge
\begin{equation}
\label{EqII2LMem}
  L_\eps := D_{g_\eps}\Ric - \Lambda g_\eps + \delta_{g_\eps}^*\delta_{g_\eps}\sfG_{g_\eps} = \frac12\Box_{g_\eps} + {\rm l.o.t.} \in \hat\rho^{-2}\Diffse^2,
\end{equation}
where $(\delta^*\omega)_{\mu\nu}=\frac12(\omega_{\mu;\nu}+\omega_{\nu;\mu})$ is the symmetric gradient, $(\delta h)_\mu=-h_{\mu\nu}{}^{;\nu}$ the divergence, and $\sfG_g h=h-\frac12(\tr_g h)g$ the trace-reversal operator. (As mentioned before, we will later on need to work with a slight modification of this operator.)

\medskip

\textbf{Uniform regularity estimates.} Based entirely on uniform (as $\eps\searrow 0$) properties of the null-geodesic flow on $(\Omega_\eps,g_\eps)$ and assuming a bound on the subprincipal symbol near the trapped set, we proved in \citeII{Theorem~\ref*{ThmEstStd}} a uniform estimate
\begin{equation}
\label{EqII2Est}
  \|h\|_{H_{\seop,\eps}^{\sfs,\alpha_\circ,\hat\alpha}(\Omega_\eps)^{\bullet,-}} \leq C\Bigl( \|L_\eps h\|_{H_{\seop,\eps}^{\sfs,\alpha_\circ,\hat\alpha-2}(\Omega_\eps)^{\bullet,-}} + \|h\|_{H_{\seop,\eps}^{\sfs_0,\alpha_\circ,\hat\alpha}(\Omega_\eps)^{\bullet,-}}\Bigr)
\end{equation}
where the orders $\alpha_\circ,\hat\alpha$ are arbitrary but fixed, and we can take $\sfs_0<\sfs-100$ (with $\sfs_0,\sfs$ subject to certain threshold conditions relative to $\alpha_\cD=\alpha_\circ-\hat\alpha$); the loss of one derivative relative to standard hyperbolic estimates is due to trapping. The verification of the subprincipal symbol condition for tensorial wave operators for the linearized gauge-fixed Einstein operator uses bounds on the parallel transport along trapped null-geodesics from \citeII{\S\ref*{SssGlDynTr}}, \cite{MarckParallelNull} following \cite{HintzPsdoInner}.

\medskip

\textbf{The role of the Kerr model.} Since the norm on $h$ on the right in~\eqref{EqII2Est} is not yet small compared to the left hand side, the estimate~\eqref{EqII2Est} does not yet imply uniform bounds on $h$ as $\eps\searrow 0$. We thus need to analyze the model operator at $\hat\rho=0$, which for the operator $\Box_{g_\eps}$ would be the tensor wave operator $\Box_{\hat g_b}$ on the limiting Kerr spacetime, and which for $L_\eps$ is the linearized Einstein vacuum operator
\begin{equation}
\label{EqII2LEin}
  L := D_{\hat g_b}\Ric + \delta_{\hat g_b}^*\delta_{\hat g_b}\sfG_{\hat g_b}
\end{equation}
on Kerr. This is a model for the family $L_\eps$ in the sense that (under suitable identifications of the symmetric 2-tensor bundles) the difference $(L_\eps-\eps^{-2}L)_{\eps\in(0,1)}$ not only lies in $\hat\rho^{-2}\Diffse^2$, but in fact has extra orders of vanishing at $\hat\rho=0$---specifically, \emph{two} orders of vanishing due to~\eqref{EqIFKerr}:
\begin{equation}
\label{EqII2Diff}
  (L_\eps - \eps^{-2}L)_{\eps\in(0,1)} \in \hat\rho^2\cdot\hat\rho^{-2}\Diffse^2 = \Diffse^2.
\end{equation}
Furthermore, if we consider the difference on a small domain~\eqref{EqII2Dom}, i.e.\ with $\lambda>0$ small, then the difference~\eqref{EqII2Diff} has \emph{small coefficients} (namely, of size $\lambda^2$).

In the coordinates $\hat t,\hat x$ appropriate for the description of the Kerr model, we need to study forward problems for $L$ on the domain~\eqref{EqII2Dom} which now takes the form
\begin{equation}
\label{EqII2Domain}
  \bigl\{ (\hat t,\hat x) \colon 0\leq\hat t\leq\lambda\eps^{-1},\ \bhm\leq|\hat x|\leq\lambda\eps^{-1}+2(\eps^{-1}\lambda-\hat t) \bigr\},
\end{equation}
and thus really globally in $\hat t\geq 0$ on $(\hat M_b^\circ,\hat g_b)$. On the level of function spaces, the vector fields~\eqref{EqII2se} are equal to
\begin{equation}
\label{EqII23b}
  \la\hat x\ra\pa_{\hat t},\qquad \la\hat x\ra\pa_{\hat x}.
\end{equation}
On a suitable compactification $\hat M_b$ of $\hat M_b^\circ$, these are the basic \emph{3b-vector fields} \cite{Hintz3b,HintzNonstat}, and thus weighted se-Sobolev norms are equivalent to weighted 3b-Sobolev norms
\[
  \|u\|_{H_\tbop^{s,\alpha_\cD,0}} := \sum_{j+|\beta|\leq s} \|\rho_\cD^{-\alpha_\cD}(\la\hat x\ra\pa_{\hat t})^j(\la\hat x\ra\pa_{\hat x})^\beta u\|_{L^2(|\dd\hat t\,\dd\hat x|)},\qquad \rho_\cD:=|\hat x|^{-1}.
\]
Concretely, we have
\begin{equation}
\label{EqII2se3b}
  \|u\|_{H_{\seop,\eps}^{s,\alpha_\circ,\hat\alpha}} = \eps^{-\hat\alpha+2}\|u\|_{H_\tbop^{s,\alpha_\cD,0}},\qquad \alpha_\cD=\alpha_\circ-\hat\alpha.
\end{equation}

Suppose now we could prove an estimate
\begin{equation}
\label{EqII2Kerr}
  \|h\|_{H_\tbop^{\sfs_0,\alpha_\cD,0}} \leq C\|L h\|_{H_\tbop^{\sfs_0,\alpha_\cD+2,0}}
\end{equation}
for symmetric 2-tensors $h$ on $\hat M_b^\circ$ which vanish in $\hat t<0$. Then using~\eqref{EqII2se3b}, we can estimate the error term in~\eqref{EqII2Est} by $\|L h\|_{H_{\seop,\eps}^{\sfs_0,\alpha_\circ+2,\hat\alpha}}$ which is in turn bounded by $\|L_\eps h\|_{H_{\seop,\eps}^{\sfs_0,\alpha_\circ,\hat\alpha-2}} + \|(L_\eps-\eps^{-2}L)h\|_{H_{\seop,\eps}^{\sfs_0,\alpha_\circ,\hat\alpha-2}}$. But when $L_\eps-\eps^{-2}L$ is small as an element of $\hat\rho^{-2}\Diffse^2$ (which is true when $\lambda$ is sufficiently small), the second term here is bounded by a small constant times $\|h\|_{H_{\seop,\eps}^{\sfs_0+2,\alpha_\circ,\hat\alpha}}$ and can thus be absorbed into the left hand side of~\eqref{EqII2Est}; this would give the \emph{uniform} estimate
\begin{equation}
\label{EqII2EstUnif}
  \|h\|_{H_{\seop,\eps}^{\sfs,\alpha_\circ,\hat\alpha}(\Omega_\eps)^{\bullet,-}} \leq C\|L_\eps h\|_{H_{\seop,\eps}^{\sfs,\alpha_\circ,\hat\alpha-2}(\Omega_\eps)^{\bullet,-}}.
\end{equation}
In \citeII{\S\ref*{SsScUnif}}, we proved~\eqref{EqII2Kerr} for $\alpha_\cD\in(-\frac32,-\frac12)$ (roughly corresponding to better than pointwise $\hat r^{-2}$ decay of the source term $f=L h$) and thus~\eqref{EqII2EstUnif} for $\alpha_\cD=\alpha_\circ-\hat\alpha$ in the case that $L_\eps=\Box_{g_\eps}$ is the scalar wave operator.

\medskip

\textbf{Estimates for the Kerr model: linearized gauge-fixed Einstein equation.} Unfortunately, the estimate~\eqref{EqII2Kerr} \emph{fails} when $L_\eps$ is the linearized gauge-fixed Einstein operator~\eqref{EqII2LEin} due to the existence of zero energy bound states (discussed below). What we \emph{will} be able to prove is the estimate
\begin{equation}
\label{EqII2KerrCorrect}
  \|h\|_{\Htb^{\sfs_0,\alpha_\cD,-2}} = \Bigl\| \Bigl(\frac{\hat r}{1+\hat t+\hat r}\Bigr)^2 h\|_{\Htb^{\sfs_0,\alpha_\cD,0}} \leq C\|L h\|_{\Htb^{\sfs_0,\alpha_\cD+2,0}}
\end{equation}
for a suitable modification of~\eqref{EqII2LEin}, for $\alpha_\cD\in(-2,-\frac32)$. (For bounded $\hat r$, this amounts to a loss of two orders of $\hat t$-decay. For $\hat r\sim\hat t$, the additional factor is $\sim 1$, so there is no further loss.) We discuss this estimate, as well as the reasons for the weaker weight (which allows for $L h$ to have pointwise $\hat r^{-q}$ decay for $q\in(\frac32,2)$), in~\S\S\ref{SssIK1}--\ref{SssIKInter}. Granted~\eqref{EqII2KerrCorrect}, we can then obtain a uniform estimate
\begin{equation}
\label{EqII2KerrUnif}
  \|h\|_{H_{\seop,\eps}^{\sfs,\alpha_\circ,\hat\alpha-2}(\Omega_\eps)^{\bullet,-}} \leq C\|L_\eps h\|_{H_{\seop,\eps}^{\sfs,\alpha_\circ,\hat\alpha-2}(\Omega_\eps)^{\bullet,-}},\qquad \alpha_\circ-\hat\alpha\in(-2,-\tfrac32),
\end{equation}
by first using~\eqref{EqII2Est} with $\hat\alpha-2$ in place of $\hat\alpha$, then relating $\|h\|_{H_{\seop,\eps}^{\sfs_0,\alpha_\circ,\hat\alpha-2}}=\eps^{-\hat\alpha+2}\|h\|_{H_\tbop^{\sfs_0,\alpha_\cD,-2}}$, then estimating this using~\eqref{EqII2KerrCorrect}, and finally replacing $L$ by $\eps^2 L_\eps$ and using that the norm of the difference,
\begin{equation}
\label{EqII2Err}
  \|(L_\eps-\eps^{-2}L)h\|_{H_{\seop,\eps}^{\sfs_0,\alpha_\circ,\hat\alpha-2}},
\end{equation}
is bounded by a \emph{small} constant times $\|h\|_{H_{\seop,\eps}^{\sfs_0,\alpha_\circ,\hat\alpha-2}}$ in view of~\eqref{EqII2Diff} when the quantity $\lambda$ controlling the size of the domain on which we work is sufficiently small; this can be absorbed as before, yielding~\eqref{EqII2KerrUnif}. (See Theorem~\ref{ThmEse}.) Thus, the loss of two orders of time decay in~\eqref{EqII2KerrCorrect} is---just barely---made up for by the fact~\eqref{EqII2Diff} that $L_\eps$ differs from the Kerr model by error terms vanishing to second order at $\hat\rho=0$ (which traces back to the corresponding statement for the glued metrics~\eqref{EqIFKerr} which, as we have already pointed out, is sharp).

The estimate~\eqref{EqII2Kerr} (for $L=\Box_{\hat g_b}$) was proved, and~\eqref{EqII2KerrUnif} (for $L$ equal to a suitable modification of~\eqref{EqII2LEin}) will be proved in this paper, by passing to the Fourier transform in $\hat t$; this gives rise to the spectral family
\[
  \hat L(\sigma),\qquad \sigma\in\C.
\]
The crucial ingredient for~\eqref{EqII2Kerr} with $L=\Box_{\hat g_b}$ is mode stability, i.e.\ the invertibility of $\hat L(\sigma)$ for all $\Im\sigma\geq 0$ on suitable function spaces. (This used \cite{ShlapentokhRothmanModeStability}.) In the case of~\eqref{EqII2LEin}, mode stability fails precisely at $\sigma=0$, as shown in \cite{AnderssonHaefnerWhitingMode}. The relevant function spaces for the analysis of $\hat L(\sigma)$ are precisely those which arise when describing $\Htb^{\sfs_0,\alpha_\cD,0}$ using the Fourier transform in $\hat t$ and Plancherel's theorem; to wit, these are scattering-b-transition Sobolev spaces \cite{GuillarmouHassellResI,HintzKdSMS} near $\sigma=0$, scattering Sobolev spaces \cite{MelroseEuclideanSpectralTheory} for bounded nonzero $\sigma$, and semiclassical scattering Sobolev spaces \cite{VasyZworskiScl} when $|\sigma|\to\infty$. For details, see \cite[Proposition~4.24]{Hintz3b}.

%%%%%%%%%%%%%%%%%%%%%%%%%%%%%%
\subsubsection{Linearized Einstein equations on Kerr, I: singularity of the resolvent}
\label{SssIK1}

We first discuss the spectral theory of $L$ as defined in~\eqref{EqII2LEin}; specifically, we discuss the resolvent family $\hat L(\sigma)^{-1}$ acting on elements of function spaces with the `usual' weights, which near zero energy roughly amounts to better than inverse quadratic decay.

The key difference between $L$ and the scalar wave operator $\Box_{\hat g_b}$ is that the former has zero energy bound states. These are already discussed in great detail in the slowly rotating case in \cite{HaefnerHintzVasyKerr}. Since they play an important role also in the present paper, we recall that there are two types of these.
\begin{enumerate}
\item{\rm (Pure gauge perturbations.)} Since $\Ric(\hat g_b)=0$, we have $D_{\hat g_b}\Ric(\cL_V\hat g_b)=0$ for all vector fields $V$ on $\hat M_b^\circ$. If moreover $\cL_V\hat g_b=o(1)$ as $\hat r\to\infty$, and $\cL_V\hat g_b$ satisfies the presently chosen gauge condition $\delta_{\hat g_b}\sfG_{\hat g_b}(\cL_V\hat g_b)=0$ (which is a wave equation for $V$), then also $\cL_V\hat g_b\in\ker\hat L(0)$. Among the vector fields $V$ for which all of these properties hold, there are asymptotic translations ($V=\pa_{\hat x^i}+{\rm l.o.t.}$) and asymptotic rotations ($V=\hat x^i\pa_{\hat x^j}-\hat x^j\pa_{\hat x^i}+{\rm l.o.t.}$).
\item{\rm (Linearized Kerr perturbations.)} Linearizations $\hat g_b'(\dot b)=\frac{\dd}{\dd s}\hat g_{b+s\dot b}|_{s=0}$ of the Kerr metric in the Kerr parameters solve the linearized Einstein vacuum equations
\[
  D_{\hat g_b}\Ric(\hat g_b'(\dot b))=0.
\]
One can add to $\hat g_b'(\dot b)$ a pure gauge term $\cL_{V(\dot b)}\hat g_b$ for a suitable vector field $V(\dot b)$ to obtain $\hat g_b^{\prime\Ups}(\dot b)=\cO(\hat r^{-1})$ which now satisfy the gauge condition, and thus $\hat g_b^{\prime\Ups}(\dot b)\in\hat L(0)$.
\end{enumerate}

One can show that $\hat L(0)$ is Fredholm of index $0$ between suitable weighted (b-)Sobolev spaces, with the weight of the domain relative to $L^2(|\dd\hat x|)$ lying in the interval $(-\frac32,-\frac12)$ (with membership in this space being implied by $\hat r^{-q}$ pointwise decay, $q\in(0,1)$); see Proposition~\ref{PropKE0} (and also \cite[Theorem~4.5]{GellRedmanHaberVasyFeynman}). The presence of a nontrivial nullspace thus implies that $\hat L(\sigma)^{-1}$ must be unbounded (say, acting on $\CIc$, with the output measured in $L^2$ on a set of positive measure) as $\sigma\to 0$.

More precisely, elements of $\hat L(0)$ are stationary solutions of $L h=0$ with appropriate spatial decay. There also exist exactly linear-in-$\hat t$ solutions of $L h=0$ with sufficient spatial decay: these are pure gauge solutions $\cL_V\hat g_b$ where $V=\hat t\pa_{\hat x^i}+\hat x^i\pa_{\hat t}+{\rm l.o.t.}$ is, asymptotically as $\hat r\to\infty$, a Lorentz boost. (For the construction of this and the earlier pure gauge solutions for a slightly modified gauge, see Lemma~\ref{LemmaKGaugePot}.) For the purposes of the present discussion, let us say that $\hat L(\sigma)^{-1}$ has a pole at $\sigma=0$ of order $D$ if
\[
  L\Biggl(\sum_{j=0}^d \hat t^j h_j\Biggr),\qquad h_0,\ldots,h_d=o(1)\ \text{as}\ \hat r\to\infty,\qquad h_d\neq 0 \implies d\leq D-1,
\]
and if moreover such a `generalized zero energy state' with $d=D-1$ and $h_{D-1}\neq 0$ does exist. One should think of this as being closely related to $\|\hat L(\sigma)^{-1}\|\sim|\sigma|^{-D}$ near $\sigma=0$, and in turn to the loss of $D$ powers of $\hat t$-decay for forward solutions of $L h=f$. Thus, $\hat L(\sigma)^{-1}$ has a pole at $\sigma=0$ of order $\geq 2$. In fact, \emph{in the Schwarzschild case $b=(\bhm,0)$, it has a pole of order $\geq 3$}, as observed around \cite[equation~(9.29)]{HaefnerHintzVasyKerr}.\footnote{We do not need this statement here, nor did we need it in \cite{HaefnerHintzVasyKerr}; we recall this only to motivate the need for modifications to $L$ below.}

To appreciate the subtleties caused by the presence of poles at zero energy, consider instead of $L_\eps$ and $L$ an ODE example,
\[
  L_\eps = \pa_t + \eps^{-1+\gamma}a(t),\qquad L=\pa_{\hat t},
\]
where $a\in\CI(\R)$, $\gamma>0$, and $\hat t=\frac{t}{\eps}$. Thus, the family $(L_\eps)_{\eps\in(0,1)}=(\eps^{-1}(\eps\pa_{\hat t}+\eps^\gamma a))_{\eps\in(0,1)}$ is an element of $\eps^{-1}\Diffse^1$, while the difference $L_\eps-\eps^{-1}L=\cO(\eps^{-1+\gamma})$ is $\gamma$ orders better. We can solve $L_\eps u=0$ with initial data\footnote{The considerations for the forcing problem $L_\eps u=f$, with $f$ and $u$ vanishing for $\hat t<0$, are completely analogous.} imposed at $t=0$ explicitly by $u(t)=\exp(-\eps^{-1+\gamma}\int_0^t a(s)\,\dd s)u(0)$. For $\gamma<1$, this has exponential growth as $\eps\searrow 0$ (unless $a$ has a sign, which is a condition unlikely to be sensible in the tensorial setting of the Einstein equations). We thus see that uniform estimates for $L_\eps$ can only hold if $L_\eps$ agrees with the model $\eps^{-1}L$ up to terms with at least \emph{one} order of decay. Note that $\hat L(\sigma)=-i\sigma$, so $\hat L(\sigma)^{-1}$ has a pole of order $1$, and indeed $L$ annihilates constants. --- Similarly, we can consider
\[
  L_\eps = \pa_t^2 + \eps^{-2+\gamma}a(t),\qquad L=\pa_{\hat t}^2.
\]
Then $\hat L(\sigma)^{-1}$ has a pole of order $2$, with $L$ annihilating linear functions in $\hat t$; and solutions of initial value problems for $L_\eps u=0$ grow exponentially as $\eps\searrow 0$ unless $\gamma\geq 2$.

Returning to the linearized gauge-fixed Einstein setting, a third order pole of $\hat L(\sigma)^{-1}$ coupled with the only second order accuracy of the Kerr model $L$ in~\eqref{EqII2Diff} means that uniform estimates for forward solutions of $L_\eps h=f$ of the sort~\eqref{EqII2KerrUnif} are (likely) false. The remedy, as already noted in \cite{HaefnerHintzVasyKerr} (see in particular \cite[Lemma~9.8]{HaefnerHintzVasyKerr}) is to modify $L$ so as to implement constraint damping; this will reduce the pole order to $2$.

%%%%%%%%%%%%%%%%%%%%%%%%%%%%%%
\subsubsection{Constraint damping; gauge modification; a first estimate for Kerr} 
\label{SssIKSpec}

Consider, as a generalization of~\eqref{EqII2LEin}, the operator
\begin{equation}
\label{EqIKSpecOp}
  L := D_{\hat g_b}\Ric + \tilde\delta_{\hat g_b}^*\breve\delta_{\hat g_b}\sfG_{\hat g_b},
\end{equation}
where $\tilde\delta_{\hat g_b}^*-\delta_{\hat g_b}^*$ and $\breve\delta_{\hat g_b}-\delta_{\hat g_b}$ are of order zero (i.e.\ vector bundle maps), stationary (i.e.\ they commute with $\hat t$-translations), and of compact spatial support. Thus, the new gauge condition for linearized perturbations $h$ is $\breve\delta_{\hat g_b}\sfG_{\hat g_b}h=0$.

The modification $\tilde\delta_{\hat g_b}^*$ does not affect solutions of $L h=0$ \emph{only} when the initial data satisfy the linearized constraints and gauge condition (since in this case $D_{\hat g_b}\Ric(h)=0$ and $\breve\delta_{\hat g_b}\sfG_{\hat g_b}h=0$ by the usual linearized second Bianchi identity argument). For general initial data however, $L h=0$ implies
\[
  \delta_{\hat g_b}\sfG_{\hat g_b}\tilde\delta_{\hat g_b}^*\omega=0
\]
where $\omega=\breve\delta_{\hat g_b}\sfG_{\hat g_b}h$ is the linearized gauge 1-form; this is a wave equation for $\omega$. When $\tilde\delta_{\hat g_b}^*=\delta_{\hat g_b}^*$, it is in fact the Hodge--d'Alembertian, which admits nontrivial stationary (Coulomb) solutions. We shall show in Proposition~\ref{PropKCD} that one can devise a (small) modification $\tilde\delta_{\hat g_b}^*-\delta_{\hat g_b}^*$ so that \emph{mode stability holds for the wave-type operator $\delta_{\hat g_b}\sfG_{\hat g_b}\tilde\delta_{\hat g_b}^*$} (which amounts to perturbing away the zero energy bound state); this is the essence of \emph{constraint damping}.\footnote{More generally, if $L h=f$, then $\delta_{\hat g_b}\sfG_{\hat g_b}\tilde\delta_{\hat g_b}^*\omega=\breve\delta_{\hat g_b}\sfG_{\hat g_b}f$, so if mode stability holds, one expects $\omega$ and thus $h$ itself to have stronger decay properties than if it does not hold; note that $L h=f$ is the type of equation we need to understand for the purposes of eventually running a nonlinear iteration scheme for the gluing problem.}

The benefit of modifying the gauge condition is the following. Consider again the problem of correcting the linearized Kerr metric $\hat g_b'(\dot b)$ by a stationary pure gauge solution $\cL_V\hat g_b=\delta_{\hat g_b}^*\omega$, $\omega=2 V^\flat$, to a metric perturbation which also satisfies the new gauge condition; this amounts to solving
\begin{equation}
\label{EqIKSpecGC}
  \breve\delta_{\hat g_b}\sfG_{\hat g_b}\delta_{\hat g_b}^*\omega = -\breve\delta_{\hat g_b}\sfG_{\hat g_b}\hat g_b'(\dot b)
\end{equation}
for a stationary 1-form $\omega$. We show in Proposition~\ref{PropKUps} how to devise a perturbation $\breve\delta_{\hat g_b}-\delta_{\hat g_b}$ so that the wave-type operator $\breve\delta_{\hat g_b}\sfG_{\hat g_b}\delta_{\hat g_b}^*$ appearing here satisfies mode stability, including at zero energy, and thus we can solve~\eqref{EqIKSpecGC}. (Without perturbing, mode stability fails at zero energy, and one can in general only solve this equation with $\omega$ being linear in $\hat t$; see \cite[Proposition~9.4]{HaefnerHintzVasyKerr}.)

Having fixed $\tilde\delta_{\hat g_b}^*$ and $\breve\delta_{\hat g_b}$ in this fashion, we then prove mode stability for $\hat L(\sigma)$ for $\sigma\in\C$, $\Im\sigma\geq 0$, $\sigma\neq 0$ (see Proposition~\ref{PropKENon0}) and describe the zero energy nullspace (Proposition~\ref{PropKE0}). We then show that $\hat L(\sigma)^{-1}$ has a second order pole at $\sigma=0$ by using the non-degenerate nature of certain pairings (Lemma~\ref{LemmaKEL2}), which we verify via reduction to the gauge-free computation in \cite[Theorem~9.6]{HintzGlueLocI}. (See Proposition~\ref{PropKELo}).

\begin{rmk}[Comparison with \cite{HaefnerHintzVasyKerr}]
\label{RmkIKSpecComp}
  Firstly, we succeed in implementing constraint damping in the full subextremal range; the key computation in the perturbative argument is the non-vanishing of a certain pairing, which we reduce here to a simple boundary pairing computation (Lemma~\ref{LemmaKCDPair}). We also take advantage of recent progress on the spectral analysis at low energies, specifically the scattering-b-transition perspective \cite{Hintz3b,HintzKdSMS}, to considerably streamline the perturbative argument. Secondly, we implement the gauge modification (which was hinted at in \cite[Remark~10.14]{HaefnerHintzVasyKerr}) in the full subextremal range. Thirdly, we study modes of the linearized gauge-fixed Einstein operator $L$ in~\eqref{EqIKSpecOp} directly, rather than perturbatively off of \eqref{EqII2LEin}.
\end{rmk}

We briefly comment on the way in which we control $\hat L(\sigma)^{-1}$ for low energies on the (scattering-b-transition) Sobolev spaces required for compatibility with 3b- and thus se-estimates. As a toy model, suppose $A(\sigma)$ is a holomorphic $N\times N$-matrix valued function, with $\ker A(0)=\mathspan\{h_0\}$ and $\ker A(0)^*=\mathspan\{h_0^*\}$. Under the non-degeneracy assumption
\[
  \la \pa_\sigma A(0)h_0,h_0^* \ra \neq 0,
\]
the inverse $A(\sigma)^{-1}$ then has a first order pole at $\sigma=0$. (Of course $A(\sigma)^{-1}$ is then meromorphic, but the same is not true for resolvents $\hat L(\sigma)^{-1}$ on asymptotically flat spaces; see \cite{HintzPrice}, but also \cite{StuckerKerrQNM,GajicWarnickKerrQNM}.) One can see this easily as follows: pick any $f^*\in\C^N$ with $\la h_0,f^*\ra\neq 0$, and consider the amalgamated operator
\[
  \wt A(\sigma) := \begin{pmatrix} A(\sigma) & A(\sigma)(\sigma^{-1}h_0) \\ \la\cdot,f^*\ra & 0 \end{pmatrix} \colon \C^N\oplus\C \to \C^N\oplus\C
\]
(Thus, we set up a \emph{Grushin problem}; see \cite[Appendix~C.1]{DyatlovZworskiBook}.) Note that $\wt A(\sigma)$ extends holomorphically across $\sigma=0$; and the zero energy operator
\[
  \wt A(0) = \begin{pmatrix} A(0) & \pa_\sigma A(0)h_0 \\ \la\cdot,f^*\ra & 0 \end{pmatrix}
\]
is invertible. (This uses that $\ran A(0)=\ker(f\mapsto\la f,h_0^*\ra)$.) We can therefore describe $A(\sigma)^{-1}f$ be writing $\wt A(\sigma)^{-1}(f,0)=:(h,c)$, which gives
\[
  f = A(\sigma) ( h + \sigma^{-1}c h_0 ),
\]
together with uniform bounds (near $\sigma=0$) for $h\in\C^N$ and $c_0\in\C$ in terms of $f$.

The strategy for controlling $\hat L(\sigma)^{-1}$ is analogous, except now the construction of the first row of an amalgamated version $\wt L(\sigma)$ of $\hat L(\sigma)$ is considerably more delicate since we need to ensure that the relevant terms lie in function spaces adapted to the low energy resolvent analysis. (For example, the naive choice $\hat L(\sigma)(\sigma^{-1}\hat g_b^{\prime\Ups}(\dot b))$ is not acceptable since it is neither outgoing ($\sim e^{i\sigma\hat r}$) for nonzero $\sigma$, nor does it decay suitably in the regime $|\sigma|\sim\hat r^{-1}\to 0$ where we would need $o(\hat r^{-2})$ decay, but only get $\hat r^{-2}$ decay.) The construction involves in particular the solution of a PDE on the \emph{transition face} governing the transition from zero to nonzero energies (Lemmas~\ref{LemmaKEtf} and \ref{LemmaKEbreve}), which introduces further logarithmic singularities at $\sigma=0$.

A careful analysis of the inverse Fourier transform of the description of $\hat L(\sigma)^{-1}$ thus obtained produces a description of forward solutions of $L h=f$; see Theorem~\ref{ThmKEFwd} for the detailed statement. Roughly speaking, for $f\in\Htb^{\sfs_0,\alpha_\cD+2,0}$, $\alpha_\cD\in(-\frac32,-\frac12)$ (which roughly means that $f$ is $L^2$ in $\hat t$ with values in $\rho_\cD^{\alpha_\cD+2}L^2(|\dd\hat x|)$, $\rho_\cD=\hat r^{-1}$, consistent with $\hat r^{-2-q}$ pointwise decay where $q=\alpha_\cD+\frac32\in(0,1)$) with support in $\hat t\geq 0$, the forward solution $h$ of $L h=f$ is of the form
\begin{equation}
\label{EqILFwd}
  h = \hat g_b^{\prime\Ups}(\dot b(\hat t_1)) + \text{(further terms involving zero energy states)} + h_0,
\end{equation}
where we set $\hat t_1=\hat t-\hat r$;\footnote{The function spaces which we use allow one to absorb the usual logarithmic correction $2\bhm\log\hat r$ to $\hat t_1$ into $h_0$.} here $\dot b\in\hat t_1 L^2$ is linearly growing relative to $L^2$, the omitted terms are quadratically growing, and $h_0\in\Htb^{\sfs_0,\alpha_\cD,0}$ arises from the regular part of $\hat L(\sigma)$ (which was the only part present in the case of the scalar wave equation, cf.\ \eqref{EqII2Kerr}).

\begin{rmk}[Spacetime bounds for forward solutions]
\label{RmkILBounds}
  The zero energy state terms in~\eqref{EqILFwd} \emph{only} lie in $\Htb^{\sfs_0,\beta_\cD,-2}$ when $\beta_\cD<-\frac32$; in particular, the estimate~\eqref{EqII2KerrCorrect} does \emph{not} hold for $\alpha_\cD\in(-\frac32,-\frac12)$. (The relevant computation in a region $|\hat t|\leq C\hat r$, $C<\infty$, is this: given $\dot b\in\hat t L^2(|\dd\hat t|)$ and $g'=\cO(\hat r^{-1})$, we have
  \[
    \int_\bhm^\infty\int_{-C\hat r}^{C\hat r} |\hat r^{\beta_\cD}g'(\hat x)\dot b(\hat t)|^2 \dd\hat t\,\hat r^2\dd\hat r \lesssim \int_\bhm^\infty \hat r^{2\beta_\cD+2}\,\dd\hat r < \infty
  \]
  if and only if $\beta_\cD<-\frac32$.)
\end{rmk}

%%%%%%%%%%%%%%%%%%%%%%%%%%%%%%
\subsubsection{Linearized Einstein equations on Kerr, II: interactions with background curvature}
\label{SssIKInter}

The considerations in~\S\S\ref{SssIK1}--\ref{SssIKSpec} are relevant already when one is only interested in linearized gravity on subextremal Kerr spacetimes. The next subtlety we encounter is germane to the gluing problem. Namely, recall that the formal solution of the gluing problem is of the form
\[
  g_{0,\eps}(t,\hat x) = \hat g_b(\hat x) + \eps^2 \hat h_{(2)}(t,\hat x) + o(\eps^2),
\]
where $\hat h_{(2)}(t,\hat x)=\hat r^2\hat h_0(t,\hat x)$ with $h_0$ bounded as $\hat r=|\hat x|\to\infty$. (Indeed, the subleading term is then $\eps^2\hat r^2\hat h_{(2)}=r^2\hat h_0$ and captures the $\cO(r^2)$ terms of the original metric $g$ in Fermi normal coordinates around the geodesic $\cC$, which are exactly the terms encoding the Riemann curvature tensor of $g$ at $\cC$.) Dropping the $t$-dependence and gauge terms for brevity, we then compute the linearized (around $g_{0,\eps}$) Einstein vacuum operator applied to a metric perturbation $h$ to next-to-leading order in $\eps$ to be
\[
  D_{\hat g_b}\Ric(h) + \eps^2 D^2_{\hat g_b}\Ric(\hat h_{(2)},h) + o(\eps^2).
\]
The second term thus describes the gravitational interactions of $h$ (arising, say, from solving the equations of linearized gravity on Kerr) with the curvature of $(M,g)$.

Consider now a term $\hat g_b^{\prime\Ups}(\dot b(\hat t_1))$ from~\eqref{EqILFwd}, which is of class $\hat r^{-1}\hat t_1 L^2(|\dd\hat t_1|)$. We then have a size estimate
\[
  \eps^2 D^2_{\hat g_b}\Ric(\hat h_{(2)},h) \sim \eps^2 \hat r^{-2} \cdot \hat r^2 \cdot \hat r^{-1}\hat t_1 L^2(|\dd\hat t_1|),
\]
the factor $\hat r^{-2}$ coming from the scaling behavior $\Ric$, and $\hat r^2$ coming from $\hat h_{(2)}$. Since we are only aiming to solve the gluing problem in regions where $\eps\hat t_1=\eps(\hat t-\hat r)=t-r$ is bounded, we write this as $(\eps\hat t_1) \eps\hat r^{-1} L^2(|\dd\hat t_1|)\subset\eps\hat r^{-1}L^2(|\dd\hat t_1|)$. In an iteration scheme, one expects this type of term to arise as the source term for a linearized equation at the next step. Following the discussion leading up to~\eqref{EqILFwd}, one should therefore check whether, on a domain of the form~\eqref{EqII2Domain}, one has uniform bounds for the norm
\[
  \|\eps\hat r^{-1}a(\hat t_1)\|_{\rho_\cD^{\alpha_\cD+2}L^2(|\dd\hat t\,\dd\hat x|)}^2 \lesssim \int_\bhm^{\eps^{-1}} |\hat r^{\alpha_\cD+2}\eps\hat r^{-1}|^2\,\hat r^2\,\dd\hat r \sim \eps^{-2\alpha_\cD-3};
\]
this evidently only holds if $\alpha_\cD\leq-\frac32$. (This is intimately related to Remark~\ref{RmkILBounds}.)

In other words, \emph{the interaction of the growing part of~\eqref{EqILFwd}, arising from zero energy states of the linearized gauge-fixed Einstein equations on Kerr, with the term $\hat h_{(2)}$ arising from the background curvature is larger than what the standard spectral theory allows for as inputs.}

We are therefore forced to revisit the spectral theory for the operator $L$ in~\eqref{EqIKSpecOp}, and concretely need to study the resolvent family $\hat L(\sigma)^{-1}$ acting on spaces which allow for $\hat r^{\alpha_\cD+2}L^2$ behavior near zero energy where $\alpha_\cD<-\frac32$, roughly corresponding to pointwise $\hat r^{-q}$ decay where $q<2$. That this is quite delicate can be illustrated already in the case of the scalar wave equation on Minkowski space: the spectral family is then $\Delta-\sigma^2$ where $\Delta=\sum_{i=1}^3 D_{\hat x^i}^2$ is the (non-negative) Laplacian, and using the explicit formula for its resolvent, one can check that
\begin{equation}
\label{EqIKInterMink}
  (\Delta-\sigma^2)^{-1} \la\hat r\ra^{-q} \sim |\sigma|^{-2+q},\qquad q<2.
\end{equation}
We use here, say, the $L^2$-norm on a compact subset of $\R^3$. (For $q=2$, one has a logarithmic singularity; this is essentially the origin of the spectral approach to Price's law \cite{DonningerSchlagSofferPrice,TataruDecayAsympFlat,HintzPrice}. For $q>2$, the output of the resolvent remains bounded.) For solutions of the wave equation, this means that weakly decaying forcing causes growth in time. Importantly, the description~\eqref{EqIKInterMink} can be sharpened, in that the singular term comes with a coefficient (a function of $\hat x$) which is a \emph{large zero energy state}, i.e.\ an element in the kernel of the zero energy operator (here $\Delta$) which does not decay at infinity (here: a constant).

We are able to describe the resolvent $\hat L(\sigma)^{-1}$ acting on weakly decaying inputs, with $L^2$-weight $\alpha_\cD\in(-2,-\frac32)$, in an analogous fashion: in addition to the $\cO(|\sigma|^{-2})$ singular terms arising from generalized zero energy states (linearized Kerr, deformation tensors of asymptotic translations, boosts, and rotations), there is a $\cO(|\sigma|^{\alpha_\cD+\frac32})$ singular term, with spatial dependence given by large zero energy states: these are (pure gauge) elements of $\ker\hat L(0)$ with asymptotic behavior $\dd\hat z^\mu\otimes_s\dd\hat z^\nu$ (where $\hat z=(\hat t,\hat x)$). (The existence of these states is shown in Lemma~\ref{LemmaKGaugePot}.) The key step is to approximately invert the resolvent on large inputs by inverting a model operator in the transitional regime $|\sigma|\sim\hat r^{-1}\searrow 0$ (Proposition~\ref{PropKELtf}). The description of forward solutions of $L h=f$ for weakly decaying $f$ is then obtained via studying the inverse Fourier transform of $\hat L(\sigma)^{-1}\hat f(\sigma)$; see Theorem~\ref{ThmKEFwdW}. Lemma~\ref{LemmaKEFwdWSize} then, finally, establishes the desired bound~\eqref{EqII2KerrCorrect}.

%%%%%%%%%%%%%%%%%%%%%%%%%%%%%%
\subsubsection{Higher regularity; closing the nonlinear iteration}
\label{SssIHigh}

At this point, we have the uniform estimate~\eqref{EqII2KerrUnif} for forward solutions of $L_\eps h=f$ on suitable weighted se-Sobolev spaces. However, recall that the differential order $\sfs$ is constrained, as it needs to satisfy a fixed upper bound (due to radial point threshold conditions involving $\alpha_\cD=\alpha_\circ-\hat\alpha$). Thus, the estimate~\eqref{EqII2KerrUnif} is not sufficient for nonlinear applications.

We thus proceed as in \cite{HintzGlueLocII} and prove \emph{higher s-regularity}. We recall that the basic s-derivatives are
\[
  \pa_t,\quad \hat\rho\pa_x,
\]
i.e.\ they control regularity in the slow time variable (in other words, \emph{adiabatic} regularity near the small black hole) rather than the fast time variable (cf.\ $\hat\rho\pa_t$ in~\eqref{EqII2se}, which near the small black hole, where $|x|\lesssim\eps$, is $\sim\eps\pa_t=\pa_{\hat t}$). The arguments of \citeII{\S\ref*{SsScS}} apply with only minimal modifications, as do those in the proof of \citeII{Theorem~\ref*{ThmNTame}}: they give \emph{tame} estimates for forward solutions of $L_\eps h=f$ on s-Sobolev spaces $H_{\sop,\eps}^{s,\alpha_\circ,\hat\alpha}(\Omega_\eps)^{\bullet,-}$ (defined as in~\eqref{EqII2seNorm} but with $\pa_t$ in place of $\hat\rho\pa_t$), of the form
\begin{equation}
\label{EqIHigh}
  \|h\|_{H_{\sop,\eps}^{k,\alpha_\circ,\hat\alpha-2}(\Omega_\eps)^{\bullet,-}} \leq C_k \Bigl( \|f\|_{H_{\sop,\eps}^{k+d,\alpha_\circ,\hat\alpha-2}(\Omega_\eps)^{\bullet,-}} + \|L_\eps\|_{k+d}\|f\|_{H_{\sop,\eps}^{d,\alpha_\circ,\hat\alpha-2}(\Omega_\eps)^{\bullet,-}} \Bigr),\quad f=L_\eps h,
\end{equation}
where $d\in\N$ is a fixed loss-of-derivatives parameter, and $\|L_\eps\|_{k+d}$ is a certain norm controlling $k+d$ many s-derivatives of the coefficients of $L_\eps$ on $\Omega_\eps$. (See Proposition~\ref{PropEsTame} for details.) Here $k$ is now arbitrary, so this is a \emph{high regularity estimate}.

\medskip

We stress that due to the second order pole of $\hat L(\sigma)^{-1}$ at $\sigma=0$, the estimate~\eqref{EqIHigh} involves a loss of two orders of decay as $\hat\rho\to 0$ relative to the forward mapping property $L_\eps\colon H_{\sop,\eps}^{k,\alpha_\circ,\hat\alpha-2}\to H_{\sop,\eps}^{k-2,\alpha_\circ,\hat\alpha-4}$, cf.\ \eqref{EqII2LMem}. This is in contrast to the scalar toy model discussed in \citeII{\S\ref*{SssNToyT}} where there was no such loss (since the Kerr model there did not have zero energy bound states). In a nonlinear iteration to solve the nonlinear gauge-fixed Einstein equations
\begin{equation}
\label{EqIHighP}
  P(g_\eps;g_{0,\eps}) := \Ric(g_\eps) - \Lambda g_\eps - \tilde\delta_{g_{0,\eps}}^*\Ups(g_\eps;g_{0,\eps}) = 0,\qquad g_\eps=g_{0,\eps}+h_\eps,
\end{equation}
where $\Ups$ is a gauge 1-form whose linearization in $g_\eps$ gives, in the limit $\hat\rho\to 0$, the gauge condition $\breve\delta_{\hat g_b}\sfG_{\hat g_b}$ from~\S\ref{SssIKSpec}, this means that one loses two orders of $\hat\rho$-decay at each iteration step; thus, \emph{the iteration scheme does not close}.

To remedy this issue, we note that the $\hat\rho$-decay loss is entirely due to the (large, generalized) zero energy states of the Kerr model. We thus upgrade the bounds~\eqref{EqIHigh} by writing
\[
  \eps^{-2}L h = f - (L_\eps-\eps^{-2}L)h \in H_{\sop,\eps}^{k-2,\alpha_\circ,\hat\alpha-2}(\Omega_\eps)^{\bullet,-}
\]
(exploiting~\eqref{EqII2Diff}) and using the precise information on forward solutions of $L$ to extract asymptotics for $h$, of the form
\begin{equation}
\label{EqIHighSplit}
  h = (\text{contributions from $\hat g_b^{\prime\Ups}$ etc.}) + h_0,\qquad h_0 \in H_{\sop,\eps}^{k-d,\alpha_\circ,\hat\alpha}(\Omega_\eps)^{\bullet,-},
\end{equation}
i.e.\ the regular part $h_0$ does not suffer from a loss of $\hat\rho$-decay relative to the source term $f$. Since in the limit $\eps\to 0$ the nonlinear operator $h\mapsto P(g_{0,\eps}+h;g_{0,\eps})$ is well-approximated by the linearized Kerr model $\eps^{-2}L$ and thus annihilates the contributions from $\hat g_b^{\prime\Ups}$ etc.\ to sufficiently high order, we then succeed in implementing a Nash--Moser iteration scheme for solving~\eqref{EqIHighP} in a space of functions which encode the contributions from the pieces of $h$, i.e.\ the zero energy states and the regular part $h_0$, \emph{separately}; see Definition~\ref{DefTrAsy}, further Theorem~\ref{ThmTrAsyTame} for the tame estimates for the different pieces of $h$, and Theorem~\ref{ThmTrLoc} for the solution of~\eqref{EqIHighP} on small domains.

There is a technical subtlety in obtaining tame estimates for the splitting~\eqref{EqIHighSplit}: from the Kerr perspective, uniform bounds on s-derivatives in particular require uniform bounds for $\eps^{-1}\pa_{\hat t}$-derivatives. (The remaining s-derivatives are handled via \emph{module regularity} similarly to \cite{HassellMelroseVasySymbolicOrderZero,BaskinVasyWunschRadMink,BaskinVasyWunschRadMink2}.) This means that we cannot separate the growing contributions to $h$ from the regular part $h_0$ by means of cutoff functions in fast time $\hat t$, as such cutoffs would have large $\eps^{-1}\pa_{\hat t}$-derivatives; instead, we are forced to use cutoff functions in slow time $t=\eps\hat t$. But after any positive amount of slow time $t$, the solution $h$ already exhibits $\gtrsim\eps^{-1}$ growth, which makes it difficult to separate the growing pieces of $h$ from the regular piece. Our strategy is instead to use cutoffs in $t=\eps\hat t$ which transition from $0$ to $1$ in an interval $[-\delta,-\frac{\delta}{2}]$ of \emph{negative} times (i.e.\ prior to the support of the source term $f$ and the forward solution $h$). (See the discussion around~\eqref{EqKHiu12} for a simple motivating example.) The individual pieces of $h$ then have support in $t\geq-\delta$, and only their sum (which is $h$) has support in $t\geq 0$. The resulting precise solution operator thus enlarges supports in $t$ by $\delta$, which comes at the cost of an operator norm $\sim\delta^{-k}$. Throughout the nonlinear iteration, we need to stay in a fixed small subset of $M$, however, forcing us to take $\delta\searrow 0$ as the iteration proceeds. We accomplish the balancing of exploding operator norms and the fast convergence of a Newton type iteration by proving a custom-made Nash--Moser theorem which is a simple adaptation of \cite{SaintRaymondNashMoser}; see Theorem~\ref{ThmNM}.

\medskip

Going from solutions of the gluing problem on small domains to compact domains is done by piecing together finitely many local solutions; see Theorem~\ref{ThmTrSemi}. We remark that if the initial error term $P(g_{0,\eps};g_{0,\eps})$ arising from the failure of the formal solution $g_{0,\eps}$ to be a true solution is only of size $\cO(\eps^N)$ for some finite $N$, then our Nash--Moser theorem would only produce a correction $h_\eps$ with uniform bounds only in a finite regularity space; due to this technical issue (which may well be avoidable with more care), it is important for us that $g_{0,\eps}$ is a formal solution \emph{to all orders in $\eps$}.

%%%%%%%%%%%%%%%%%%%%%%%%%%%%%%%%%%%%%%%%%%%%%%%%%%
\subsection{Outline of the paper}
\label{SsIO}

In~\S\ref{SN}, we recall the relevant notions from geometric singular and microlocal analysis, including the total gluing spacetime, the Kerr spacetime manifold, and the relevant Sobolev spaces on which our analysis will take place.

The technical heart of the paper is~\S\ref{SK}, in which present a detailed analysis of the linearized gauge-fixed Einstein equations on subextremal Kerr spacetimes.

In~\S\ref{SE}, we combine the analysis of the Kerr model with the uniform se-regularity estimates proved in \cite{HintzGlueLocII} to obtain uniform control of the linearized gauge-fixed Einstein equations on glued spacetimes. We also prove tame estimates on s-Sobolev spaces.

In~\S\ref{STr}, we introduce a sharper description of forward solutions of the linearized gauge-fixed Einstein equations on glued spacetimes which is then shown to be compatible with a Nash--Moser theorem. We use this to prove the main theorem of this paper (Theorem~\ref{ThmTrSemi}).

In~\S\ref{SApp}, we describe a simple application of our main result and construct spacetimes describing a black hole merger event, followed by the relaxation of the resulting single black hole to a stationary (slowly rotating Kerr--de~Sitter) state.

%%%%%%%%%%%%%%%%%%%%%%%%%%%%%%%%%%%%%%%%%%%%%%%%%%
\subsection*{Acknowledgments}

This series of papers would not have been possible without fundamental contributions of many fantastic colleagues. The singular geometric analysis perspective is strongly inspired by works of Richard Melrose, Rafe Mazzeo, and Michael Singer. Similarly, the influence of Andr\'as Vasy is visible throughout, especially through his introduction of a powerful framework for the global analysis of non-elliptic operators via microlocal analysis as well as his inspirational work on low energy resolvent bounds. The deep insights into trapping by Maciej Zworski, Semyon Dyatlov, and Jared Wunsch play a crucial role as well. Many thanks are also due to Gunther Uhlmann and Sara Kali\v{s}nik for their continued support and encouragement.

%%%%%%%%%%%%%%%%%%%%%%%%%%%%%%%%%%%%%%%%%%%%%%%%%%%%%%%%%%%%%%%%%%%%%%
\section{Notation and preliminaries}
\label{SN}

%%%%%%%%%%%%%%%%%%%%%%%%%%%%%%%%%%%%%%%%%%%%%%%%%%
\subsection{Geometric singular analysis structures; Sobolev spaces}
\label{SsNG}

In this section, we recall some notions of geometric singular analysis which are used throughout the paper. Most of these notions were already recalled in \cite{HintzGlueLocII}, and hence we shall be terse.

Manifolds with corners $M$ are required to have embedded boundary hypersurfaces $H$ which thus admit defining functions $\rho_H\in\CI(M)$, i.e.\ $\rho_H\geq 0$, $H=\rho_H^{-1}(0)$, and $\dd\rho_H\neq 0$ on $H$. Oftentimes, the precise choice of boundary defining function does not matter (e.g.\ when defining weighted Sobolev spaces where the weights are products of powers of boundary defining functions); this is due to the fact that the quotient of two defining functions of the same boundary hypersurface is a positive smooth function. Furthermore, we often work in some open subset $U\subset M$, in which case by a mild abuse of terminology we call a function $\rho_H\in\CI(U)$ a (local) defining function of $H$ if on each compact subset $K\subset U$ it is the restriction to $U$ of a defining function defined on all of $X$. To avoid overburdening the notation, we shall use the same symbol $\rho_H$ to denote (local) defining functions, the precise choice of which may change. (We shall always point out when such changes occur.)

In this paper, manifolds with corners arise via the process of \emph{real blow-up}. Thus, if $M$ is a manifold with corners and $Y\subset\pa M$ is a boundary p-submanifold \cite{MelroseDiffOnMwc} (i.e.\ around any point of $Y$ there exist local coordinates $x^1,\ldots,x^k\geq 0$, $y^1,\ldots,y^{n-k}\in\R$, so that $Y=\{x^1=\ldots=x^p=0,\ y^1=\ldots,y^q=0\}$ for some $p\geq 1$ and $q\in\{0,\ldots,n-k\}$), we set
\[
  [M;Y]:=(M\setminus Y)\sqcup S{}^+N Y,
\]
where $S{}^+N Y$ is the inward pointing spherical normal bundle. The latter is defined as the quotient of ${}^+N Y:={}^+T_Y M/T Y$ by dilations in the fibers, where ${}^+T_Y M\subset T_Y M$ consists of all (non-strictly) inward pointing tangent vectors. The space $S{}^+N Y$ is called the \emph{front face} of the blow-up (denoted $\ff[M;Y]$), and the map $[M;Y]\to M$ which is the identity on $M\setminus Y$ and the base projection on the front face is the \emph{blow-down map}. The space $[M;Y]$ carries a unique smooth structure in which polar coordinates around $Y$ are smooth down to the front face; the blow-down map is then a diffeomorphism from the complement of the front face to $M\setminus Y$. The \emph{lift} $\upbeta^*X\subset[M;Y]$ of a subset $X\subset M$ is defined to be $\upbeta^{-1}(X)$ when $X\subset\ff[M;Y]$, and the closure in $[M;Y]$ of $\upbeta^{-1}(X\setminus\ff[M;Y])$ otherwise.

The \emph{radial compactification} of $\R^n$ is the manifold with boundary $\ol{\R^n}=(\R^n\sqcup([0,\infty)\times\Sph^{n-1}))/\sim$ where $0\leq x=r\omega\sim(r^{-1},\omega)$, $r=|x|$. Important spaces for this paper are
\begin{equation}
\label{EqNSpcM}
  \cM := [\ol{\R\times\R^3};\pa(\ol\R\times\{0\})],\qquad
  \cX := \ol{\R^3},
\end{equation}
which will be the carriers of stationary and asymptotically flat metrics such as Minkowski and Kerr (see~\S\ref{SsNKerr}), with $\cX$ being the spatial manifold (a cross section for the action of time translations).

For a manifold $X$ with boundary or corners, we write $\Vb(X)$ for the space of \emph{b-vector fields} \cite{MelroseMendozaB,MelroseAPS}, i.e.\ vector fields tangent to all boundary hypersurfaces. When $X$ is a manifold with boundary, and $\rho\in\CI(X)$ is a boundary defining function, an \emph{unweighted b-density} on $X$ is a smooth positive density $\mu$ on $X^\circ$ so that $0<\rho\mu\in\CI(X;\Omega X)$ is a smooth density on $X$; a \emph{weighted b-density} is the product of an unweighted b-density and a power $\rho^{-w}$ for some $w\in\R$. We moreover recall the notation $\Vsc(X)=\rho\Vb(X)$ for the space of \emph{scattering vector fields}. When $Y\subset X$ is a boundary p-submanifold, then the space $\Vtsc([X;Y])$ of \emph{3-body-scattering vector fields} is the $\CI([X;Y])$-span of the lift of $\Vsc(X)$ along $[X;Y]\to X$. In the special case $X=\ol{\R^n}$, the space $\Vsc(\ol{\R^n})$ is spanned over $\CI(\ol{\R^n})=S^0_{\rm cl}(\R^n)$ by $\pa_{x^i}$, $i=1,\ldots,n$. A global frame of the corresponding scattering cotangent bundle $\Tsc^*X$, $X=\ol{\R^n}$, is thus $\dd x^i$, $i=1,\ldots,n$; and this is also a global frame of the 3-body-scattering cotangent bundle $\Ttsc^*[X;Y]$; the case of interest in this paper is $\Ttsc^*\cM$, with $\cM$ as in~\eqref{EqNSpcM}.

For $m\in\N_0$, we write $\Diffb^m(X)$ for the space of $m$-th order b-differential operators, i.e.\ locally finite sums of up to $m$-fold compositions of b-vector fields, with a $0$-fold composition defined to be multiplication by an element of $\CI(X)$. If $\rho_1,\ldots,\rho_N\in\CI(X)$ are boundary defining functions (for $X$ with corners), we write
\[
  \Diffb^{m,\ell_1,\ldots,\ell_N}(X) := \rho_1^{-\ell_1}\cdots\rho_N^{-\ell_N}\Diffb^m(X) := \{ \rho_1^{-\ell_1}\cdots\rho_N^{-\ell_N}A \colon A\in\Diffb^m(X) \}
\]
for the space of weighted b-differential operators. For $X$ with boundary, we define $\Diffsc^m(X)$ and its weighted analogues $\Diffsc^{m,\ell}(X)=\rho^{-\ell}\Diffsc^m(X)$ analogously. Given $P\in\Diffb^m(X)$ and a collar neighborhood $[0,\delta)_\rho\times\pa X$ of $\pa X$, we can write
\[
  P = \sum_{j=0}^m P_j(\rho)(\rho\pa_\rho)^j,\qquad P_j\in\CI([0,\delta)_\rho;\Diff^{m-j}(\pa M)),
\]
and then define the \emph{normal operator} by $N(P):=\sum_{j=0}^m P_j(0)(\rho\pa_\rho)^j$. The Mellin transform diagonalizes $N(P)$, and we write $N(P,\lambda):=\sum_{j=0}^m\lambda^j P_j(0)\in\Diff^m(\pa M)$, $\lambda\in\C$, for the Mellin-transformed normal operator family. The \emph{boundary spectrum} $\specb(P)\subset\C$ is the set of all $\lambda\in\C$ for which $N(P,\lambda)\colon\CI(\pa M)\to\CI(\pa M)$ is not invertible. We write $\cA^0(X)$ for the space of \emph{conormal functions}, i.e.\ $u\in\cA^0(X)$ if and only if $u\in\CI(X^\circ)$ and $P u\in L^\infty_\loc(X)$ (which equals $L^\infty(X)$ when $X$ is compact) for all $P\in\Diffb(X)$. When $X$ only has one boundary hypersurface, with boundary defining function $\rho\in\CI(X)$, we write $\cA^\alpha(X):=\rho^\alpha\cA^0(X)$.

The space of semiclassical scattering vector fields on a manifold with boundary $X$ is equal to $h\CI([0,1)_h;\Vsc(X))$, with a global frame for $X=\ol{\R^n}$ given by $h\pa_{x^i}$. With $\rho\in\CI(X)$ denoting a boundary defining function, the corresponding class of differential operators
\[
  \Diffsch^{s,r,b}(X) = \rho^{-r}h^{-b}\Diffsch^s(X)
\]
possesses three orders, which we refer to as the semiclassical scattering differential order ($s$), scattering decay order ($r$), and semiclassical order ($b$).

We only explicitly discuss differential operators acting on complex-valued functions, and similarly only Sobolev spaces of $\C$-valued distributions; we omit the simple notational modifications required to treat spaces of sections of vector bundles and operators acting between them.

%%%%%%%%%%%%%%%%%%%%%%%%%%%%%%
\subsubsection{Scattering-b-transition structures}
\label{SssNscbt}

Let $X$ be a manifold with boundary; denote by $\rho\in\CI(X)$ a boundary defining function. We then recall the \emph{sc-b-transition single space} from~\cite[Appendix~A]{HintzKdSMS} (following \cite{GuillarmouHassellResI}) to be
\[
  X_\scbtop := [[0,1)\times X;\{0\}\times\pa X].
\]
We write the parameter on the interval $[0,1)$ as $|\sigma|$. The three boundary hypersurfaces of $X_\scbtop$ are:
\begin{itemize}
\item the \emph{scattering face} $\scface$, with defining function $\rho_\scface$ (e.g.\ $\frac{\rho}{\rho+|\sigma|}$);
\item the \emph{transition face} $\tface$, with defining function $\rho_\tface$ (e.g.\ $\rho+|\sigma|$);
\item the \emph{zero face} $\zface$, with defining function $\rho_\zface$ (e.g.\ $\frac{|\sigma|}{\rho+|\sigma|}$).
\end{itemize}
The Lie algebra $\Vscbt(X)$ of \emph{sc-b-transition vector fields} consists of all elements of the space $\rho_\scface\Vb(X_\scbtop)$ which annihilate $\dd|\sigma|$. (These are thus families of vector fields on $X^\circ$ which behave in a particular manner near $\pa X$ or as $|\sigma|\to 0$.) In local coordinates $\rho\geq 0$, $\omega\in\R^{n-1}$ near a boundary point of $X$, these are thus linear combinations of $\frac{\rho}{\rho+|\sigma|}\rho\pa_\rho$, $\frac{\rho}{\rho+|\sigma|}\pa_\omega$ with $\CI(X_\scbtop)$ coefficients. These vector fields are a local frame of the sc-b-tangent bundle $\Tscbt X\to X_\scbtop$. A local frame of the dual \emph{sc-b-transition cotangent bundle} $\Tscbt^*X\to X_\scbtop$ is correspondingly given by $\frac{\rho+|\sigma|}{\rho}\frac{\dd\rho}{\rho}$, $\frac{\rho+|\sigma|}{\rho}\dd\omega$.

We write $\Diffscbt^m(X)$ for the space of $m$-th order sc-b-transition differential operators, i.e.\ locally finite up to $m$-fold compositions of elements of $\Vscbt(X)$. Our convention for spaces of weighted operators is
\begin{equation}
\label{EqNscbtDiff}
  \Diffscbt^{m,\ell_\scface,\ell_\tface,\ell_\zface}(X) := \rho_\scface^{-\ell_\scface}\rho_\tface^{-\ell_\tface}\rho_\zface^{-\ell_\zface}\Diffscbt^m(X).
\end{equation}

%%%%%%%%%%%%%%%%%%%%%%%%%%%%%%
\subsubsection{3b-structures}
\label{SssN3b}

The manifold $\cM=[\ol{\R_{\hat t}\times\R_{\hat x}^3};\pa(\ol\R\times\{0\})]$ in~\eqref{EqNSpcM} has two boundary hypersurfaces:
\begin{itemize}
\item the front face $\cT=\cT^+\sqcup\cT^-$, with defining function $\rho_\cT$ (e.g.\ $\rho_\cT=\frac{\la\hat r\ra}{\la(\hat t,\hat r)\ra}$ where $\hat r=|\hat x|$);
\item the lift $\cD$ of the original boundary $\pa\ol{\R^4}$, with defining function $\rho_\cD$ (e.g.\ $\rho_\cD=\la\hat r\ra^{-1}$).
\end{itemize}
We then recall from \cite{Hintz3b} the space of \emph{3b-vector fields} $\Vtb(\cM):=\rho_\cD^{-1}\Vtsc(\cM)$; this is the $\CI(\cM)$-span of the vector fields $\la\hat r\ra\pa_{\hat t}$ and $\la\hat r\ra\pa_{\hat x}$. Spaces of weighted 3b-differential operators are denoted
\[
  \Difftb^{m,\ell_\cD,\ell_\cT}(\cM) := \rho_\cD^{-\ell_\cD}\rho_\cT^{-\ell_\cT}\Difftb^m(\cM).
\]
The 3b-operators of main interest in the present paper are invariant under time-translations; they thus necessarily have $\cT$-weight $0$. We denote the subspace of time-translation-invariant 3b-operators by
\[
  \Diff_{\tbop,\rm I}^{m,\ell_\cD}(\cM) = \rho_\cD^{-\ell_\cD}\Diff_{\tbop,\rm I}^m(\cM) \subset \Difftb^{m,\ell_\cD,0}(\cM).
\]
For $P\in\rho_\cD^{-\ell_\cD}\Diff_{\tbop,\rm I}(\cM)$, we define its spectral family by $\hat P(\sigma):=e^{i\sigma\hat t}P e^{-i\sigma\hat t}$ (acting on functions on $\R^3_{\hat x}$), $\sigma\in\C$. Since the spectral families of the basic 3b-vector fields $\la\hat r\ra\pa_{\hat t}$, $\la\hat r\ra\pa_{\hat x}$ are $-i\sigma\rho_\cD^{-1}$, $\la\hat r\ra\pa_{\hat x}$, we conclude that
\begin{equation}
\label{EqNDiff3bFT}
  P\in\rho_\cD^{-\ell_\cD}\Difftb^m(\cM) \implies \pa_\sigma^j\hat P(0) \in \rho_\cD^{-\ell_\cD-j}\Diffb^{m-j}(\cX)\ \ (0\leq j\leq m),
\end{equation}
where we recall $\cX=\ol{\R^3_{\hat x}}$. Furthermore, such $P$ have \emph{transition face normal operators} $N_\tface(P,\hat\sigma)=P_\tface(\hat\sigma)$, defined for each $\hat\sigma\in\C$, $|\hat\sigma|=1$, by restricting $|\sigma|^{\ell_\cD}\hat P(|\sigma|\hat\sigma)$ to $\tface$. Concretely, if in $r>1$ we write $P=\sum_{j+k+|\alpha|\leq m}\hat r^{\ell_\cD}a_{j k\alpha}(\hat r^{-1},\omega)(\hat r\pa_{\hat t})^j(\hat r\pa_{\hat r})^k\pa_\omega^\alpha$, then in terms of $r':=\hat r|\sigma|$, we have
\begin{align*}
  &\hat P_\tface(\hat\sigma) = r'{}^{\ell_\cD}\sum_{j+k+|\alpha|\leq m} a_{j k\alpha}(0,\omega) (-i\hat\sigma r')^j (r'\pa_{r'})^k\pa_\omega^\alpha \in \Diff_{\scop,\bop}^{m,\ell_\cD+m,-\ell_\cD}(\tface), \\
  &\qquad \tface=[0,\infty]_{r'}\times\Sph^2 \subset X_\scbtop.
\end{align*}

%%%%%%%%%%%%%%%%%%%%%%%%%%%%%%
\subsubsection{se- and s-structures}
\label{SssNse}

Let $M$ be a smooth manifold, and let $\cC\subset M$ be a closed 1-dimensional submanifold. Following~\citeI{\S\ref*{SG}}, we then define
\begin{equation}
\label{EqNsewtM}
  \wt M = [[0,1)_\eps\times M;\{0\}\times\cC],
\end{equation}
and write $M_\eps=\{\eps\}\times M\subset\wt M$, $\eps>0$, for the $\eps$-level sets. Let $r\in\CI([M;\cC])$ be a boundary defining function; in local coordinates $t\in\R$, $x\in\R^3$ in which $x=0$ at $\cC$, one can take $|x|$ as a local defining function. The boundary hypersurfaces of $\wt M$ are:
\begin{itemize}
\item the lift $M_\circ$ of $\{0\}\times M$, with defining function $\rho_\circ$ (e.g.\ $\rho_\circ=(1+\frac{r^2}{\eps^2})^{-\frac12}=\frac{\eps}{(\eps^2+r^2)^{\frac12}}$);
\item the front face $\hat M$, with defining function $\hat\rho$ (e.g.\ $\hat\rho=(\eps^2+r^2)^{\frac12})$.
\end{itemize}
Local coordinates on $\hat M^\circ$ are $t$ and $\hat x=\frac{x}{\eps}$. The front face is the total space of a fibration $\ol{\R^3_{\hat x}}-\hat M\to\cC$; we denote the fiber over $p\in\cC$ by $\hat M_p$. See \citeI{Figure~\ref*{FigGTot}}. Following \citeI{Definition~\ref*{DefGVf}} and \citeII{Definition~\ref*{DefFs}}, we then define the space $\Vs(\wt M)$ of \emph{s-vector fields} to consist of all $V\in\Vb(\wt M)$ with $\dd\eps(V)=0$, and $\Vse(\wt M)\subset\Vs(\wt M)$ as the Lie subalgebra consisting of those vector fields which are in addition tangent to the fibers of $\hat M$. (Regarding the terminology, see \citeI{Remark~\ref*{RmkGVf}}.) We can equivalently define $\Vse(\wt M)=\hat\rho\CI(\wt M;\wt T\wt M)$, where
\[
  \wt T\wt M\to\wt M
\]
is the pullback of the vertical tangent bundle $[0,1)\times T M\to[0,1)\times M$ along the blow-down map
\begin{equation}
\label{EqNseBlowdown}
  \wt\upbeta \colon \wt M \to [0,1)\times M.
\end{equation}
We also introduce the notation $\upbeta_\circ\colon M_\circ=[M;\cC]\to M$ for the restriction of $\wt\upbeta$ to $M_\circ$. Thus, every se-vector field is a linear combination, with coefficients in $\CI(\wt M)$, of the vector fields
\begin{equation}
\label{EqNseVF}
  \hat\rho\pa_t=\la\hat r\ra\pa_{\hat t},\quad \hat\rho\pa_x=\la\hat r\ra\pa_{\hat x}.\qquad (\hat r=|\hat x|.)
\end{equation}
Similarly, every s-vector field is a linear combination of
\begin{equation}
\label{EqNsVF}
  \pa_t=\eps^{-1}\pa_{\hat t},\quad \hat\rho\pa_x=\la\hat r\ra\pa_{\hat x}.
\end{equation}

Given an se-differential operator
\[
  \wt L \in \Diffse^{m,\ell_\circ,\hat\ell}(\wt M) := \rho_\circ^{-\ell_\circ}\hat\rho^{-\hat\ell}\Diffse^m(\wt M),
\]
we can define its $\hat M$-normal operators $N_{\hat M_p}(\wt L)$, $p\in\cC$, as follows. Geometrically, we can restrict $\eps^{\hat\ell}\wt L$ to the front face of $[\wt M;\hat M_p]$ which is diffeomorphic to $\cM$; see \citeII{Lemma~\ref*{LemmaFseRelGeo}} for details. In terms of local coordinates $t,x$ as above, with $p=(t_0,0)$, and writing $\hat r=|\hat x|=\frac{|x|}{\eps}$, we write
\[
  \eps^{\hat\ell}\wt L = \sum_{i+|\alpha|\leq m} \la\hat r\ra^{-\ell_\circ+\hat\ell} \wt\ell_{i\alpha} (\hat\rho\pa_t)^i(\hat\rho\pa_x)^\alpha,
\]
where $\wt\ell_{i\alpha}\in\CI(\wt M)$; fixing $\hat\rho=(\eps^2+|x|^2)^{\frac12}$, writing $\hat t=\frac{t-t_0}{\eps}$, $\hat x=\frac{x}{\eps}$, and using~\eqref{EqNseVF}, we thus have
\[
  N_{\hat M_p}(\wt L) = \sum_{i+|\alpha|\leq m} \la\hat r\ra^{-\ell_\circ+\hat\ell}\wt\ell_{i\alpha}|_{\hat M_p} (\la\hat r\ra\pa_{\hat t})^i(\la\hat r\ra\pa_{\hat x})^\alpha \in \Diff_{\tbop,\rm I}^{m,\ell_\circ-\hat\ell,0}(\cM).
\]
This in turn has a spectral family and transition face normal operators, as discussed in~\S\ref{SssN3b}. In this paper, we will only encounter (tensorial) se-operators whose $\hat M$-normal operators, in suitable coordinates (and bundle trivializations), are independent of the point $p$.

The equalities~\eqref{EqNseVF} induce a bundle isomorphism
\begin{equation}
\label{EqNseBundleIso}
  \Tse^*_z\wt M \cong \Ttb^*_{\hat z}\cM
\end{equation}
for all $z\in\hat M$ and $\hat z\in\cM$ so that the $\hat x$-coordinates of $z$ and $\hat z$ coincide; cf.\ \citeII{(\ref*{EqFseBundle3b})}.

%%%%%%%%%%%%%%%%%%%%%%%%%%%%%%
\subsubsection{Function spaces}
\label{SssNSob}

Corresponding to each of the above Lie algebras of vector fields (and the associated classes of weighted differential operators), we have scales of Sobolev spaces on the underlying manifold $X$; we only consider the case that $X$ is compact. Fix a (weighted or unweighted) b-density on $X$. Then for $s\in\N_0$, we define $\Hb^s(X)\subset L^2(X)$ to consist of all $u\in L^2(X)$ so that $P u\in L^2(X)$ for all $P\in\Diffb^s(X)$, and then $\Hb^{s,\alpha}(X)=\rho^\alpha\Hb^s(X)$. These spaces can be given the structure of Hilbert spaces. The spaces $\Hsc^{s,\alpha}(X)$ and $\Htb^{s,\alpha_\cD,\alpha_\cT}(\cM)$ are defined analogously.

For Lie algebras of vector fields depending on a parameter, the associated Sobolev spaces have parameter-dependent norms. Thus, we set
\[
  H_{\scbtop,|\sigma|}^{s,\alpha_\scface,\alpha_\tface,\alpha_\zface}(X) = \Hsc^{s,\alpha_\scface}(X),\qquad |\sigma|\neq 0,
\]
as a vector space, but the norm is defined as
\[
  \|u\|_{H_{\scbtop,|\sigma|}^{s,\alpha_\scface,\alpha_\tface,\alpha_\zface}(X)}^2 := \sum_{i=1}^N \|(P_i)_{|\sigma|}u\|_{L^2(X)}^2
\]
where $P_i=((P_i)_{|\sigma|})_{|\sigma|\in[0,1)}\in\Diffscbt^{m,\alpha_\scface,\alpha_\tface,\alpha_\zface}(X)$, $i=1,\ldots,N$ is any \emph{fixed} finite set of operators spanning $\Diffscbt^{m,\alpha_\scface,\alpha_\tface,\alpha_\zface}(X)$ over $\CI(X_\scbtop)$. The spaces $H_{\scop,h}^{s,r,b}(X)$ (equal to $\Hsc^{s,r}(X)$ as a vector space for $h>0$) and $H_{\seop,\eps}^{s,\alpha_\circ,\hat\alpha}(M)$ (equal to $H^s(M)$ as a vector space for $\eps>0$) are defined analogously. (In the setting of interest in the present paper, with $M$ globally hyperbolic and $\cC\subset M$ an inextendible timelike geodesic, we only consider distributions on $M$ whose support lies in a fixed compact subset. One can then reduce the definition of norms on se-Sobolev spaces to the case of compact $M$. See~\S\ref{SsEstFn} for details.)

Recalling $\cM$ and $\cX$ from~\eqref{EqNSpcM}, it is a simple consequence of Plancherel's theorem that, for all $\sfs\in\N_0$ and $\alpha_\cD\in\R$, the Fourier transform in $\hat t$ defines an isomorphism
\begin{equation}
\label{EqNFT3b}
  \cF \colon \Htb^{\sfs,\alpha_\cD,0}(\cM) \to L^2\bigl(\R_\sigma; H_{\wh\tbop,\sigma}^{\sfs,\alpha_\cD,0}(\cX)\bigr),\qquad H_{\wh\tbop,\sigma}^{\sfs,\alpha_\cD,0}(\cX) := \begin{cases} H_{\scbtop,|\sigma|}^{\sfs,\sfs+\alpha_\cD,\alpha_\cD,0}(\cX), & |\sigma|\leq 1, \\ H_{\scop,|\sigma|^{-1}}^{\sfs,\sfs+\alpha_\cD,\sfs}(\cX), & |\sigma|\geq 1. \end{cases}
\end{equation}

All algebras of differential operators discussed so far sit inside of algebras of \emph{pseudodifferential operators}; all of these are discussed in detail in \cite[\S{2}]{Hintz3b}, with the sole exception of the se-algebra which is introduced in \citeII{\S\ref*{SsFVar}}. Furthermore, those orders in which an algebra is commutative to leading order (i.e.\ the differential order in all cases, the semiclassical order for semiclassical operators, and the scattering decay order for scattering and sc-b-transition operators) can be \emph{variable}; and~\eqref{EqNFT3b} continues to hold for time-translation-invariant variable orders $\sfs\in\CI_{\rm I}(\Stb^*\cM)$, where the orders on the Fourier transform side are induced by $\sfs$ via appropriate phase space relationships, as described in detail in the aforementioned references.

Of central importance for us is the following relationship between se- and 3b-Sobolev spaces, which we recall from \citeII{Lemma~\ref*{LemmaFseHse3b} and Proposition~\ref*{PropFVarseSobRel}}; we only discuss this for $M=\R_t\times\R^3_x$ (thus $n=3$ in the reference).

\begin{prop}[se- and 3b-Sobolev spaces]
\label{PropNse3b}
  Fix the densities $|\dd t\,\dd x|$ and $|\dd\hat t\,\dd\hat x|$ on $M$ and $\cM$, respectively. Let $\sfs\in\CI(\Sse^*\wt M)$, $\alpha_\circ,\hat\alpha\in\R$. Let $t_0\in\R$, and set $\hat\sfs:=\sfs|_{\Sse^*_{\hat M_{t_0}}}\wt M\in\CI_{\rm I}(\Stb^*\cM)$ where we use the identification~\eqref{EqNseBundleIso}. Then for all $\delta>0$ there exist $c_0>0$ and $C$ so that the following holds for all $\eps\in(0,1)$. Let $\lambda\in(0,c_0)$, and define $\Omega=\{|t-t_0|<\lambda,\ |x|<\lambda\}$ and $\hat\Omega_\eps=\{|\hat t|<\eps^{-1}\lambda,\ |\hat x|<\eps^{-1}\lambda\}\subset\cM$. Write $\Psi_\eps(t,x)=(\hat t,\hat x)=(\frac{t-t_0}{\eps},\frac{x}{\eps})$. Then for all $u$ with support in $\Omega$,
  \begin{align*}
    \|u\|_{H_{\seop,\eps}^{\sfs,\alpha_\circ,\hat\alpha}(M)} &\leq C\eps^{-\hat\alpha+2}\|(\Psi_\eps)_*u\|_{\Htb^{\hat\sfs+\delta,\alpha_\circ-\hat\alpha,0}(\cM)}, \\
    \eps^{-\hat\alpha+2}\|(\Psi_\eps)_*u\|_{\Htb^{\hat\sfs-\delta,\alpha_\circ-\hat\alpha,0}(\cM)} &\leq C\|u\|_{H_{\seop,\eps}^{\sfs,\alpha_\circ,\hat\alpha}(M)}.
  \end{align*}
\end{prop}

We can similarly define $L^\infty$-based function spaces (with integer regularity orders). We thus recall from \citeII{Definition~\ref*{DefFseCont}} that
\[
  \cC_{\seop,\eps}^{k,\alpha_\circ,\hat\alpha}(M) = \cC^k(M), \qquad
  \|u\|_{\cC_{\seop,\eps}^{k,\alpha_\circ,\hat\alpha}(M)} := \sum_j \| \rho_\circ^{-\alpha_\circ}\hat\rho^{-\hat\alpha} P_j u \|_{L^\infty(M)},
\]
where $\{P_j\}\subset\Diffse^k(M)$ is a finite subset which spans $\Diffse^k(M)$ over $\CI(\wt M)$. The normed space $\cC_{\sop,\eps}^{m,\alpha_\circ,\hat\alpha}(M)=\cC^m(M)$ is defined similarly using a spanning set $\{Q_l\}\subset\Diffs^m(M)$; and the space $\cC_{(\seop;\sop),\eps}^{(k;m),\alpha_\circ,\hat\alpha}(M)=\cC^{k+m}(M)$ is then equipped with the norm
\[
  \|u\|_{\cC_{(\seop;\sop),\eps}^{(k;m),\alpha_\circ,\hat\alpha}(M)} := \sum_{j,l} \|\rho_\circ^{-\alpha_\circ}\hat\rho^{-\hat\alpha}P_j Q_l u\|_{L^\infty(M)},
\]
cf.\ \citeII{Definition~\ref*{DefFsFn}}. We write $\cC_\seop^k(\wt M)$ for the space of families $\wt u=(u_\eps)_{\eps\in(0,1)}$ of functions $u_\eps\colon M\to\C$ (or equivalently, $\wt u\colon\wt M\setminus(M_\circ\cup\hat M)\to\C$) which are uniformly bounded in $\cC_{\seop,\eps}^k$; similarly for weighted spaces and their s- and (se;s)-analogues. We also note that
\begin{equation}
\label{EqNCsinfty}
  \eps^\infty\cC_\sop^\infty(\wt M) = \bigcap_{k\in\N_0} \eps^k\cC_\sop^k(\wt M) = \eps^\infty L^\infty([0,1)_\eps;\CI(M))
\end{equation}
is the space of all functions $\wt u$ which vanish, together with all derivatives along fixed vector fields on $M$, faster than $\eps^N$ for all $N$ as $\eps\searrow 0$. If $(\eps\pa_\eps)^i\wt u\in\eps^\infty\cC_\sop^\infty(\wt M)$ for all $i\in\N_0$, then $\wt u\in\CIdot(\wt M)=\eps^\infty\CI(\wt M)$ (and vice versa).

For our nonlinear analysis, we need tame estimates involving $H_{\sop,\eps}$- and $\cC_{\sop,\eps}$-norms. We shall use the schematic notation $D_\sop^j$ for the vector of all $j$-fold compositions of the elements of a fixed finite subset of $\cV_\sop(\wt M)$ which spans $\cV_\sop(\wt M)$ over $\CI(\wt M)$.

\begin{lemma}[Tame estimates]
\label{LemmaNTame}
  Let $\Omega\subset M$ be a compact manifold with corners, $\dim M=4$; let $\wt\Omega=\wt\upbeta^{-1}([0,1)_\eps\times\Omega)\subset\wt M$ and write $\Omega_\eps=\wt\Omega\cap M_\eps$. All norms in the estimates below are taken on $\Omega_\eps$. Fix $2<d\in\N_0$. Then for all $k\in\N_0$ there exists a constant $C_k$ so that
  \begin{align}
  \label{EqNTameDerHs}
    \|(D_\sop^j u)(D_\sop^{k-j}v)\|_{L^2} &\leq C_k\Bigl(\|u\|_{\cC_{\sop,\eps}^0}\|v\|_{H_{\sop,\eps}^k} + \|u\|_{\cC_{\sop,\eps}^k}\|v\|_{H_{\sop,\eps}^d}\Bigr), \\
  \label{EqNTameDerCs}
    \|(D_\sop^j u)(D_\sop^{k-j}v)\|_{L^\infty} &\leq C_k\Bigl(\|u\|_{\cC_{\sop,\eps}^0}\|v\|_{\cC_{\sop,\eps}^k} + \|u\|_{\cC_{\sop,\eps}^k}\|v\|_{\cC_{\sop,\eps}^0}\Bigr), \\
  \label{EqNTameHs}
    \|u v\|_{H_{\sop,\eps}^k} &\leq C_k\Bigl( \|u\|_{\cC_{\sop,\eps}^0}\|v\|_{H_{\sop,\eps}^k} + \|u\|_{H_{\sop,\eps}^k}\|v\|_{\cC_{\sop,\eps}^0} \Bigr), \\
  \label{EqNTameCs}
    \|u v\|_{\cC_{\sop,\eps}^k} &\leq C_k\Bigl( \|u\|_{\cC_{\sop,\eps}^0}\|v\|_{\cC_{\sop,\eps}^k} + \|u\|_{\cC_{\sop,\eps}^k}\|v\|_{\cC_{\sop,\eps}^0}\Bigr).
  \end{align}
  Finally, when $u$ is non-vanishing, then $\|u^{-1}\|_{\cC_{\sop,\eps}^k}\leq C_k(1+\|u\|_{\cC_{\sop,\eps}^k})$ where $C_k$ only depends on $\wt\Omega$, $k$, and $\inf_{\Omega_\eps}|u|$.
\end{lemma}
\begin{proof}
  It suffices to prove these estimates in the case that $\Omega=M$ is a closed manifold. Indeed, once this case is established, the estimates in the general case follow by applying the estimates in the case $\Omega=M$ to the functions $E_\eps u$, $E_\eps v$ in place of $u$, $v$, where $E_\eps$ is a uniformly bounded linear extension operator mapping $\cC_{\sop,\eps}^j(\Omega_\eps)\to\cC_{\sop,\eps}^j(M)$ and $H_{\sop,\eps}^j(\Omega_\eps)\to H_{\sop,\eps}^j(M)$ for all $j\in\N_0$, as constructed in \citeII{Lemma~\ref*{LemmaNTameExt}} using \cite{SeeleyExtension}.

  The estimates~\eqref{EqNTameDerHs} and \eqref{EqNTameHs} can then be reduced to standard results on $\R^n$ (here $n=4$); see \citeII{Lemma~\ref*{LemmaNTameDaDb}, Remark~\ref{RmkNTameMultSpecial}}. The estimate~\eqref{EqNTameDerCs} is likewise classical: see \cite[(4.2.17)']{HormanderNonlinearLectures} for the case of functions in 1 variable, which via H\"older's inequality as in the proof of \cite[Chapter~13, Proposition~3.6]{TaylorPDE3} gives~\eqref{EqNTameDerCs} and \eqref{EqNTameCs}. The final statement follow from the observation that $D_\sop^k(u^{-1})$ is the sum of terms of the form $(D_\sop^{k_1}u)\cdots(D_\sop^{k_N}u) u^{-1-N}$ where $k_1,\ldots,k_N\in\N$ with $k_1+\cdots+k_N=k$.
\end{proof}

Moreover, in the notation of this Lemma, we have the (uniform) Sobolev embedding
\begin{equation}
\label{EqNSobEmbed}
  \|u\|_{\cC_{\sop,\eps}^{k,\alpha_\circ,\hat\alpha-\frac32}(\Omega_\eps)} \leq C\|u\|_{H_{\sop,\eps}^{k+d,\alpha_\circ,\hat\alpha}(\Omega_\eps)},\qquad d\geq 3,
\end{equation}
as follows from \citeII{Proposition~\ref*{PropFsSobEmb}} with $n=3$.

%%%%%%%%%%%%%%%%%%%%%%%%%%%%%%%%%%%%%%%%%%%%%%%%%%
\subsection{Minkowski and Kerr metrics}
\label{SsNKerr}

We use the coordinates $\hat t,\hat x$ on $\R_{\hat t}\times\R^3_{\hat x}$; the Minkowski metric will be denoted
\begin{equation}
\label{EqNSpMink}
  \hat{\ubar g} = -\dd\hat t^2 + \dd\hat x^2.
\end{equation}
This is a (non-degenerate) scattering metric on $\ol{\R^4}$, and thus a 3-body-scattering metric on the space $\cM$ defined in~\eqref{EqNSpcM}; that is, $\hat{\ubar g}\in\CI(\cM;S^2\,\Ttsc^*\cM)$. We shall often identify stationary tensors on ($\hat t$-translation-invariant subsets of) $\cM$ such as $\hat{\ubar g}$ with their restrictions to a cross section, so for example $\hat{\ubar g}\in\CI(\cX;S^2\,\Ttsc^*_\cX\cM)$.

We shall similarly consider the Kerr metric $\hat g_b=\hat g_{\bhm,\bha}$ with subextremal parameters $b=(\bhm,\bha)$, $\bhm>0$, $\bha\in\R^3$, $|\bha|<\bhm$, as a stationary 3sc-metric on
\begin{equation}
\label{EqNKerrMb}
  \hat M_b := \cM \cap \{ \hat r=|\hat x|\geq\bhm \},
\end{equation}
that is,
\begin{equation}
\label{EqNKerrXb}
  \hat g_b \in \CI(\hat X_b;S^2\,\Ttsc^*_{\hat X_b}\hat M_b),\qquad \hat X_b := \cX \cap \{ \hat r\geq\bhm \} = \cX \setminus \hat K_b^\circ,\quad \hat K_b = \{ \hat x\in\R^3\colon \hat r\leq \bhm\}.
\end{equation}
Concretely, fix polar coordinates on the unit sphere $\Sph^2\subset\R^3_{\hat x}$ so that $\theta=0$ is aligned with $\frac{\bha}{|\bha|}$ when $\bha\neq 0$ (and arbitrary when $\bha=0$). We then define $\hat g_b$ to be equal to the analytic extension of the Kerr metric $\hat g_{\bhm,a}$, given in Boyer--Lindquist coordinates $\hat t_{\rm BL},\hat r,\theta,\phi$ as
\begin{equation}
\label{EqNKerrMet}
\begin{split}
  &-\frac{\mu}{\varrho^2}(\dd\hat t_{\rm BL}-a\sin^2\theta\,\dd\phi)^2 + \varrho^2\Bigl(\frac{\dd\hat r^2}{\mu}+\dd\theta^2\Bigr) + \frac{\sin^2\theta}{\varrho^2}\bigl((\hat r^2+a^2)\dd\phi-a\,\dd\hat t_{\rm BL}\bigr)^2, \\
  &\qquad \mu:=\hat r^2-2\bhm\hat r+a^2,\quad \varrho^2:=\hat r^2+a^2\cos^2\theta,
\end{split}
\end{equation}
when expressed in the coordinates $\hat t=\hat t_{\rm BL}-T(\hat r)$, $\phi_*=\phi-\Phi(\hat r)$; here $T,\Phi$ are chosen according to
\[
  T'(\hat r) = -\frac{\hat r^2+a^2}{\mu(\hat r)} + \tilde T(\hat r),\qquad
  \Phi'(\hat r) = -\frac{a}{\mu(\hat r)} + \tilde\Phi(\hat r),
\]
where the functions $\tilde T,\tilde\Phi\colon[0,\infty)\to\R$ are analytic on $[0,4\bhm]$, the functions $T,\Phi$ vanish for large $\hat r$, and $\tilde T,\tilde\Phi$ depend smoothly on $\bhm,a$ near fixed subextremal parameters $\bhm_0,a_0$; and we arrange for $\dd\hat t$ to be past timelike.

Unless otherwise specified, we equip $\hat M_b$ with the metric density $|\dd\hat g_b|$, and $\hat X_b$ with the density $|\dd\hat g_b|_{\hat X_b}|$ defined by the relationship $|\dd\hat t|\otimes|\dd\hat g_b|_{\hat X_b}=|\dd\hat g_b|$.

%%%%%%%%%%%%%%%%%%%%%%%%%%%%%%%%%%%%%%%%%%%%%%%%%%
\subsection{Criteria for memberships in variable order or module regularity spaces}
\label{SsNMem}

We record a few results which give estimates for certain $\sigma$- or $\hat t$-dependent families of functions in sc-b-transition or 3b-Sobolev norms.

%%%%%%%%%%%%%%%%%%%%%%%%%%%%%%
\subsubsection{Estimates on scattering-b-transition spaces}

We begin with an $L^2$-estimate for $L^\infty$-conormal functions.

\begin{lemma}[$L^\infty$ vs.\ $L^2$ on the sc-b-transition single space]
\label{LemmaNLinftyL2}
  Let $X=\ol{\R_x^n}$ and equip $X$ with an unweighted b-density. Let $\rho_\scface,\rho_\tface,\rho_\zface\in\CI(X_\scbtop)$ be defining functions of $\scface,\tface,\zface\subset X_\scbtop$. Write
  \[
    \|u\|_{H_{\bop,|\sigma|}^{k,a,b,c}} := \|\rho_\scface^{-a}\rho_\tface^{-b}\rho_\zface^{-c}u\|_{\Hb^k(X)}.
  \]
  Let $\alpha,\beta,\gamma\in\R$, and let
  \begin{equation}
  \label{EqNLinftyL2Incl}
    u \in \cA^{\alpha,\beta,\gamma}(X_\scbtop) := \rho_\scface^\alpha\rho_\tface^\beta\rho_\zface^\gamma\cA^{0,0,0}(X_\scbtop).
  \end{equation}
  Then for all $k\in\N_0$ and $\eta>0$ there exists a constant $C_{k,\eta}$ so that
  \begin{equation}
  \label{EqNLinftyL2Est}
    \|u\|_{H_{\bop,|\sigma|}^{k,\alpha-\eta,\beta-\eta,\gamma}} + \|u\|_{H_{\bop,|\sigma|}^{k,\alpha-\eta,\beta,\gamma-\eta}} \leq C_{k,\eta},\qquad |\sigma|\in[0,1].
  \end{equation}
\end{lemma}
\begin{proof}
  Since the $|\sigma|$-independent lifts of b-vector fields on $X$ span (over $\CI(X_\scbtop)$) the space of vertical (i.e.\ in the kernel of $\dd|\sigma|$) b-vector fields on $X_\scbtop$, it suffices to prove~\eqref{EqNLinftyL2Est} for $k=0$ and for $\alpha=\beta=\gamma=0$. In the region $|x|\leq 2$ then, where we can take $\rho_\scface=\rho_\tface=1$ and $\rho_\zface=|\sigma|$, the function $u$ is pointwise bounded, so~\eqref{EqNLinftyL2Est} is valid. In $|x|\geq 1$, we introduce $\rho=|x|^{-1}$ and $\rho'=\frac{\rho}{|\sigma|}$. We drop the spherical variables $\frac{x}{|x|}$ from the notation.
  
  In the region $\rho'\leq 2$, which is disjoint from $\zface$, we take $\rho_\scface=\rho'$, $\rho_\tface=|\sigma|$ (and $\rho_\zface=1$). We note that $|\frac{\dd\rho}{\rho}|=|\frac{\dd\rho'}{\rho'}|$; the estimate
  \[
    \int_0^2 \rho'{}^{2\eta}|u|^2\,\frac{\dd\rho'}{\rho'} \leq C_\eta\|u\|_{L^\infty}^2
  \]
  proves the desired uniform boundedness (without the $\eta$-loss at $\tface$) in this region.

  In the region $\rho'\geq 1$ (so $|\sigma|\leq\rho\leq 1$) on the other hand, which is disjoint from $\scface$, we take $\rho_\tface=\rho$, $\rho_\zface=\rho'{}^{-1}$ (and $\rho_\scface=1$), and note that
  \[
    \int_{|\sigma|}^1 \rho^{2\eta} + |\rho'{}^{-1}|^{2\eta} \,\frac{\dd\rho}{\rho} = \int_{|\sigma|}^1 \rho^{2\eta}\,\frac{\dd\rho}{\rho} + \int_1^{|\sigma|^{-1}} \rho'{}^{-2\eta}\,\frac{\dd\rho'}{\rho'} \leq \int_0^1 \rho^{2\eta}\,\frac{\dd\rho}{\rho} + \int_1^\infty \rho'{}^{-2\eta}\,\frac{\dd\rho'}{\rho'} < \infty.
  \]
  This completes the proof.
\end{proof}

Next, we capture the sc-b-transition properties of certain `outgoing' functions in the low and high frequency regimes. To avoid (purely notational) modifications near $\hat r=0$, we work on $\hat X_b$ (see~\eqref{EqNKerrXb}) where $\hat r\geq\bhm>0$.

\begin{lemma}[Oscillatory conormal functions I: low complex frequencies]
\label{LemmaNMult}
  Let $\Omega\subset\cV(\Sph^2)$ be a finite spanning set over $\CI(\Sph^2)$. Let $\tilde T=c\log\hat r+\tilde T'$, where $c\in\R$ and $\tilde T'\in\cA^0(\hat X_b)$ is a real-valued function. Let $\alpha,\beta,\gamma\in\R$, $u\in\cA^{\alpha,\beta,\gamma}((\hat X_b)_\scbtop)$. We fix the density $|\dd\hat x|$ (or $|\dd\hat g_b|_{\hat X_b}|$) on $\hat X_b$.
  \begin{enumerate}
  \item\label{ItNMultMod}{\rm (Module regularity.)} For all $\eta>0$ and $j\in\N_0$, $\delta\in\N_0^{|\Omega|}$,
    \begin{equation}
    \label{EqNMultMod}
      \|(\hat r(\pa_{\hat r}-i\sigma))^j\Omega^\delta e^{i\sigma(\hat r+\tilde T)}u\|_{H_{\scbtop,|\sigma|}^{0,-\frac32+\alpha-\eta,-\frac32+\beta-\eta,\gamma}(\hat X_b)},\qquad \sigma=\hat\sigma|\sigma|,
    \end{equation}
    is uniformly bounded for $\hat\sigma\in e^{i[0,\pi]}$ and $|\sigma|\leq 1$. For $\hat\sigma\in e^{i[\frac{\pi}{4},\frac{3\pi}{4}]}$, the uniform bounds~\eqref{EqNMultMod} hold in the space $H_{\scbtop,|\sigma|}^{0,\alpha',-\frac32+\beta-\eta,\gamma}$ for arbitrary $\alpha'\in\R$.
  \item\label{ItNMultVar}{\rm (Variable orders.)} Denote by $\rho'=\frac{\hat r^{-1}}{|\sigma|}$ an affine coordinate along $\tface\subset(\hat X_b)_\scbtop$. For $\hat\sigma\in e^{i[0,\frac{\pi}{4}]}$, resp.\ $\hat\sigma\in e^{i[\frac{3\pi}{4},\pi]}$, the family $e^{i\sigma(\hat r+\tilde T)}u$ is uniformly bounded in $H_{\scbtop,|\sigma|}^{\sfs,\sfr,-\frac32+\beta-\eta,\gamma}$ for $|\sigma|\leq 1$ for all (possibly variable) orders $\sfs,\sfr$ provided $\sfr<-\frac32+\alpha$ at the graph of $-\frac{\dd\rho'}{\rho'{}^2}$, resp.\ $+\frac{\dd\rho'}{\rho'{}^2}$ over $\{\rho'=0\}=\scface\subset(\hat X_b)_\scbtop$ (which is thus a subset of $\Tscbt^*_\scface\hat X_b$). For $\hat\sigma\in e^{i[\frac{\pi}{4},\frac{3\pi}{4}]}$, $e^{i\sigma(\hat r+\tilde T)}u$ is uniformly bounded in $H_{\scbtop,|\sigma|}^{\sfs,\sfr,-\frac32+\beta-\eta,\gamma}$ for all orders $\sfs,\sfr$.
  \end{enumerate}
  The same conclusions remain valid when replacing the orders $-\frac32+\beta-\eta,\gamma$ by $-\frac32+\beta,\gamma-\eta$.
\end{lemma}
\begin{proof}
  Since $e^{-i\sigma\hat r}\hat r(\pa_{\hat r}-i\sigma) e^{i\sigma\hat r}=\hat r\pa_{\hat r}$, the uniform boundedness of~\eqref{ItNMultMod} is equivalent to that of
  \begin{equation}
  \label{EqNMultMod2}
    \|e^{i\sigma\hat r} (\hat r\pa_{\hat r})^j\Omega^\delta e^{i\sigma\tilde T}u\|_{H_{\scbtop,|\sigma|}^{0,-\frac32+\alpha-\eta,-\frac32+\beta-\eta,\gamma}(\hat X_b)}.
  \end{equation}
  Write $u=\rho_\scface^\alpha\rho_\tface^\beta\rho_\zface^\gamma u_0$ with $u_0\in\cA^{0,0,0}$. Since $\hat r\pa_{\hat r}\in\Vb(\hat X_b)$ and $\Omega\subset\Vb(\hat X_b)$---and thus also their $\sigma$-independent extensions to $[0,1)_\sigma\times\hat X_b^\circ$ lie in $\Vb((\hat X_b)_\scbtop)$---, the uniform boundedness of~\eqref{EqNMultMod2} is in turn equivalent to that of $\|e^{i\sigma\hat r}(\hat r\pa_{\hat r})^j\Omega^\delta e^{i\sigma\tilde T}u_0\|_{H_{\scbtop,|\sigma|}^{0,-\frac32-\eta,-\frac32-\eta,0}}$. In other words, it suffices to consider the case $\alpha=\beta=\gamma=0$.

  Note further that in $e^{i\sigma\hat r}e^{i\sigma\tilde T}=e^{i\sigma\hat r}\hat r^{i c\sigma}e^{i|\sigma|\hat\sigma\tilde T'}$, the factor $e^{i|\sigma|\hat\sigma\tilde T'}$ is pointwise bounded with all its derivatives along $\hat r\pa_{\hat r}$ and $\Omega$. The same is true for $e^{i\sigma\hat r}\hat r^{i c\sigma}=\exp(i\sigma(\hat r+c\log\hat r))$ in $\hat r\geq\bhm>0$ since $\Re(i\sigma)\leq 0$; and in fact in the notation $\rho'=\frac{\hat r^{-1}}{|\sigma|}$, we have $e^{i\sigma\hat r}\hat r^{i c\sigma}=e^{i\hat\sigma/\rho'}\rho'{}^{-i c\sigma}|\sigma|^{-i c\sigma}$, which for $\hat\sigma\in e^{i[\frac{\pi}{4},\frac{3\pi}{4}]}$ is exponentially decaying towards the scattering face $\rho'=0$ of $(\hat X_b)_\scbtop$.

  It remains to show that
  \[
   \|u\|_{H_{\scbtop,|\sigma|}^{0,-\frac32-\eta,-\frac32-\eta,0}(\hat X_b)}
   = \|u\|_{H_{\bop,|\sigma|}^{0,-\frac32-\eta,-\frac32-\eta,0}(\hat X_b)}
  \]
  is uniformly bounded for $|\sigma|\leq 1$; but this is the content of Lemma~\ref{LemmaNLinftyL2}. (To prove the final claim of the lemma, we note that also the uniform boundedness of $\|u\|_{H_{\scbtop,|\sigma|}^{0,-\frac32-\eta,-\frac32,-\eta}}$ follows from Lemma~\ref{LemmaNLinftyL2}.)

  Part~\eqref{ItNMultVar} is a consequence of part~\eqref{ItNMultMod}. Indeed, fixing $\hat\sigma\in e^{i[0,\pi]}$, consider the intersection $\sC$ of the characteristic sets of the sc-b-transition operators $\hat r(\pa_{\hat r}-i\sigma)_{|\sigma|\in[0,1]},\Omega_k\in\Diffscbt^{1,1,0,0}(\hat X_b)$ ($k=1,\ldots,|\Omega|$) (see~\eqref{EqNscbtDiff} for the notation): for $\hat\sigma=\pm 1$, the set $\sC$ is equal to the graph of $\sigma\,\dd\hat r=\pm|\sigma|\,\dd\hat r=\mp\frac{\dd\rho'}{\rho'{}^2}$ over $\rho'=0$, whereas for $\hat\sigma\in e^{i(0,\pi)}$, the set $\sC$ is empty. The claim thus follows from elliptic regularity in the sc-b-transition calculus, which implies that the sc-b-transition regularity and scattering decay order of $e^{i\sigma(\hat r+\tilde T)}u$ are equal to $+\infty$ away from $\sC\subset\Tscbt^*_\scface\hat X_b$.
\end{proof}

\begin{lemma}[Oscillatory conormal functions II: high real frequencies]
\label{LemmaNMultHi}
  We use the notation $\Omega,\tilde T$ from Lemma~\usref{LemmaNMult}. Let $u\in\cA^\alpha(\hat X_b)$ where $\alpha\in\R$.
  \begin{enumerate}
  \item\label{ItNMultHiMod}{\rm (Module regularity.)} For all $\eta>0$, $j\in\N_0$, $\delta\in\N_0^{|\Omega|}$, the norm
    \[
      \|(\hat r(\pa_{\hat r}-i\sigma))^j\Omega^\delta e^{i\sigma(\hat r+\tilde T)}u\|_{H_{\scop,|\sigma|^{-1}}^{0,-\frac32+\alpha-\eta,0}(\hat X_b)},
    \]
    is uniformly bounded for $\sigma\in\R$, $|\sigma|\geq 1$.
  \item\label{ItNMultHiVar}{\rm (Variable orders.)} For all (variable) orders $\sfs$ (semiclassical scattering regularity), $\sfr$ (semiclassical scattering decay), $\sfb$ (semiclassical order, i.e.\ power of $h=|\sigma|^{-1}$) subject only to the conditions that $\sfb\leq 0$, and $\sfr<-\frac32+\alpha$ at the graph of $\sigma\,\dd\hat r=\pm\,\frac{\dd\hat r}{|\sigma|^{-1}}$ over $\hat r=\infty$ (for $\pm\sigma\geq 1$), the function $e^{i\sigma(\hat r+\tilde T)}u$ is uniformly bounded in $H_{\scop,|\sigma|^{-1}}^{\sfs,\sfr,\sfb}(\hat X_b)$.
  \end{enumerate}
\end{lemma}
\begin{proof}
  Analogously to the proof of Lemma~\ref{LemmaNMult}, this can be reduced to the case $\alpha=0$, in which case part~\eqref{ItNMultHiMod} again follows from $\|\hat r^{-\frac32-\eta}u\|_{L^2(\hat X_b)}<\infty$. Part~\eqref{ItNMultHiVar}, specifically the arbitrarily high regularity $\sfs$ and the arbitrarily decay order $\sfr$ away from the graph of $\sigma\,\dd\hat r$ follow from elliptic regularity in the semiclassical scattering setting.
\end{proof}

We next record a result on the inverse Fourier transform of terms arising in expansions of the low energy resolvent.

\begin{lemma}[$L^2$ estimates for low energy resolvent expansion terms]
\label{LemmaNL2FT}
  Let $\chi\in\sS(\R)$, and fix a boundary defining function $\hat\rho\in\CI(\ol{\R^3_{\hat x}})$ (such as $\hat\rho=\la\hat r\ra^{-1}$). Let $\alpha\in[0,\frac12)$ and $\hat a\in|\sigma|^{-\alpha}L^2_\cp(\R)$. Then
  \begin{equation}
  \label{EqNL2FTMem}
    \cF^{-1}(\chi(\hat\rho^{-1}\cdot)\hat a(\cdot))\in\bigcap_{\eta>0}\Htb^{\infty,\,-\frac32-\alpha-\eta,\,-\alpha}(\cM) \cap \Htb^{\infty,\,-\frac32-\alpha,\,-\alpha-\eta}(\cM),
  \end{equation}
  where we use the density $|\dd\hat t\,\dd\hat x|$ on $\cM$.
\end{lemma}
\begin{proof}
  For $\alpha\in(0,\frac12)$, we have $|\sigma|^{-\alpha}L^2(\R)\subset L^1_\loc(\R)$ and also
  \[
    |\sigma|^{-\alpha}L^2(\R)\subset H^{-\alpha}(\R).
  \]
  Indeed, this follows from \cite[\S{4}, Lemma~5.4]{TaylorPDE1}---which states that multiplication by $|\sigma|^{-\alpha}$ is a continuous map $H^\alpha(\R)\to L^2(\R)$---by taking adjoints. We can thus factor $a=\cF^{-1}\hat a$ as
  \[
    a(\hat t)=\la\hat t\ra^\alpha a_0(\hat t),\qquad a_0\in H^\infty(\R).
  \]
  Since $\chi(\hat\rho^{-1}\cdot)\hat a$ is uniformly bounded in $|\sigma|^{-\alpha}L^2_\cp(\R)$, we conclude that
  \[
    u:=\cF^{-1}(\chi(\hat\rho^{-1}\cdot)\hat a) \in L^\infty\bigl(\R^3_{\hat x};\la\hat t\ra^\alpha H^\infty(\R)\bigr).
  \]
  Moreover, $\pa_{\hat x}^\alpha u\in L^\infty_\loc(\R^3_{\hat x};\la\hat t\ra^\alpha H^\infty(\R))$ for all $\alpha$, and therefore $\phi u\in\Htb^{\infty,\infty,-\alpha}(\cM)$ for all $\phi\in\CIc(\R^3_{\hat x})$.

  It thus suffices to estimate $u$ in the region $\hat r>2$. In this region, we test for 3b-regularity using the vector fields $\hat\rho^{-1}\pa_t$, $\hat\rho\pa_{\hat\rho}$, $\cV(\Sph^2)$. The conjugation of these vector fields by the Fourier transform is given by multiplication with $-i\hat\rho^{-1}\sigma$ and differentiation along $\hat\rho\pa_{\hat\rho}$, $\cV(\Sph^2)$. These operations leave the form of $\chi(\hat\rho^{-1}\cdot)\hat a(\cdot)$ intact (using in particular that $\hat\rho\pa_{\hat\rho}(\chi(\hat\rho^{-1}\cdot))=\chi_1(\hat\rho^{-1}\cdot)$ where $\chi_1(r'):=-r'\chi'(r')\in\sS(\R)$). Thus, we only need to establish the $L^2$-bounds
  \begin{equation}
  \label{EqNL2FTPf}
    u \in \rho_\cD^{-\frac32-\alpha-\eta}\rho_\cT^{-\alpha}L^2,\qquad u \in \rho_\cD^{-\frac32-\alpha}\rho_\cT^{-\alpha-\eta}L^2,
  \end{equation}
  where $\rho_\cD=\la\hat r\ra^{-1}$, $\rho_\cT=\frac{\la\hat r\ra}{\la(\hat t,\hat r)\ra}$. The first membership follows from
  \[
    \la\hat t\ra^{-\alpha}u\in L^\infty\bigl(\R_{\hat x}^3;L^2_{\hat t}(\R)\bigr)\subset\la\hat r\ra^{\frac32+\eta}L^2(\R_{\hat t}\times\R_{\hat x}^3)
  \]
  and the fact that $(1+|\hat t|+|\hat r|)^{-\alpha}\la\hat t\ra^\alpha\in L^\infty(\R\times\R^3)$ in view of $\alpha\geq 0$.

  The second membership in~\eqref{EqNL2FTPf} is more subtle. Passing to an unweighted b-density in space, we need to prove
  \begin{equation}
  \label{EqNL2FTbDens}
    u=\cF^{-1}(\chi(\hat\rho^{-1}\cdot)\hat a) \in \rho_\cD^{-\alpha}\rho_\cT^{-\alpha-\eta}L^2\Bigl(\cM;\Bigl|\dd\hat t\,\frac{\dd\hat r}{\hat r}\,\dd\slg\Bigr|\Bigr)
  \end{equation}
  where $\slg$ is the standard metric on $\Sph^2$. We shall work with $\hat\rho=\hat r^{-1}$. Moreover, we shall drop spherical arguments and integrations. We write $\hat a(\sigma)=|\sigma|^{-\alpha}\breve a(\sigma)$, $\breve a\in L^2$.

  %%%%%%%%%%
  \pfstep{Cutoffs vanishing at $0$.} Suppose that $\chi(0)=0$. Then we can write $\chi(r')=r'\chi_0(r')$, $\chi_0\in\sS(\R)$, and therefore $\chi(\hat\rho^{-1}\sigma)\hat a(\sigma)=\chi_0(\sigma\hat r)\hat r \sigma|\sigma|^{-\alpha}\breve a(\sigma)$; since $|\chi_0(r')|\leq C\la r'\ra^{-N}$, this obeys the bound
  \begin{align*}
    \int_1^\infty |\hat r^{-\alpha}\chi(\hat\rho^{-1}\sigma)\hat a(\sigma)|^2\,\frac{\dd\hat r}{\hat r} &= \int_1^\infty |\chi_0(\sigma\hat r)|^2 (\hat r|\sigma|)^{2(1-\alpha)} |\breve a(\sigma)|^2\,\frac{\dd\hat r}{\hat r} \\
      &= C^2|\breve a(\sigma)|^2 \int_{|\sigma|}^\infty \la r'\ra^{-2 N} \hat r'^{2-2\alpha}\,\frac{\dd r'}{r'} \\
      &\leq C'|\breve a(\sigma)|^2;
  \end{align*}
  here we use $\alpha<1$ for convergence near $r'=0$, and we take $N$ large (for $\alpha\geq 0$, we can take $N=2$) for convergence near $r'=\infty$. By Plancherel, this implies that $u\in\rho_\cD^{-\alpha}L^2(\cM\cap\{\hat r\geq 2\})=\hat r^\alpha L^2(\cM\cap\{\hat r\geq 2\})$, with norm bounded by $\|\breve a\|_{L^2}$. (This is stronger than~\eqref{EqNL2FTbDens} in that the factor $\rho_\cT^{-\alpha-\eta}$ is not needed.)

  %%%%%%%%%%
  \pfstep{Cutoffs with compact inverse Fourier support.} Given any $\chi\in\sS(\R)$, note that there exists $\tilde\chi\in\sS(\R)$ with $\tilde\chi(0)=\chi(0)$ and $\cF^{-1}\tilde\chi\in\CIc((-1,1))$; indeed, it suffices to take a function $\phi\in\CIc((-1,1))$ with $\int\phi=1$ and to set $\tilde\chi(\sigma)=\chi(0)(\cF\phi)(\sigma)$. Since $\chi-\tilde\chi\in\sS(\R)$ vanishes at $0$, it suffices, by the previous step, to prove~\eqref{EqNL2FTbDens} for $\tilde\chi$ in place of $\chi$. Relabeling $\tilde\chi$ as $\chi$, we thus assume that
  \begin{equation}
  \label{EqNL2FTbinvchi}
    \cF^{-1}\chi \in \CIc((-1,1)).
  \end{equation}

  We now work in the region $\hat t>-\hat r$. (The region $\hat t<\hat r$ is treated similarly.) Note then that, for any $c\in\R$,
  \[
    u(\hat t-c\hat r,\hat r) = \frac{1}{2\pi}\int e^{-i\sigma(\hat t-c\hat r)}\chi(\sigma\hat r)\hat a(\sigma)\,\dd\sigma = \frac{1}{2\pi}\int e^{-i\sigma\hat t}\chi_c(\sigma\hat r)\hat a(\sigma)\,\dd\sigma
  \]
  where $\chi_c(r'):=e^{i c r'}\chi(r')$ is of class $\sS(\R)$; since we can replace $\chi_c$ by $\tilde\chi_c\in\cF(\CIc((-1,1))$ as before, it thus suffices to estimate $u$ in the region $\hat t>(-1+c)\hat r$. Taking $c=9$, we use $\rho_\cD=\hat r^{-1}$ and $\rho_\cT=\frac{\hat r}{\hat t}$ as local defining functions of $\cD$ and $\cT$, respectively, in $\hat t>8\hat r>16$. For $T\geq 16$, define
  \[
    U_T := \sup_{2\leq\hat r\leq T/8} \int_T^{2 T} |u(\hat t,\hat r)|^2\,\dd\hat t.
  \]
  We then claim that
  \begin{equation}
  \label{EqNL2FTUT}
    U_T \lesssim T^{2\alpha}\int_{T/2}^{4 T} |a_0(\hat t)|^2\,\dd\hat t,
  \end{equation}
  where here and below we write $A\lesssim B$ if $A\leq C B$ for some constant $C$ only depending on $\chi,\alpha,\eta$. But since $u(\hat t,\hat r)=\hat r^{-1}\cF^{-1}\chi(\frac{\cdot}{\hat r})*a$, we can use~\eqref{EqNL2FTbinvchi} to bound
  \[
    |u(\hat t,\hat r)| \lesssim \int_{-\hat r}^{\hat r} \hat r^{-1} \la\hat t+s\ra^\alpha a_0(\hat t+s)\,\dd s = \int_{-1}^1 \la\hat t+\hat r s\ra^\alpha a_0(\hat t+\hat r s)\,\dd s;
  \]
  for $\hat r\leq T/8$ this gives
  \[
    \int_T^{2 T} |u(\hat t,\hat r)|^2\,\dd\hat t \lesssim \int_T^{2 T} \int_{-1}^1 \la\hat t+\hat r s\ra^{2\alpha} |a_0(\hat t+\hat r s)|^2\,\dd s\,\dd\hat t \lesssim T^{2\alpha} \int_{T-\hat r}^{2 T+\hat r} |a_0(\hat t)|^2\,\dd\hat t,
  \]
  which in view of $\frac{T}{2}\leq T-\hat r$ and $2 T+\hat r\leq 4 T$ implies~\eqref{EqNL2FTUT}.

  Using~\eqref{EqNL2FTUT}, we can now bound
  \begin{align*}
    &\int_2^\infty \int_{8\hat r}^\infty \Bigl|\hat r^{-\frac32-\alpha}\Bigl(\frac{\hat r}{\hat t}\Bigr)^{\alpha+\eta}u(\hat t,\hat r)\Bigr|^2 \,\dd\hat t\,\hat r^2\,\dd\hat r = \int_2^\infty \hat r^{2\eta} \int_{8\hat r}^\infty |\hat t^{-\alpha-\eta}u(\hat t,\hat r)|^2\,\dd\hat t\,\frac{\dd\hat r}{\hat r} \\
    &\qquad \lesssim \sum_{j=1}^\infty 2^{2\eta j} \sum_{k=j+3}^\infty 2^{-2(\alpha+\eta)k} U_{2^k} = \sum_{k=4}^\infty 2^{-2(\alpha+\eta)k} U_{2^k} \sum_{j=1}^{k-3} 2^{2\eta j} \\
    &\qquad\lesssim \sum_{k=4}^\infty 2^{-2\alpha k}U_{2^k} \lesssim \sum_{k=4}^\infty \int_{2^{k-1}}^{2^{k+2}} |a_0(\hat t)|^2\,\dd\hat t \leq 3\|a_0\|_{L^2}^2 \lesssim \|\hat a\|_{H^{-\alpha}}^2 \lesssim \||\sigma|^\alpha\hat a\|_{L^2},
  \end{align*}
  completing the proof.
\end{proof}

%%%%%%%%%%%%%%%%%%%%%%%%%%%%%%
\subsubsection{Estimates on 3b-spaces}

We introduce some function spaces which will capture the dominant terms in the late time expansion of solutions of the linearized gauge-fixed Einstein equations.

\begin{definition}[Function spaces]
\label{DefNHdashb}
  Let $s\in\N_0\cup\{\infty\}$, $k\in\N_0$. Then $\tau^k\dot H_{-;\bop}^{(s;k)}([1,\infty])$ is the space of all distributions $u=u(\tau)$ with support in $[1,\infty)$ for which $\pa_\tau^k u\in H^s([1,\infty))$. The norm is $\|u\|_{\tau^k\dot H_{-;\bop}^{(s;k)}}=\|\pa_\tau^k u\|_{H^s}$. The spaces $\tau^k\dot H_{-;\bop}^{(s;k)}([\tau_0,\infty])$ for $\tau_0>0$ are defined analogously.
\end{definition}

For $k=0$, $\dot H_{-;\bop}^{(s;0)}([1,\infty])=\dot H^s([1,\infty))$ is a standard Sobolev space (of supported distributions). For $s=0$ on the other hand, $\tau^k\dot H_{-;\bop}^{(0;k)}=\tau^k\dot H_\bop^k([1,\infty])$ is a weighted b-Sobolev space (as we show below). For $k\geq 1$, $\tau^k\dot H_{-;\bop}^{(s;k)}$ combines standard and b (at $\infty$) Sobolev space properties:

\begin{lemma}[Equivalent characterization]
\label{LemmaNHdashb}
  For a distribution $u$ with support in $[1,\infty)$, we have $u\in\tau^k\dot H_{-;\bop}^{(s;k)}([1,\infty])$ if and only if $(\tau\pa_\tau)^j (\tau^{-k}u)\in\dot H^{s+k-j}([1,\infty))$ for $j=0,\ldots,k$. In particular, the norms $\|u\|_{\tau^k\dot H_{-;\bop}^{(s;k)}}$ and $\sum_{j=0}^k\|(\tau\pa_\tau)^j(\tau^{-k}u)\|_{H^{s+k-j}}$ are equivalent.
\end{lemma}
\begin{proof}
  The second condition is equivalent to $\pa_\tau^j u\in\tau^{k-j}\dot H^{s+k-j}([1,\infty))$. For $j=k$, this is the defining condition for membership in $\tau^k\dot H_{-;\bop}^{(s;k)}$. It thus suffices to show that if $\pa_\tau^j u\in\tau^{k-j}\dot H^{s+k-j}$ for some $j\in\{1,\ldots,k\}$, then $\pa_\tau^{j-1}u\in\tau^{k-(j-1)}\dot H^{s+k-(j-1)}$, which follows once we show that for $q\geq 0$, we have
  \[
    \pa_\tau v\in\tau^q\dot H^s,\quad \supp v\subset\{\tau\geq 1\} \implies v\in\tau^{q+1}\dot H^{s+1}.
  \]
  Once this is shown for $s=0$, then for general $s$ we can combine the conclusion $v\in\tau^{q+1}\dot H^1$ with the assumption $\pa_\tau v\in\tau^q\dot H^s$ to conclude that $v\in\tau^{q+1}\dot H^{s+1}$ indeed. For $s=0$, we only need to show the conclusion $v\in\tau^{q+1}L^2$. But this is a simple instance of the Hardy inequality: fixing $q'\in(-\frac12,q)$, we have
  \begin{align*}
    \|v\|_{\tau^{q+1}L^2}^2 &= \int_1^\infty |\tau^{-q-1}v(\tau)|^2\,\dd\tau = \int_1^\infty \tau^{-2 q}\Bigl|\tau^{-1}\int_1^\tau \pa_\tau v(s)\,\dd s\Bigr|^2\,\dd\tau \\
      &\leq \int_1^\infty \tau^{-2 q-2}\Bigl(\int_1^\tau s^{2 q'}\,\dd s\Bigr)\Bigl(\int_1^\tau s^{-2 q'}|\pa_\tau v(s)|^2\,\dd s\Bigr)\,\dd\tau \\
      &\lesssim \int_1^\infty \tau^{-2 q-1+2 q'}\int_1^\tau s^{-2 q'}|\pa_\tau v(s)|^2\,\dd s\,\dd\tau \\
      &= \int_1^\infty s^{-2 q'} \Bigl(\int_s^\infty \tau^{-2(q-q')-1}\,\dd\tau\Bigr) |\pa_\tau v(s)|^2\,\dd s \\
      &\lesssim \|\pa_\tau v\|_{\tau^q L^2}^2.\qedhere
  \end{align*}
\end{proof}

Since every $u\in\tau^k\dot H_{-;\bop}^{(s;k)}([1,\infty])$ is thus in particular a tempered distribution, we can equivalently define
\begin{equation}
\label{EqNHdashb}
  \tau^k\dot H_{-;\bop}^{(s;k)}([1,\infty]) := \bigl\{ u\in\dot\sS'([1,\infty)) \colon \sigma^k\hat u \in \la\sigma\ra^{-s}L^2(\R_\sigma) \bigr\}.
\end{equation}

\begin{lemma}[Spacetime bounds for expansion terms]
\label{LemmaNExpTerm}
  Let $k\in\N_0$ and $\alpha\in\R$. Let $\tilde T=c\log\hat r+\tilde T'$, where $c\in\R$ and $\tilde T'\in\cA^0(\hat X_b)$ is a real-valued function. Set $\hat t_1=\hat t-\hat r-\tilde T$. Then for $u\in\hat t_1^k\dot H_{-;\bop}^{(\infty;k)}([1,\infty])$ and $h\in\cA^\alpha(\hat X_b)$, we have
  \begin{equation}
  \label{EqNExpTerm}
    u h \in \Htb^{\sfs,\alpha_\cD,-k}(\hat M_b)
  \end{equation}
  for all $\alpha_\cD<-\frac32-k+\alpha$ and time-translation-invariant (variable) order functions $\sfs\in\CI_{\rm I}(\Stb^*\hat M_b)$ for which $\sfs+\alpha_\cD<-\frac32-k+\alpha$ at the \emph{time-translation-invariant outgoing radial set}, defined to be (boundary at fiber infinity of) the subbundle spanned by $\hat r^{-1}(\dd\hat t-\dd\hat r)$ over $\cD=\hat r^{-1}(\infty)\subset\hat M_b$.
\end{lemma}
\begin{proof}
  By multiplying $h$ by $\hat r^\alpha=\rho_\cD^{-\alpha}$, we can reduce the Lemma to the case $\alpha=0$.

  We first treat the case $k=0$ using the Fourier transform in $\hat t$; note that the Fourier transform $v=v(\sigma)$ of $u=u(\hat t_1)\in H^\infty$ is
  \[
    \int e^{i\sigma\hat t}u(\hat t-\hat r-\tilde T)\,\dd\hat t = e^{i\sigma(\hat r+\tilde T)} \hat u(\sigma).
  \]
  We thus need to show that, in the notation~\eqref{EqNFT3b},
  \[
    e^{i\sigma(\hat r+\tilde T)}\hat u(\sigma)h \in L^2\bigl(\R_\sigma;H_{\wh\tbop,\sigma}^{\sfs,\alpha_\cD,0}(\hat X_b)\bigr)
  \]
  when $\alpha_\cD<-\frac32$ and $\sfs+\alpha_\cD<-\frac32$ at the outgoing radial set. But this follows from the fact that $\hat u\in\la\sigma\ra^{-N}L^2(\R_\sigma)$ for all $N$ by~\eqref{EqNHdashb} in combination with Lemma~\ref{LemmaNMult}\eqref{ItNMultVar} (for $\sfr=\sfs+\alpha_\cD$ and $\alpha=\beta=\gamma=0$) and Lemma~\ref{LemmaNMultHi}\eqref{ItNMultHiVar} (for $\alpha=0$, $\sfr=\sfs+\alpha_\cD$, and $\sfb=0$).

  We treat the case $k\geq 1$ by induction. Having established~\eqref{EqNExpTerm} for $k-1$ in place of $k$, we need to show for $v:=u h$ the implication
  \begin{equation}
  \label{EqNExpTermv}
    \pa_{\hat t}v=f:=(\pa_{\hat t_1}u)h\in\Htb^{\sfs,\beta+1,-k+1}(\hat M_b),\quad v=0\ \text{for}\ \hat t_1<1 \implies v\in\Htb^{\sfs,\beta,-k}(\hat M_b),
  \end{equation}
  which we shall show for any $\beta\in\R$.

  Since $\pa_{\hat t}$ commutes with multiplication by functions of $\hat r$, we only need to show~\eqref{EqNExpTermv} for $\beta=0$. We can in fact prove~\eqref{EqNExpTermv} for any time-translation-invariant order $\sfs$: threshold conditions play no role anymore. We shall drop spherical variables from the notation. We first consider the case $\sfs=0$. By shifting $\hat t$, we may assume that $v$ is supported in $\bhm\leq\hat r\leq 2\hat t$; we thus work with $\rho_\cD=\hat r^{-1}$ and $\rho_\cT=\frac{\hat r}{\hat t}$, and then estimate, for any fixed $\delta\in(0,\frac12)$ (e.g.\ $\delta=\frac14$)
  \begin{align*}
    \|v\|_{\Htb^{0,0,-k}(\hat M_b)}^2 &= \int_\bhm^\infty \int_{\hat r/2}^\infty \Bigl(\frac{\hat r}{\hat t}\Bigr)^{2 k}\biggl|\int_{\hat r/2}^{\hat t} f(\hat s,\hat r)\,\dd\hat s\biggr|^2\,\hat r^2 \dd\hat t\,\dd\hat r \\
      &\leq \int_\bhm^\infty \int_{\hat r/2}^\infty \Bigl(\frac{\hat r}{\hat t}\Bigr)^{2 k} \biggl( \int_{\hat r/2}^{\hat t} \Bigl(\frac{\hat s}{\hat r}\Bigr)^{2(k-1)}\hat s^{-2\delta}\hat r^{-2}\,\dd\hat s\biggr) \\
      &\quad \hspace{7em} \times \biggl(\int_{\hat r/2}^{\hat t} \hat s^{2\delta}\Bigl(\frac{\hat r}{\hat s}\Bigr)^{2(k-1)}\hat r^2|f(\hat s,\hat r)|^2\,\dd\hat s\biggr)\,\hat r^2\,\dd\hat t\,\dd\hat r \\
      &\leq C\int_\bhm^\infty \int_{\hat r/2}^\infty \hat t^{-1-2\delta} \int_{\hat r/2}^{\hat t} \hat s^{2\delta}|(\rho_\cD^{-1}\rho_\cT^{k-1}f)(\hat s,\hat r)|^2\,\dd\hat s\,\hat r^2\,\dd\hat t\,\dd\hat r \\
      &= C\int_\bhm^\infty \int_{\hat r/2}^\infty \biggl(\int_{\hat s}^\infty \hat t^{-1-2\delta}\,\dd\hat t\biggr) \hat s^{2\delta} |(\rho_\cD^{-1}\rho_\cT^{k-1}f)(\hat s,\hat r)|^2\,\dd\hat s\,\hat r^2\,\dd\hat r \\
      &\leq C \|f\|_{\Htb^{0,1,-k+1}(\hat M_b)}^2.
  \end{align*}
  Next, note that~\eqref{EqNExpTermv} follows for $\sfs\in\N$ from the case $\sfs=0$ by direct differentiation along $\hat r\pa_{\hat t}$, $\hat r\pa_{\hat r}$.

  Finally, for general $\sfs\in\CI_{\rm I}(\Stb^*\hat M_b)$, let us consider $v$ with support in $\hat r>\bhm+\delta$ for some $\delta>0$ (the general case is dealt with using an extension/restriction procedure). Fix a translation-invariant operator $A\in\Psi_{\tbop,\rm I}^\sfs(\cM)$ which is elliptic in $\hat r>\bhm+\frac{\delta}{4}$, and let $B\in\Psi_{\tbop,\rm I}^{-\sfs}(\cM)$ be a parametrix in $\hat r>\bhm+\frac{\delta}{2}$ with $I=B A+R$, $R\in\Psi_{\tbop,\rm I}^0(\cM)$, $\WF'_\tbop(R)\subset\{\hat r<\bhm+\frac{3\delta}{4}\}$; we can pick $B,A,R$ so that their Schwartz kernels are supported in $\{\hat r>\bhm\}\times\{\hat r>\bhm\}$. Since $[A,\pa_{\hat t}]=0$ and $[R,\pa_{\hat t}]=0$, we can estimate
  \begin{align*}
    \|v\|_{\Htb^{\sfs,0,-k}} &\leq \|B(A v)\|_{\Htb^{\sfs,0,-k}} + \|R v\|_{\Htb^{\sfs,0,-k}} \\
      &\leq C\bigl(\|A v\|_{\Htb^{0,0,-k}} + \|R v\|_{\Htb^{\sfs,0,-k}}\bigr) \\
      &\leq C'\bigl( \|\pa_{\hat t}A v\|_{\Htb^{0,1,-k+1}} + \|\pa_{\hat t}R v\|_{\Htb^{\sfs,1,-k+1}}\bigr) \\
      &= C'\bigl( \|A\pa_{\hat t}v\|_{\Htb^{0,1,-k+1}} + \|R\pa_{\hat t}v\|_{\Htb^{\sfs,1,-k+1}}\bigr) \\
      &\leq C'' \|\pa_{\hat t}v\|_{\Htb^{\sfs,1,-k+1}},
  \end{align*}
  as desired. This completes the proof.
\end{proof}

%%%%%%%%%%%%%%%%%%%%%%%%%%%%%%%%%%%%%%%%%%%%%%%%%%
\subsection{Some expressions for wave operators}
\label{SsNMink}

Consider first Minkowski space $(\R^{1+3}_{\hat t,\hat x},\hat{\ubar g})$ from~\S\ref{SsNKerr}. Working in $\hat r>1$, let $\rho_\cD=\hat r^{-1}=|\hat x|^{-1}$. We compute the conjugation of the spectral family $\wh{\Box_{\hat{\ubar g}}}(\sigma)=\rho_\cD^2(-(\rho_\cD\pa_{\rho_\cD})^2+\rho_\cD\pa_{\rho_\cD}+\slDelta)-\sigma^2$ by $e^{-i\sigma\hat r}$ to be
\begin{equation}
\label{EqNMinkConj}
  {}^{\rm o}\wh{\Box_{\hat{\ubar g}}}(\sigma) := e^{-i\sigma\hat r} \wh{\Box_{\hat{\ubar g}}}(\sigma) e^{i\sigma\hat r} = 2 i\sigma\rho_\cD(\rho_\cD\pa_{\rho_\cD}-1) + \rho_\cD^2\bigl(-(\rho_\cD\pa_{\rho_\cD})^2+\rho_\cD\pa_{\rho_\cD}+\slDelta\bigr).
\end{equation}
This is formally the Fourier transform of $\Box_{\hat{\ubar g}}$ with respect to the foliation given by the level sets of $\hat t-\hat r$, hence the superscript `o' for `outgoing'. In terms of the coordinates $\hat\ft_1:=\hat t-\hat r$, $\hat r$, $\omega=\frac{\hat x}{|\hat x|}$, we indeed have
\begin{equation}
\label{EqNMinkOpZeroCoord}
  \Box_{\hat{\ubar g}} = 2\hat r^{-1}\pa_{\hat\ft_1}(\hat r\pa_{\hat r}+1) + \hat r^{-2}\bigl( -(\hat r\pa_{\hat r})^2-\hat r\pa_{\hat r} + \slDelta\bigr).
\end{equation}
The transition face normal operator for parameters $\hat\sigma\in\C$, $|\hat\sigma|=1$, is, in terms of $r':=\hat r|\sigma|$, the restriction of $|\sigma|^{-2}\wh{\Box_{\hat{\ubar g}}}(\hat\sigma|\sigma|)$ (expressed in terms of $\hat\sigma,|\sigma|,r',\omega$) to $|\sigma|=0$, so given by
\begin{equation}
\label{EqNMinktf}
\begin{split}
  (\Box_{\hat{\ubar g}})_\tface(\hat\sigma) &= r'{}^{-2}\bigl(-(r'\pa_{r'})^2-r'\pa_{r'}+\slDelta\bigr) - \hat\sigma^2, \\
  e^{-i\hat\sigma r'}(\Box_{\hat{\ubar g}})_\tface(\hat\sigma) e^{i\hat\sigma r'} &= -2 i\hat\sigma r'{}^{-1}(r'\pa_{r'}+1) + r'{}^{-2}\bigl(-(r'\pa_{r'})^2-r'\pa_{r'}+\slDelta\bigr).
\end{split}
\end{equation}

Next, consider a time-translation-invariant subset $\R_{\hat t}\times\cX\subset\R^{1+3}_{\hat t,\hat x}$, $\cX\subset\R^3_{\hat x}$, and a stationary vector bundle $\cE$, i.e.\ $\cE=\hat\pi^*\hat\cE$ where $\hat\cE\to\cX$ is a vector bundle and $\hat\pi\colon(\hat t,\hat x)\mapsto\hat x$ is the projection. Given an operator $L\in\Diff^2(\R\times\cX;\cE)$ which is stationary, i.e.\ commutes with translations in $\hat t$, we can then write
\[
  L = \sum_{j=0}^2 L_j\pa_{\hat t}^j,\qquad L_j\in\Diff^{2-j}(\cX;\hat\cE).
\]
The action of $L$ on stationary sections of $\cE$ is given by the zero energy operator $\hat L(0)=L_0$; conversely, in the chosen splitting $\R\times\cX$, we can regard $\hat L(0)$ as an element of $\Diff^2(\R\times\cX;\cE)$ by $(\hat L(0)u)(\hat t,\hat x)=(\hat L(0)u(\hat t,\cdot))(\hat x)$. Next, $[L,\hat t]=L_1+2 L_2\pa_{\hat t}$, so $L_1=\wh{[L,\hat t]}(0)$ and $L_2=\frac12[[L,\hat t],\hat t]$. We thus conclude that
\begin{equation}
\label{EqNMinkLdecomp}
  L = \hat L(0) + \wh{[L,\hat t]}(0)\pa_{\hat t} + \frac12[[L,\hat t],\hat t]\pa_{\hat t}^2.
\end{equation}
Recalling that the spectral family $\hat L(\sigma)$ of $L$ with respect to $\hat t$ is defined as the action of $e^{i\sigma\hat t}L e^{-i\sigma\hat t}=\sum_{j=0}^2 L_j(\pa_{\hat t}-i\sigma)^j$ on stationary tensors, we thus have
\begin{equation}
\label{EqNMinkLdecompSpec}
  \hat L(\sigma) = \hat L(0) - i\sigma\wh{[L,\hat t]}(0) - \frac12\sigma^2[[L,\hat t],\hat t].
\end{equation}
In particular, we read off
\begin{equation}
\label{EqNMinkLderiv}
  i\pa_\sigma\hat L(0) = \wh{[L,\hat t]}(0),\qquad
  -\pa_\sigma^2\hat L(0) = [[L,\hat t],\hat t].
\end{equation}

%%%%%%%%%%%%%%%%%%%%%%%%%%%%%%%%%%%%%%%%%%%%%%%%%%%%%%%%%%%%%%%%%%%%%%
\section{Linearized gauge-fixed Einstein equations on Kerr}
\label{SK}

We work on a subextremal Kerr spacetime $(\hat M_b,\hat g_b)$; see~\S\ref{SsNKerr}. We shall study the mapping properties of the linearized gauge-fixed Einstein operator
\begin{subequations}
\begin{equation}
\label{EqKEOp}
\begin{split}
  L &= 2\bigl(D_{\hat g_b}\Ric + \delta_{\hat g_b,\gamma_C}^*\delta_{\hat g_b,\gamma_\Ups}\sfG_{\hat g_b}\bigr), \\
  &\qquad \delta_{\hat g_b,\gamma_C}^* := \delta_{\hat g_b}^* + \gamma_C E_{\rm CD}, \quad
          \delta_{\hat g_b,\gamma_\Ups} := \delta_{\hat g_b} + \gamma_\Ups E_\Ups, \quad
          \sfG_{\hat g_b} := \Id - \frac12\hat g_b\tr_{\hat g_b},
\end{split}
\end{equation}
on weighted 3b-Sobolev spaces; here we define the vector bundle maps
\begin{equation}
\label{EqKECUps}
\begin{split}
  E_{\rm CD}\omega &:= 2\cd_C\otimes_s\omega - \la\cd_C,\omega\ra_{\hat g_b^{-1}}\hat g_b, \\
  E_\Ups h &:= 2\iota_{\cd_\Ups^\sharp}h-\cd_\Ups\tr_{\hat g_b}h
\end{split}
\end{equation}
\end{subequations}
for suitable stationary and compactly supported 1-forms $\cd_C,\cd_\Ups\in\CIc(\hat X_b^\circ;T^*_{\hat X_b^\circ}\hat M_b^\circ)$, and $\gamma_C,\gamma_\Ups\in\R$ are small constants. (The data $\gamma_C,\gamma_\Ups,\cd_C,\cd_\Ups$ are fixed in~\S\S\ref{SsKCD} and \ref{SsKUps} below. Only in~\S\ref{SsKE} do we begin the analysis of $L$; since at that point these data are fixed, we do not adorn $L$ with any subscripts.) The term $E_\Ups$ leads to a modification of the standard generalized wave coordinate condition $\delta_{\hat g_b}\sfG_{\hat g_b}h=0$ in linearized gravity; the term $E_{\rm CD}$ is used for the purpose of constraint damping. (Many other choices of $E_{\rm CD}$ and $E_\Ups$ would work equally well. The present choices are made for similarity with \cite[equation~(8.1)]{HintzVasyKdSStability} and \cite[equation~(3.3)]{HintzVasyMink4}.) Elements $\omega$ in the kernel of the \emph{gauge potential wave operator}
\[
  \Box^\Ups_{\hat g_b,\gamma_\Ups} := 2\delta_{\hat g_b,\gamma_\Ups}\sfG_{\hat g_b}\delta_{\hat g_b}^*
\]
give rise to pure gauge solutions $\delta_{\hat g_b}^*\omega\in\ker L$, while elements $\omega^*$ in the kernel of the adjoint of the \emph{constraint propagation wave operator}
\[
  \Box^{\rm CD}_{\hat g_b,\gamma_C} := 2\delta_{\hat g_b}\sfG_{\hat g_b}\delta_{\hat g_b,\gamma_C}^*
\]
give rise to dual pure gauge solutions $\sfG_{\hat g_b}\delta_{\hat g_b}^*\omega\in\ker L^*$. For $\gamma_C=\gamma_\Ups=0$, both operators are equal to the Hodge and tensor wave operators $\Box_{\hat g_b}$ on 1-forms; this uses $\Ric(\hat g_b)=0$. Also, $L=\Box_{\hat g_b}$ for $\gamma_C=\gamma_\Ups=0$.

We study $L$, $\Box^\Ups_{\hat g_b,\gamma_\Ups}$, and $\Box^{\rm CD}_{\hat g_b,\gamma_C}$ via their spectral families. \emph{We define $\hat L(\sigma)$ etc.\ to be the spectral families with respect to the time function $\hat t$ on $\hat M_b\subset\cM$.} The plan for this section is as follows.

\begin{itemize}
\item Following preliminary observations on structural and symbolic properties of $L$, $\Box_{\hat g_b,\gamma_\Ups}^\Ups$, and $\Box^{\rm CD}_{\hat g_b,\gamma_C}$, we implement \emph{constraint damping} \cite{BrodbeckFrittelliHubnerReulaSCP,GundlachCalabreseHinderMartinConstraintDamping} for $L$ via a careful choice of $\cd_C,\gamma_C$; see Proposition~\ref{PropKCD} in~\S\ref{SsKCD}. Our proof introduces a simpler and cleaner approach to the relevant perturbation theory compared to \cite[\S{10}]{HaefnerHintzVasyKerr}, and it is moreover carried out for the full subextremal range of Kerr parameters. Additionally, we construct dual pure gauge potentials which play an important role in the low energy analysis of $L$; see Lemma~\ref{LemmaKCDPot}.
\item In~\S\ref{SsKUps}, we use the same approach to choose \emph{gauge modifications} $\cd_\Ups,\gamma_\Ups$, as indicated (but not carried out) in~\cite[Remark~10.14]{HaefnerHintzVasyKerr}; see Proposition~\ref{PropKUps}. We moreover construct pure gauge potentials whose symmetric gradients (i.e.\ deformation tensors) appear in the low energy analysis of $L$; see Lemma~\ref{LemmaKGaugePot}.
\item In~\S\ref{SsKE}, we prove the mode stability of $L$ at nonzero frequencies following \cite{AnderssonHaefnerWhitingMode} (see Proposition~\ref{PropKENon0}), and give a detailed description of the low energy resolvent; the second order pole exhibited in \cite[Theorem~11.5]{HaefnerHintzVasyKerr} will arise quite directly (in a weak form) from a suitable Grushin problem setup. See Proposition~\ref{PropKELo}.
\item In~\S\ref{SsKE2}, we turn the description of the resolvent into a description of forward solutions of $L h=f$ when $f$ decays faster than $r^{-2}$ (roughly speaking); see Theorem~\ref{ThmKEFwd}.
\item For our gluing application, it is crucial to have control on $h$ also when $f$ has slightly less decay; this is achieved in~\S\ref{SsKEL} (see Theorem~\ref{ThmKEFwdW}).
\item In~\S\ref{SsKHi}, we prove higher regularity for solutions of $L h=f$ (see Theorem~\ref{ThmKHi}).
\end{itemize}

We begin with the preliminary observations.

\begin{lemma}[Structure of $L$]
\label{LemmaKStruct}
  Recall that $\rho_\cD\in\CI(\hat M_b)$ is a defining function of $\cD\subset\hat M_b$ (see~\eqref{SssN3b}). Define $\ubar L:=2(D_{\hat{\ubar g}}\Ric+\delta^*_{\hat{\ubar g}}\delta_{\hat{\ubar g}}\sfG_{\hat{\ubar g}})$. Then
  \begin{equation}
  \label{EqKStruct}
    L\in\rho_\cD^2\Diff_{\tbop,\rm I}^2(\hat M_b;S^2\,\Ttsc^*\hat M_b),\qquad
    L-\ubar L\in\rho_\cD^3\Diff_{\tbop,\rm I}^2(\hat M_b;S^2\,\Ttsc^*\hat M_b).
  \end{equation}
  The analogous memberships are true for $\Box^\Ups_{\hat g_b,\gamma_\Ups}$, $\Box^{\rm CD}_{\hat g_b,\gamma_\cC}$, and their differences with $\ubar\Box=\Box_{\hat{\ubar g}}$ as operators acting on sections of $\Ttsc^*\hat M_b$.
\end{lemma}
\begin{proof}
  The membership of $L$ follows from \cite[(4.3)]{HaefnerHintzVasyKerr}, where we note that the bundle $\wt{\Tsc}{}^*X$ in the reference (see also \cite[(3.24)]{HaefnerHintzVasyKerr}) is the same as $\Ttsc^*_{\hat X_b}\hat M_b$ in present notation: both bundles have as frames the differentials of the standard coordinates on $\hat M_b^\circ\subset\R^4_{\hat t,\hat x}$. The fact that $L$ agrees with the Minkowski model $\ubar L$ to leading order at $r=\infty$ similarly follows from \cite[Equation~(3.43) until Lemma~3.4]{HaefnerHintzVasyKerr}.

  A more efficient proof proceeds as follows. Let $\hat z=(\hat t,\hat x)$. First, using that $\pa_{\hat z^\mu}\in\Vtsc(\hat M_b)$ maps $\CI(\hat M_b)\to\rho_\cD\CI(\hat M_b)$, one finds (by inspecting the formula for the Christoffel symbols) that the Levi-Civita connection of $\hat g_b$ (and indeed of any smooth 3sc-metric on $\hat M_b$) defines an element
  \[
    \nabla^{\hat g_b} \in \rho_\cD\Diff_\tbop^1(\hat M_b;\Ttsc\hat M_b,\Ttsc^*\hat M_b\otimes\Ttsc\hat M_b).
  \]
  Furthermore, since $\hat g_b-\hat{\ubar g}\in\rho_\cD\CI(\hat M_b;S^2\,\Ttsc^*\hat M_b)$, one finds that
  \[
    \nabla^{\hat g_b}-\nabla^{\hat{\ubar g}} \in \rho_\cD^2\CI(\hat M_b;\End(\Ttsc\hat M_b)).
  \]
  The analogous conclusions hold for the covariant derivative acting on any tensor product of $\Ttsc\hat M_b$ and $\Ttsc^*\hat M_b$. The proof of~\eqref{EqKStruct} is now straightforward.
\end{proof}

The operator $\ubar L=\Box_{\hat{\ubar g}}$ is the tensor wave operator on symmetric 2-tensors on Minkowski space, and thus in the trivialization of $S^2\,\Ttsc^*\hat M_b$ induced by the differentials of the standard coordinate functions given by 10 copies of the scalar wave operator. Therefore, the transition face normal operators of $L$, and likewise those of $\Box^\Ups_{\hat g_b,\gamma_\Ups}$ and $\Box^{\rm CD}_{\hat g_b,\gamma_C}$, are copies of that of the scalar wave operator on Minkowski space. In the terminology of \citeII{Definition~\ref*{DefEstInvtf}}, we thus have by \cite[Lemma~3.20]{HintzKdSMS}:

\begin{lemma}[Invertibility of transition face normal operators]
\label{LemmaKtfInv}
  The transition face normal operators of $L$, resp.\ $\Box^\Ups_{\hat g_b,\gamma_\Ups}$ and $\Box^{\rm CD}_{\hat g_b,\gamma_C}$ are equal to those of $\Box_{\hat{\ubar g}}$ on symmetric 2-tensors, resp.\ 1-forms, and they are invertible at all weights $\alpha_\cD\in(-\frac32,-\frac12)$.
\end{lemma}

The spectral theory of \citeII{Proposition~\ref*{PropEstFThi}} is immediately applicable to the wave operators of interest here in view of the following result.

\begin{lemma}[Subprincipal symbol at trapping]
\label{LemmaKSpecEst}
  The condition on the subprincipal symbol at the trapped set stated in \citeII{(\ref*{EqEstFTSubpr})} is satisfied for the tensor wave operator on any tensor bundle\footnote{We shall need this result only for the bundles of 1-forms and symmetric 2-tensors.} on the subextremal Kerr spacetime $(\hat M_b,\hat g_b)$. It is also satisfied, for fixed $\cd_C,\cd_\Ups$, for the operators $L$, $\Box_{\hat g_b,\gamma_\Ups}^\Ups$, $\Box_{\hat g_b,\gamma_C}^{\rm CD}$ when $\gamma_\Ups,\gamma_C$ are sufficiently small in absolute value.
\end{lemma}
\begin{proof}
  This uses the results of \citeII{\S\ref*{SssGlDynTr}}. We first consider the wave operator on 1-forms whose subprincipal operator is $-i$ times the pullback of the Levi-Civita connection along the base projection $\pi\colon T^*\hat M_b^\circ\to\hat M_b^\circ$, i.e.\ $S_{\rm sub}(\Box_{\hat g_b})=-i\nabla_{H_{\hat G_b}}^{\pi^*T^*\hat M_b^\circ}$, as shown in general in \cite[Proposition~4.1]{HintzPsdoInner}. But then \citeII{Proposition~\ref*{PropGlDynTrNabla}} gives a stationary frame (see also \citeII{Lemma~\ref*{LemmaGlDynTrFrame}}) in which
  \[
    S_{\rm sub}(\Box_{\hat g_b})=-i H_{\hat G_b}+A
  \]
  over the trapped set, where $A$ is a constant nilpotent matrix (in Jordan block form). Since tensor products of $A$ with itself (and with its adjoint) are then also nilpotent, we conclude that for the wave operator on any tensor bundle on Kerr there exists a stationary frame of the bundle so that again $S_{\rm sub}(\Box_{\hat g_b})=-i H_{\hat G_b}+A$ over the trapped set, with $A$ a constant nilpotent matrix of the appropriate size. But since $-i H_{\hat G_b}$ is formally self-adjoint with respect to any constant (in this frame) inner product, it remains to observe that for every nilpotent matrix $A\in\C^{N\times N}$ and for all $\delta>0$, one can find a positive definite inner product on $\C^N$ so that $\|\frac{1}{2 i}(A-A^*)\|<\delta$; this is shown in \cite[\S{3.4}]{HintzPsdoInner}. In particular, condition \citeII{(\ref*{EqEstFTSubpr})} holds for all sufficiently small $\delta$.

  The second statement follows from the openness of the condition~\citeII{(\ref*{EqEstFTSubpr})}.
\end{proof}

We next record that the threshold quantities of \citeII{Definition~\ref*{DefEstAdm}} are
\[
  \vartheta_{\rm in}=\vartheta_{\rm out}=0\ \text{for}\ L,\ \Box^\Ups_{\hat g_b,\gamma_\Ups},\ \Box^{\rm CD}_{\hat g_b,\gamma_C}.
\]
Indeed, for their computation only the leading order terms of $L$ etc.\ (as elements of $\rho_\cD^2\Difftb^2$) are relevant, so it suffices to compute these quantities for the (tensor) wave operator $\Box_{\hat{\ubar g}}$ on Minkowski space $(\R^4_{\hat t,\hat x},\hat{\ubar g})$; but in the bundle splitting induced by the differentials of the standard coordinates, this operator consists of copies of the scalar wave operator and is thus formally symmetric with respect to the diagonal Euclidean fiber inner product (which is positive definite). We shall not compute the threshold quantity $\vartheta_{\cH^+}$ from \citeII{Definition~\ref*{DefEstAdm}} here; instead we simply fix
\begin{equation}
\label{EqKbartheta}
  \bar\vartheta := \vartheta_{\cH^+}+1,\qquad \bar s:=\frac12(1+\bar\vartheta),
\end{equation}
where $\vartheta_{\cH^+}$ is the maximum of this threshold quantity for $\Box_{\hat g_b}$ on 1-forms and symmetric 2-tensors. The addition of $1$ in the definition of $\bar\vartheta$ ensures that $\bar\vartheta$ exceeds the threshold quantity also for $L$, $\Box^\Ups_{\hat g_b,\gamma_\Ups}$, and $\Box^{\rm CD}_{\hat g_b,\gamma_C}$ for fixed $\cd_\Ups,\cd_C$ when $\gamma_\Ups,\gamma_C$ are sufficiently small.

\begin{definition}[Strongly Kerr-admissible orders]
\label{DefKAdm}
  We call $\sfs\in\CI(\Stb^*\hat M_b)$, $\alpha_\cD\in\R$ \emph{strongly Kerr-admissible} if the Kerr-admissibility properties of \citeII{Definition~\ref*{DefEstAdm}} hold for $\sfs,\alpha_\cD$ for $\vartheta_{\rm in}=0$, $\vartheta_{\rm out}=0$, and with $\bar\vartheta$ in place of the quantity denoted $\vartheta_{\cH^+}$ there.
\end{definition}

As far as the threshold conditions are concerned, this means that $\sfs+\alpha_\cD>-\frac12$ at the incoming and $\sfs+\alpha_\cD<-\frac12$ and the outgoing radial set, and $\sfs>\bar s$ at the radial set over the event horizon; and $\sfs$ is independent of $\hat t$.

%%%%%%%%%%%%%%%%%%%%%%%%%%%%%%%%%%%%%%%%%%%%%%%%%%
\subsection{Constraint damping and dual pure gauge potentials}
\label{SsKCD}

Constraint damping a\-mounts to choosing the modification $E_{\rm CD}$ in~\eqref{EqKECUps} in such a way that mode stability holds for $\Box^{\rm CD}_{\hat g_b,\gamma_C}$ in the closed upper half plane, including at zero energy. We first recall the properties of the (unmodified) 1-form wave operator $\Box_{\hat g_b}=\Box^{\rm CD}_{\hat g_b,0}$. In the Boyer--Lindquist coordinates used in~\eqref{EqNKerrMet}, we introduce the functions
\[
  \hat t_j := \hat t_{\rm BL}+T_j(\hat r),\qquad
  \phi_j := \phi+\Phi_j(\hat r)
\]
for $j=0,1$. Here $T_0(\hat r)=\hat r_*:=\int_{4\bhm}^{\hat r}\frac{\hat r^2+a^2}{\mu}\,\dd\hat r$ and $\Phi_0(\hat r)=\int_{4\bhm}^{\hat r}\frac{a}{\mu}\,\dd\hat r$, so $\hat t_0,\hat r,\theta,\phi_0$ are ingoing Kerr-star coordinates, and $\hat t_1,\phi_1$ are the functions $t_*$, $\phi_*$ introduced in \cite[\S{3.1}]{AnderssonHaefnerWhitingMode}: they are defined using $T_1,\Phi_1$ which satisfy $T_1(\hat r)=\hat r_*$ for $\hat r\leq 3\bhm$ and $T_1(\hat r)=-\hat r_*$ for $\hat r\geq 4\bhm$, whereas $\Phi_1=\int\frac{a}{\mu}\,\dd\hat r$ for $\hat r\leq 3\bhm$ and $\Phi_1=0$ for $\hat r\geq 4\bhm$. In particular, $\hat t_j,\hat r,\theta,\phi_j$ are ingoing Kerr-star coordinates near the event horizon (with $\hat t_j-\hat t$ and $\phi_j-\phi_*$ being smooth across $\hat r=\hat r_b$), and $\hat t_1,\hat r,\theta,\phi_1$ are outgoing Kerr-star coordinates near $\hat r=\infty$. See Figure~\ref{FigKCDTime}. We furthermore record the relationship
\begin{equation}
\label{EqKCDStarCoord}
  \hat t_1 = \hat t + \cT_1(\hat r),\qquad \cT_1\in\CI([\bhm,\infty));\quad \cT_1(\hat r)=-\hat r_*\ \text{for large}\ \hat r.
\end{equation}

\begin{figure}[!ht]
\centering
\includegraphics{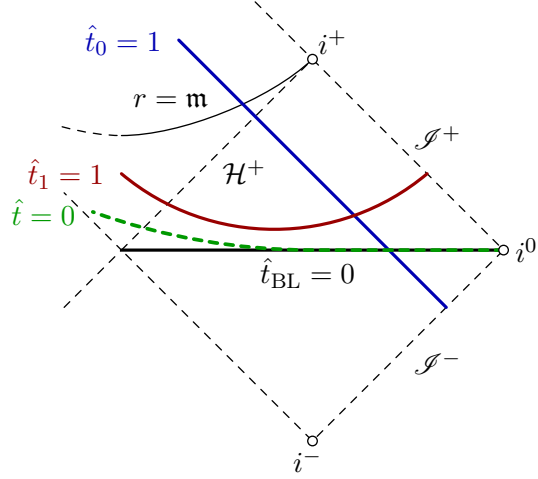}
\caption{Illustration of the time functions $\hat t$ (green), $\hat t_{\rm BL}$ (black), $\hat t_0$ (blue), and $\hat t_1$ (red) inside of a Penrose diagram of Kerr. As usual, $\cH^+$ is the future event horizon, $\scri^\pm$ is future/past null infinity, $i^\pm$ is future/past timelike infinity, and $i^0$ is spacelike infinity.}
\label{FigKCDTime}
\end{figure}

\begin{thm}[Spectral theory for the 1-form wave operator]
\label{ThmKCDUnmod}
  Let $\sfs,\alpha_\cD$ be strongly Kerr-admissible orders (or more generally Kerr-admissible orders for the 1-form wave operator $\Box_{\hat g_b}$), and suppose that $\alpha_\cD\in(-\frac32,-\frac12)$.
  \begin{enumerate}
  \item\label{ItKCDUnmod0} The operator
    \begin{align*}
      \wh{\Box_{\hat g_b}}(0) &\colon \Bigl\{ \omega\in\Hb^{\sfs,\alpha_\cD}(\hat X_b;\Ttsc^*_{\hat X_b}\hat M_b) \colon \wh{\Box_{\hat g_b}}(0)\omega\in\Hb^{\sfs-1,\alpha_\cD+2}(\hat X_b;\Ttsc^*_{\hat X_b}\hat M_b) \Bigr\} \\
        &\qquad \to \Hb^{\sfs-1,\alpha_\cD+2}(\hat X_b;\Ttsc^*_{\hat X_b}\hat M_b)
    \end{align*}
    is Fredholm of index $0$. Its kernel is spanned by the 1-form\footnote{Of course $\omega_{\rms 0}$ depends on the Kerr parameters $b$; but we shall not make this explicit in the notation here.}
    \begin{equation}
    \label{EqKCDUnmodOmega}
      \omega_{\rms 0} = \frac{\hat r}{\varrho^2}(\dd\hat t_0-a\,\sin^2\theta\,\dd\phi_0) + \frac{\hat r_b-\hat r}{\mu}\dd\hat r \in \hat r^{-1}\CI(\hat X_b;\Ttsc^*_{\hat X_b}\hat M_b),
    \end{equation}
    which is divergence-free; and the cokernel is spanned by the closed 1-form
    \begin{equation}
    \label{EqKCDUnmodOmegaStar}
      \omega_{\rms 0}^* := \delta(\hat r-\hat r_b)\,\dd\hat r.
    \end{equation}
  \item\label{ItKCDUnmodNon0} For $\sigma\in\C$, $\sigma\neq 0$, $\Im\sigma\geq 0$, the operator
    \begin{equation}
    \label{EqKCDUnmodNon0}
    \begin{split}
      \wh{\Box_{\hat g_b}}(\sigma) &\colon \Bigl\{ \omega\in\Hsc^{\sfs,\sfs+\alpha_\cD}(\hat X_b;\Ttsc^*_{\hat X_b}\hat M_b) \colon \wh{\Box_{\hat g_b}}(\sigma)\omega\in\Hsc^{\sfs-1,\sfs+\alpha_\cD+1}(\hat X_b;\Ttsc^*_{\hat X_b}\hat M_b) \Bigr\} \\
        &\qquad \to \Hsc^{\sfs-1,\sfs+\alpha_\cD+1}(\hat X_b;\Ttsc^*_{\hat X_b}\hat M_b)
    \end{split}
    \end{equation}
    is invertible.
  \end{enumerate}
\end{thm}

\begin{rmk}[Orders]
\label{RmkKCDOrders}
  In part~\eqref{ItKCDUnmod0}, we work with the b-regularity order induced by $\sfs$ according to \citeII{Remark~\ref*{RmkEstAdmInd}}. In part~\eqref{ItKCDUnmodNon0}, we work with the scattering regularity and scattering decay orders induced by $\sfs$ for positive, resp.\ negative frequencies when $\arg\sigma\in[0,\frac{\pi}{4}]$, resp.\ $\arg\sigma\in[\frac{3\pi}{4},\pi]$; and the orders can be arbitrary near $\hat r=\infty$ (but still need to satisfy the threshold condition at the radial set over the event horizon, and the monotonicity along the future null-bicharacteristic flow) when $\arg\sigma\in[\frac{\pi}{4},\frac{3\pi}{4}]$. In this latter case, one can replace $\sfs+\alpha_\cD+1$ in~\eqref{EqKCDUnmodNon0} by $\sfs+\alpha_\cD$.
\end{rmk}

\begin{proof}[Proof of Theorem~\usref{ThmKCDUnmod}]
  At zero energy, this is a re-statement of \cite[Theorem~5.1(2)]{AnderssonHaefnerWhitingMode}. The fact that $\delta_{\hat g_b}\omega_{\rms 0}=0$ is observed towards the end of the proof in the reference.

  For nonzero $\sigma$ with $\Im\sigma\geq 0$, we argue as follows. By \citeII{Proposition~\ref*{PropEstFTbdd}}, the operator $\wh{\Box_{\hat g_b}}(\sigma)$ is Fredholm of index $0$; so to prove its invertibility it suffices to demonstrate that it has dense range. But by \cite[Theorem~5.1(1)]{AnderssonHaefnerWhitingMode}, given $f'\in\CIc(\hat X_b^\circ;\Ttsc^*_{\hat X_b^\circ}\hat M_b^\circ)$, there exists $\omega'\in\Hb^{\infty,q}(\hat X_b;\Ttsc^*_{\hat X_b}\hat M_b)$ with $\Box_{\hat g_b}(e^{-i\sigma\hat t_1}\omega')=e^{-i\sigma\hat t_1}f'$, where $q<-\frac12$. Given $f\in\CIc(\hat X_b^\circ;T^*_{\hat X_b^\circ}\hat M_b^\circ)$, we apply this to $f':=e^{i\sigma\cT_1(\hat r)}f$ and obtain $\omega:=e^{-i\sigma\cT_1(\hat r)}\omega'$ with $\wh{\Box_{\hat g_b}}(\sigma)\omega=e^{i\sigma\hat t}\Box_{\hat g_b}(e^{-i\sigma\hat t}\omega)=f$. But since $\cT_1(\hat r)=-\hat r+\cO(\log\hat r)$, $\omega$ is exponentially decaying as $\hat r\to\infty$ when $\Im\sigma>0$ and thus lies in the domain of~\eqref{EqKCDUnmodNon0}. When $\sigma\in\R\setminus\{0\}$, we note that $\omega'\in\Hsc^{\infty,\sfr}$ for all decay orders $\sfr\in\CI(\ol{\Tsc^*_{\pa\hat X_b}}\hat X_b)$ with $r\leq q$ at the zero section. Therefore, $\omega$ lies in the scattering Sobolev space $\Hsc^{\infty,\sfs+\alpha_\cD}\subset\Hsc^{\sfs,\sfs+\alpha_\cD}$ for all orders $\sfs,\alpha_\cD$ for which $\sfs+\alpha_\cD\leq q<-\frac12$ at the outgoing radial set (i.e.\ the graph of $\sigma\,\dd\hat r$ over $\hat r=\infty$), and therefore in the domain of~\eqref{EqKCDUnmodNon0}. This completes the proof.
\end{proof}

\begin{prop}[Constraint damping]
\label{PropKCD}
  There exists $\cd_C\in\CIc(\hat X_b^\circ;T^*_{\hat X_b^\circ}\hat M_b^\circ)$ so that, for all sufficiently small $\gamma_C>0$, the operator $\wh{\Box^{\rm CD}_{\hat g_b,\gamma_C}}(\sigma)$ is invertible for all $\sigma\in\C$, $\Im\sigma\geq 0$, as an operator between the function spaces of Theorem~\usref{ThmKCDUnmod}, for all strongly Kerr-admissible orders $\sfs,\alpha_\cD$ with $\alpha_\cD\in(-\frac32,-\frac12)$ for which also $\sfs-1,\alpha_\cD$ are strongly Kerr-admissible.
\end{prop}

This can be proved in a manner analogous to \cite[\S{10}]{HaefnerHintzVasyKerr} (where a concrete choice of $\cd_C$ was given in the Schwarzschild case $b=(\bhm,0)$); we give a simpler proof below. The key calculations concern certain $L^2$-pairings:

\begin{lemma}[$L^2$-pairings]
\label{LemmaKCDPair}
   Write $\la\cdot,\cdot\ra$ for the $L^2$-pairing on $[\bhm,\infty)_r\times\Sph^2_{\theta,\phi_*}$ with the spatial volume density $|\dd\hat g_b|_{\hat X_b}|=\varrho^2\sin\theta\,|\dd r\,\dd\theta\,\dd\phi_0|$ of the Kerr metric; we moreover use the (not positive definite, but non-degenerate) fiber inner product on $T^*_{\hat X_b^\circ}\hat M_b^\circ$ induced by $\hat g_b$. Then
  \begin{equation}
  \label{EqKCDPair1}
    \la [\Box_{\hat g_b},\hat t]\omega_{\rms 0},\omega_{\rms 0}^*\ra = 4\pi.
  \end{equation}
  This also holds for $\hat t+F(\hat r)$ in place of $\hat t$, for any smooth function $F$. Moreover, there exists $\cd_C\in\CIc(\hat X_b^\circ;T^*_{\hat X_b^\circ}\hat M_b^\circ)$ so that, in the notation~\eqref{EqKECUps},
  \begin{equation}
  \label{EqKCDPair2}
    \la 2\delta_{\hat g_b}\sfG_{\hat g_b}E_{\rm CD}\omega_{\rms 0},\omega_{\rms 0}^*\ra = 4\pi.
  \end{equation}
\end{lemma}
\begin{proof}
  For brevity, we write $\omega=\omega_{\rms 0}$ and $\omega^*=\omega^*_{\rms 0}$ in this proof, and we omit the subscript `$\hat g_b$'. Fix a function $\chi\in\CI([0,\infty))$ which equals $0$ on $[0,1)$ and $1$ on $[2,\infty)$ and set $\chi_\eps(\rho)=\chi(\frac{\rho}{\eps})$ where $\rho=\hat r^{-1}$. Since $\omega^*=\dd u^*$ where $u^*=H(\hat r-\hat r_b)$, we can then write
  \begin{align*}
    \la[\Box,\hat t]\omega,\omega^*\ra &= \la[\Box,\hat t]\omega,\dd u^*\ra \\
      &= \lim_{\eps\searrow 0} \la\chi_\eps[\Box,\hat t]\omega,\dd u^*\ra \\
      &= \lim_{\eps\searrow 0} \Bigl(\la[\delta,\chi_\eps][\Box,\hat t]\omega,u^*\ra + \la \chi_\eps\delta[\Box,\hat t]\omega,u^*\ra\Bigr).
  \end{align*}
  In the second pairing on the right, we write $\delta[\Box,\hat t]=\delta\Box\hat t-\delta\hat t\Box=\Box\delta\hat t-\delta\hat t\Box=[\Box,\delta\hat t]$ (using that $[\Box,\delta]=0$ since $\Box$ is the Hodge d'Alembertian), so $\delta[\Box,\hat t]\omega=\Box(\delta\hat t\omega)=\Box[\delta,\hat t]\omega$ since $\delta\omega=0$ and $\Box\omega=0$. We thus have
  \begin{align*}
    \la[\Box,\hat t]\omega,\omega^*\ra &= \lim_{\eps\searrow 0} \Bigl( -\la [\Box,\hat t]\omega,[\dd,\chi_\eps]u^*\ra + \la [\delta,\hat t]\omega,\Box\chi_\eps u^*\ra \Bigr) \\
      &= \lim_{\eps\searrow 0} \Bigl( -\la [\Box,\hat t]\omega,[\dd,\chi_\eps]u^*\ra + \la [\delta,\hat t]\omega,[\Box,\chi_\eps]u^*\ra \Bigr)
  \end{align*}
  since $\Box u^*=0$. Due to the presence of a commutator with $\chi_\eps$ in one factor of each pairing, this number only depends on the $r^{-1}$, resp.\ $r^0$, resp.\ $r^3$ leading order term of $\omega$, resp.\ $u^*$, resp.\ the volume density (regarded as a weighted b-density); moreover, it is unchanged by modifications of $\Box$ by operators of class $\rho_\cD^3\Diff_{\tbop,\rm I}^2$. (This follows from~\eqref{EqNDiff3bFT} with $\ell_\cD=-3$, $j=1$, and~\eqref{EqNMinkLderiv}.) We can therefore evaluate it by plugging in the Minkowski 1-form wave operator $\Box$, the 1-form $\omega=\hat r^{-1}\dd\hat t$, $u^*=1$, and the Euclidean volume density $\hat r^2\,|\dd\hat r\,\dd\slg|=\hat r^3\,|\frac{\dd\hat r}{\hat r}\,\dd\slg|$. In this case, we have $[\Box,\hat t]\omega=\Box(\hat t\hat r^{-1})\,\dd\hat t=0$ (where on the right $\Box$ is the scalar wave operator on Minkowski space, and we work in $\hat r>0$)---so the first pairing vanishes---and $[\delta,\hat t]\omega=-\iota_{\nabla\hat t}\omega=\hat r^{-1}$. The remaining integral $4\pi\lim_{\eps\searrow 0}\int_0^\infty \hat r^{-1}\,[\Box,\chi_\eps]1\,\hat r^2\,\dd\hat r$ is easily evaluated to equal $4\pi$. (This also follows from \citeI{Lemma~\ref*{LemmaBgBdyPair}} with $L=\hat r^2\hat\Box(0)=\hat r^2\Delta$ being the rescaled Euclidean Laplacian for which $N(r^2 L,\lambda)=-\lambda^2+\lambda$ and thus $\pa_\lambda N(L,0)=1$.) This establishes~\eqref{EqKCDPair1}. Changing $\hat t$ by adding $F(\hat r)$ to it gives the additional contribution $\la\Box F\omega,\omega^*\ra=\la F\omega,\Box^*\omega^*\ra=0$; the integration by parts does not produce boundary terms here since $\omega^*$ has compact support in $\hat r^{-1}((\bhm,\infty))$.

  Turning to~\eqref{EqKCDPair2}, we note that $\sfG E_{\rm CD}=2\cd_C\otimes_s(-)$, and therefore the left hand side of~\eqref{EqKCDPair2} is equal to $4$ times
  \[
    \la\delta(\cd_C\otimes_s\omega), \dd u^*\ra = \la \cd_C\otimes_s\omega, \delta^*\dd u^*\ra = \la \cd_C, \iota_{\omega^\sharp}\delta^*\dd u^* \ra.
  \]
  The existence of $\cd_C$ is thus equivalent to the nonvanishing of $\delta^*\dd u^*(\omega^\sharp,\cdot)$ as a distribution on $\hat X_b^\circ$. But modulo 1-forms whose coefficients are smooth multiples of $\delta(\hat r-\hat r_b)$, we have $\delta^*\dd u^*\equiv\delta'(\hat r-\hat r_b)\,\dd\hat r^2$; it then remains to note that $\dd\hat r(\omega^\sharp)=\hat g_b^{-1}(\dd\hat r,\omega)=\frac{2\hat r_b}{\varrho^2}$, which is nonzero at $\hat r=\hat r_b$.
\end{proof}

\begin{proof}[Proof of Proposition~\usref{PropKCD}]
  We fix $\cd_C$ and thus $E_{\rm CD}$ using Lemma~\ref{LemmaKCDPair}. For better readability, we write $\gamma_C:=\gamma$ and
  \[
    \Box_\gamma := \Box^{\rm CD}_{\hat g_b,\gamma},\qquad
    \Box := \Box_0,\qquad
    {}^{\rm o}\wh\Box(\sigma) := e^{-i\sigma\hat r}\wh\Box(\sigma)e^{i\sigma\hat r},
  \]
  and similarly $\delta=\delta_{\hat g_b}$, ${}^{\rm o}\wh\delta(\sigma)=e^{-i\sigma\hat r}\wh\delta(\sigma)e^{i\sigma\hat r}$. Then $\wh{\Box_\gamma}(\sigma)=\wh\Box(\sigma)+2\gamma\wh\delta(\sigma)\sfG E_{\rm CD}$. We furthermore drop the bundle $\Ttsc^*_{\hat X_b}\hat M_b$ from the notation, and we write $\rho_\cD=\hat r^{-1}$.

  %%%%%%%%%%
  \pfstep{Step 1.~Low energy analysis.} While $\omega_{\rms 0}$ lies in the kernel of $\wh{\Box_0}(0)$, the 1-form $\wh{\Box_0}(\sigma)\omega_{\rms 0}$ only vanishes to the third order at the transition face $\tface\subset(\hat X_b)_\scbtop$ (due to $\omega_{\rms 0}\in\hat r^{-1}\CI(\hat X_b)$ and $\wh{\Box_0}\in\Diffscbt^{2,1,2,0}(\hat X_b)$). For present purposes we need one more order. We thus instead introduce $\omega_{\rms 0}(\sigma):=e^{i\sigma\hat r}\omega_{\rms 0}$ and compute\footnote{For the purposes of this argument, one can equally well use the tortoise coordinate in place of $\hat r$, since in the computation of the decay rate of 1-forms on $(\hat X_b)_\scbtop$ we only use the Minkowskian normal operator (and $\wh\Box(0)\omega_{\rms 0}=0$).}
  \begin{align*}
    e^{-i\sigma\hat r}\wh{\Box_\gamma}(\sigma)\omega_{\rms 0}(\sigma) &= {}^{\rm o}\wh{\Box_\gamma}(\sigma)\omega_{\rms 0} \\
      &= {}^{\rm o}\wh\Box(\sigma)\omega_{\rms 0} + 2\gamma\,{}^{\rm o}\wh\delta(\sigma)\sfG E_{\rm CD}\omega_{\rms 0} \\
      &= \sigma \pa_\sigma{}^{\rm o}\wh\Box(0)\omega_{\rms 0} + 2\gamma\,{}^{\rm o}\wh\delta(0)\sfG E_{\rm CD}\omega_{\rms 0} \\
      &\quad\qquad + \frac{\sigma^2}{2}\pa_\sigma^2{}^{\rm o}\wh\Box(0)\omega_{\rms 0} + 2\gamma\sigma\pa_\sigma{}^{\rm o}\wh\delta(0)\sfG E_{\rm CD}\omega_{\rms 0}.
  \end{align*}
  But since $\pa_\sigma{}^{\rm o}\wh\Box(0)\equiv\pa_\sigma{}^{\rm o}\wh{\Box_{\hat{\ubar g}}}(0)\bmod\rho_\cD^2\Diffb^1(\hat X_b)$ as a consequence of Lemma~\ref{LemmaKStruct} and~\eqref{EqNDiff3bFT}, and since by~\eqref{EqNMinkConj} the operator $\pa_\sigma{}^{\rm o}\wh{\Box_{\hat{\ubar g}}}(0)$ annihilates $\hat r^{-1}$ (which gives the desired gain), the first term is of class $\sigma\rho_\cD^3\CI(\hat X_b)$; and it equals $-i\sigma[\Box,\hat t-\hat r]\omega_{\rms 0}$ by~\eqref{EqNMinkLderiv} (now with time function $\hat t-\hat r$). The second term equals $2\gamma\delta\sfG E_{\rm CD}\omega_{\rms 0}$ (and has compact support). The third term is of class $\sigma^2\rho_\cD^2\CI(\hat X_b)$ (now using $\pa_\sigma^2{}^{\rm o}\wh\Box(0)\in\rho_\cD\Diffb^0$). The final term is $\gamma\sigma$ times a compactly supported 1-form.

  We now introduce
  \[
    \lambda := |(\sigma,\gamma)|,\qquad
    (\hat\sigma,\hat\gamma)=\lambda^{-1}(\sigma,\gamma),
  \]
  and define
  \begin{align}
    f(\lambda,\hat\sigma,\hat\gamma) &:= \wh{\Box_\gamma}(\sigma)\bigl( \lambda^{-1}\omega_{\rms 0}(\sigma)\bigr) \nonumber\\
  \label{EqKCDf}
      &= e^{i\sigma\hat r}\bigl ( -i\hat\sigma[\Box,\hat t-\hat r]\omega_{\rms 0} + 2\hat\gamma\delta\sfG E_{\rm CD}\omega_{\rms 0} \bigr) + f'(\lambda,\hat\sigma,\hat\gamma), \\
      &\quad \hspace{4em}  e^{-i\sigma\hat r}f'\in \hat\sigma\sigma\rho_\cD^2\CI(\hat X_b)+\hat\gamma\sigma\CIc(\hat X_b^\circ). \nonumber
  \end{align}
  While initially defined only for $\lambda\neq 0$, this extends by continuity to $\lambda=0$. We now set up a Grushin problem by introducing the rank $1$ augmentation
  \begin{equation}
  \label{EqKCDGrushin}
    \wt{\Box}(\lambda,\hat\sigma,\hat\gamma) := \begin{pmatrix} \wh{\Box_\gamma}(\sigma) & f(\lambda,\hat\sigma,\hat\gamma) \\ \la\cdot,f^*\ra & 0 \end{pmatrix}
  \end{equation}
  where we fix any $f^*\in\CIc(\hat X_b^\circ)$ so that $\la\omega_{\rms 0},f^*\ra=1$. Introduce the norm
  \[
    \| (u,c) \|_{\tilde H_{\scbtop,|\sigma|}^{\sfs,\sfr,\alpha_\cD,0}} := \|u\|_{H_{\scbtop,|\sigma|}^{\sfs,\sfr,\alpha_\cD,0}} + |c|
  \]
  for 1-forms $u$ and complex numbers $c$; we fix $\sfr=\sfs+\alpha_\cD$. We shall establish the existence of $\lambda_0>0$ and $C$ so that
  \begin{equation}
  \label{EqKCDEst}
    \Im\hat\sigma\geq 0,\ \hat\gamma\geq 0,\ 0\leq\lambda<\lambda_0 \implies \|(u,c)\|_{\tilde H_{\scbtop,|\sigma|}^{\sfs,\sfr,\alpha_\cD,0}} \leq C\| \wt\Box(\lambda,\hat\sigma,\hat\gamma)(u,c) \|_{\tilde H_{\scbtop,|\sigma|}^{\sfs-1,\sfr+1,\alpha_\cD+2,0}}.
  \end{equation}
  We shall mimic the proof of \citeII{Proposition~\ref*{PropEstFTlo}}. Importantly, rather than use the invertibility of the zero energy operator $\wt\Box(\lambda',0,\hat\gamma)$ for fixed $\hat\gamma$ and small $\lambda'>0$,\footnote{This invertibility is a consequence of our arguments below, and was already proved in \cite[Proposition~10.3]{HaefnerHintzVasyKerr}. However, it would only allow one to prove the estimate in~\eqref{EqKCDEst} for each $\gamma>0$ for $|\sigma|$ smaller than a $\gamma$-dependent constant which may shrink to $0$ as $\gamma$ becomes small.} we perturb off of $\wt\Box(0,\hat\sigma,\hat\gamma)$.

  %%%%%%%%%%
  \pfsubstep{(1.1.)}{Symbolic and transition face estimate.} Symbolic estimates (microlocal elliptic regularity, radial point estimates at the incoming and outgoing radial sets at infinity, radial point estimates at the conormal bundle of $\hat r=\hat r_b$, and real principal type propagation) for $\wh{\Box_\gamma}(\sigma)$ in sc-b-transition Sobolev spaces together with the invertibility of the transition face normal operators (Lemma~\ref{LemmaKtfInv}) give, as in \citeII{(\ref*{EqEstFTloEst2})}, for all sufficiently small\footnote{The smallness of $\gamma$ only enters since we \emph{fixed} the orders $\sfs,\alpha_\cD$ for the operator $\Box$; these orders are Kerr-admissible for $\Box_\gamma$ when $\gamma$ is smaller than some absolute constant $\gamma_0$ in absolute value, and so are $\sfs-1,\alpha_\cD$. (The fact that one can take $\gamma_0>0$ to be independent of $\sfs,\alpha_\cD$ is due to the \emph{strong} Kerr-admissibility of these orders.)} $\gamma$ and for $\sigma\in\C$, $|\sigma|\leq 1$, $\Im\sigma\geq 0$,
  \[
    \|(u,c)\|_{\tilde H_{\scbtop,|\sigma|}^{\sfs,\sfr,\alpha_\cD,0}} \leq C\Bigl( \|\wh{\Box_\gamma}(\sigma)u\|_{H_{\scbtop,|\sigma|}^{\sfs-1,\sfr+1,\alpha_\cD+2,0}} + \|u\|_{H_{\scbtop,|\sigma|}^{\sfs_0,\sfr_0+1,\alpha_\cD-1,0}}\Bigr) + |c|
  \]
  where we fix Kerr-admissible orders $\sfs_0<\sfs-1$, $\alpha_\cD$ for $\Box$ (which are then also Kerr-admissible for $\Box_\gamma$ when $\gamma$ is sufficiently small); and we set $\sfr_0:=\sfs_0+\alpha_\cD<\sfr-1$.

  We wish to replace $\wh{\Box_\gamma}(\sigma)u$ in this estimate by $\wt\Box(\lambda,\hat\sigma,\hat\gamma)(u,c)$. To this end, note first that by (the discussion leading up to) \eqref{EqKCDf}, the 1-form $e^{-i\sigma\hat r}f(\lambda,\hat\sigma,\hat\gamma)$ is uniformly bounded, for $\lambda\leq 1$ and unit $(\hat\sigma,\hat\gamma)$, in the space $\rho_\cD^3\CI(\hat X_b)+\sigma\rho_\cD^2\CI(\hat X_b)\subset\cA^{2,3,0}((\hat X_b)_\scbtop)$. By Lemma~\ref{LemmaNMult}\eqref{ItNMultVar} (with $\tilde T=0$ and $\alpha=2$, $\beta=3$, $\gamma=0$), this implies that
  \begin{equation}
  \label{EqKCDfMem}
    \text{$f(\lambda,\hat\sigma,\hat\gamma)$ is uniformly bounded in $H_{\scbtop,|\sigma|}^{\sfs-1,\sfr+1,\alpha_\cD+2,0}$}
  \end{equation}
  in view of $\alpha_\cD<-\frac12$ and the upper bound $\sfr<-\frac12$ at the outgoing radial set. (The orders $\sfs,\sfr$ here are induced by the 3b-differential order at positive or negative real frequencies depending on whether $\sigma=\hat\sigma|\sigma|$ with $\hat\sigma\in\exp(i[0,\frac{\pi}{4}])$ or $\hat\sigma\in\exp(i[\frac{3\pi}{4},\pi])$; and the orders can be arbitrary near $\hat r=\infty$ for $\hat\sigma\in\exp(i[\frac{\pi}{4},\frac{3\pi}{4}])$. See Remark~\ref{RmkKCDOrders}.) Since also $\la\cdot,f^*\ra$ is a uniformly bounded linear functional on $H_{\scbtop,|\sigma|}^{\sfs,\sfr,\alpha_\cD,0}$, we therefore obtain, as desired, the estimate
  \begin{equation}
  \label{EqKCDEst0}
    \|(u,c)\|_{\tilde H_{\scbtop,|\sigma|}^{\sfs,\sfr,\alpha_\cD,0}} \leq C\Bigl( \|\wt\Box(\lambda,\hat\sigma,\hat\gamma)(u,c)\|_{\tilde H_{\scbtop,|\sigma|}^{\sfs-1,\sfr+1,\alpha_\cD+2,0}} + \|(u,c)\|_{\tilde H_{\scbtop,|\sigma|}^{\sfs_0,\sfr_0+1,\alpha_\cD-\eta,0}} \Bigr).
  \end{equation}
  Here one can take $\eta=1$, but of course the estimate remains valid for all $\eta\leq 1$; we fix $\eta>0$ so that $\alpha_\cD-\eta\in(-\frac32,-\frac12)$ still.
  
  %%%%%%%%%%
  \pfsubstep{(1.2)}{Estimate for $\lambda=|(\sigma,\gamma)|=0$.} Next, we fix a cutoff $\chi\in\CI((\hat X_b)_\scbtop)$ which equals $1$ near the zero energy face and has support in a collar neighborhood of it; e.g.\ $\chi=\psi(\frac{|\sigma|}{\rho_\cD})$ where $\psi\in\CIc([0,1))$ equals $1$ near $0$. Then
  \begin{equation}
  \label{EqKCDEst05}
    \|(1-\chi)u\|_{H_{\scbtop,|\sigma|}^{\sfs_0,\sfr_0+1,\alpha_\cD-\eta,0}}\leq C_N\|u\|_{H_{\scbtop,|\sigma|}^{\sfs_0,\sfr_0+1,\alpha_\cD-\eta,-N}}
  \end{equation}
  for every $N$. We proceed to estimate $(\chi u,c)$. Recall the uniform norm equivalence
  \[
    \|(\chi u,c)\|_{\tilde H_{\scbtop,|\sigma|}^{\sfs_0,\sfr_0+1,\alpha_\cD-\eta,0}}\sim \|(\chi u,c)\|_{\tilde H_\bop^{\sfs_0,\alpha_\cD-\eta}}:=\|\chi u\|_{\Hb^{\sfs_0,\alpha_\cD-\eta}(\hat X_b)}+|c|
  \]
  from~\citeII{(\ref*{EqF3scbtNormzf})}. We claim that, for $\hat\sigma,\hat\gamma$ with $|(\hat\sigma,\hat\gamma)|=1$, $\Im\hat\sigma\geq 0$, and $\hat\gamma\geq 0$,
  \begin{equation}
  \label{EqKCDEst1}
  \begin{split}
    \|(v,c)\|_{\tilde H_\bop^{\sfs_0,\alpha_\cD-\eta}} &\leq C \|\wt\Box(0,\hat\sigma,\hat\gamma)(v,c)\|_{\tilde H_\bop^{\sfs_0,-1,\alpha_\cD-\eta+2}} \\
      &= C\left\| \begin{pmatrix} \wh\Box(0) & f(0,\hat\sigma,\hat\gamma) \\ \la\cdot,f^*\ra & 0 \end{pmatrix}\begin{pmatrix}v\\c\end{pmatrix}\right\|_{\tilde H_\bop^{\sfs_0-1,\alpha_\cD-\eta+2}}
  \end{split}
  \end{equation}
  for a constant $C$ which does not depend on $\hat\sigma,\hat\gamma$. To show this, note that $\wh\Box(0)$ is Fredholm of index $0$ as a map $\{v\in\Hb^{\sfs_0,\alpha_\cD-\eta}\colon\wh\Box(0)v\in\Hb^{\sfs_0-1,\alpha_\cD-\eta+2}\}\to\Hb^{\sfs_0-1,\alpha_\cD-\eta+2}$ due to Theorem~\ref{ThmKCDUnmod} and our choice of $\eta$, and thus also $\wt\Box(0,\hat\sigma,\hat\gamma)$ is Fredholm of index $0$. The estimate~\eqref{EqKCDEst1} thus follows once we establish the injectivity of $\wt\Box(0,\hat\sigma,\hat\gamma)$; but $\wt\Box(0,\hat\sigma,\hat\gamma)(v,c)=0$ implies $\wh\Box(0)v+f(0,\hat\sigma,\hat\gamma)c=0$, which upon pairing with $\omega_{\rms 0}^*$ gives, in view of~\eqref{EqKCDf} (note in particular the vanishing of $f'$ at $\sigma=0$ and thus at $\lambda=0$) and using Lemma~\ref{LemmaKCDPair},
  \[
    0 = \bigl(-i\hat\sigma\la[\Box,\hat t-\hat r]\omega_{\rms 0},\omega_{\rms 0}^*\ra + \hat\gamma \la 2\delta\sfG E_{\rm CD}\omega_{\rms 0},\omega_{\rms 0}^*\ra\bigr)c = 4\pi(-i\hat\sigma+\hat\gamma)c.
  \]
  Since $-i\hat\sigma+\hat\gamma$ is nonzero when $\Im\hat\sigma\geq 0$, $\hat\gamma\geq 0$, $|(\hat\sigma,\hat\gamma)|=1$, we conclude that $c=0$, and thus $\wh\Box(0)v=0$, so $v=c'\omega_{\rms 0}$; but then $0=\la v,f^*\ra=c'\la\omega_{\rms 0},f^*\ra=c'$ implies $v=0$.

  %%%%%%%%%%
  \pfsubstep{(1.3)}{Estimating the difference of $\wt\Box(\lambda,\hat\sigma,\hat\gamma)$ and $\wt\Box(0,\hat\sigma,\hat\gamma)$.} Applying~\eqref{EqKCDEst1} to $v=\chi u$, we next wish to replace $\wh\Box(0,\hat\sigma,\hat\gamma)$ on the right by $\wh\Box(\lambda,\hat\sigma,\hat\gamma)$, and to this end need to estimate the resulting error $\|(\wt\Box(\lambda,\hat\sigma,\hat\gamma)-\wt\Box(0,\hat\sigma,\hat\gamma))(\chi u,c)\|_{\tilde H_\bop^{\sfs_0-1,\alpha_\cD-\eta+2}}$. We have
  \[
    \bigl(\wt\Box(\lambda,\hat\sigma,\hat\gamma) - \wt\Box(0,\hat\sigma,\hat\gamma)\bigr) \circ (\chi\times\Id_\C) = \begin{pmatrix} A_{1 1} & f'(\lambda,\hat\sigma,\hat\gamma) \\ 0 & 0 \end{pmatrix}
  \]
  where
  \begin{align*}
    A_{1 1} &:= \bigl(\wh{\Box_{\lambda\hat\gamma}}(\lambda\hat\sigma)-\wh{\Box_0}(0)\bigr)\chi = \bigl(\wh\Box(\lambda\hat\sigma) - \wh\Box(0)\bigr)\chi + 2\lambda\hat\gamma\wh\delta(\lambda\hat\sigma)\sfG E_{\rm CD}\chi \\
    &= \lambda\hat\sigma\pa_\sigma\wh\Box(0)\chi + \frac{\lambda^2\hat\sigma^2}{2}\pa_\sigma^2\wh\Box(0)\chi + 2\lambda\hat\gamma\wh\delta(0)\sfG E_{\rm CD}\chi + 2\lambda^2\hat\gamma\hat\sigma \pa_\sigma\wh\delta(0)\sfG E_{\rm CD}\chi.
  \end{align*}
  Consider the first term of $A_{1 1}$. Write
  \begin{equation}
  \label{EqKCDWeights}
    \lambda\hat\sigma=\sigma=\Bigl(\Bigl|\frac{\sigma}{\lambda}\Bigr|^\eta\frac{\sigma}{|\sigma|}\Bigr)\cdot\lambda^\eta|\sigma|^{1-\eta},
  \end{equation}
  with the prefactor bounded in absolute value by $1$; then, in view of $\pa_\sigma\wh\Box(0)\in\rho_\cD\Diffb^1(\hat X_b)$, we can estimate
  \begin{align*}
    \|\lambda\hat\sigma\pa_\sigma\wh\Box(0)(\chi u)\|_{\Hb^{\sfs_0-1,\alpha_\cD-\eta+2}} &\leq C\lambda^\eta \| |\sigma|^{1-\eta}\pa_\sigma\wh\Box(0)(\chi u)\|_{H_{\scbtop,|\sigma|}^{\sfs_0-1,-N,\alpha_\cD-\eta+2,0}} \\
      &\leq C\lambda^\eta \|\chi u\|_{H_{\scbtop,|\sigma|}^{\sfs_0,-N-1,(\alpha_\cD-\eta+2)-(1-\eta+1),-1+\eta}} \\
      &\leq C\lambda^\eta \|\chi u\|_{H_{\scbtop,|\sigma|}^{\sfs_0,-N-1,\alpha_\cD,0}}
  \end{align*}
  where in the last step we gave up the improvement at $\zface\subset(\hat X_b)_\scbtop$ (since the presence of the small prefactor $\lambda^\eta$ is already sufficient in the sequel). The same bound holds for $\|\lambda^2\hat\sigma^2\pa_\sigma^2\wh\Box(0)\chi u\|_{\Hb^{\sfs_0-1,\alpha_\cD-\eta+2}}$ (with $\lambda^2\hat\sigma^2=(\frac{|\sigma|}{\lambda})^\eta(\frac{\sigma}{|\sigma|})^2\cdot\lambda^\eta|\sigma|^{2-\eta}$ and $\pa_\sigma^2\wh\Box(0)\in\Diffb^0(\hat X_b)$). The norms corresponding to the terms involving $E_{\rm CD}$ are bounded by $C\lambda\|\chi u\|_{H_{\scbtop,|\sigma|}^{\sfs_0,-N,-N,0}}$ due to the compact support property of the 1-form $\cd_C$.
  
  Furthermore, for $f'(\lambda,\hat\sigma,\hat\gamma)\in \hat\sigma\lambda\hat\sigma\rho_\cD^2\CI(\hat X_b)+\hat\gamma\lambda\hat\sigma\CIc(\hat X_b^\circ)$ we again use~\eqref{EqKCDWeights} to deduce that $\lambda^{-\eta}f'$ is uniformly bounded in $\cA^{2,3-\eta,0}((\hat X_b)_\scbtop)$ (again giving up the improvement at $\zface$). Since $3-\eta>(\alpha_\cD-\eta+2)+\frac32$, we have thus established
  \[
    \|(\wt\Box(\lambda,\hat\sigma,\hat\gamma)-\wt\Box(0,\hat\sigma,\hat\gamma))(\chi u,c)\|_{\tilde H_\bop^{\sfs_0-1,\alpha_\cD-\eta+2}} \leq C\lambda^\eta\|(\chi u,c)\|_{\tilde H_\scbtop^{\sfs_0,-N,\alpha_\cD,0}}.
  \]
  Altogether, we thus deduce from~\eqref{EqKCDEst1} the estimate
  \begin{equation}
  \label{EqKCDEst11}
    \|(\chi u,c)\|_{\tilde H_\scbtop^{\sfs_0,\sfr_0+1,\alpha_\cD-\eta,0}} \leq C\Bigl( \|\wt\Box(\lambda,\hat\sigma,\hat\gamma)(\chi u,c)\|_{\tilde H_{\scbtop,|\sigma|}^{\sfs_0-1,-N,\alpha_\cD-\eta+2,0}} + \lambda^\eta\|(\chi u,c)\|_{\tilde H_{\scbtop,|\sigma|}^{\sfs_0,-N,\alpha_\cD,0}}\Bigr).
  \end{equation}

  %%%%%%%%%%
  \pfsubstep{(1.4)}{Commuting through the cutoff $\chi$; conclusion.} We now estimate
  \begin{equation}
  \label{EqKCDEst2}
  \begin{split}
    \|\wt\Box(\lambda,\hat\sigma,\hat\gamma)(\chi u,c)\|_{\tilde H_{\scbtop,|\sigma|}^{\sfs_0-1,-N,\alpha_\cD-\eta+2,0}} &\leq \|\wt\Box(\lambda,\hat\sigma,\hat\gamma)(u,c)\|_{\tilde H_{\scbtop,|\sigma|}^{\sfs_0-1,-N,\alpha_\cD-\eta+2,0}} \\
    &\quad\quad + \|[\wt\Box(\lambda,\hat\sigma,\hat\gamma),\chi\times\Id_\C](u,c)\|_{\tilde H_{\scbtop,|\sigma|}^{\sfs_0-1,-N,\alpha_\cD-\eta+2,0}}.
  \end{split}
  \end{equation}
  We have
  \[
    [\wt\Box(\lambda,\hat\sigma,\hat\gamma),\chi\times\Id_\C] = \begin{pmatrix} [\wh{\Box_\gamma}(\sigma),\chi] & (1-\chi)f(\lambda,\hat\sigma,\hat\gamma) \\ \la(\chi-1)\cdot,f^*\ra & 0 \end{pmatrix},
  \]
  with $(\chi-1)f^*=0$ for small $|\sigma|$ (and thus for small $\lambda$) in view of $f^*\in\CIc(\hat X_b^\circ)$. The commutator term in~\eqref{EqKCDEst2} is thus bounded by
  \begin{equation}
  \label{EqKCDEst21}
    \|[\wt\Box(\lambda,\hat\sigma,\hat\gamma),\chi\times\Id_\C](u,c)\|_{\tilde H_{\scbtop,|\sigma|}^{\sfs_0-1,-N,\alpha_\cD-\eta+2,0}} \leq \|u\|_{H_{\scbtop,|\sigma|}^{\sfs_0,-N,\alpha_\cD-\eta,0}} + C\lambda^\eta|c|,
  \end{equation}
  where the second term arises from the observation~\eqref{EqKCDfMem} (for the orders $\sfs_0-1,-N$ instead of $\sfs-1,\sfr+1$) and the fact that $|\sigma|^{-(\alpha_\cD+2-\eta)}\leq\lambda^\eta|\sigma|^{-(\alpha_\cD+2)}$, with $|\sigma|$ a defining function of $\tface$ on $\supp(1-\chi)$ (which contains the support of $(1-\chi)f$).

  Combining~\eqref{EqKCDEst0}, \eqref{EqKCDEst05}, \eqref{EqKCDEst11}, \eqref{EqKCDEst2}, and \eqref{EqKCDEst21} establishes the estimate~\eqref{EqKCDEst} except for an additional summand $C\lambda^\eta\|(u,c)\|_{\tilde H_{\scbtop,|\sigma|}^{\sfs_0+1,\sfr_0+1,\alpha_\cD,0}}$ on the right hand side---which in view of $\sfs_0+1<\sfs$, $\sfr_0+1<\sfr$ can be absorbed into the left hand side when $\lambda$ is sufficiently small. This proves~\eqref{EqKCDEst}.\footnote{Geometrically speaking, the proof takes place on the total space $[\halfopen_{|\sigma|}\times\halfopen_\gamma\times\hat X_b;\{|\sigma|=\gamma=0\};\{|\sigma|=0\},\{|\gamma|=0\}]$, on sc-b-transition Sobolev spaces with parameter $|\sigma|$ (i.e.\ relative to b-vector fields tangent to the level sets of $|\sigma|,\gamma$ which vanish at the lifts of $\halfopen\times\halfopen\times\pa\hat X_b$ and $\halfopen\times\{0\}\times\pa\hat X_b$). In addition to symbolic estimates, we use estimates for the $\tface$-model operators at the front faces corresponding to the first and second blow-ups (these are given by $N_\tface(\wh\Box,\frac{\sigma}{|\sigma|})$), and in addition each model operator along the front face corresponding to the final blow-up (given by $\wt\Box(0,\hat\sigma,\hat\gamma)$).}

  By the Fredholm index $0$ property of $\wh{\Box_\gamma}(\sigma)$ for small $\gamma$ and $\sigma\in\C$ with $\Im\sigma\geq 0$ (established for $\sigma\neq 0$ in \citeII{Proposition~\ref*{PropEstFTbdd}}, and following for $\sigma=0$ from Theorem~\ref{ThmKCDUnmod}\eqref{ItKCDUnmod0} and the Fredholm statement in \citeII{Proposition~\ref*{PropEstFT0}})---which implies the corresponding property of $\wt\Box(\lambda,\hat\sigma,\hat\gamma)$---the estimate~\eqref{EqKCDEst} implies the invertibility of $\wt\Box(\lambda,\hat\sigma,\hat\gamma)$ on the direct sum with $\C$ of the spaces in Theorem~\ref{ThmKCDUnmod}. By inspection of the first row of $\wt\Box(\lambda,\hat\sigma,\hat\gamma)$, we thus deduce the surjectivity of $\wh{\Box_\gamma}(\sigma)$ as a map between the spaces in Theorem~\ref{ThmKCDUnmod} for all $\sigma\in\C$, $\Im\sigma\geq 0$, $\gamma\geq 0$, with $0<|(\sigma,\gamma)|<\lambda_0$. In particular, for every $\gamma\in(0,\frac12\lambda_0)$, $\wh{\Box_\gamma}(\sigma)$ is invertible for the \emph{$\gamma$-independent} set of frequencies $\{\sigma\in\C\colon \Im\sigma\geq 0,\ |\sigma|\leq\frac12\lambda_0\}$.

  %%%%%%%%%%
  \pfstep{Step 2.~Conclusion of the proof.} We now argue as in the proof of \cite[Proposition~10.12]{HaefnerHintzVasyKerr}. The above low energy estimates and the high energy estimates for $\wh{\Box_\gamma}(\sigma)$ from \citeII{Proposition~\ref*{PropEstFThi}} show that there exists $c>0$ so that for all $\gamma\in(0,\frac12\lambda_0)$, the operator $\wh{\Box_\gamma}(\sigma)$ is invertible for all $\sigma\in\C$ with $\Im\sigma\geq 0$ and $|\sigma|\leq c$ or $|\sigma|\geq c^{-1}$. But with $c$ thus fixed, we note that in the compact set $\{|\sigma|\in[c,c^{-1}],\ \Im\sigma\geq 0\}$, the invertibility of $\wh{\Box_\gamma}(\sigma)$ for $\gamma=0$ persists for sufficiently small $\gamma$. This completes the proof.
\end{proof}

\begin{rmk}[Simplifications]
  The conceptual simplifications of the present proof compared to \cite[\S{10}]{HaefnerHintzVasyKerr} are twofold. First, we do not need to prove the (partial) differentiability of the resolvent or a relative resolvent. Second, control of the resolvent for small $(\gamma,\sigma)\neq(0,0)$ is established by a variation of the uniform low energy resolvent analysis of \cite[Proposition~3.21]{HintzKdSMS}, as recalled in \citeII{Proposition~\ref*{PropEstFTlo}} and going back, albeit in a different form, to \cite{VasyLowEnergyLAG}.
\end{rmk}

\textit{In the sequel, we fix $\cd_C$ and $\gamma_C>0$ so that the conclusions of Proposition~\usref{PropKCD} hold.}

We next study special elements in the nullspace of the zero energy operator of $(\Box^{\rm CD}_{\hat g_b,\gamma_C})^*=2(\delta_{\hat g_b,\gamma_C}^*)^*\sfG_{\hat g_b}\delta_{\hat g_b}^*$. We shall use that the operator
\begin{equation}
\label{EqKCDBox0}
\begin{split}
  \wh{\Box^{\rm CD}_{\hat g_b,\gamma_C}}(0)^* \colon \bigl\{ u\in\dot H_\bop^{-s+1,-\alpha_\cD-2}(\hat X_b;\Ttsc^*_{\hat X_b}\hat M_b) \colon \wh{\Box^{\rm CD}_{\hat g_b,\gamma_C}}(0)^* u \in \dot H_\bop^{-s,-\alpha_\cD} \bigr\} \to \dot H_\bop^{-s,-\alpha_\cD}
\end{split}
\end{equation}
is invertible for $s>\bar s$ (see~\eqref{EqKbartheta}) and $\alpha_\cD\in(-\frac32,-\frac12)$; here $\dot H_\bop$ denotes spaces of distributions with supported character at $\hat r=\bhm$, i.e.\ the vanish in $\hat r<\bhm$. (The Fredholm property is established as part of the proof of Theorem~\ref{ThmKCDUnmod}; the invertibility then follows by duality from Proposition~\ref{PropKCD}.) By the hyperbolic nature of the operator $\wh{\Box^{\rm CD}_{\hat g_b,\gamma_C}}(0)^*$ in $\hat r<\hat r_b$, elements of its nullspace vanish in $\hat r<\hat r_b$. (They are, moreover, smooth in $\hat r>\hat r_b$ by a propagation of singularities argument, and conormal at $\hat r^{-1}(\infty)$ since this operator is an elliptic weighted b-differential operator there; and they are only singular at the conormal bundle of the event horizon.)

\begin{definition}[Scalar and vector types]
\label{DefKCDScalVect}
  Write geometric operators on $\Sph^2$, equipped with the standard metric $\slg$, as $\sld$ (exterior derivative), $\slstar$ (Hodge star), $\slDelta$ (nonnegative scalar Laplacian). We then write
  \[
    \scalspace_1=\{\sld u\colon u\in\CI(\Sph^2),\ \slDelta u=2 u\},\qquad
    \vectspace_1=\{\slstar\sld u\colon u\in\CI(\Sph^2),\ \slDelta u=2 u\}
  \]
  for the spaces of scalar type $1$, resp.\ vector type $1$ 1-form spherical harmonics.
\end{definition}

Thus $\scalspace_1,\vectspace_1$ are (complex) 3-dimensional vector spaces; explicit parameterizations are given in~\eqref{EqKEL2VectScal} below.

\begin{lemma}[Dual pure gauge potentials]
\label{LemmaKCDPot}
  Let $\chi=\chi(\hat r)$ be equal to $0$ for $\hat r\leq 3\bhm$ and equal to $1$ for $\hat r\geq 4\bhm$. There exist unique stationary 1-forms $\omega_{\rms 0}^*$, $\omega_{\rms 1}^*(\scal)$ $(\scal\in\scalspace_1)$, $\omega_{\rmv 1}^*(\vect)$ $(\vect\in\vectspace_1)$ on $\hat M_b^\circ$, which we equivalently regard as sections of $T^*_{\hat X_b^\circ}\hat M_b^\circ\to\hat X_b^\circ$, with the following properties:
  \begin{enumerate}
  \item\label{ItKCDPotKer} $\omega_{\rms 0}^*$, $\omega_{\rms 1}^*(\scal)$, $\omega_{\rmv 1}^*(\vect)\in\ker\wh{\Box_{\hat g_b,\gamma_C}^{\rm CD}}(0)^*$;
  \item $\omega_{\rms 0}^*-\chi\pa_{\hat t}^\flat\in\Hbsupp^{-\infty,-\frac12-}(\hat X_b;\Ttsc^*_{\hat X_b}\hat M_b)$;\footnote{This difference in fact vanishes in $\hat r<\hat r_b$ as a consequence of the support property of $\chi$ and part~\eqref{ItKCDPotKer}. Furthermore, this difference is smooth in $\hat r>\hat r_b$ and conormal at $\hat r=\infty$ with weight $-\frac12-$, with a pointwise $\cO(\hat r^{-1+})$ bound of its components in $(\hat t,\hat x)$-coordinates by Sobolev embedding. Similar comments apply to parts~\eqref{ItKCDPots1} and \eqref{ItKCDPotv1}.}
  \item\label{ItKCDPots1} $\omega_{\rms 1}^*(\scal)-\chi\dd(\hat r\scal)\in\Hbsupp^{-\infty,-\frac12-}$;
  \item\label{ItKCDPotv1} $\omega_{\rmv 1}^*(\vect)-\chi\hat r^2\vect\in\Hbsupp^{-\infty,-\frac12-}$.
  \end{enumerate}
  Their trace-reversed symmetric gradients
  \begin{equation}
  \label{EqKCDPotDels}
    h_{\rms 0}^* := \sfG_{\hat g_b}\delta_{\hat g_b}^*\omega_{\rms 0}^*,\qquad
    h_{\rms 1}^*(\scal) := \sfG_{\hat g_b}\delta_{\hat g_b}^*\omega_{\rms 1}^*(\scal),\qquad
    h_{\rmv 1}^*(\vect) := \sfG_{\hat g_b}\delta_{\hat g_b}^*\omega_{\rmv 1}^*(\vect)
  \end{equation}
  are elements of $\Hbsupp^{-\infty,\frac12-}(\hat X_b;S^2\,\Ttsc^*_{\hat X_b}\hat M_b)$.
\end{lemma}
\begin{proof}
  The statement and proof are related to (the proof of) \cite[Lemma~10.8]{HaefnerHintzVasyKerr}. Write $\wh\Box(0)^*:=\wh{\Box_{\hat g_b,\gamma_C}^{\rm CD}}(0)^*$ for brevity. Since $\pa_{\hat t}$ is a Killing vector field for $\hat g_b$, we have $\delta_{\hat g_b}^*\pa_{\hat t}^\flat=0$ and thus $\wh\Box(0)^*(\chi\pa_{\hat t}^\flat)=[\wh\Box(0)^*,\chi]\pa_{\hat t}^\flat=:f\in\CIc(\hat X_b^\circ)$, with $\supp f\subset\supp\dd\chi$. Using the invertibility of~\eqref{EqKCDBox0} (with $\alpha_\cD=-\frac12-$), we can then find a unique $\breve\omega\in\Hbsupp^{-\infty,-\frac12-}$ with $\wh\Box(0)^*\breve\omega=-f$; therefore $\omega_{\rms 0}^*=\chi\pa_{\hat t}^\flat+\breve\omega$.

  The argument for $\omega_{\rms 1}^*(\scal)$ is similar, except now we use that the Minkowski 1-form wave operator $\Box_{\hat{\ubar g}}$ annihilates the Killing 1-form $\dd(\hat r\scal)\in\CI(\hat X_b;\Ttsc^*_{\hat X_b}\hat M_b)$. Since we compute adjoints with respect to the volume density and fiber inner product induced by $\hat g_b$ (not $\hat{\ubar g}$), we have $\Box_{\hat{\ubar g}}^*-\Box_{\hat{\ubar g}}\in\rho_\cD^3\Diffb^1$; moreover, $\wh\Box(0)^*-\wh{\Box_{\hat{\ubar g}}}(0)^*\in\rho_\cD^3\Diffb^2$, and thus
  \begin{align*}
    \wh\Box(0)^*(\chi\dd(\hat r\scal)) &= \bigl(\wh\Box(0)^*-\wh{\Box_{\hat{\ubar g}}}(0)^*\bigr)(\chi\dd(\hat r\scal)) + \bigl(\wh{\Box_{\hat{\ubar g}}}(0)^*-\wh{\Box_{\hat{\ubar g}}}(0)\bigr)(\chi\dd(\hat r\scal)) \\
      &\qquad + \bigl[ \wh{\Box_{\hat{\ubar g}}}(0), \chi \bigr] \dd(\hat r\scal) \\
      &\quad \in \rho_\cD^3\CI + \rho_\cD^3\CI + \CIc
  \end{align*}
  lies in $\Hbsupp^{\infty,\frac32-}$ and thus in the range of $\wh\Box(0)^*$ on $\Hbsupp^{-\infty,-\frac12-}$.

  The argument for $\omega_{\rmv 1}^*(\vect)$ is again similar, except now we compare the operator $\wh\Box(0)^*$ on Kerr with parameters $b=(\bhm,\bha)$ with the operator $\wh{\Box_{\hat g_{b_0}}}(0)^*$ on the Schwarzschild spacetime with $b_0=(\bhm,0)$ (the metric of which agrees with $\hat g_b$ modulo terms of size $\hat\rho_\cD^2=\hat r^{-2}$ as sections of $S^2\,\Ttsc^*_{\hat X_b}\hat M_b$). Thus, we use that $\Box_{\hat g_{b_0}}$ annihilates $\hat r^2\vect$, further $\wh{\Box_{\hat g_{b_0}}}(0)-\wh{\Box_{\hat g_{b_0}}}(0)^*\in\rho_\cD^4\Diffb^1$ and $\wh\Box(0)-\wh{\Box_{\hat g_{b_0}}}(0)\in\rho_\cD^4\Diffb^2$, and therefore $\wh\Box(0)^*(\chi\hat r^2\vect)\in\rho_\cD^3\CI$ again.

  For the final claim, we use $\wh{\delta_{\hat g_b}^*}(0)\in\rho_\cD\Diffb^1(\hat X_b;\Ttsc^*_{\hat X_b}\hat M_b,S^2\,\Ttsc^*_{\hat X_b}\hat M_b)$ and compute
  \[
    \delta_{\hat g_b}^*\omega_{\rms 0}^* \in [\delta_{\hat g_b}^*,\chi]\pa_{\hat t}^\flat + \rho_\cD\Diffb^1(\hat X_b)\Hbsupp^{-\infty,-\frac12-}\subset\CIc(\hat X_b^\circ)+\Hbsupp^{-\infty,\frac12-}.
  \]
  Next, we use that $\dd(\hat r\scal)\in\ker\delta_{\hat{\ubar g}}^*$, so modulo $\Hbsupp^{-\infty,\frac12-}$ we have
  \begin{equation}
  \label{EqKCDPotDelss1}
    \delta_{\hat g_b}^*\omega_{\rms 1}^*(\scal) \equiv \bigl(\wh{\delta_{\hat g_b}^*}(0) - \wh{\delta_{\hat{\ubar g}}^*}(0)\bigr)(\chi\dd(\hat r\scal)) + [\delta_{\hat{\ubar g}}^*,\chi]\dd(\hat r\scal) \in \rho_\cD^2\CI \subset \Hbsupp^{-\infty,\frac12-}.
  \end{equation}
  Finally, we use $\hat r^2\vect\in\ker\delta_{\hat g_{b_0}}^*$ to obtain $\delta_{\hat g_b}^*\omega_{\rmv 1}^*(\vect)\in\Hbsupp^{-\infty,-\frac12-}$ in an analogous fashion.
\end{proof}

%%%%%%%%%%%%%%%%%%%%%%%%%%%%%%%%%%%%%%%%%%%%%%%%%%
\subsection{Gauge potential wave operator and large gauge potentials}
\label{SsKUps}

As indicated in \cite[Remark~10.14]{HaefnerHintzVasyKerr} in the case of Schwarzschild and slowly rotating Kerr, one can similarly modify the gauge potential wave operator to ensure its mode stability in the closed upper half plane. Here we carry this out in the full subextremal range.

\begin{prop}[Gauge modification]
\label{PropKUps}
  There exists $\cd_\Ups\in\CIc(\hat X_b^\circ;T^*_{\hat X_b^\circ}\hat M_b^\circ)$ so that, for all sufficiently small $\gamma_\Ups>0$, the operator $\wh{\Box^\Ups_{\hat g_b,\gamma_\Ups}}(\sigma)$ is invertible for all $\sigma\in\C$, $\Im\sigma\geq 0$, as an operator between the function spaces of Theorem~\usref{ThmKCDUnmod}, for all strongly Kerr-admissible orders $\sfs,\alpha_\cD$ with $\alpha_\cD\in(-\frac32,-\frac12)$ so that also $\sfs-1,\alpha_\cD$ are strongly Kerr-admissible.
\end{prop}
\begin{proof}
  The proof is essentially the same as that of Proposition~\ref{PropKCD}. The only modification concerns the choice of 1-form $\cd_\Ups$. To wit, since
  \[
    \Box^\Ups_{\hat g_b,\gamma_\Ups} = \Box_{\hat g_b} + 2 \gamma_\Ups E_\Ups \sfG_{\hat g_b}\delta_{\hat g_b}^*,
  \]
  we wish to choose $\cd_\Ups$ so that $\la 2 E_\Ups\sfG_{\hat g_b}\delta_{\hat g_b}^*\omega_{\rms 0},\omega_{\rms 0}^*\ra=4\pi$ (or any other positive number). But using $E_\Ups\sfG_{\hat g_b}=2\iota_{\cd_\Ups^\sharp}$, we compute
  \[
    \la E_\Ups\sfG_{\hat g_b}\delta_{\hat g_b}^*\omega_{\rms 0},\omega_{\rms 0}^*\ra = \la \delta_{\hat g_b}^*\omega_{\rms 0}, 2\cd_\Ups\otimes_s\omega_{\rms 0}^*\ra = \la 2\iota_{(\omega_{\rms 0}^*)^\sharp}\delta_{\hat g_b}^*\omega_{\rms 0}, \cd_\Ups\ra.
  \]
  We thus only need to show that $0\neq\iota_{(\omega_{\rms 0}^*)^\sharp}\delta_{\hat g_b}^*\omega_{\rms 0}=\delta(\hat r-\hat r_b) \iota_{\nabla\hat r}\delta_{\hat g_b}^*\omega_{\rms 0}$. But since
  \[
    \iota_{\pa_{\hat t_0}}\iota_{\nabla\hat r}\delta_{\hat g_b}^*\omega_{\rms 0} = -\frac{\hat r^2-a^2\cos^2\theta}{4\varrho^6}
  \]
  is nonzero at $\hat r=\hat r_b$, we are done.
\end{proof}

\textit{In the sequel, we fix $\cd_\Ups$ and $\gamma_\Ups>0$ so that the conclusions of Proposition~\usref{PropKCD} hold.} We now use that
\begin{equation}
\label{EqKBoxUps0}
  \wh{\Box^\Ups_{\hat g_b,\gamma_\Ups}}(0) \colon \bigl\{ u\in\Hb^{s,\alpha_\cD}(\hat X_b;\Ttsc^*_{\hat X_b}\hat M_b) \colon \wh{\Box^\Ups_{\hat g_b,\gamma_\Ups}}(0)u\in\Hb^{s-1,\alpha_\cD+2} \bigr\} \to \Hb^{s-1,\alpha_\cD+2}
\end{equation}
is invertible for $s>\bar s$ and $\alpha_\cD\in(-\frac32,-\frac12)$. For $\alpha_\cD\in(-\frac52,-\frac32)$, this operator is Fredholm by \citeII{Proposition~\ref*{PropEstFT0}}, and indeed surjective (since upon decreasing the weight, the cokernel can only get smaller) with a 4-dimensional nullspace spanned by the 1-forms $\pa_{\hat t}^\flat$ and $\omega_{\rms 1}(\scal)$, $\scal\in\scalspace_1$, constructed below. This follows from a normal operator argument:\footnote{Here, $(-\frac52,-\frac32)+\frac32=(-1,0)$ does not contain the real part of any element of $\specb\wh{\Box^\Ups_{\hat g_b,\gamma_\Ups}}(0)=\specb\wh\Box(0)=\Z$. The shift by $\frac32$ arises from the switch from an asymptotically Euclidean volume density to an unweighted b-density.} the elements of this nullspace are characterized by their leading order term at $\pa\hat X_b$, which is an indicial solution, i.e.\ here a stationary 1-form on Minkowski space with coefficients, in the frame $\dd\hat t$, $\dd\hat x^i$, $i=1,2,3$, which are invariant under spatial dilations and lie in the kernel of $\Box_{\hat{\ubar g}}$ (away from $\hat x=0$); these 1-forms are linear combinations of $\dd\hat t$, $\dd\hat x^i$.

\begin{lemma}[Gauge potentials]
\label{LemmaKGaugePot}
  For $\scal,\scal'\in\scalspace_1$, there exist stationary 1-forms $\omega_{\rms 1}(\scal)$, $\omega_{(0)}(\scal,\scal')$, $\omega_{(1)}(\scal)$ in the kernel of $\wh{\Box^\Ups_{\hat g_b,\gamma_\Ups}}(0)$ of the form
  \begin{align*}
    \omega_{\rms 1}(\scal)&\equiv\dd(\hat r\scal)\bmod\cA^{1-}(\hat X_b;\Ttsc^*_{\hat X_b}\hat M_b), \\
    \omega_{(0)}(\scal,\scal')&\equiv \hat r\scal\,\dd(\hat r\scal')\bmod\cA^{0-}, \\
    \omega_{(1)}(\scal) &\equiv \hat r\scal\,\dd\hat t\bmod\cA^{0-};
  \end{align*}
  and moreover there exists $\breve\omega_{(2)}\in\cA^{0-}$ so that
  \[
    \omega_{(2)} := -\hat t\pa_{\hat t}^\flat + \breve\omega_{(2)} \in \ker \Box^\Ups_{\hat g_b,\gamma_\Ups}.
  \]
  Write $\scal_i=\frac{\hat x^i}{|\hat x^i|}$ in standard coordinates $\hat x=(\hat x^1,\hat x^2,\hat x^3)\in\hat X_b^\circ\subset\R^3_{\hat x}$. Then\footnote{We use the notation $\cA^{\cE,\alpha}=\cA^\cE_\phg+\cA^\alpha$.}
  \begin{subequations}
  \begin{alignat}{2}
  \label{EqKPureGauges1}
    h_{\rms 1}(\scal) &:= \delta_{\hat g_b}^*\omega_{\rms 1}(\scal) &&\in \cA^{(2,0),3-}(\hat X_b;S^2\,\Ttsc^*_{\hat X_b}\hat M_b)=\rho_\cD^2\CI+\cA^{3-}, \\
  \label{EqKPureGaugeij}
    h^\Ups_{i j} &:= \delta_{\hat g_b}^*\omega_{(0)}(\scal_i,\scal_j) &&\equiv \dd\hat x^i\otimes_s\dd\hat x^j \bmod \cA^{1-}, \\
  \label{EqKPureGauge0j}
    h^\Ups_{0 j} &:= \delta_{\hat g_b}^*\omega_{(1)}(\scal_j) &&\equiv \dd\hat t\otimes_s\dd\hat x^j \bmod \cA^{1-}, \\
  \label{EqKPureGauge00}
    h^\Ups_{0 0} &:= \delta_{\hat g_b}^*\omega_{(2)} &&\equiv \dd\hat t\otimes\dd\hat t \bmod \cA^{1-}.
  \end{alignat}
  Finally, there exists $\breve\omega_{\rms 1}(\scal)\equiv-\hat r\scal\,\dd\hat t\bmod\cA^{0-}$ so that $\hat t\omega_{\rms 1}(\scal)+\breve\omega_{\rms 1}(\scal)\in\ker\Box^\Ups_{\hat g_b,\gamma_\Ups}$ and
  \begin{equation}
  \label{EqKPureGaugeGens1}
    \hat t h_{\rms 1}(\scal)+\breve h_{\rms 1}(\scal) := \delta_{\hat g_b}^*\bigl(\hat t\omega_{\rms 1}(\scal)+\breve\omega_{\rms 1}(\scal)\bigr),\qquad \breve h_{\rms 1}(\scal)\in\cA^{(1,1),2-}.
  \end{equation}
  \end{subequations}
\end{lemma}
\begin{proof}
  Write $\Box:=\Box_{\hat g_b,\gamma_\Ups}^\Ups$, $\ubar\Box:=\Box_{\hat{\ubar g}}$, and $\ubar\delta^*=\delta_{\hat{\ubar g}}^*$. Then we set
  \[
    \omega_{\rms 1}(\scal)=\dd(\hat r\scal)+\omega'_{\rms 1}(\scal), \qquad \omega'_{\rms 1}(\scal):=-\wh\Box(0)^{-1}f,
  \]
  where
  \[
    f:=\wh\Box(0)(\dd(\hat r\scal)) = \bigl(\wh\Box(0)-\wh{\ubar\Box}(0)\bigr)(\dd(\hat r\scal)) \in \rho_\cD^3\Diffb^2(\hat X_b)\CI \subset \rho_\cD^3\CI.
  \]
  Using the invertibility of~\eqref{EqKBoxUps0} for $s=\infty$ and $\alpha_\cD=-\frac12-$ together with Sobolev embedding gives $\omega'_{\rms 1}(\scal)\in\cA^{1-}$ as claimed. The membership $h_{\rms 1}(\scal)\in\cA^{2-}$ can be proved analogously to~\eqref{EqKCDPotDelss1}. The sharper statement~\eqref{EqKPureGauges1} follows from $0=L h_{\rms 1}(\scal)=\hat L(0)h_{\rms 1}(\scal)$: since the normal operator of $\hat L(0)$, in the splitting of $S^2\,\Ttsc^*_{\hat X_b}\hat M_b$ induced by the differentials $\dd\hat t$, $\dd\hat x^i$ ($i=1,2,3$), is given by the Euclidean Laplacian tensored with the $10\times 10$ identity matrix, we have $\Specb(\hat L(0))=\Z\times\{0\}$ (the divisor of $N(\rho_\cD^{-2}\hat L(0),\lambda)^{-1}$, cf.\ \cite[equation~(5.10)]{MelroseAPS}). Therefore, every element in $\cA^{2-}\cap\ker\hat L(0)$ is of the form stated in~\eqref{EqKPureGauges1} (and in fact log-smooth).

  To construct $\omega_{(0)}(\scal,\scal')$, we make the ansatz $\omega_{(0)}(\scal,\scal')=\hat r\scal\,\omega_{\rms 1}(\scal')-\breve\omega_{(0)}$. Since $\wh{\ubar\Box}(0)(\hat r\scal\,\dd(\hat r\scal'))=0$ (i.e.\ $\wh{\ubar\Box}(0)(\hat x^1\,\dd\hat x^2)=0$ etc.), the correction term $\breve\omega_{(0)}$ must satisfy
  \begin{align*}
    \wh\Box(0)\breve\omega_{(0)} &= \wh\Box(0)\bigl(\hat r\scal\omega_{\rms 1}(\scal')\bigr) \\
      &= \bigl(\wh\Box(0)-\wh{\ubar\Box}(0)\bigr)\bigl(\hat r\scal\,\dd(\hat r\scal')\bigr) + \wh{\Box}(0)\bigl( \hat r\scal \omega'_{\rms 1}(\scal') \bigr) \\
      &\in \rho_\cD^3\Diffb^2(\hat X_b)\rho_\cD^{-1}\CI + \rho_\cD^2\Diffb^2(\hat X_b)\cA^{0-} \subset \cA^{2-}(\hat X_b;\Ttsc^*_{\hat X_b}\hat M_b).
  \end{align*}
  This equation has a solution $\breve\omega_{(0)}\in\cA^{0-}$, as required. We moreover compute
  \[
    \delta_{\hat g_b}^*\omega_{(0)}(\scal,\scal') = \dd(\hat r\scal)\otimes_s \omega_{\rms 1}(\scal') + \hat r\scal h_{\rms 1}(\scal) - \wh{\delta_{\hat g_b}^*}(0)\breve\omega_{(0)},
  \]
  which gives~\eqref{EqKPureGaugeij} by the first part and $\wh{\delta_{\hat g_b}^*}(0)\in\rho_\cD\Diffb^1$.

  The construction of $\omega_{(1)}(\scal)=-\hat r\scal\pa_{\hat t}^\flat-\breve\omega_{(1)}$ is completely analogous. Since $\pa_{\hat t}^\flat=-\dd\hat t+\rho_\cD\CI$, also~\eqref{EqKPureGauge0j} follows immediately.

  Next, we compute
  \begin{align*}
    \Box(\hat t\pa_{\hat t}^\flat) &= [\Box,\hat t]\pa_{\hat t}^\flat \\
      &= [\ubar\Box,\hat t](-\dd\hat t + (\pa_{\hat t}^\flat+\dd\hat t)) + [\Box-\ubar\Box,\hat t]\pa_{\hat t}^\flat \\
      &= -\ubar\Box(\hat t\,\dd\hat t) + [\ubar\Box,\hat t](\pa_{\hat t}^\flat+\dd\hat t) + [\Box-\ubar\Box,\hat t]\pa_{\hat t}^\flat \\
      &\in 0 + \rho_\cD\Diffb^1(\hat X_b)\rho_\cD\CI + \rho_\cD^2\Diffb^1(\hat X_b)\CI \\
      &\subset \rho_\cD^2\CI(\hat X_b;\Ttsc^*_{\hat X_b}\hat M_b).
  \end{align*}
  This can thus be written as $\wh\Box(0)\breve\omega_{(2)}$ for some $\breve\omega_{(2)}\in\cA^{0-}$, and we then let $\omega_{(2)}=-\hat t\pa_{\hat t}^\flat+\breve\omega_{(2)}$. Finally,~\eqref{EqKPureGauge00} follows from $\delta_{\hat g_b}^*\omega_{(2)}=-\dd\hat t\otimes_s\pa_{\hat t}^\flat+\wh{\delta_{\hat g_b}^*}(0)\breve\omega_{(2)}$.

  Finally, to construct $\breve\omega_{\rms 1}(\scal)$ we argue similarly to the proof of \cite[Proposition~7.13]{HaefnerHintzVasyKerr}: we use that $\ubar\Box(\hat t\,\dd(\hat r\scal)-\hat r\scal\,\dd\hat t)=0$ (Lorentz boost in the direction determined by $\scal\in\scalspace_1$) and compute
  \begin{align*}
    &\Box(\hat t\omega_{\rms 1}(\scal)-\hat r\scal\,\dd\hat t) \\
    &\quad= [\Box,\hat t]\bigl(\dd(\hat r\scal)+\omega'_{\rms 1}(\scal)\bigr) + \ubar\Box(-\hat r\scal\,\dd\hat t) + \bigl(\wh\Box(0)-\wh{\ubar\Box}(0)\bigr)(-\hat r\scal\,\dd\hat t) \\
    &\quad = \bigl([\ubar\Box,\hat t]\dd(\hat r\scal) + \ubar\Box(-\hat r\scal\,\dd\hat t)\bigr) + [\ubar\Box,\hat t]\omega'_{\rms 1}(\scal) + [\Box-\ubar\Box,\hat t]\omega_{\rms 1}(\scal) + \bigl(\wh\Box(0)-\wh{\ubar\Box}(0)\bigr)(-\hat r\scal\,\dd\hat t) \\
    &\quad \in 0 + \rho_\cD\Diffb^1(\hat X_b)\cA^{1-} + \rho_\cD^2\Diffb^1(\hat X_b)\cA^0 + \rho_\cD^3\Diffb^2(\hat X_b)\CI \\
    &\subset \cA^{2-}(\hat X_b).
  \end{align*}
  This can thus be written as $-\wh\Box(0)\breve\omega'$ where $\breve\omega'\in\cA^{0-}$, and we then set $\breve\omega_{\rms 1}(\scal)=-\hat r\scal\,\dd\hat t+\breve\omega'$.

  Using $\ubar\delta^*(\hat t\,\dd(\hat r\scal)-\hat r\scal\,\dd\hat t)=0$, we furthermore compute
  \begin{align*}
    &\delta_{\hat g_b}^*(\hat t\omega_{\rms 1}(\scal)-\hat r\scal\,\dd\hat t+\breve\omega') \\
    &\qquad = \hat t h_{\rms 1}(\scal) + [\delta_{\hat g_b}^*,\hat t]\omega_{\rms 1}(\scal) + \ubar\delta^*(-\hat r\scal\,\dd\hat t) + \bigl(\wh{\delta_{\hat g_b}^*}(0)-\wh{\ubar\delta^*}(0)\bigr)(-\hat r\scal\,\dd\hat t) + \wh{\delta_{\hat g_b}^*}(0)\breve\omega' \\
    &\qquad = \hat t h_{\rms 1}(\scal) + \bigl([\ubar\delta^*,\hat t]\dd(\hat r\scal)+\ubar\delta^*(-\hat r\scal\,\dd\hat t)\bigr) + [\delta_{\hat g_b}^*-\ubar\delta^*,\hat t]\dd(\hat r\scal) \\
    &\qquad \qquad + [\delta_{\hat g_b}^*,\hat t]\omega'_{\rms 1}(\scal) + \bigl(\wh{\delta_{\hat g_b}^*}(0)-\wh{\ubar\delta^*}(0)\bigr)(-\hat r\scal\,\dd\hat t) + \wh{\delta_{\hat g_b}^*}(0)\breve\omega' \\
    &\qquad = \hat t h_{\rms 1}(\scal) + \breve h_{\rms 1}(\scal),
  \end{align*}
  where we read off $\breve h_{\rms 1}(\scal)\in 0+\rho_\cD\CI+\cA^{1-}+\rho_\cD\CI+\cA^{1-}=\cA^{1-}$. Note furthermore that
  \begin{equation}
  \label{EqKPureGaugebrevehs1}
    0 = L\bigl(\hat t h_{\rms 1}(\scal)+\breve h_{\rms 1}(\scal)\bigr) = \hat L(0)\breve h_{\rms 1}(\scal) + [L,\hat t]h_{\rms 1}(\scal);
  \end{equation}
  since $[L,\hat t]h_{\rms 1}(\scal)\in\cA^{(3,0),4-}$, a normal operator argument shows that $\breve h_{\rms 1}(\scal)$ is the sum of a $\hat r^{-1}\log\hat r$ leading order term (with coefficient equal to an indicial solution with homogeneity $-1$, i.e.\ an element of $\ker N(\hat L(0),1)$), an $\hat r^{-1}$ term, and a $\cA^{2-}$ remainder. This completes the proof of~\eqref{EqKPureGaugeGens1}.
\end{proof}

Finally, we consider the linearized Kerr metrics $\hat g'_b(\dot b)=\frac{\dd}{\dd s}\hat g_{b+s\dot b}|_{s=0}$. Since they are defined using a somewhat arbitrary presentation of the Kerr family of metrics, one does not expect $\hat g'_b(\dot b)$ to satisfy the gauge condition. This is easily remedied:

\begin{lemma}[Gauged linearized Kerr metrics]
\label{LemmaKLinKerr}
  There exist stationary 1-forms $\omega(\dot b)\in\cA^{0-}(\hat X_b;T^*_{\hat X_b}\hat M_b)$, depending linearly on $\dot b\in\C^4$, so that
  \[
    \hat g^{\prime\Ups}_b(\dot b) := \hat g'_b(\dot b) + \delta_{\hat g_b}^*\omega(\dot b) \in \cA^{(1,0),2-}(\hat X_b;\Ttsc^*_{\hat X_b}\hat M_b) = \rho_\cD\CI+\cA^{2-}
  \]
  lies in $\ker\delta_{\hat g_b,\gamma_\Ups}\sfG_{\hat g_b}$.
\end{lemma}
\begin{proof}
  For $\dot b=(0,\dot\bha)$, we note that $\hat g'_b(0,\dot\bha)\in\rho_\cD^2\CI$, so $\delta_{\hat g_b}^*\hat g'_b(0,\dot\bha)\in\rho_\cD^3\CI$. This can be written in a unique manner as $-\wh{\Box^\Ups_{\hat g_b,\gamma_\Ups}}(0)\omega(0,\dot\bha)$ where $\omega(0,\dot\bha)\in\Hb^{\infty,-\frac12-}=\cA^{1-}$. (In particular, $\omega(0,\dot\bha)$ is automatically linear in $\dot\bha$.) Therefore, $\hat g^{\prime\Ups}_b(0,\dot\bha)\in\rho_\cD^2\CI+\wh{\delta_{\hat g_b}^*}(0)\omega(0,\dot\bha)\in\cA^{2-}$ indeed.

  Consider now a linearized mass perturbation $\dot b=(1,0)$; then $\delta_{\hat g_b}^*\hat g'_b(1,0)\in\rho_\cD^2\CI$, which we can write as $=-\wh{\Box_{\hat g_b,\gamma_\Ups}^\Ups}(0)\omega(1,0)$ for some $\omega(1,0)\in\cA^{0-}$, and therefore $\hat g^{\prime\Ups}_b(1,0)=\hat g'_b(1,0)+\wh{\delta_{\hat g_b}^*}(0)\omega(1,0)\in\cA^{1-}$. But since $\hat L(0)\hat g^{\prime\Ups}_b(1,0)=0$, a normal operator argument as in the proof of Lemma~\ref{LemmaKGaugePot} implies that in fact $\hat g^{\prime\Ups}_b(1,0)\in\hat r^{-1}\CI+\cA^{2-}$.
\end{proof}

%%%%%%%%%%%%%%%%%%%%%%%%%%%%%%%%%%%%%%%%%%%%%%%%%%
\subsection{Spectral theory of the linearized gauge-fixed Einstein operator}
\label{SsKE}

We continue to work with the fixed choices of $\cd_C,\gamma_C,\cd_\Ups,\gamma_\Ups$ from Propositions~\ref{PropKCD} and \ref{PropKUps}, which we thus use to define the operator $L$ in~\eqref{EqKEOp}. To analyze this operator, we rely fundamentally on \cite{AnderssonHaefnerWhitingMode}. We work with strongly Kerr-admissible orders $\sfs$, $\alpha_\cD$, where $\alpha_\cD\in(-\frac32,-\frac12)$ unless otherwise noted.

\begin{prop}[Mode stability for nonzero frequencies]
\label{PropKENon0}
  For $\sigma\in\C$, $\Im\sigma\geq 0$, $\sigma\neq 0$, the operator
  \[
    \hat L(\sigma) \colon \bigl\{ h\in\Hsc^{\sfs,\sfs+\alpha_\cD}(\hat X_b;S^2\,\Ttsc^*_{\hat X_b}\hat M_b) \colon \hat L(\sigma)h\in\Hsc^{\sfs-1,\sfs+\alpha_\cD+1} \bigr\} \to \Hsc^{\sfs-1,\sfs+\alpha_\cD+1}
  \]
  is invertible.
\end{prop}
\begin{proof}
  Given $h\in\Hsc^{\sfs,\sfs+\alpha_\cD}$ with $\hat L(\sigma)h=0$, we have
  \[
    0 = \frac12L(e^{-i\sigma\hat t}h)=D_{\hat g_b}\Ric(e^{-i\sigma\hat t}h) + \delta_{\hat g_b,\gamma_C}^*\delta_{\hat g_b,\gamma_\Ups}\sfG_{\hat g_b}(e^{-i\sigma\hat t}h).
  \]
  Applying $\delta_{\hat g_b}\sfG_{\hat g_b}$ gives the equation $\Box^{\rm CD}_{\hat g_b,\gamma_C}\eta=0$ where $\eta=\delta_{\hat g_b,\gamma_\Ups}\sfG_{\hat g_b}(e^{-i\sigma\hat t}h)$. Proposition~\ref{PropKCD} implies $\eta=0$ and thus $D_{\hat g_b}\Ric(e^{-i\sigma\hat t}h)=0$.

  Rather than applying \cite[Theorem~6.1]{AnderssonHaefnerWhitingMode} directly (which would require proving that $e^{-i\sigma\hat r}h$ is conormal at $\hat r=\infty$, which can be done by following the arguments of \cite[\S{2.4}]{GellRedmanHassellShapiroZhangHelmholtz}), we use \cite[equation~(1)]{AksteinerAnderssonBaeckdahlKerrId} as in \cite[\S{6.4}]{AnderssonHaefnerWhitingMode} to conclude that
  \[
    e^{-i\sigma\hat t}h=i\sigma^{-1}\delta_{\hat g_b}(e^{-i\sigma\hat t}\omega)
  \]
  where the 1-form $e^{-i\sigma\hat t}\omega$ is obtained from $e^{-i\sigma\hat t}h$ by applying an explicit third order differential operator (with coefficients which are smooth in $\hat r^{-1}$). Since the scattering wave front set of $h$ is contained in the outgoing radial set over $\hat r=\infty$ (in view of $\hat L(\sigma)h=0$), the same is therefore true for $\omega$. But then $0=2\eta=i\sigma^{-1}\wh{\Box^\Ups_{\hat g_b,\gamma_\Ups}}(\sigma)\omega=0$ implies by Proposition~\ref{PropKUps} that $\omega=0$ and thus $h=0$.
\end{proof}

As a consequence, besides the high energy estimates \citeII{Proposition~\ref*{PropEstFThi}}, we also have uniform estimates on the resolvent on any compact subset of the punctured closed upper half plane. We next turn to the low energy behavior.

\begin{prop}[Zero energy operator]
\label{PropKE0}
  For $\alpha_\cD\in(-\frac32,-\frac12)$, the operator
  \[
    \hat L(0) \colon \bigl\{ h\in\Hb^{\sfs,\alpha_\cD}(\hat X_b;S^2\,\Ttsc^*_{\hat X_b}\hat M_b) \colon \hat L(0)h\in\Hb^{\sfs-1,\alpha_\cD+2} \bigr\} \to \Hb^{\sfs-1,\alpha_\cD+2}
  \]
  is Fredholm of index $0$. Its kernel and cokernel are 7-dimensional, and they are explicitly given by
  \begin{alignat}{2}
  \label{EqKE0Ker}
    \cK &:= \ker_{\Hb^{\infty,-\frac12-}}\hat L(0) &&= \mathspan \Bigl( \{ g_b^{\prime\Ups}(\dot b) \colon \dot b\in\C^4 \} \cup \{ h_{\rms 1}(\scal) \colon \scal\in\scalspace_1 \} \Bigr), \\
  \label{EqKE0Coker}
    \cK^* &:= \ker_{\dot H_\bop^{-\infty,-\frac12-}}\hat L(0)^* &&= \mathspan \Bigl(\{ h_{\rms 0}^* \} \cup \{ h_{\rmv 1}^*(\vect) \colon \vect\in\vectspace_1 \} \cup \{ h_{\rms 1}^*(\scal) \colon \scal\in\scalspace_1 \} \Bigr)
  \end{alignat}
  in the notation of Lemmas~\usref{LemmaKCDPot}, \usref{LemmaKGaugePot}, and \usref{LemmaKLinKerr}. Furthermore, in the notation~\eqref{EqKPureGaugeij}--\eqref{EqKPureGauge00}, we have
  \begin{equation}
  \label{EqKE0KerL}
    \alpha\in\bigl(-\tfrac52,-\tfrac32\bigr) \implies \ker_{\Hb^{\infty,\alpha}}\hat L(0) = \cK \oplus \mathspan \{ h_{0 0}^\Ups,\ h_{0 j}^\Ups,\ h_{i j}^\Ups \colon 1\leq i\leq j\leq 3 \}.
  \end{equation}
\end{prop}
\begin{proof}
  The Fredholm index $0$ property can be proved in the same fashion as in the proof of \cite[Theorem~4.3]{HaefnerHintzVasyKerr}. Let now $\beta\in\R$. Any $h\in\Hb^{\sfs,\beta}\cap\ker\hat L(0)$ is automatically conormal by elliptic regularity (in the b-setting) near infinity and (above threshold) radial point estimates at the event horizon combined with real principal type propagation estimates; so $h\in\Hb^{\infty,\beta}$. Furthermore, as noted before, the b-normal operator of $\hat L(0)$ at $\hat r=\infty$, in the bundle splitting induced by the differentials $\dd\hat t,\dd\hat x^i$ ($i=1,2,3$), is the tensor product of the Euclidean Laplacian with the $10\times 10$ identity matrix; hence, $\specb(L)=\Z\times\{0\}$. Therefore, by a normal operator argument, one concludes $h\in\Hb^{\infty,-\frac32+\ell-}$ where $\ell\in\Z$ is the smallest integer with $\beta<-\frac32+\ell$.

  We consider now the case $\beta=\alpha_\cD\in(-\frac32,-\frac12)$. The equation $L h=0$ implies that
  \[
    \wh{\Box_{\hat g_b,\gamma_C}^{\rm CD}}(0)\eta=0,\qquad \eta:=\delta_{\hat g_b,\gamma_\Ups}\sfG_{\hat g_b}h\in\Hb^{\infty,\beta+1}.
  \]
  Since $\beta+1>-\frac32$, this implies $\eta=0$ by Proposition~\ref{PropKUps}, and therefore $h\in\Hb^{\infty,-\frac12-}$ is a stationary solution of $D_{\hat g_b}\Ric(h)=0$. By \cite[Theorem~6.1]{AnderssonHaefnerWhitingMode}, there exist $\dot b\in\C^4$ and a stationary 1-form $\omega\in\Hb^{\infty,-\frac32-}=\cA^{0-}$ so that $h=g'_b(\dot b)+\delta_{\hat g_b}^*\omega$; by Lemma~\ref{LemmaKLinKerr}, we can replace $g'_b(\dot b)$ by $g^{\prime\Ups}_b(\dot b)$ upon modifying $\omega$ by an element of $\cA^{0-}$. But then $0=2\eta=\wh{\Box^\Ups_{\hat g_b,\gamma_\Ups}}(0)\omega$ implies (in view of the 4-dimensionality of $\ker_{\cA^{0-}}\wh{\Box^\Ups_{\hat g_b,\gamma_\Ups}}(0)$ discussed before Lemma~\ref{LemmaKGaugePot}) that $\omega=c\pa_{\hat t}^\flat+\omega_{\rms 1}(\scal)$ for some $c\in\C$, $\scal\in\scalspace_1$. Since $\delta_{\hat g_b}^*\pa_{\hat t}^\flat=0$, this implies that $h=g^{\prime\Ups}_b(\dot b)+h_{\rms 1}(\scal)$, thus proving~\eqref{EqKE0Ker}. That the right hand side of~\eqref{EqKE0Ker} is indeed $7$-dimensional can be seen using Lemma~\ref{LemmaKEL2} below: the inner product $\la[L,\hat t](\hat g^{\prime\Ups}_b(\dot b)+h_{\rms 1}(\scal)),h^*\ra$ computes the components of $\dot b$ when $h^*=h_{\rms 0}^*,h_{\rmv 1}^*(\vect)$, so the vanishing of $\hat g^{\prime\Ups}_b(\dot b)+h_{\rms 1}(\scal)$ implies $\dot b=0$, and then also $\scal=0$.

  Since $\dim\cK^*=\dim\cK=7$, one can prove~\eqref{EqKE0Coker} by simply noting that the space on the right is, indeed, $7$-dimensional. This in turns follows from the fact that $\delta_{\hat g_b}^*\omega^*=0$ for $\omega^*=c\omega_{\rms 0}^*+\omega_{\rms 1}^*(\scal)+\omega_{\rmv 1}^*(\vect)=0$, with $c\in\C$, $\scal\in\scalspace_1$, $\vect\in\vectspace_1$, requires $\omega^*$ to be a Killing vector field on Kerr; but since $\omega^*=0$ in $\hat r<\hat r_b$, we must have $\omega^*=0$, and an inspection of the asymptotic behavior of $\omega_{\rms 0}^*$ etc.\ as $\hat r\to\infty$ implies that $c=0$, $\scal=0$, $\vect=0$, as required.

  The proof of~\eqref{EqKE0KerL} is based on a normal operator argument. The elements in the kernel of the normal operator of $\hat L(0)$ (i.e.\ the Euclidean Laplacian tensored with the $10\times 10$ identity matrix) which are quasi-homogeneous of degree $0$ are spanned by $\dd\hat z^\mu\otimes_s\dd\hat z^\nu$, $0\leq\mu,\nu\leq 3$, where $\hat z=(\hat t,\hat x)$; there is thus a 10-dimensional space of them. Correspondingly, the dimension of $\cQ:=\ker_{\Hb^{\infty,\alpha}}\hat L(0)/\cK$ (this space being independent of $\alpha\in(-\frac52,-\frac32)$) is at most 10-dimensional. The tensors~\eqref{EqKPureGaugeij}--\eqref{EqKPureGauge00} are elements in the nullspace of $\hat L(0)$ whose leading order terms at $\hat r=\infty$ are precisely these tensors $\dd\hat z^\mu\otimes_s\dd\hat z^\nu$. Since they (or more precisely: their images in $\cQ$) thus span a 10-dimensional subspace in the quotient space $\cQ$, they must, in fact, be a basis of $\cQ$. This establishes~\eqref{EqKE0KerL} and finishes the proof.
\end{proof}

In order to obtain control on $\hat L(\sigma)^{-1}$ near $\sigma=0$, we use a Grushin problem setup. We first briefly consider the simpler case of the 1-form wave operator $\Box_{\hat g_b}$: similarly to (but simpler than)~\eqref{EqKCDGrushin}, one considers the augmented operator
\begin{equation}
\label{EqKEToyAug}
  \begin{pmatrix} \wh{\Box_{\hat g_b}}(\sigma) & \wh{\Box_{\hat g_b}}(\sigma)(\sigma^{-1}e^{i\sigma\hat r}\omega_{\rms 0}) \\ \la\cdot,f^*\ra & 0 \end{pmatrix}
\end{equation}
where $f^*\in\CIc(\hat X_b^\circ;T^*_{\hat X_b^\circ}\hat M_b^\circ)$ is chosen so that $\la\omega_{\rms 0},f^*\ra\neq 0$. Indeed, in addition to the symbolic and transition face normal operator estimates, its zero energy operator can be shown to be invertible as a consequence of the non-degeneracy of the pairing~\eqref{EqKCDPair1} (see also the discussion preceding~\eqref{EqKELcommbreveh} below). By inspection of the first line of this operator, we thus get (a fortiori) uniform $|\sigma|^{-1}$ resolvent bounds, and in fact a decomposition of the low energy resolvent into a regular part and a singular part which is a multiple of $\sigma^{-1}e^{i\sigma\hat r}\omega_{\rms 0}$.

We proceed to implement this approach in the more difficult setting of $\hat L(\sigma)$. First, we record the relevant non-degenerate pairings.

\begin{lemma}[$L^2$-pairings]
\label{LemmaKEL2}
  For $\fq\in\R^3$, write
  \begin{equation}
  \label{EqKEL2VectScal}
    \vect(\fq)=\Bigl(\fq\times\frac{\hat x}{|\hat x|}\Bigr)\cdot\frac{\dd\hat x}{|\hat x|}\in\vectspace_1,\qquad
    \scal(\fq) = \fq\cdot\frac{\hat x}{|\hat x|} \in \scalspace_1.
  \end{equation}
  We use the notation~\eqref{EqKCDPotDels}.
  \begin{enumerate}
  \item\label{ItKEL2Mass}{\rm (Mass changes.)} We have $\la[L,\hat t]\hat g^{\prime\Ups}_b(\dot\bhm,0),h^*\ra=-16\pi\dot\bhm$ (for $h^*=h_{\rms 0}^*$), resp.\ $-16\pi(\fq\cdot\bha)\dot\bhm$ (for $h^*=h_{\rmv 1}^*(\vect(\fq))$), resp.\ $0$ (for $h^*=h_{\rms 1}^*(\scal)$, $\scal\in\scalspace_1$).
  \item\label{ItKEL2Ang}{\rm (Angular momentum changes.)} We have $\la[L,\hat t]\hat g^{\prime\Ups}_b(0,\dot\bha),h^*\ra=0$ (for $h^*=h_{\rms 0}^*$), resp.\ $-16\pi\bhm(\fq\cdot\dot\bha)$ (for $h^*=h_{\rmv 1}^*(\vect(\fq))$), resp.\ $0$ (for $h^*=h_{\rms 1}^*(\scal)$, $\scal\in\scalspace_1$).
  \item\label{ItKE2L2COM}{\rm (Center of mass changes.)} We have $\la\frac12[[L,\hat t],\hat t]h_{\rms 1}(\scal(c))+[L,\hat t]\breve h_{\rms 1}(\scal(c)),h^*\ra=0$ (for $h^*=h_{\rms 0}^*,h_{\rmv 1}^*(\vect)$, $\vect\in\vectspace_1$), resp.\ $8\pi\bhm(\fq\cdot c)$ (for $h^*=h_{\rms 1}^*(\scal(\fq))$).
  \end{enumerate}
  The same results remain true if one replaces $\hat t$ by $\hat t'=\hat t+F$ with $F\in\cA^{-1}$ and $\breve h_{\rms 1}$ by $\breve h'_{\rms 1}:=\breve h_{\rms 1}-F h_{\rms 1}$ (so that $\hat t h_{\rms 1}+\breve h_{\rms 1}=\hat t' h_{\rms 1}+\breve h'_{\rms 1}$).
\end{lemma}

This is analogous to \citeI{Theorem~\ref*{ThmAhKCoker}}; the difference with the reference is that we work here with the gauge-fixed linearized Einstein operator $L$ and the gauged linearized Kerr metrics $\hat g^{\prime\Ups}_b(\dot b)$ instead of the linearized Ricci curvature operator $D_{\hat g_b}\Ric$ and $\hat g'_b(\dot b)$. That this does not affect the values of the pairings will be shown to be due to the fact that the zero energy (dual) states which we work with here satisfy a gauge condition.

\begin{proof}[Proof of Lemma~\usref{LemmaKEL2}]
  %%%%%%%%%% 
  \pfstep{Parts~\eqref{ItKEL2Mass} and \eqref{ItKEL2Ang}.} We first consider the contribution of the gauge term $2\delta_{\hat g_b,\gamma_C}^*\delta_{\hat g_b,\gamma_\Ups}\sfG_{\hat g_b}$ of $L$ to the pairings. Since $\delta_{\hat g_b,\gamma_\Ups}\sfG_{\hat g_b}\hat g^{\prime\Ups}_b(\dot b)=0$ for all $\dot b\in\C^4$, we have
  \begin{align*}
    &\la [ \delta_{\hat g_b,\gamma_C}^*\delta_{\hat g_b,\gamma_\Ups}\sfG_{\hat g_b}, \hat t ] \hat g^{\prime\Ups}_b(\dot b), h^* \ra = \la \delta_{\hat g_b,\gamma_C}^* \eta,h^* \ra, \\
    &\hspace{8em} \eta := [\delta_{\hat g_b,\gamma_\Ups}\sfG_{\hat g_b},\hat t]\hat g^{\prime\Ups}_b(\dot b) \in \cA^1(\hat X_b;S^2\,\Ttsc^*_{\hat X_b}\hat M_b)
  \end{align*}
  since $\hat g_b^{\prime\Ups}(\dot b)\in\cA^1$. Since $h^*$ vanishes for $\hat r<\hat r_b$ and has coefficients (in the frame $\dd\hat t,\dd\hat x^i$) of size $\cO(\hat r^{-2})$ as $\hat r\to\infty$, we can integrate by parts and thus compute this pairing to be
  \[
    \big\la \wh{\delta_{\hat g_b,\gamma_C}^*}(0)\eta,h^* \big\ra = \big\la \eta, \wh{\delta_{\hat g_b,\gamma_C}^*}(0)^*h^* \big\ra = 0;
  \]
  for the final equality we use that $(\delta_{\hat g_b,\gamma_C}^*)^*h^*=0$ since $h^*=\sfG_{\hat g_b}\delta_{\hat g_b}^*\omega^*$ with $\omega^*\in\ker(\Box^{\rm CD}_{\hat g_b,\gamma_C})^*$. We may thus replace $L$ by $2 D_{\hat g_b}\Ric$.

  Next, we argue that we can replace $\hat g^{\prime\Ups}_b(\dot b)$ by $\hat g'_b(\dot b)$. By Lemma~\ref{LemmaKLinKerr}, the difference of the two tensors is $\delta_{\hat g_b}^*\omega$ for some $\omega\in\cA^{0-}$. But
  \begin{align*}
    \la[D_{\hat g_b}\Ric,\hat t]\delta_{\hat g_b}^*\omega,h^*\ra = \la D_{\hat g_b}\Ric(\hat t\delta_{\hat g_b}^*\omega), h^* \ra = -\la \wh{D_{\hat g_b}\Ric}(0)([\delta_{\hat g_b}^*,\hat t]\omega), h^* \ra,
  \end{align*}
  with $[\delta_{\hat g_b}^*,\hat t]\omega=\dd\hat t\otimes_s\omega\in\cA^{0-}$. Again due to the vanishing, resp.\ decay properties of $h^*$ for $\hat r<\hat r_b$, resp.\ as $\hat r\to\infty$, we can integrate by parts and obtain $0$ for the value of the pairing since $h^*=\sfG_{\hat g_b}\delta_{\hat g_b}^*\omega^*$ lies in the kernel of $\wh{D_{\hat g_b}\Ric}(0)^*=\sfG_{\hat g_b}\wh{D_{\hat g_b}\Ric}(0)\sfG_{\hat g_b}$.

  The pairings of interest are thus equal to twice $\la[D_{\hat g_b}\Ric,\hat t]\hat g'_b(\dot b),\sfG_{\hat g_b}\delta_{\hat g_b}^*\omega^*\ra$ for $\omega^*=\omega_{\rms 0}^*,\omega_{\rms 1}^*(\scal),\omega_{\rmv 1}^*(\vect)$; these pairings in turn were computed in \citeI{Theorem~\ref*{ThmAhKCoker}} and thus have the stated values.

  %%%%%%%%%%
  \pfstep{Part~\eqref{ItKE2L2COM}.} The contribution of the gauge term of $L$ again vanishes. Indeed, since $\delta_{\hat g_b,\gamma_\Ups}\sfG_{\hat g_b} h_{\rms 1}(\scal)=0$, we have
  \begin{align*}
    &\frac12\bigl[ [\delta_{\hat g_b,\gamma_C}^*\delta_{\hat g_b,\gamma_\Ups}\sfG_{\hat g_b},\hat t], \hat t\bigr]h_{\rms 1}(\scal) + [\delta_{\hat b_b,\gamma_C}^*\delta_{\hat g_b,\gamma_\Ups}\sfG_{\hat g_b},\hat t]\breve h_{\rms 1}(\scal) \\
    &\ \  = \frac12[\delta_{\hat g_b,\gamma_C}^*\delta_{\hat g_b,\gamma_\Ups}\sfG_{\hat g_b},\hat t]\hat t h_{\rms 1}(\scal) - \frac12\hat t[\delta_{\hat g_b,\gamma_C}^*\delta_{\hat g_b,\gamma_\Ups}\sfG_{\hat g_b},\hat t]h_{\rms 1}(\scal) + [\delta_{\hat b_b,\gamma_C}^*\delta_{\hat g_b,\gamma_\Ups}\sfG_{\hat g_b},\hat t]\breve h_{\rms 1}(\scal) \\
    &\ \  = \frac12[\delta_{\hat g_b,\gamma_C}^*,\hat t]\delta_{\hat g_b,\gamma_\Ups}\sfG_{\hat g_b}(\hat t h_{\rms 1}(\scal)) + \frac12\delta_{\hat g_b,\gamma_C}^*[\delta_{\hat g_b,\gamma_\Ups}\sfG_{\hat g_b},\hat t]\hat t h_{\rms 1}(\scal) - \frac12\hat t \delta_{\hat g_b,\gamma_C}^*[\delta_{\hat g_b,\gamma_\Ups}\sfG_{\hat g_b},\hat t]h_{\rms 1}(\scal) \\
    &\ \ \quad \hspace{3em} + [\delta_{\hat g_b,\gamma_C}^*,\hat t]\delta_{\hat g_b,\gamma_\Ups}\sfG_{\hat g_b}\breve h_{\rms 1}(\scal) + \delta_{\hat g_b,\gamma_C}^*[\delta_{\hat g_b,\gamma_\Ups}\sfG_{\hat g_b},\hat t]\breve h_{\rms 1}(\scal) \\
    &\ \ = [\delta_{\hat g_b,\gamma_C}^*,\hat t]\delta_{\hat g_b,\gamma_\Ups}\sfG_{\hat g_b}\bigl(\hat t h_{\rms 1}(\scal)+\breve h_{\rms 1}(\scal)\bigr) + \delta_{\hat g_b,\gamma_C}^*[\delta_{\hat g_b,\gamma_\Ups}\sfG_{\hat g_b},\hat t]\breve h_{\rms 1}(\scal) \\
    &\ \ = \wh{\delta_{\hat g_b,\gamma_C}^*}(0)\eta,\qquad \eta := [\delta_{\hat g_b,\gamma_\Ups}\sfG_{\hat g_b},\hat t]\breve h_{\rms 1}(\scal),
  \end{align*}
  where for the third equality sign we use that $[\delta_{\hat g_b,\gamma_\Ups}\sfG_{\hat g_b},\hat t]$ is a zeroth order operator and thus commutes with multiplication by $\hat t$. Since $\eta\in\cA^{1-}$, we can integrate by parts as above to conclude that $\big\la\wh{\delta_{\hat g_b,\gamma_C}^*}(0)\eta, h^* \big\ra = \big\la \eta, \wh{\delta_{\hat g_b,\gamma_C}^*}(0)^*h^* \big\ra = 0$. We may thus replace $L$ by $2 D_{\hat g_b}\Ric$.

  We next show that we can replace $h_{\rms 1}(\scal)=\delta_{\hat g_b}^*\omega_{\rms 1}(\scal)$, where $\scal=\scal(\hat c)$, by $h_{b,\hat c}:=\delta_{\hat g_b}^*(\hat c\cdot\pa_{\hat x}^\flat)$ and $\breve h_{\rms 1}(\scal)=[\delta_{\hat g_b}^*,\hat t]\omega_{\rms 1}(\scal)+\delta_{\hat g_b}^*\breve\omega_{\rms 1}(\scal)$ by $\breve h_{b,\hat c}:=[\delta_{\hat g_b}^*,\hat t](\hat c\cdot\pa_{\hat x}^\flat)+\delta_{\hat g_b}^*((\hat c\cdot\hat x)\pa_{\hat t}^\flat)$. (The terms $h_{b,\hat c}$ and $\breve h_{b,\hat c}$ were used in \citeII{\S\ref*{SsAh0}}.) But
  \begin{alignat*}{2}
    \omega_{\rms 1}(\scal) &= \hat c\cdot\pa_{\hat x}^\flat + \omega, &\qquad& \omega\in\cA^{1-}, \\
    \breve\omega_{\rms 1}(\scal) &= (\hat c\cdot\hat x)\pa_{\hat t}^\flat + \breve\omega, &\qquad& \breve\omega\in\cA^{0-}.
  \end{alignat*}
  We thus need to show that the $L^2$-inner product of $\frac12[[D_{\hat g_b}\Ric,\hat t],\hat t]\delta_{\hat g_b}^*\omega+[D_{\hat g_b}\Ric,\hat t]([\delta_{\hat g_b}^*,\hat t]\omega+\delta_{\hat g_b}^*\breve\omega)$ with $h^*$ vanishes. For the term involving $\breve\omega$, this follows via integration by parts as before since
  \[
    [D_{\hat g_b}\Ric,\hat t]\delta_{\hat g_b}^*\breve\omega = D_{\hat g_b}\Ric(\hat t\delta_{\hat g_b}^*\breve\omega) = -D_{\hat g_b}\Ric\bigl( [\delta_{\hat g_b}^*,\hat t]\breve\omega \bigr) \in \wh{D_{\hat g_b}\Ric}(0)(\cA^{0-}).
  \]
  On the other hand, the term involving $\omega$ vanishes, for it equals
  \begin{align*}
    &\frac12[D_{\hat g_b}\Ric,\hat t]\hat t\delta_{\hat g_b}^*\omega - \frac12\hat t[D_{\hat g_b}\Ric,\hat t]\delta_{\hat g_b}^*\omega + [D_{\hat g_b}\Ric,\hat t][\delta_{\hat g_b}^*,\hat t]\omega \\
    &\quad = \frac12[D_{\hat g_b}\Ric,\hat t]\delta_{\hat g_b}^*(t\omega) + \frac12[D_{\hat g_b}\Ric,\hat t][\delta_{\hat g_b}^*,\hat t]\omega - \frac12\hat t D_{\hat g_b}\Ric(\hat t\delta_{\hat g_b}^*\omega) \\
    &\quad = \frac12 D_{\hat g_b}\Ric\bigl(\hat t\delta_{\hat g_b}^*(\hat t\omega)\bigr) + \frac12[D_{\hat g_b}\Ric,\hat t][\delta_{\hat g_b}^*,\hat t]\omega + \frac12\hat t\hat D_{\hat g_b}\Ric\bigl([\delta_{\hat g_b}^*,\hat t]\omega\bigr) \\
    &\quad = \frac12 D_{\hat g_b}\Ric\bigl( -[\delta_{\hat g_b}^*,\hat t]\hat t\omega + \hat t[\delta_{\hat g_b}^*,\hat t]\omega\bigr) \\
    &\quad = 0.
  \end{align*}
  The claim now follows from \citeI{Theorem~\ref*{ThmAhKCoker}}.
\end{proof}

To guide our next step, note that the top right entry of~\eqref{EqKEToyAug} evaluates at $\sigma=0$ to $-i[\Box_{\hat g_b},\hat t-\hat r]\omega_{\rms 0}$ (see also~\eqref{EqNMinkLdecomp} and the subsequent discussion), which complements the range of $\wh{\Box_{\hat g_b}}(0)$ as it lies in $\cA^3$ and is not orthogonal to $\omega_{\rms 0}^*$. For the augmentation of $L$, a naive expectation, given Lemma~\ref{LemmaKEL2}, is then that its top right entry should evaluate at $\sigma=0$ to a linear combination of $[L,\hat t_1]\hat g^{\prime\Ups}_b(\dot b)$ and $\frac12[[L,\hat t_1],\hat t_1]h_{\rms 1}(\scal)+[L,\hat t_1]\breve h_{1,\rms 1}(\scal)$ (see~\eqref{EqKEbreveh1} for the notation). The only issue is that (ultimately due to $\breve h_{1,\rms 1}(\scal)$ having an $\hat r^{-1}\log\hat r$ leading order term, cf.\ \eqref{EqKPureGaugeGens1} and \eqref{EqKEbreveh1})
\begin{equation}
\label{EqKELcommbreveh}
  [L,\hat t_1]\breve h_{1,\rms 1}(\scal)\in\cA^{(2,0),3-}(\hat X_b;S^2\,\Ttsc^*_{\hat X_b}\hat M_b)
\end{equation}
fails (just barely) to lie in the target space $\Hb^{\sfs-1,\alpha_\cD+2}$ of $\hat L(0)$ when $\alpha_\cD\in(-\frac32,-\frac12)$ due to (borderline) lack of decay (since $(\alpha_\cD+2)+\frac32>2$). Before we can define the correct augmentation of $L$, we thus need to approximately invert $\hat L(\sigma)$ on the input~\eqref{EqKELcommbreveh} (times an appropriate oscillatory factor). We do this in steps.

\begin{enumerate}
\item The first step is the inversion of the tf-normal operator $L_\tface(\hat\sigma)$ of $L$. (In the bundle splitting of $S^2\,\Ttsc^*_{\hat X_b}\hat M_b$ induced by the standard coordinate differentials, this is equal to the operator $(\Box_{\hat{\ubar g}})_\tface(\hat\sigma)$ in~\eqref{EqNMinktf} tensored with the $10\times 10$ identity matrix.) See Lemma~\ref{LemmaKEtf}.
\item We then use the information from the first step to find a good approximation (near $\zface\cup\tface\subset(\hat X_b)_\scbtop$) for $\hat L(\sigma)^{-1}(e^{i\sigma\hat r}[L,\hat t_1]\breve h_{1,\rms 1}(\scal))$; see Lemma~\ref{LemmaKEbreve}.
\item We can now define the appropriate terms in the first row of the augmentation of $L$; see Lemma~\ref{LemmaKECompl}.
\item Uniform resolvent bounds for the augmented operator now follow easily; see Proposition~\ref{PropKELo}.
\end{enumerate}

We introduce the notation
\begin{equation}
\label{EqKESph1plus}
  \Sph^1_+=\{\hat\sigma\in\C\colon|\hat\sigma|=1,\ \Im\hat\sigma\geq 0\}.
\end{equation}

\begin{lemma}[Transition face equation]
\label{LemmaKEtf}
  Let $\breve f_\pa\in\CI(\pa\hat X_b;S^2\,\Ttsc^*_{\pa\hat X_b}\hat M_b)$.\footnote{More simply put, $\breve f_\pa\in\CI(\Sph^2;\C^{10})$ when working with, say, the trivialization by differentials of standard coordinates.} For each $\hat\sigma\in\Sph^1_+$, set
  \[
    h_\tface(\hat\sigma,\hat r,\omega) := e^{-i\hat\sigma r'}L_\tface(\hat\sigma)^{-1}\bigl(e^{i\hat\sigma r'}r'{}^{-2} \breve f_\pa\bigr)\qquad\text{on}\ \tface=[0,\infty]_{r'}\times\Sph^2_\omega,
  \]
  where the inverse is the one given by \citeII{Lemma~\ref*{LemmaEstMcInvft}}. Then the following conclusions hold.
  \begin{enumerate}
  \item We have\footnote{See \cite[Definition~2.13]{HintzPrice} for the notation.}
    \begin{equation}
    \label{EqKEtfhtfMem}
      h_\tface(\hat\sigma,r',\omega) \in \cA^{1-,((0,1),1-)}(\tface;S^2\,\Ttsc^*_{\pa\hat X_b}\hat M_b),
    \end{equation}
    with continuous dependence on $\hat\sigma\in\Sph^1_+$; here, the weights/index sets refer to $r'=\infty$ and $r'=0$ (in this order).
  \item\label{ItKEtfLog} There exists a symmetric matrix $\Ups=(\Ups_{\mu\nu})\in\C^{4\times 4}_{\rm sym}$ so that, writing $\Ups\ubar h^\Ups:=\sum_{0\leq\mu\leq\nu\leq 3}\Ups_{\mu\nu}\ubar h^\Ups_{\mu\nu}$ where $\ubar h^\Ups_{\mu\nu}=\dd\hat z^\mu\otimes_s\dd\hat z^\nu$, we have\footnote{The function $\log\frac{\hat\sigma r'}{\hat\sigma r'+i}$ is equal to $\log r'$ plus smooth terms near $r'=0$, and is conormal with decay $|\hat\sigma r'|^{-1}$ as $|\hat\sigma r'|\to\infty$. We write $h_\tface$ in this manner---rather than with cutoff functions localizing the logarithmic term near $r'=0$---with an eye towards the holomorphicity statement in part~\eqref{ItKEtfHolo}.}
  \begin{equation}
  \label{EqKEtftildehtf}
    \tilde h_\tface(\hat\sigma,r',\omega) := \log\Bigl(\frac{\hat\sigma r'+i 0}{\hat\sigma r'+i}\Bigr)\Ups\ubar h^\Ups + h_\tface(\hat\sigma,r',\omega) \in \cA^{1-,((0,0),1-)}([0,\infty]_{\rho'}\times\Sph^2)
  \end{equation}
  \item\label{ItKEtfHolo} We have $h_\tface(\hat\sigma,r',\omega)=h_\tface(\hat\sigma r',\omega)$, with $h_\tface(z,\omega)$ holomorphic in $\Im z>0$ and continuous down to $\R_z\setminus\{0\}$. The analogous statements holds also for $\tilde h_\tface$.
  \end{enumerate}
\end{lemma}
\begin{proof}
  The conormality of $h_\tface(\hat\sigma,r',\omega)$ follows from \cite[Proposition~5.4]{VasyLowEnergyLAG} with $l<-\frac12$, $\tilde r=\infty$, and $\nu\in(\frac12,\frac32)$ arbitrarily close to $\frac32$; this implies $h_\tface(\hat\sigma)\in\cA^{1-,0-}(\tface)$ via Sobolev embedding. The more precise behavior at $r'=0$ follows from~\eqref{EqNMinktf}, which gives
  \[
    e^{-i\hat\sigma r'}L_\tface(\hat\sigma)e^{i\hat\sigma r'} \equiv r'{}^{-2}\bigl(-(r'\pa_{r'})^2-r'\pa_{r'}+\slDelta) \bmod r'{}^{-1}\Diffb^1([0,1)_{r'}\times\Sph^2),
  \]
  acting component-wise when decomposing tensors into $\ubar h^\Ups_{\mu\nu}$, via a normal operator argument: the normal operator here maps $-\log r'$ into $r'{}^{-2}$ and $(l(l+1))^{-1}Y_{l m}$, $l\in\N$, into $r'{}^{-2}Y_{l m}$. Therefore, near $r'=0$, $h_\tface(\hat\sigma)$ lies in $\cA^{(0,1),1-}([0,1)_{r'}\times\Sph^2)$. The decomposition~\eqref{EqKEtftildehtf}, with $\Ups$ such that $\breve f_\pa-\sum_{0\leq\mu\leq\nu\leq 3}\Ups_{\mu\nu}\ubar h^\Ups_{\mu\nu}$ is supported in $l\geq 1$ spherical harmonics, is now an immediate consequence since $\log(\frac{\hat\sigma r'+i 0}{\hat\sigma r'+i})-\log r'=\log(\frac{\hat\sigma+i 0}{\hat\sigma r'+i})$ is smooth down to $r'=0$, and $\log(\frac{\hat\sigma r'+i 0}{\hat\sigma r'+i})\in\cA^1([0,1)_{\rho'})$ in view of $\frac{\hat\sigma r'}{\hat\sigma r'+i}=1-\frac{i}{\hat\sigma r'+i}$.

  Finally, we exploit scaling in $(\sigma,\hat r^{-1})$: define $h'(\sigma,\hat r,\omega):=h_\tface(\frac{\sigma}{|\sigma|},|\sigma|\hat r,\omega)$; then
  \[
    e^{-i\sigma\hat r}\wh{\Box_{\hat{\ubar g}}}(\sigma)\bigl(e^{i\sigma\hat r}h'\bigr) = \hat r^{-2}\breve f_\pa
  \]
  can be solved using \cite[Proposition~5.4]{VasyLowEnergyLAG} again, or simply via the explicit (oscillatory integral) formula
  \[
    h'(\sigma,\hat r,\omega) = \int_{\R^3} \frac{e^{i\sigma(|\hat r\omega-\hat y|-\hat r+|\hat y|)}}{4\pi|\hat y-\hat r\omega|} |\hat y|^{-2}\breve f_\pa\Bigl(\frac{\hat y}{|\hat y|}\Bigr)\,\dd\hat y.
  \]
  This gives the holomorphicity of $h'(\cdot,\hat r,\omega)$ in $\Im\sigma>0$ and its continuity down to $\Im\sigma=0$. The identity $h_\tface(\hat\sigma,r',\omega)=h'(\hat\sigma r',1,\omega)$ thus proves part~\eqref{ItKEtfHolo}.
\end{proof}

From now on, we work with $\hat t_1=\hat t+\cT_1(\hat r)$ where $\cT_1$ is either defined as in~\eqref{EqKCDStarCoord}, or more permissively
\begin{equation}
\label{EqKEt1}
  \hat t_1=\hat t+\cT_1(\hat r),\qquad \cT_1(\hat r)\in-\hat r+c\log\hat r+\cA^0(\hat X_b),\ c\in\R.
\end{equation}
(The only difference is that when $c\neq-2\bhm$, the quantity $|\dd\hat t_1|_{\hat g_b^{-1}}$ is only of size $\hat r^{-1}$ instead of $\hat r^{-2}$; but this is sufficient for our purposes.\footnote{The function spaces on which we study waves on Kerr are so permissive that in Theorem~\ref{ThmKEFwd} below they allow for the usage of $\hat t-\hat r$ or $\hat t_1$ interchangeably (if one re-defines $\breve h_{1,\rms 1}$ accordingly); this would \emph{not} be the case if we were to blow up the light cone at infinity and record radiation fields, as done e.g.\ in \cite{HintzVasyMink4,HintzPrice}. Cf.\ the flexibility in the choice of $\tilde T$ in Lemmas~\ref{LemmaNMult} and \ref{LemmaNMultHi}.}) Define
\begin{equation}
\label{EqKEbreveh1}
  \breve h_{1,\rms 1}(\scal)=\breve h_{\rms 1}(\scal)-\cT_1 h_{\rms 1}(\scal),
\end{equation}
so $\hat t h_{\rms 1}(\scal)+\breve h_{\rms 1}(\scal)=\hat t_1 h_{\rms 1}(\scal)+\breve h_{1,\rms 1}(\scal)$. By~\eqref{EqKPureGauges1} and \eqref{EqKPureGaugeGens1}, we have $\breve h_{1,\rms 1}(\scal)\in\cA^{(1,1),2-}(\hat X_b;S^2\,\Ttsc^*_{\hat X_b}\hat M_b)$.

\begin{lemma}[Solving away a leading order term at $\tface$]
\label{LemmaKEbreve}
  For $\scal\in\scalspace_1$, set
  \[
    \breve f_\pa(\scal):=(\hat r^2[L,\hat t_1]\breve h_{1,\rms 1}(\scal))|_{\hat r=\infty} \in \CI(\pa\hat X_b;S^2\,\Ttsc^*_{\pa\hat X_b}\hat M_b).
  \]
  For $\hat\sigma\in\Sph^1_+$, use Lemma~\usref{LemmaKEtf} to define
  \[
    h_\tface(\scal,\hat\sigma,r',\omega):=e^{-i\hat\sigma r'}L_\tface(\hat\sigma)^{-1}(e^{i\hat\sigma r'}r'{}^{-2}\breve f_\pa(\scal)).
  \]
  Define the linear map $\Ups\colon\scalspace_1\to\C^{4\times 4}_{\rm sym}$ using Lemma~\usref{LemmaKEtf}\eqref{ItKEtfLog}, and define $\tilde h_\tface(\scal,z,\omega)$ (linear in $\scal\in\scalspace_1$, holomorphic in $\Im z>0$, continuous down to $\R_z\setminus\{0\}$) by
  \[
    h_\tface(\scal,\hat\sigma,r',\omega)=\log\Bigl(\frac{\hat\sigma r'}{\hat\sigma r'+i}\Bigr)\Ups(\scal)\ubar h^\Ups+\tilde h_\tface(\scal,\hat\sigma r',\omega).
  \]
  Writing $\Ups(\scal)h^\Ups:=\sum_{0\leq\mu\leq\nu\leq 3}\Ups_{\mu\nu}(\scal)h_{\mu\nu}^\Ups$ in the notation of~\eqref{EqKPureGaugeij}--\eqref{EqKPureGauge00}, set
  \begin{subequations}
  \begin{equation}
  \label{EqKEbreveRes}
  \begin{split}
    \breve f_{1,\rms 1}(\scal,\sigma) &:= e^{-i\sigma\cT_1(\hat r)}[L,\hat t_1]\breve h_{1,\rms 1}(\scal) \\
      &\quad \qquad + \hat L(\sigma)\Bigl( e^{-i\sigma\cT_1(\hat r)}\Bigl[ \log\Bigl(\frac{\sigma\hat r}{\sigma\hat r+i}\Bigr)\Ups(\scal)h^\Ups + \tilde h_\tface(\scal,\sigma\hat r,\omega)\Bigr]\Bigr)
  \end{split}
  \end{equation}
  for $\Im\sigma\geq 0$, $\sigma\neq 0$. Then, for $\hat\sigma\in\Sph^1_+$ and with $|\sigma|$ as the parameter on $(\hat X_b)_\scbtop$,
  \begin{equation}
  \label{EqKEbreveRes2}
  \begin{split}
    &\breve f_{1,\rms 1}(\scal,\hat\sigma\cdot) \in e^{-i\sigma\cT_1(\hat r)} \cA^{2-,3-,((0,0),1-)}((\hat X_b)_\scbtop;S^2\,\Ttsc^*_{\hat X_b}\hat M_b), \\
    &\breve f_{1,\rms 1}(\scal,\hat\sigma\cdot)|_\zface = [L,\hat t_1]\breve h_{1,\rms 1}(\scal) + \hat L(0)\breve h_{1,\rms 1}'(\scal) \in \cA^{3-}(\hat X_b;S^2\,\Ttsc^*_{\hat X_b}\hat M_b),
  \end{split}
  \end{equation}
  \end{subequations}
  where $\breve h_{1,\rms 1}'(\scal)\in\cA^{0-}(\hat X_b;S^2\,\Ttsc^*_{\hat X_b}\hat M_b)$.\footnote{The term $\hat L(0)\breve h'_{1,\rms 1}(\scal)$ thus eliminates the $\cO(\hat r^{-2})$ leading order term of $[L,\hat t_1]\breve h_{1,\rms 1}(\scal)$.}
\end{lemma}

The emergence of the logarithmic term here is analogous to \cite[Lemma~2.23]{HintzPrice} (where, however, only $\hat\sigma=1$ was considered, and holomorphicity considerations were sidestepped in a somewhat ad hoc manner, see \cite[Remarks~2.25 and 3.5]{HintzPrice}).

Thus, $e^{i\sigma\cT_1(\hat r)}$ times the second line in~\eqref{EqKEbreveRes} cancels the $\hat r^{-2}$ leading order term of $[L,\hat t_1]\breve h_{1,\rms 1}(\scal)$. Carefully note moreover that unlike $h_\tface(\scal,\frac{\sigma}{|\sigma|},|\sigma|\hat r)$, the tensor $\breve f_{1,\rms 1}$ does \emph{not} have a logarithmic leading order term at $\zface$ (cf.\ the index set $(0,0)$ in~\eqref{EqKEbreveRes2}), whose presence one may naively have expected in view of the logarithm in~\eqref{EqKEbreveRes}; the reason for its absence is that each $h^\Ups_{\mu\nu}$ (and thus $\Ups(\scal)h^\Ups$) lies in the nullspace of $\hat L(0)$. That is, we crucially rely on the fact that for all indicial solutions $\ubar h_{\mu\nu}=\dd\hat z^\mu\otimes_s\dd\hat z^\nu$ of $\hat L(0)$ at weight $0$ there exists a true solution $h_{\mu\nu}$ with leading order term $\ubar h_{\mu\nu}$.\footnote{While all $h_{\mu\nu}^\Ups$ are pure gauge, we do not actually need this information in this paper.} If $\hat L(0)$ were invertible on $L^2$-spaces with weight in the interval $(-\frac32,-\frac12)$, this would be automatic: acting on spaces with more permissive weights (which avoid indicial roots minus $\frac32$), the Fredholm index increases by the dimension of the space of indicial solutions when one crosses an indicial root of $\hat L(0)$, and in the absence of a cokernel this means that the dimension of nullspace increases by this amount. By contrast, in the present situation where $\hat L(0)$ has nontrivial cokernel for weights in $(-\frac32,-\frac12)$, the dimension of its cokernel may decrease upon crossing an indicial root and thus the dimension of its nullspace may increase only by a correspondingly smaller amount than expected from the dimension of the space of indicial solutions; this does not occur when passing the weight $-\frac32$ (since the elements of the cokernel decay at least quadratically rather than only like $\hat r^{-1}$), but it does occur when passing the weight $-\frac52$ (since then the quadratically vanishing elements of the cokernel get eliminated when passing to more negative weights). These considerations can be viewed as instances of the relative index theorem \cite[\S{6.1}]{MelroseAPS}.

\begin{proof}[Proof of Lemma~\usref{LemmaKEbreve}]
  Write $\breve h_2(\scal,\sigma,\hat r,\omega):=(\log\frac{\sigma\hat r}{\sigma\hat r+i})\Ups(\scal)h^\Ups+\tilde h_\tface(\scal,\sigma\hat r,\omega)$, which, as a function of $|\sigma|$ for fixed $\hat\sigma=\frac{\sigma}{|\sigma|}$, lies in $\cA^{1-,((0,0),1-),((0,1),1-)}((\hat X_b)_\scbtop)\subset\cA^{1-,0,0-}((\hat X_b)_\scbtop)$ in view of~\eqref{EqKEtfhtfMem}, with continuous dependence on $\hat\sigma=\frac{\sigma}{|\sigma|}\in\Sph^1_+$. Define $\breve{\ubar h}_2$ similarly, except with $\ubar h^\Ups$ in place of $h^\Ups$. Write moreover $r'=\hat r|\sigma|$. Since $[L,\hat t_1]\breve h_{1,\rms 1}(\scal)\equiv\hat r^{-2}\breve f_\pa(\scal)\bmod\cA^{3-}(\hat X_b)$, we then have
  \begin{align}
    e^{i\sigma\cT_1(\hat r)}\breve f_{1,\rms 1}(\scal,\sigma) & = e^{i\sigma\cT_1(\hat r)}\hat L(\sigma)\bigl(e^{-i\sigma\cT_1(\hat r)}\breve h_2(\scal)\bigr) + [L,\hat t_1]\breve h_{1,\rms 1}(\scal) \nonumber\\
  \label{EqKEbreveDiff}
  \begin{split}
      &\in \bigl( e^{i\sigma\cT_1(\hat r)}\hat L(\sigma) e^{-i\sigma\cT_1(\hat r)} - e^{-i\sigma\hat r}\hat L(\sigma)e^{i\sigma\hat r}\bigr)\breve h_2(\scal) \\
      &\quad + \bigl(e^{-i\sigma\hat r}\hat L(\sigma)e^{i\sigma\hat r} - e^{-i\hat\sigma r'}|\sigma|^2 L_\tface(\hat\sigma) e^{i\hat\sigma r'}\bigr)\breve h_2(\scal) \\
      &\quad + \bigl(|\sigma|^2 r'{}^{-2}\breve f_\pa(\scal) + \cA^{3-}\bigr) + e^{-i\hat\sigma r'}|\sigma|^2 L_\tface(\hat\sigma)\bigl(e^{i\hat\sigma r'}\breve{\ubar h}_2(\scal)\bigr) \\
      &\quad + e^{-i\hat\sigma r'}|\sigma|^2 L_\tface(\hat\sigma)e^{i\hat\sigma r'}\bigl(\breve h_2(\scal)-\breve{\ubar h}_2(\scal)\bigr).
  \end{split}
  \end{align}
  In the first line of the expression on the right, we note that $e^{i\sigma\cT_1}\hat L(\sigma)e^{-i\sigma\cT_1}$, as a sc-b-transition operator (for fixed $\hat\sigma\in\Sph^1_+$), is the conjugation of $e^{-i\sigma\hat r}\hat L(\sigma)e^{i\sigma\hat r}\in\Diffscbt^{2,1,2,0}(\hat X_b)$ (cf.\ \eqref{EqNscbtDiff} and~\eqref{EqNMinktf}) by $e^{i\sigma(\cT_1(\hat r)+\hat r)}$; but $\cT_1(\hat r)+\hat r=c\log\hat r+\tilde\cT_1(\hat r)$ with $\tilde\cT_1\in\cA^0$, so $e^{i\sigma(\cT_1(\hat r)+\hat r)}=\hat r^{i c\sigma}e^{i\sigma\tilde\cT_1(\hat r)}$, the conjugation by which of the frame $\hat R:=(1+\hat r|\sigma|)^{-1}\hat r\pa_{\hat r}$, $\Omega:=(1+\hat r|\sigma|)^{-1}\pa_\omega$ of $\Vscbt(\hat X_b)$ is $(1+\hat r|\sigma|)^{-1}\hat r\pa_{\hat r}-\frac{\sigma}{1+\hat r|\sigma|}(i c+i\hat r\pa_{\hat r}\tilde\cT_1)\equiv\hat R\bmod\cA^{1,1,1}((\hat X_b)_\scbtop)$ and $\Omega$. Therefore, the first line of~\eqref{EqKEbreveDiff} lies in
  \[
    \cA^{1,1,1}\Diffscbt^{2,1,2,0}(\hat X_b)\cA^{1-,0,0-}((\hat X_b)_\scbtop) \subset \cA^{3-,3,1-}((\hat X_b)_\scbtop).
  \]
  Turning to the second line of~\eqref{EqKEbreveDiff}, we note that $\sigma\hat r=\hat\sigma r'$ and $\hat L(\sigma)-|\sigma|^2 L_\tface(\hat\sigma)\in\Diffscbt^{2,2,3,0}(\hat X_b)$; so the second line lies in $\cA^{3-,3,((0,1),1-)}$. The third line lies in $\cA^{3-}(\hat X_b)\subset\cA^{3-,3-,(0,0)}((\hat X_b)_\scbtop)$ since, by definition of $\breve{\ubar h}_2(\scal)$, the terms involving $\breve f_\pa(\scal)$ and $\breve{\ubar h}_2(\scal)$ cancel. The fourth line finally lies in
  \[
    \Diffscbt^{2,1,2,0}(\hat X_b)\cA^{1-,1-,(0,1)}((\hat X_b)_\scbtop)\subset\cA^{2-,3-,(0,1)}((\hat X_b)_\scbtop).
  \]

  Altogether, we have shown that
  \[
    e^{i\sigma\cT_1(\hat r)}\breve f_{1,\rms 1}(\scal,\sigma)\in\cA^{2-,3-,((0,1),1-)}((\hat X_b)_\scbtop).
  \]
  But directly from~\eqref{EqKEbreveRes}, its $\log r'$ coefficient at $\zface^\circ\subset(\hat X_b)_\scbtop$ is equal to $\hat L(0)(\breve\Ups(\scal)h^\Ups)$ and thus vanishes. We can therefore restrict $e^{i\sigma\cT_1(\hat r)}\breve f_{1,\rms 1}(\scal,\hat\sigma\cdot)$ to $|\sigma|=0$, where in view of $\log(\sigma\hat r+i 0)=\log\hat r+\log(\sigma+i 0)$ it equals $[L,\hat t_1]\breve h_{1,\rms 1}(\scal)+\hat L(0)((\log\hat r)\breve\Ups(\scal)h^\Ups+\tilde h_\tface(\scal,0,\omega))$; the argument of $\hat L(0)$ here, denoted $\breve h'_{1,\rms 1}(\scal)$ in~\eqref{EqKEbreveRes2}, lies in $\cA^{(0,1),1-}(\hat X_b)\subset\cA^{0-}(\hat X_b)$. The proof is complete.
\end{proof}

We can now introduce the key objects appearing in the augmentation of $L$.

\begin{lemma}[Complement of the range of the zero energy operator]
\label{LemmaKECompl}
  We use the notation of Lemmas~\usref{LemmaKGaugePot}, \usref{LemmaKLinKerr}, and \usref{LemmaKEbreve}, and~\eqref{EqKEt1}. Define, for $\dot b\in\C^4$, $\scal\in\scalspace_1$, and $\sigma\neq 0$,
  \begin{equation}
  \label{EqKECompl}
  \begin{split}
    f_{\rm Kerr}(\sigma,\dot b) &:= \hat L(\sigma)\bigl( e^{-i\sigma\cT_1(\hat r)}i\sigma^{-1}\hat g^{\prime\Ups}_b(\dot b)\bigr), \\
    f_{\rm COM}(\sigma,\scal) &:= \hat L(\sigma)\Bigl( e^{-i\sigma\cT_1(\hat r)} \Bigl[ -\sigma^{-2}h_{\rms 1}(\scal) + i\sigma^{-1}\breve h_{1,\rms 1}(\scal) \\
      &\quad \hspace{8em} + \log\Bigl(\frac{\sigma\hat r}{\sigma\hat r+i}\Bigr)\Ups(\scal)h^\Ups + \tilde h_\tface(\scal,\sigma\hat r,\omega) \Bigr] \Bigr).
  \end{split}
  \end{equation}
  Then for all $\hat\sigma\in\Sph^1_+$, we have
  \begin{equation}
  \label{EqKECompl2}
  \begin{split}
    &e^{i\sigma\cT_1}f_{\rm Kerr}(\hat\sigma\cdot,\dot b) \in \cA^{2,3-,(0,0)}((\hat X_b)_\scbtop;S^2\,\Ttsc^*_{\hat X_b}\hat M_b), \\
      &\hspace{4em} f_{\rm Kerr}(\hat\sigma\cdot,\dot b)|_\zface=[L,\hat t_1]\hat g^{\prime\Ups}_b(\dot b)\in\cA^{3-}(\hat X_b;S^2\,\Ttsc^*_{\hat X_b}\hat M_b), \\
    &e^{i\sigma\cT_1}f_{\rm COM}(\hat\sigma\cdot,\scal) \in \cA^{2-,3-,((0,0),1-)}((\hat X_b)_\scbtop;S^2\,\Ttsc^*_{\hat X_b}\hat M_b)), \\
      &\hspace{4em} f_{\rm COM}(\hat\sigma\cdot,\scal)|_\zface=\frac12[[L,\hat t_1],\hat t_1]h_{\rms 1}(\scal)+[L,\hat t_1]\breve h_{1,\rms 1}(\scal)+\hat L(0)\breve h'_{1,\rms 1}(\scal) \\
      &\hspace{4em} \phantom{f_{\rm COM}(\hat\sigma,\cdot,\scal)|_\zface}\,\in \cA^{3-}(\hat X_b;S^2\,\Ttsc^*_{\hat X_b}\hat M_b).
  \end{split}
  \end{equation}
  Moreover, $f_{\rm Kerr}$ and $f_{\rm COM}$ are holomorphic in $\{\sigma\in\C\colon\Im\sigma>0\}$, with values in the space $\CIdot(\hat X_b;S^2\,\Ttsc^*_{\hat X_b}\hat M_b)$ (i.e.\ the space of tensors which are, component-wise, Schwartz).
\end{lemma}
\begin{proof}
  By~\eqref{EqNMinkLdecompSpec}, the action of $e^{i\sigma\hat t_1}L e^{-i\sigma\hat t_1}$ on stationary symmetric 2-tensors $h=h(\hat x)$ is given by
  \[
    e^{i\sigma\cT_1(\hat r)}\hat L(\sigma)(e^{-i\sigma\cT_1(\hat r)}h) = e^{i\sigma\hat t_1}L(e^{-i\sigma\hat t_1}h) = \Bigl(L - i\sigma[L,\hat t_1] - \frac{\sigma^2}{2}[[L,\hat t_1],\hat t_1]\Bigr)h.
  \]
  Recall that $\cT_1\equiv-(\hat r-c\log\hat r)+\tilde\cT_1$, $\tilde\cT_1\in\cA^0$. Using Lemma~\ref{LemmaKStruct}, the commutator term can be written as
  \begin{align*}
    [L,\hat t_1]h &= [\Box_{\hat{\ubar g}},\hat t-\hat r]h + [L-\Box_{\hat{\ubar g}},\hat t-\hat r]h + [L,\hat t_1-(\hat t-\hat r)]h \\
      &= i\pa_\sigma{}^{\rm o}\wh{\Box_{\hat{\ubar g}}}(0)h + \rho_\cD^2\Diffb^1(\hat X_b)h + \cA^2\Diffb^1(\hat X_b)h,
  \end{align*}
  where for the final term we use that the function $\hat t_1-(\hat t-\hat r)=c\log\hat r+\tilde\cT_1$ is stationary, and its commutator with $\hat L(0)\in\rho_\cD^2\Diffb^2(\hat X_b)$ is of class $\rho_\cD^2\Diffb^1+\cA^2\Diffb^1=\cA^2\Diffb^1$. In particular, for $h\in\cA^{(1,0),2-}(\hat X_b;S^2\,\Ttsc^*_{\hat X_b}\hat M_b)$, we have, due to~\eqref{EqNMinkConj} mapping $\rho_\cD\CI(\hat X_b)$ into $\rho_\cD^3\CI(\hat X_b)$,
  \begin{equation}
  \label{EqKEComplComm}
    [L,\hat t_1]h \in \cA^{3-}.
  \end{equation}

  For later use, we moreover note that $\frac12[[L,\hat t_1],\hat t_1]$ is the scalar operator given by multiplication by the function $|\dd\hat t_1|^2_{\hat g_b^{-1}}$, which lies in $\rho_\cD\CI(\hat M_b)$ or equivalently (due to its stationarity) in $\hat r^{-1}\CI(\hat X_b)$. (For the choice $c=-2\bhm$, it lies in $\hat r^{-2}\CI(\hat X_b)$, but we do not need this stronger decay for this quantity below.)

  Consider now $h=\hat g^{\prime\Ups}_b(\dot b)$, for which we do have~\eqref{EqKEComplComm} by Lemma~\ref{LemmaKLinKerr}. Since $L(\hat g^{\prime\Ups}_b(\dot b))=0$, we obtain
  \begin{align*}
    e^{i\sigma\cT_1}f_{\rm Kerr}(\sigma,\dot b)&=i\sigma^{-1}\Bigl(-i\sigma[L,\hat t_1]\hat g^{\prime\Ups}_b(\dot b) - \frac{\sigma^2}{2}[[L,\hat t_1],\hat t_1]\hat g^{\prime\Ups}_b(\dot b)\Bigr) \\
      &\in [L,\hat t_1]\hat g^{\prime\Ups}_b(\dot b) + \sigma\cA^2(\hat X_b) \\
      &\subset [L,\hat t_1]\hat g^{\prime\Ups}_b(\dot b) + \cA^{2,3,(1,0)}((\hat X_b)_\scbtop) \\
      &\subset \cA^{2,3-,(0,0)}((\hat X_b)_\scbtop);
  \end{align*}
  this implies the first part of~\eqref{EqKECompl2}.
  
  For $f_{\rm COM}$, we note that, using the notation~\eqref{EqKEbreveRes},
  \begin{align*}
    e^{i\sigma\cT_1}f_{\rm COM}(\sigma,\scal) &= i\sigma^{-1}\bigl(L \breve h_{1,\rms 1}(\scal) + [L,\hat t_1]h_{\rms 1}(\scal)\bigr) \\
      &\quad + \frac12[[L,\hat t_1],\hat t_1]h_{\rms 1}(\scal) + e^{i\sigma\cT_1}\breve f_{1,\rms 1}(\sigma,\scal) - \frac{i\sigma}{2}[[L,\hat t_1],\hat t_1]\breve h_{1,\rms 1}(\scal).
  \end{align*}
  The first line vanishes since $L(\hat t_1 h_{\rms 1}(\scal)+\breve h_{\rms 1}(\scal))=0$. The first term in the second line lies in $\rho_\cD\cA^2(\hat X_b)\subset\cA^{3,3,(0,0)}((\hat X_b)_\scbtop)$, and the second term lies in $\cA^{2-,3-,((0,0),1-)}$ by~\eqref{EqKEbreveRes2}. The final term lies in $\sigma\rho_\cD\cA^{1-}(\hat X_b)\subset\cA^{2-,3-,(1,0)}((\hat X_b)_\scbtop)$.

  The holomorphicity statement follows from the definition~\eqref{EqKECompl} and the fact that, for $\Im\sigma>0$, the function $e^{-i\sigma\cT_1(\hat r)}$ is exponentially decaying as $\hat r\to\infty$.
\end{proof}

\begin{prop}[Low energy resolvent]
\label{PropKELo}
  In the notation of Lemmas~\usref{LemmaKEbreve} and \usref{LemmaKECompl}, define the symmetric 2-tensor
  \begin{equation}
  \label{EqKELomod}
  \begin{split}
    f_{\rm mod}(\sigma,(\dot b,\scal))&:=f_{\rm Kerr}(\sigma,\dot b)+f_{\rm COM}(\sigma,\scal) \\
      &= \hat L(\sigma)\Bigl( e^{-i\sigma\cT_1(\hat r)} \Bigl[ -\sigma^{-2}h_{\rms 1}(\scal) + i\sigma^{-1}\bigl(\breve h_{1,\rms 1}(\scal)+\hat g^{\prime\Ups}_b(\dot b)\bigr) \\
      &\quad \hspace{8em} + \log\Bigl(\frac{\sigma\hat r}{\sigma\hat r+i}\Bigr)\Ups(\scal)h^\Ups + \tilde h_\tface(\scal,\sigma\hat r,\omega) \Bigr] \Bigr),
  \end{split}
  \end{equation}
  which is linear in $(\dot b,\scal)$. Fix $f_1^*,\ldots,f_7^*\in\CIc(\hat X_b^\circ;S^2 T^*_{\hat X_b^\circ}\hat M_b^\circ)$ so that the $7\times 7$ matrix of inner products of $\hat g^{\prime\Ups}_b(\dot b_j)$, $\dot b_1=(1,0,0,0)$, $\dot b_2=(0,1,0,0)$, $\dot b_3=(0,0,1,0)$, $\dot b_4=(0,0,0,1)$, and $h_{\rms 1}(\scal_j)$, $j=1,2,3$ (where $\scal_j=\frac{\hat x_j}{|\hat x|}$) with the $f_k^*$ is invertible; set $f^*(h)=(\la h,f_j^*\ra)_{j=1,\ldots,7}$. Identify $\scalspace_1\cong\C^3$ and define the operator
  \begin{equation}
  \label{EqKELowtL}
    \wt L(\sigma) := \begin{pmatrix} \hat L(\sigma) & f_{\rm mod}(\sigma,\cdot) \\ f^* & 0 \end{pmatrix} \colon \sD'(\hat X_b^\circ;S^2 T^*_{\hat X_b^\circ}\hat M_b^\circ) \oplus \C^7 \to \sD'(\hat X_b^\circ;S^2 T^*_{\hat X_b^\circ}\hat M_b^\circ) \oplus \C^7.
  \end{equation}
  Define
  \[
    \|(h,(\dot b,\scal))\|_{\tilde H_{\scbtop,|\sigma|}^{\sfs,\sfr,\ell,0}}:=\|h\|_{H_{\scbtop,|\sigma|}^{\sfs,\sfr,\ell,0}(\hat X_b;S^2\,\Ttsc^*_{\hat X_b}\hat M_b)}+|\dot b|+|\scal|.
  \]
  Let $\sfs,\alpha_\cD$ be strongly Kerr-admissible orders with $\alpha_\cD\in(-\frac32,-\frac12)$ so that also $\sfs-1,\alpha_\cD$ are strongly Kerr-admissible, and work with the sc-b-transition orders induced by $\sfs$ (depending on $\arg\hat\sigma$ where $\hat\sigma=\frac{\sigma}{|\sigma|}$) analogously to Remark~\usref{RmkKCDOrders}. Then there exist constants $\sigma_0\in(0,1]$ and $C>0$ so that, for all $\hat\sigma\in e^{i[0,\pi]}$, we have the estimate
  \begin{equation}
  \label{EqKELo}
    \| (h,(\dot b,\scal)) \|_{\tilde H_{\scbtop,|\sigma|}^{\sfs,\sfs+\alpha_\cD,\alpha_\cD,0}} \leq C\| \wt L(\hat\sigma|\sigma|)(h,(\dot b,\scal)) \|_{\tilde H_{\scbtop,|\sigma|}^{\sfs-1,\sfs+\alpha_\cD+1,\alpha_\cD+2,0}},\qquad |\sigma|\leq\sigma_0.
  \end{equation}
  Furthermore, for $|\sigma|\leq\sigma_0$, $\Im\sigma>0$, the operator $\wt L(\sigma)^{-1}\colon\tilde H_\scop^{\sfs-1,\sfs+\alpha_\cD}:=\Hsc^{\sfs-1,\sfs+\alpha_\cD}\oplus\C^7\to\tilde H_\scop^{\sfs,\sfs+\alpha_\cD}$ is holomorphic. (Here, $\sfs$ is arbitrary near $\hat r=\infty$ in the sense of Remark~\usref{RmkKCDOrders}.)
\end{prop}
\begin{proof}
  The proof is analogous to (but due to the absence of any additional parameters beyond $\sigma$ simpler than) that of Proposition~\ref{PropKCD}. We use Lemma~\ref{LemmaNMult} with $\tilde T=-\cT_1-\hat r$ and the memberships~\eqref{EqKECompl2} to deduce the uniform boundedness of $f_{\rm mod}(\hat\sigma\cdot,(\dot b,\scal))$ in $H_{\scbtop,|\sigma|}^{\sfs-1,\sfs+\alpha_\cD+1,\alpha_\cD+2,0}$ (using that $\alpha_\cD<-\frac12$, and $\sfs+\alpha_\cD<-\frac12$ at the outgoing radial set). On the level of estimates, we again combine symbolic estimates and the invertibility of the transition face normal operators (Lemma~\ref{LemmaKtfInv}) with the invertibility of the zero energy operator $\wt L(0)$---the latter of which we proceed to demonstrate. If $h\in\Hb^{\sfs,\alpha_\cD}(\hat X_b;S^2\,\Ttsc^*_{\hat X_b}\hat M_b)$ and $\dot b\in\C^4$, $\scal\in\scalspace_1$ satisfy $\wt L(\sigma)(h,(\dot b,\scal))=0$, then Lemma~\ref{LemmaKECompl} implies
  \[
    \hat L(0)\bigl(h+\breve h'_{1,\rms 1}(\scal)\bigr) + [L,\hat t_1]\hat g_b^{\prime\Ups}(\dot b) + \frac12[[L,\hat t_1],\hat t_1]h_{\rms 1}(\scal) + [L,\hat t_1]\breve h_{1,\rms 1}(\scal) = 0.
  \]
  Write $\dot b=(\dot\bhm,\dot\bha)$. Taking the inner product of this equation with $h^*_{\rms 0}$ implies that $\dot\bhm=0$ by Lemma~\ref{LemmaKEL2}, where we use that $\la\hat L(0)h',h^*_{\rms 0}\ra=0$ for $h'=h+\breve h'_{1,\rms 1}(\scal)$---in fact, for any $h'\in\cA^{0-}$ or more generally $h'\in\cA^{-1+}$---upon integration by parts. Taking inner products with $h^*_{\rmv 1}(\vect)$ for $\vect\in\vectspace_1$ implies $\dot\bha=0$; and then taking inner products with $h^*_{\rms 1}(\scal')$, $\scal'\in\scalspace_1$ implies $\scal=0$. Therefore, $\hat L(0)h=0$, which by Proposition~\ref{PropKE0} implies $h\in\cK$ (see~\eqref{EqKE0Ker}). But since $f^*(h)=0$, the nondegeneracy condition on $f^*$ forces $h=0$, as desired.

  The invertibility of $\wt L$ for $\Im\sigma\geq 0$, $|\sigma|\leq\sigma_0$, on the direct sum of the spaces in Theorem~\ref{ThmKCDUnmod} (except with bundles $S^2\,\Ttsc^*_{\hat X_b}\hat M_b$) follows from its injectivity by a Fredholm index argument as at the end of Step 1 in the proof of Proposition~\ref{PropKCD}. To prove the holomorphicity of $\wt L(\sigma)^{-1}$ in $\Im\sigma>0$, we first formally write
  \begin{equation}
  \label{EqKELoDer}
    \pa_\sigma\wt L(\sigma)^{-1} = -\wt L(\sigma)^{-1}\circ\pa_\sigma\wt L(\sigma)\circ\wt L(\sigma)^{-1},
  \end{equation}
  and note that
  \[
    \pa_\sigma\wt L(\sigma) = \begin{pmatrix} \pa_\sigma\hat L(\sigma) & \pa_\sigma f_{\rm mod}(\sigma,\cdot) \\ 0 & 0 \end{pmatrix} \colon \Hsc^{\sfs,\sfs+\alpha_\cD}\oplus\C^7 \to \Hsc^{\sfs-1,\sfs+\alpha_\cD}\oplus\C^7;
  \]
  this follows from the exponential decay in $\hat r$ of $\pa_\sigma f_{\rm mod}(\sigma,(\dot b,\scal))$, which is a consequence of~\eqref{EqKECompl2}. Thus the right hand side of~\eqref{EqKELoDer} is well-defined as an operator $\tilde H_\scop^{\sfs-1,\sfs+\alpha_\cD}\to\tilde H_\scop^{\sfs,\sfs+\alpha_\cD}$. The equation~\eqref{EqKELoDer} can be justified by taking the limit of finite difference quotients. In a similar vein, the $\pa_{\bar\sigma}$-derivative of $\wt L(\sigma)$ vanishes in view of the holomorphic dependence of the coefficients of $\wt L(\sigma)$---in particular, that of $f_{\rm mod}$, which we carefully arranged---on $\sigma$.
\end{proof}

%%%%%%%%%%%%%%%%%%%%%%%%%%%%%%%%%%%%%%%%%%%%%%%%%%
\subsection{Forward solutions of the linearized gauge-fixed Einstein equations}
\label{SsKE2}

We shall now control forward solutions of $L h=f$ using our analysis of $\hat L(\sigma)$. Inspecting the first line of $\wt L(\sigma)$ in~\eqref{EqKELowtL}, we see that the solution of $\hat L(\sigma)\hat u(\sigma)=\hat f(\sigma)$ is the sum of singular terms---of size $\sigma^{-2}$ and $\sigma^{-1}$ relative to $\hat f(\sigma)$, with spatial dependence given by $\hat g^{\prime\Ups}_b$, $h_{\rms 1}$, $\breve h_{\rms 1}$ to leading order at $\zface\subset(\hat X_b)_\scbtop$---and a term of the same size as $\hat f(\sigma)$. The spaces introduced in Definition~\ref{DefNHdashb} will be used to capture the inverse Fourier transform of the singular terms.

The following result concerns forcing terms $f$ with, roughly speaking, more than inverse quadratic decay in $\hat r$ (see also \citeII{Remark~\ref*{RmkSc3bSpace}}). The analysis of forward solutions when $f$ has weaker decay (which is crucial for our application) in~\S\ref{SsKEL} will be similar but more involved, and therefore we start with a simpler setting first.

\begin{thm}[3b-estimates for forward solutions]
\label{ThmKEFwd}
  Recall $\hat t_1$ from~\eqref{EqKEt1}. Fix $\hat{\ubar t}_1\geq 1$ so that $\{\hat t_1\geq\hat{\ubar t}_1\}\subset\{\hat t\geq 0\}$. Let $\hat M_b^+:=\ol{\{\hat t\geq 0\}}\subset\hat M_b$. Fix
  \[
    \phi_+\in\CI(\hat M_b),\qquad \phi_+=0\ \text{for}\ \hat t\leq 2\hat r,\quad \phi_+=1\ \text{for}\ \hat t\geq 3\hat r,
  \]
  and fix $\chi_0\in\sS(\R)$ with $\chi_0(0)=1$. Let $\sfs\in\CI(\Stb^*\hat M_b)$, $\alpha_\cD\in\R$ be strongly Kerr-admissible orders with $\alpha_\cD\in(-\frac32,-\frac12)$ so that also $\sfs-1,\alpha_\cD$ are strongly Kerr-admissible. Then there exists a continuous linear map
  \begin{equation}
  \label{EqKEFwdMap}
  \begin{split}
    L_+^{-1} &\colon \dot H_\tbop^{\sfs,\alpha_\cD+2,0}(\hat M_b^+;S^2\,\Ttsc^*\hat M_b) \\
    &\quad \to \bigcap_{\eta>0}\dot H_\tbop^{\sfs,-\frac32-\eta,0}(\hat M_b^+;S^2\,\Ttsc^*\hat M_b) \oplus \hat t_1\dot H_{-;\bop}^{(\infty;1)}([\hat{\ubar t}_1,\infty];\R^4) \oplus \hat t_1^2\dot H_{-;\bop}^{(\infty;2)}([\hat{\ubar t}_1,\infty];\scalspace_1) \\
    &\quad\qquad \oplus (\log|\sigma|)L^2([-1,1]_\sigma;\C^{4\times 4}_{\rm sym})
  \end{split}
  \end{equation}
  with the following properties.
  \begin{enumerate}
  \item\label{ItKEFwdSol}{\rm (Solution.)} Given $f\in\dot H_\tbop^{\sfs,\alpha_\cD+2,0}$ and $L_+^{-1}f=:\bigl(h_0,\dot b,\scal,\wh{\Ups_{(0)}}\bigr)$, set\footnote{The membership of $h_+^\Ups\bigl(\wh{\Ups_{(0)}}\bigr)$ follows, for $\wh{\Ups_{(0)}}\in(\log|\sigma|)L^2([-1,1];\C^{4\times 4}_{\rm sym})$, by applying Lemma~\ref{LemmaNL2FT} with arbitrary $\alpha>0$.}
    \begin{equation}
    \label{EqKEFwdhUps}
      h_+^\Ups\bigl(\wh{\Ups_{(0)}}\bigr) := \phi_+\frac{1}{2\pi}\int_{-1}^1 e^{-i\sigma\hat t_1} \chi_0(\sigma\hat r) \wh{\Ups_{(0)}}(\sigma)h^\Ups\,\dd\sigma \in \bigcap_{\eta>0} \dot H_\tbop^{\infty,-\frac32-\eta,-\eta}(\hat M_b^+),
    \end{equation}
    where $\wh{\Ups_{(0)}}(\sigma)h^\Ups=\sum_{0\leq\mu\leq\nu\leq 3}\wh{\Ups_{(0)}}(\sigma)_{\mu\nu}h_{\mu\nu}^\Ups$. Set further
    \begin{equation}
    \label{EqKEFwdSol}
      h := h_0 + \hat g^{\prime\Ups}_b\bigl(\dot b(\hat t_1)\bigr) + h_{\rms 1}\bigl(\scal(\hat t_1)\bigr) + \breve h_{1,\rms 1}\bigl(\pa_{\hat t_1}\scal(\hat t_1)\bigr) + h_+^\Ups\bigl(\wh{\Ups_{(0)}}\bigr).
    \end{equation}
    (Here, we use the notation~\eqref{EqKEbreveh1}.) Then $L h=f$.
  \item\label{ItKEFwdIndep}{\rm (Independence of orders.)} $L_+^{-1}$ is independent of the choice of $\sfs,\alpha_\cD$, in the sense that it is the unique continuous extension of a map with domain of definition equal to $\CIc((\hat M_b^+)^\circ;S^2 T^*\hat M_b^\circ)$.
  \end{enumerate}
\end{thm}

While the uniqueness of the forward solution $h$ of $L h=f$ follows from standard hyperbolic theory, the \emph{decomposition}~\eqref{EqKEFwdSol} of $h$ is not unique; the point of part~\eqref{ItKEFwdIndep} is then that our construction of $L_+^{-1}$ produces a particular decomposition~\eqref{EqKEFwdSol} which is independent of the orders $\sfs,\alpha_\cD$.

We explain the notation in~\eqref{EqKEFwdSol} in two equivalent ways. The first way is to foliate $\hat M_b^\circ=\R_{\hat t_1}\times\hat X_b^\circ$, and to define $\hat g^{\prime\Ups}_b(\dot b(\hat t_1))$ to be the linearized Kerr metric (i.e.\ a symmetric 2-tensor built from $\dd\hat t,\dd\hat x$) with parameters $\dot b(\hat t_1)$ on the $\hat t_1$-level set. The second way is to consider $\hat g^{\prime\Ups}_b\colon\C^4\to\CI(\hat M_b;S^2\,\Ttsc^*\hat M_b)$ as a linear map, or equivalently as an element $\hat\fg^{\prime\Ups}_b\in\CI(\hat M_b;(\ul\C^4)^*\otimes S^2\,\Ttsc^*\hat M_b)$ where $\ul\C^4=\hat M_b\times\C^4$ is the trivial bundle. Then $\hat g^{\prime\Ups}_b(\dot b(\hat t_1))$, at a point $p\in\hat M_b$, is defined as $\hat\fg^{\prime\Ups}_b|_p(\dot b(\hat t_1(p)))\in S^2\,\Ttsc^*_p\hat M_b$.

\begin{rmk}[Linear stability of subextremal Kerr]
\label{RmkKEFwdLinStab}
  It is a natural question whether one can conclude the linear stability of subextremal Kerr spacetimes from Theorem~\ref{ThmKEFwd} (and from Theorem~\ref{ThmKEFwdW} below for more general initial data with weaker decay). The issue of only having limited (3b-)regularity is easily overcome by proving additional module regularity (see Theorem~\ref{ThmKHi} as well as Remarks~\ref{RmkKHiSpacesNoEps} and \ref{RmkKHiFwdSolNoEps} below). Even then, Sobolev embedding only gives the pointwise boundedness for $h_0$ (as it has weight $0$ at $\cT\subset\hat M_b$) and $\hat t_1^{\frac12}$ upper bounds for $\dot b(\hat t_1)$ (and $\hat t_1^{\frac32}$ for $\scal$). What is missing is, on the spectral side, a more precise description of $\hat L(\sigma)^{-1}$ near $\sigma=0$ \emph{which encodes higher regularity in $\sigma$}. (This is discussed in a slightly different functional analytic framework in \cite[\S{12}]{HaefnerHintzVasyKerr}.) As $\sigma$-regularity is not relevant for the purposes of the present paper, we leave this open for future investigations.
\end{rmk}

\begin{proof}[Proof of Theorem~\usref{ThmKEFwd}]
  Consider first $f\in\CIc((\hat M_b^+)^\circ;S^2 T^*\hat M_b^\circ)$. Then $\hat f(\sigma)$ is holomorphic in $\sigma$, of uniformly compact support, and for all $m$ there exists a constant $C_m$ so that $\|\hat f(\sigma)\|_{H^m}\leq C_m$ for all $\Im\sigma\geq 0$. Using Propositions~\ref{PropKENon0} and \ref{PropKELo}, we now define
  \begin{alignat*}{2}
    \hat h(\sigma) &:= \hat L(\sigma)^{-1}\hat f(\sigma),&\qquad& \Im\sigma\geq 0,\ \sigma\neq 0, \\
    \bigl(\wh{h_0}(\sigma), \dot b_0(\sigma), \scal_0(\sigma) \bigr) &:= \wt L(\sigma)^{-1} \bigl( \hat f(\sigma), 0 \bigr), &\qquad& \Im\sigma\geq 0,\ |\sigma|\leq\sigma_0.
  \end{alignat*}
  We set
  \begin{equation}
  \label{EqKEFwdb1S1}
    \wt{\dot b}_1(\sigma) := i\sigma^{-1}\dot b_0(\sigma),\qquad
    \wt\scal_1(\sigma) := -\sigma^{-2}\scal_0(\sigma).
  \end{equation}
  By the invertibility of $\hat L(\sigma)$, we have
  \begin{equation}
  \label{EqKEFwdDecomp}
  \begin{split}
    \hat h(\sigma) &= \wh{h_0}(\sigma) + e^{-i\sigma\cT_1(\hat r)}\Bigl[\hat g^{\prime\Ups}_b\bigl(\wt{\dot b}_1(\sigma)\bigr) + h_{\rms 1}\bigl(\wt\scal_1(\sigma)\bigr) + \breve h_{1,\rms 1}\bigl(-i\sigma\wt\scal_1(\sigma)\bigr) \\
      &\quad \hspace{8em} + \log\Bigl(\frac{\sigma\hat r}{\sigma\hat r+i}\Bigr)\Ups\bigl(-\sigma^2\wt\scal_1(\sigma)\bigr)h^\Ups + \tilde h_\tface\bigl(-\sigma^2\wt\scal_1(\sigma),\sigma\hat r,\omega\bigr) \Bigr]
  \end{split}
  \end{equation}
  for $\Im\sigma\geq 0$, $0<|\sigma|\leq\sigma_0\leq 1$. Fix a cutoff $\psi\in\CIc([0,\sigma_0))$ which equals $1$ on $[0,\frac{\sigma_0}{2}]$. Using the uniform bounds for $\dot b_0(\sigma),\scal_0(\sigma)$ in terms of $\hat f(\sigma)$, Plancherel's theorem gives
  \begin{equation}
  \label{EqKEFwdCorrBd}
    \|\cF^{-1}(\psi(|\cdot|)\dot b_0(\cdot))\|_{H^m(\R)},\ \ \|\cF^{-1}(\psi(|\cdot|)\scal_0(\cdot))\|_{H^m(\R)} \leq C_m\|f\|_{\Htb^{\sfs,\alpha_\cD+2,0}}.
  \end{equation}

  %%%%%%%%%%
  \pfstep{Forward solution.} The large $|\sigma|$ estimates \citeII{Proposition~\ref*{PropEstFThi}} and the invertibility of $\hat L(\sigma)$ for bounded $\sigma$ in the upper half plane imply that, for all $\phi\in\CIc(\hat X_b^\circ)$, the function $\phi\hat h(\sigma)$ is uniformly bounded in $H^m$ for $\Im\sigma\geq\gamma$ for any fixed $\gamma>0$. Therefore,
  \begin{equation}
  \label{EqKEFwdh}
    h(\hat t) := \frac{1}{2\pi}\int_{\Im\sigma=\gamma} e^{-i\sigma\hat t}\hat h(\sigma)\,\dd\sigma
  \end{equation}
  is well-defined and defines an element of $\CI(\hat M_b^\circ;S^2 T^*\hat M_b^\circ)$ which is supported in $\hat t\geq 0$ by the Paley--Wiener theorem. Moreover, this function is independent of $\gamma>0$ by Cauchy's integral theorem. It is therefore equal to the unique forward solution of $L h=f$.

  %%%%%%%%%%
  \pfstep{First contour shifting.} We note that $\phi\hat h(\sigma)$ is continuous in $H^m$ for $\{\sigma\in\C\colon \sigma\neq 0,\ \Im\sigma\geq 0\}$ and some $m\in\R$; this follows from the uniform boundedness of
  \[
    \hat L(\sigma)^{-1}\colon H_\scop^{\sfs-1,\sfs+\alpha_\cD+1}\to H_\scop^{\sfs,\sfs+\alpha_\cD}
  \]
  down to $\R\setminus\{0\}$ and a weak compactness argument \cite[\S{2.7}]{VasyMicroKerrdS}. Therefore, we can further shift the contour in~\eqref{EqKEFwdh} to the union
  \begin{equation}
  \label{EqKEFwdContour}
    \gamma_- \cup \gamma_0 \cup \gamma_+,\quad
    \gamma_-:=(-\infty,-\sigma_0],\ 
    \gamma_0:=\Bigl[-\sigma_0,-\frac{\sigma_0}{2}\Bigr] \cup \frac{\sigma_0}{2}e^{i[\pi,0]} \cup \Bigl[\frac{\sigma_0}{2},\sigma_0\Bigr],\ 
    \gamma_+:=[\sigma_0,\infty).
  \end{equation}
  See Figure~\ref{FigKEFwdContour}.

  \begin{figure}[!ht]
  \centering
  \includegraphics{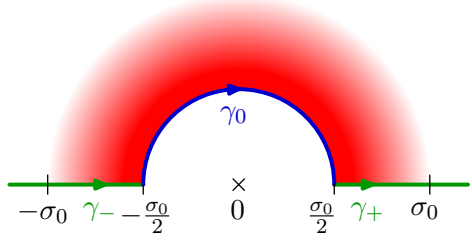}
  \caption{The contours $\gamma_-,\gamma_0,\gamma_+$ defined in~\eqref{EqKEFwdContour}. We also indicate the function $\psi(|\cdot|)$ which is equal to $1$ in the red region, and equal to $0$ in the white region.}
  \label{FigKEFwdContour}
  \end{figure}

  Define
  \begin{equation}
  \label{EqKEFwdb1scal1}
    \dot b_1(\hat t_1) := \frac{1}{2\pi} \int_{\gamma_0} e^{-i\sigma\hat t_1}\psi(|\sigma|)\wt{\dot b}_1(\sigma)\,\dd\sigma,\qquad
    \scal_1(\hat t_1) := \frac{1}{2\pi} \int_{\gamma_0} e^{-i\sigma\hat t_1}\psi(|\sigma|)\wt\scal_1(\sigma)\,\dd\sigma.
  \end{equation}
  In view of~\eqref{EqKEFwdDecomp} and using $\hat t_1=\hat t+\cT_1(\hat r)$, we can then further write
  \begin{alignat}{2}
  \label{EqKEFwdhSum}
    &h = h_{\rm hi} + h_{\rm lo} + h_{\rm Kerr} + h_{\rm COM} + h_{\rm COM,\tface,1} + h_{\rm COM,\tface,2}, \hspace{-22em}&& \\
    &\qquad& h_{\rm hi} &:= \frac{1}{2\pi}\int_\R e^{-i\sigma\hat t} \bigl(1-\psi(|\sigma|)\bigr)\hat h(\sigma)\,\dd\sigma, \nonumber\\
    &\qquad& h_{\rm lo} &:= \frac{1}{2\pi}\int_{\gamma_0} e^{-i\sigma\hat t} \psi(|\sigma|) \wh{h_0}(\sigma)\,\dd\sigma, \nonumber\\
    &\qquad& h_{\rm Kerr} &:= \hat g^{\prime\Ups}_b(\dot b_1(\hat t_1)), \nonumber\\
    &\qquad& h_{\rm COM} &:= h_{\rms 1}(\scal_1(\hat t_1)) + \breve h_{1,\rms 1}\bigl(\pa_{\hat t_1}\scal_1(\hat t_1)\bigr), \nonumber \\
  \label{EqKEFwdhCOMtf}
    &\qquad& h_{\rm COM,\tface} &:= \frac{1}{2\pi}\int_{\gamma_0} e^{-i\sigma\hat t}e^{-i\sigma\cT_1(\hat r)}\psi(|\sigma|)\Bigl[\log\Bigl(\frac{\sigma\hat r}{\sigma\hat r+i}\Bigr)\Ups\bigl(-\sigma^2\wt\scal_1(\sigma)\bigr)h^\Ups \\
    &\qquad&&\hspace{16em} + \tilde h_\tface\bigl(-\sigma^2\wt\scal_1(\sigma),\sigma\hat r,\omega)\Bigr]\,\dd\sigma; \nonumber
  \end{alignat}
  these are smooth tensors, resp.\ functions on $\hat M_b^\circ$.

  In view of the holomorphicity of $\phi\wh{h_0}(\sigma)$ (for any $\phi\in\CIc(\hat X_b^\circ)$) for $\sigma\neq 0$ in the upper half plane, and in view of its uniform boundedness near $\sigma=0$, we can shift the part $\frac{\sigma_0}{2}e^{i[\pi,0]}$ of the contour $\gamma_0$ in the integral $h_{\rm lo}$ to $[-\frac{\sigma_0}{2},\frac{\sigma_0}{2}]$; so $h_{\rm lo}(\hat t)=(2\pi)^{-1}\int_\R e^{-i\sigma\hat t}\psi(|\hat\sigma|)\wh{h_0}(\sigma)\,\dd\sigma$. Using the uniform bounds on $\hat L(\sigma)^{-1}$ and $\wt L(\sigma)^{-1}$ for real $\sigma$ given in \citeII{Propositions~\ref*{PropEstFTbdd} and \ref*{PropEstFThi}} and Proposition~\ref{PropKELo}, and recalling~\eqref{EqNFT3b}, we obtain the quantitative bound
  \begin{equation}
  \label{EqKEFwdBd1}
    \|h_{\rm hi}\|_{\Htb^{\sfs,\alpha_\cD,0}} + \|h_{\rm lo}\|_{\Htb^{\sfs,\alpha_\cD,0}} \leq C\|f\|_{\Htb^{\sfs,\alpha_\cD+2,0}}.
  \end{equation}

  Similarly, we can shift the contour in the integral of $h_{\rm COM,\tface}$ to the real axis (since the logarithmic singularity at $\sigma=0$ is suppressed by the length $\sim|\sigma|$ of the semicircle contour from $-|\sigma|$ to $|\sigma|$). Let $\chi\in\CI((\hat X_b)_\scbtop)$ be such that $\chi|_\zface=1$, and $\chi$ vanishes to infinite order at $\scface$. We then rewrite the resulting integral as a sum $h_{\rm COM,\tface}=h_{\rm COM,\tface,1}+h_{\rm COM,\tface,2}$ where
  \begin{align}
  \label{EqKEFwdtf1}
    h_{\rm COM,\tface,1} &:= \frac{1}{2\pi}\int_{-\sigma_0}^{\sigma_0} e^{-i\sigma\hat t_1}\psi(|\sigma|)\chi\log(\sigma+i 0)\Ups\bigl(-\sigma^2\wt\scal_1(\sigma)\bigr)h^\Ups\,\dd\sigma, \\
  \label{EqKEFwdtf2}
  \begin{split}
    h_{\rm COM,\tface,2} &:= \frac{1}{2\pi}\int_{-\sigma_0}^{\sigma_0} e^{-i\sigma\hat t_1}\psi(|\sigma|)\Bigl[ \chi\log\Bigl(\frac{\hat r}{\sigma\hat r+i}\Bigr)\Ups\bigl(-\sigma^2\wt\scal_1(\sigma)\bigr)h^\Ups \\
      &\quad + (1-\chi)\log\Bigl(\frac{\sigma\hat r}{\sigma\hat r+i}\Bigr)\Ups\bigl(-\sigma^2\wt\scal_1(\sigma)\bigr)h^\Ups + \tilde h_\tface\bigl(\sigma^2\wt\scal_1(\sigma),\sigma\hat r,\omega\bigr)\Bigr]\,\dd\sigma.
  \end{split}
  \end{align}
  Now, $\psi(|\sigma|)\log(\sigma+i 0)\sigma^2\wt\scal_1(\scal)\in(\log|\sigma|)L^2_\cp((-\sigma_0,\sigma_0))\subset|\sigma|^{-\eta}L^2_\cp$ for all $\eta>0$; therefore, Lemma~\ref{LemmaNL2FT} gives
  \begin{equation}
  \label{EqKEFwdBdtf1}
    h_{\rm COM,\tface,1} \in \bigcap_{\eta>0}\Htb^{\infty,-\frac32-\eta,-\eta}(\hat M_b).
  \end{equation}
  The terms in the integrand for $h_{\rm COM,\tface,2}$ in square brackets involving $h^\Ups$ lie in
  \[
    L^2_\cp((-\sigma_0,\sigma_0))\cdot\cA^{1-,0-,0}((\hat X_b)_\scbtop),
  \]
  and therefore Lemma~\ref{LemmaNMult}\eqref{ItNMultVar} and \eqref{EqNFT3b} imply that their contribution to $h_{\rm COM,\tface,2}$ lies in the space $\Htb^{\sfs,-\frac32-,0}(\hat M_b)$. The term $\tilde h_\tface$, for $\pm\sigma\in(0,\sigma_0]$, is, in view of~\eqref{EqKEtftildehtf} and the fact that $|\sigma|\hat r$ is an affine coordinate on $\tface$, a finite sum (over a basis of $\scalspace_1$) of products of elements of $L^2_\cp$ (the components of $\sigma^2\wt\scal_1(\sigma)$) and elements of $\cA^{1-,0,0}((\hat X_b)_\scbtop)$. Using Lemma~\ref{LemmaNMult}\eqref{ItNMultVar} and \eqref{EqNFT3b} again, we thus conclude that
  \begin{equation}
  \label{EqKEFwdBdtf2}
    h_{\rm COM,\tface,2} \in \Htb^{\sfs,-\frac32-,0}(\hat M_b).
  \end{equation}
  We record (for $h_{\rm COM,\tface,2}$ using Plancherel) the bounds
  \begin{equation}
  \label{EqKEFwdBdtf}
    \|h_{\rm COM,\tface,1}\|_{\Htb^{\sfs,-\frac32-\eta,-\eta}} + \|h_{\rm COM,\tface,2}\|_{\Htb^{\sfs,-\frac32-\eta,0}} \leq C_\eta\|f\|_{\Htb^{\sfs,\alpha_\cD+2,0}}.
  \end{equation}

  Carefully note that neither $h_{\rm hi}$ nor $h_{\rm lo}$ (nor their sum) are supported in $\hat t\geq 0$; likewise for $h_{\rm Kerr}$, $h_{\rm COM}$, and $h_{\rm COM,\tface,j}$ ($j=1,2$).

  %%%%%%%%%%
  \pfstep{Control of the singular terms.} Note that
  \[
    \pa_{\hat t_1}\dot b_1=\frac{1}{2\pi}\int_{\gamma_0} e^{-i\sigma\hat t_1}\psi(|\hat\sigma|)\bigl(-i\sigma\wt{\dot b}_1(\sigma)\bigr)\,\dd\sigma.
  \]
  Since $-i\sigma\wt{\dot b}_1(\sigma)=\dot b_0(\sigma)$ is holomorphic in $\Im\sigma\geq 0$, $0<|\sigma|\leq\sigma_0$, and uniformly bounded near $\sigma=0$, we can shift the contour and obtain
  \begin{subequations}
  \begin{equation}
  \label{EqKEFwdb1}
    \pa_{\hat t_1}\dot b_1 = \frac{1}{2\pi}\int_\R e^{-i\sigma\hat t_1}\psi(|\sigma|)\bigl(-i\sigma\wt{\dot b}_1(\sigma)\bigr)\,\dd\sigma,\qquad
    \|\pa_{\hat t_1}\dot b_1\|_{H^m(\R_{\hat t_1})} \leq C_m\|f\|_{\Htb^{\sfs,\alpha_\cD+2,0}}
  \end{equation}
  for all $m\in\N_0$ in view of~\eqref{EqKEFwdCorrBd}. In the same fashion, one proves\footnote{If $\dot b_1$ and $\scal_1$ vanished for $\hat t_1\leq 1$, this would give $\dot b_1\in\hat t_1\dot H_{-;\bop}^{(m;1)}$ and $\scal_1\in\hat t_1^2\dot H_{-;\bop}^{(m;2)}$. This vanishing property is however not true, and we deal with this issue below.}
  \begin{equation}
  \label{EqKEFwdScal1}
    \|\pa_{\hat t_1}^2\scal_1\|_{H^m(\R_{\hat t_1})} \leq C_m\|f\|_{\Htb^{\sfs,\alpha_\cD+2,0}}.
  \end{equation}
  \end{subequations}

  Recalling~\eqref{EqKEFwdtf1}, define
  \begin{equation}
  \label{EqKEFwdUps1}
    \Ups_1(\hat t_1) = \frac{1}{2\pi}\int_{-\sigma_0}^{\sigma_0} e^{-i\sigma\hat t_1}\psi(|\sigma|)\Ups\bigl(-\sigma^2\log(\sigma+i 0)\wt\scal_1(\sigma)\bigr)\,\dd\sigma.
  \end{equation}
  We shall prove that $\dot b_1$, $\scal_1$, and $\Ups_1$ are approximately supported in $\hat t_1\geq 0$, in the sense that for $\chi_-\in\CI(\R)$ with $\chi_-|_{(-\infty,-2]}=1$, $\chi_-|_{[-1,\infty)}=0$,
  \begin{equation}
  \label{EqKEFwdBd3}
    \|\chi_-\dot b_1\|_{H^m} + \|\chi_-\scal_1\|_{H^m} + \|\chi_-\Ups_1\|_{H^m} \leq C_m\|f\|_{\Htb^{\sfs,\alpha_\cD+2,0}}.
  \end{equation}
  (By contrast, they typically grow relative to $L^2$ as $\hat t_1\to\infty$.) Since $\psi\in\CIc(\R)$, it suffices to prove~\eqref{EqKEFwdBd3} for a single (arbitrary) value of $m\in\R$. Let now $\phi\in\CIc(\hat X_b^\circ)$ be such that $\phi h_{\rms 1}(\scal)\neq 0$ unless $\scal=0$. Since $\hat t\geq 0$ on $\supp h$, we then have
  \begin{equation}
  \label{EqKEFwdBd3Pf}
  \begin{split}
    0 &= \chi_-(\hat t)\phi h \\
      &= \chi_-(\hat t)\phi(h_{\rm hi}+h_{\rm lo}+h_{\rm COM,\tface,2}) + \chi_-(\hat t)\phi\hat g^{\prime\Ups}_b\bigl(\dot b_1(\hat t_1)\bigr) + \chi_-(\hat t)\phi h_{\rms 1}\bigl(\scal_1(\hat t_1)\bigr) \\
      &\qquad + \chi_-(\hat t)\phi\breve h_{1,\rms 1}\bigl(\pa_{\hat t_1}\scal_1(\hat t_1)\bigr) + \chi_-(\hat t)\phi h_{\rm COM,\tface,1}.
  \end{split}
  \end{equation}
  Since $\chi_-(\hat t)\phi$ is a uniformly bounded multiplication operator on all weighted 3b-Sobolev spaces, we conclude from~\eqref{EqKEFwdBd1} and \eqref{EqKEFwdBdtf2} that $\|\chi_-\phi(h_{\rm hi}+h_{\rm lo}+h_{\rm COM,\tface,2})\|_{\Htb^{\sfs,-\frac32-\eta,0}}\leq C_\eta\|f\|_{\Htb^{\sfs,\alpha_\cD+2,0}}$, so also the sum of the remaining terms on the right hand side is bounded in $\Htb^{\sfs,-\frac32-\eta,0}$ by $\|f\|_{\Htb^{\sfs,\alpha_\cD+2,0}}$. Let $\tilde\chi_-\in\CI(\R)$ be such that $\tilde\chi_-|_{(-\infty,-3]}=1$ and $\tilde\chi_-|_{[-2,\infty)}=0$; then differentiating~\eqref{EqKEFwdBd3Pf} along $\pa_{\hat t}=\pa_{\hat t_1}$ and multiplying by $\tilde\chi_-(\hat t)$ shows, using~\eqref{EqKEFwdb1}--\eqref{EqKEFwdScal1}, that (omitting the weight at $\cD\subset\pa\hat M_b$ for distributions which are supported in a region of bounded $\hat r$)
  \[
    \| \tilde\chi_-(\hat t)\phi h_{\rms 1}\bigl(\pa_{\hat t_1}\scal_1(\hat t_1)\bigr) \|_{\Htb^{\sfs-1}} \leq C\|f\|_{\Htb^{\sfs,\alpha_\cD+2,0}}.
  \]
  (This uses $\pa_{\hat t}h_{\rm COM,\tface,1}=\frac{1}{2\pi}\int e^{-i\sigma\hat t_1}\psi(|\sigma|)\chi\Ups(i\sigma^3(\log\sigma)\wt\scal_1(\sigma))h^\Ups\,\dd\sigma$, with $\sigma^3(\log\sigma)\wh\scal_1(\sigma)\in L^2$.) Therefore, for $\chi_{1,-}\in\CI(\R)$ which equals $1$ near $-\infty$ and satisfies $\supp(\chi_{1,-}(\hat t_1)\phi)\subset\{\tilde\chi_-(\hat t)=1\}$, we conclude that
  \begin{equation}
  \label{EqKEFwdBd3Pf2}
    \| \chi_{1,-}\pa_{\hat t_1}\scal_1 \|_{H^m} \leq C\|f\|_{\Htb^{\sfs,\alpha_\cD+2,0}}
  \end{equation}
  for $m\leq\inf\sfs-1$.

  We localize~\eqref{EqKEFwdBd3Pf} to more negative times, but use the same notation $\chi_-$ for the new cutoff with smaller support, and plug in the information~\eqref{EqKEFwdBd3Pf2} to deduce
  \[
    \chi_-\phi\hat g^{\prime\Ups}_b(\dot b_1(\hat t_1))+\chi_-\phi h_{\rms 1}(\scal_1(\hat t_1))+\chi_-\phi h_{\rm COM,\tface,1} \in \Htb^m.
  \]
  Now, the tensors $\hat g^{\prime\Ups}_b(\dot b_i)$, $h_{\rms 1}(\scal_j)$, $h^\Ups_{\mu\nu}$, where $\dot b_0=(1,0,0,0)$, $\dot b_1=(0,1,0,0)$, $\dot b_2=(0,0,1,0)$, $\dot b_3=(0,0,0,1)$, and $\scal_j=\frac{\hat x^j}{|\hat x|}$, are linearly independent; and this remains true for $\phi\hat g^{\prime\Ups}_b(\dot b_i)$, $\phi h_{\rms 1}(\scal_j)$, $\phi h^\Ups_{\mu\nu}$ when $\phi$ is equal to $1$ on a sufficiently large compact subset of $\hat X_b^\circ$. Taking $\chi$ in~\eqref{EqKEFwdtf1}--\eqref{EqKEFwdtf2} to be a function of $\sigma$ on $\supp\phi$ which equals $1$ for $|\sigma|\leq\sigma_0$, we obtain (with a cutoff $\chi_{1,-}$ satisfying the same conditions as before, now relative to a new cutoff $\tilde\chi_-$) the bound
  \begin{equation}
  \label{EqKEFwdBdNeg}
    \| \chi_{1,-}\dot b_1 \|_{H^m} + \| \chi_{1,-}\scal_1 \|_{H^m} + \| \chi_{1,-}\Ups_1 \|_{H^m} \leq C\|f\|_{\Htb^{\sfs,\alpha_\cD+2,0}}.
  \end{equation}
  Using the fundamental theorem of calculus and the derivative bounds~\eqref{EqKEFwdb1}--\eqref{EqKEFwdScal1}, we now obtain the estimate~\eqref{EqKEFwdBd3}. This argument in fact shows that this estimate is valid for \emph{any} fixed cutoff $\chi_-$ which equals $1$ near $-\infty$ and $0$ near $+\infty$.

  %%%%%%%%%%
  \pfstep{Further analysis of $h_{\rm COM,\tface,1}$.} We now take $\chi$ in~\eqref{EqKEFwdtf1}--\eqref{EqKEFwdtf2} to be $\chi=\chi_0(\sigma\hat r)$, with $\chi_0\in\sS(\R)$, $\chi_0(0)=1$. We also arrange $\cF^{-1}\chi_0\in\CIc((2,3))$ for now. Set $\phi_-=1-\phi_+$. We then decompose (using the notation~\eqref{EqKEFwdUps1})
  \[
    h_{\rm COM,\tface,1} = h_{\rm COM,\tface,+} + h_{\rm COM,\tface,-},\qquad
    h_{\rm COM,\tface,\pm} := \Bigl(\hat r^{-1}(\cF^{-1}\chi_0)\Bigl(\frac{\cdot}{\hat r}\Bigr) * \Ups_1\Bigr)\phi_\pm h^\Ups,
  \]
  with convolution in $\hat t_1$. On $\supp\phi_-$ and thus on $\supp h_{\rm COM,\tface,-}$, we have $\hat t\leq 2\hat r$, and therefore $h_{\rm COM,\tface,-}$ only depends on the restriction of $\Ups_1$ to $\hat t\leq(2-2)\hat r=0$; but for $\hat t\leq 0$, the function $\hat t_1\equiv\hat t-\hat r\bmod\cA^{0-}(\hat X_b)$ has a finite upper bound $C<\infty$. Using a partition of unity subordinate to $(-\infty,C+1)\cup(C,\infty)$, split then
  \[
    \Ups_1 = \Ups_1^- + \Ups_1^+,\qquad \Ups_1^-\in H^\infty(\R);
  \]
  for the membership of $\Ups_1^-$, we use~\eqref{EqKEFwdBd3}. Then, by Lemma~\ref{LemmaNL2FT} with $\alpha=0$,
  \begin{equation}
  \label{EqKEFwdhComtfm}
    h_{\rm COM,\tface,-} = \Bigl(\hat r^{-1}(\cF^{-1}\chi_0)\Bigl(\frac{\cdot}{\hat r}\Bigr) * \Ups_1^-\Bigr)\phi_- h^\Ups \in \bigcap_{\eta>0} \Htb^{\infty,-\frac32-\eta,0}.
  \end{equation}

  %%%%%%%%%%
  \pfstep{Reshuffling of the terms.} We now take $\chi_-\in\CI(\R)$ to be equal $1$ on $(-\infty,\hat{\ubar t}_1]$ and $0$ on $[\hat{\ubar t}_1+1,\infty)$. The functions
  \[
    \dot b := (1-\chi_-)\dot b_1,\qquad
    \scal := (1-\chi_-)\scal_1
  \]
  then vanish for $\hat t_1\leq\hat{\ubar t}_1$. In view of~\eqref{EqKEFwdb1}--\eqref{EqKEFwdScal1}, they satisfy the quantitative bound
  \[
    \|\dot b\|_{\hat t_1\dot H_{-;\bop}^{(m;1)}} + \|\scal\|_{\hat t_1^2\dot H_{-;\bop}^{(m;2)}} \leq C_m\|f\|_{\Htb^{\sfs,\alpha_\cD+2}}
  \]
  for all $m\in\R$. We correspondingly rewrite~\eqref{EqKEFwdhSum} in the form
  \begin{align}
  \label{EqKEFwduSumFin}
  \begin{split}
    &h = (h_{\rm hi}+h_{\rm lo} + h_{\rm COM,\tface,-} + h_{\rm COM,\tface,2} + h_{\rm mod,-}) \\
    &\qquad + \Bigl(\hat g^{\prime\Ups}_b(\dot b(\hat t_1)) + h_{\rms 1}(\scal(\hat t_1)) + \breve h_{1,\rms 1}\bigl(\pa_{\hat t_1}\scal(\hat t_1)\bigr) + h_{\rm COM,\tface,+}\Bigr),
  \end{split} \\
    &h_{\rm mod,-} := \hat g^{\prime\Ups}_b\bigl(\chi_-(\hat t_1)\dot b_1(\hat t_1)\bigr) + h_{\rms 1}\bigl(\chi_-(\hat t_1)\scal_1(\hat t_1)\bigr) + \breve h_{1,\rms 1}\bigl(\pa_{\hat t_1}(\chi_-(\hat t_1)\scal_1(\hat t_1))\bigr). \nonumber
  \end{align}
  Since the second sum in~\eqref{EqKEFwduSumFin} vanishes for $\hat t\leq 0$, so does the first sum.

  Now, for $m$ large enough ($m\geq\sup\sfs+1$ being sufficient), we have
  \begin{equation}
  \label{EqKEFwdhmodBound}
    \|h_{\rm mod,-}\|_{\Htb^{\sfs,\alpha_\cD,0}} \leq C\bigl(\|\chi_-\dot b_1\|_{H^m}+\|\chi_-\scal_1\|_{H^m}+\|\pa_{\hat t_1}(\chi_-\scal_1)\|_{H^m}\bigr) \leq C'\|f\|_{\Htb^{\sfs,\alpha_\cD+2,0}}.
  \end{equation}
  Indeed, the first bound follows from Lemma~\ref{LemmaNExpTerm} with $k=0$ and $\alpha=1$ since $\alpha_\cD<-\frac32+1=-\frac12$, and $\sfs+\alpha_\cD<-\frac12$ at the outgoing radial set; and the second bound follows from~\eqref{EqKEFwdBd3} and \eqref{EqKEFwdBd3Pf2}. We can finally define
  \[
    L_+^{-1}f =: \bigl(h_{\rm hi}+h_{\rm lo}+h_{\rm COM,\tface,-} + h_{\rm COM,\tface,2} + h_{\rm mod,-}, \dot b,\scal, \wh{\Ups_1^+}\bigr)
  \]
  for $f\in\CIc((\hat M_b^+)^\circ;S^2 T^*\hat M_b^\circ)$; and we have shown that this is a linear and continuous map between the spaces~\eqref{EqKEFwdMap}.

  While above we have required the cutoff function $\chi_0\in\sS(\R)$, $\chi_0(0)=1$, featuring in~\eqref{EqKEFwdhUps} to satisfy $\cF^{-1}\chi_0\in\CIc((2,3))$, this requirement can be dropped since passing from such a cutoff to a general one in $\sS(\R)$ which equals $1$ at $0$ gives a contribution to $h_+^\Ups(\scal_{(0)})$ which is equal to $\phi_+\frac{1}{2\pi}\int_{-1}^1 e^{-i\sigma\hat t_1}|\sigma\hat r|^{1-\eta}\chi_1(\sigma\hat r) |\sigma|^\eta\wh\Ups_1(\sigma) \hat r^\eta h^\Ups\,\dd\sigma$ for some $\chi_1=\chi_1(r')$ which is Schwartz as $|r'|\to\infty$ and bounded conormal at $r'=0$; here $\eta$ is arbitrary. This is thus the inverse Fourier transform of an element of $\bigcap_{\eta>0}\cA^{\infty,-\eta,0}((\hat X_b)_\scbtop)\cdot L^2_\cp(\R)$ and thus of class $\Htb^{\infty,-\frac32-,0}$. Multiplied with $\phi_+$, this can therefore be absorbed into the term $h_0$ in~\eqref{EqKEFwdSol}. The proof is complete.
\end{proof}

%%%%%%%%%%%%%%%%%%%%%%%%%%%%%%%%%%%%%%%%%%%%%%%%%%
\subsection{Forcing terms with weak spatial decay}
\label{SsKEL}

As already explained in~\S\ref{SssIKInter}, we shall need to study forward problems for $L$ with forcing terms whose decay rate $\alpha_\cD+2$ at $\cD$ lies (slightly) below the threshold $-\frac32+2$ of Theorem~\ref{ThmKEFwd}. Now, the low energy resolvent, even in the absence of a kernel of the zero energy operator, is no longer uniformly bounded when acting on spaces of functions whose decay order (with respect to $L^2$) at $\tface\subset(\hat X_b)_\scbtop$ is below the threshold value $-\frac32+2$. Rather, for forcing with decay rate $\alpha_\cD+2$ where $\alpha_\cD<-\frac32$, the resolvent has a $|\sigma|^{\alpha_\cD+\frac32}$ singularity (localized to a neighborhood of $\zface$), with spatial dependence given by \emph{large zero energy states}, i.e.\ in the present context elements of~\eqref{EqKE0KerL}. The main step towards making this precise is:

\begin{prop}[Approximating the resolvent near tf]
\label{PropKELtf}
  Let $\sfs\in\CI(\Stb^*\hat M_b)$, $\alpha_\cD$ be strongly Kerr-admissible orders, with $\alpha_\cD\in(-\frac52,-\frac32)$. Let $\eta>0$. Then there exists a constant $C$ so that the following holds. Let $|\sigma|\leq 1$ and $\hat\sigma=\frac{\sigma}{|\sigma|}\in\Sph^1_+$ (see~\eqref{EqKESph1plus}); we work with the sc-b-transition orders induced by $\sfs$ as in Remark~\usref{RmkKCDOrders}. Then for $\hat f\in H_{\scbtop,|\sigma|}^{\sfs-1,\sfs+\alpha_\cD+1,\alpha_\cD+2,0}(\hat X_b;S^2\,\Ttsc^*_{\hat X_b}\hat M_b)$, there exist
  \[
    (\hat a_{\mu\nu}) \in \C^{4\times 4}_{\rm sym},\qquad
    \hat h_\tface \in H_{\scbtop,|\sigma|}^{\sfs,\sfs+\alpha_\cD,\alpha_\cD,0}(\hat X_b;S^2\,\Ttsc^*_{\hat X_b}\hat M_b),
  \]
  so that, writing $\hat a h^\Ups:=\sum_{0\leq\mu\leq\nu\leq 3}\hat a_{\mu\nu}h_{\mu\nu}^\Ups$,
  \begin{equation}
  \label{EqKELtf}
  \begin{split}
    &|\sigma|^{-(\alpha_\cD+\frac32)}|\hat a| + \|\hat h_\tface\|_{H_{\scbtop,|\sigma|}^{\sfs,\sfs+\alpha_\cD,\alpha_\cD,0}} \\
    &\quad + \Bigl\|\hat f - \hat L(\sigma)\bigl( e^{-i\sigma\cT_1(\hat r)}(1-i\sigma\hat r)^{-2}\hat a h^\Ups + \hat h_\tface \bigr)\Bigr\|_{H_{\scbtop,|\sigma|}^{\sfs-2,\sfs+\alpha_\cD+1,\alpha_\cD+3-\eta,0}} \\
    &\hspace{23em} \leq C\|\hat f\|_{H_{\scbtop,|\sigma|}^{\sfs-1,\sfs+\alpha_\cD+1,\alpha_\cD+2,0}}.
  \end{split}
  \end{equation}
  Moreover, we can define $\hat a,\hat h_\tface$ so that the map $\hat f\mapsto(\hat a,\hat h_\tface)$ is linear, and it depends holomorphically on $\sigma$ in $\Im\sigma>0$ as a map $\Hsc^{\sfs-1,\sfs+\alpha_\cD+1}\to\C^{4\times 4}_{\rm sym}\oplus\Hsc^{\sfs,\sfs+\alpha_\cD}$.
\end{prop}

The point is that the $\tface$-decay order of the remainder term in~\eqref{EqKELtf} (i.e.\ the third term on the left) is (almost) one order better than that of $f$ itself; the loss of regularity will be inconsequential for us. The proof, much like that of Lemma~\ref{LemmaKEbreve} above, will exploit the fact that all indicial solutions of the zero energy operator $\hat L(0)$ which are (quasi-)homogeneous of degree $0$ can be extended to elements of the kernel of $\hat L(0)$, cf.\ (the proof of) \eqref{EqKE0KerL}; the relevance of degree $0$ homogeneity is that $0$ is the lower endpoint of the indicial gap for $\hat L(0)$ (on the $L^2$-level, where weights are shifted by $-\frac32$, corresponding to the lower endpoint $-\frac32$ of the interval of weights for which the zero energy operator has Fredholm index $0$).

\begin{proof}[Proof of Proposition~\usref{PropKELtf}]
  We work with the trivialization of $S^2\,\Ttsc^*_{\hat X_b}\hat M_b$ given by coordinate differentials, and drop the bundle from the notation. We pass to the unweighted b-density $|\frac{\dd\hat r}{\hat r}\,\dd\slg|$ on $\hat X_b$; thus
  \[
    \hat f\in H_{\scbtop,|\sigma|}^{\sfs-1,\sfr+1,-\alpha+2,0}(\hat X_b),\qquad \sfr=(\sfs+\alpha_\cD)+\frac32,\quad \alpha:=-\Bigl(\alpha_\cD+\frac32\Bigr)\in(0,1).
  \]
  Using a partition of unity, we may assume that $\hat r\geq 4\bhm$ on $\supp\hat f$, since for $\hat f$ supported away from $\hat r=\infty$ the conclusion~\eqref{EqKELtf} holds for $\hat a_{\mu\nu}=0$, $\hat h_\tface=0$.

  %%%%%%%%%%
  \pfstep{Inversion of the tf-normal operator.} Passing from $\hat r$ to the affine coordinate $r':=\hat r|\sigma|$ on $\tface$, set
  \[
    f'(r',\omega) := |\sigma|^{-2}(\sigma+i 0)^\alpha\hat f(|\sigma|^{-1}r',\omega) = |\sigma|^{-2}(\sigma+i 0)^\alpha\hat f(\hat r,\omega);
  \]
  then (cf.\ \citeII{(\ref*{EqF3scbtNormzf})})
  \[
    f' \in H_{\scop,\bop}^{\sfs-1,\sfr+1,-2+\alpha}(\tface),\qquad \tface=[0,\infty]_{r'}\times\Sph^2,
  \]
  with the orders referring to the regularity, decay at $r'=\infty$, and decay at $r'=0$, relative to $L^2(\tface;|\frac{\dd r'}{r'}\,\dd\slg|)$; and the norm of $f'$ is bounded by that of $\hat f$ in $H_{\scbtop,|\sigma|}^{\sfs-1,\sfr+1,-\alpha+2,0}$; here and for the remainder of the proof, we mean by this that the implicit constant in this bound can be taken to be independent of $\sigma$.

  In view of Lemma~\ref{LemmaKtfInv} and \citeII{Lemma~\ref*{LemmaEstMcInvft}}, the function $h'(\hat\sigma)=h'(\hat\sigma,r',\omega)$ defined by $h'(\hat\sigma):=L_\tface(\hat\sigma)^{-1}f'$ satisfies
  \begin{equation}
  \label{EqKELtfup}
    h'(\hat\sigma)\in H_{\scop,\bop}^{\sfs,\sfr,-\eta}(\tface)
  \end{equation}
  for all $\eta>0$, with norm bounded by that of $\hat f$. Regarded as a function of $\hat r=|\sigma|^{-1}r'$ and $\omega\in\Sph^2$, we have $h'=(\sigma+i 0)^\alpha\wh{\Box_{\hat{\ubar g}}}(\sigma)^{-1}\hat f$ (with the Minkowskian resolvent $\wh{\Box_{\hat{\ubar g}}}(\sigma)^{-1}$ acting component-wise in the $\dd\hat t,\dd\hat x$ splitting), which gives the holomorphicity in $\Im\sigma>0$ for fixed $\hat f$. We now use a normal operator argument (in the coordinates $\sigma,\hat r,\omega$) at $\hat r=0$ applied to $\wh{\Box_{\hat{\ubar g}}}(\sigma)h'=(\sigma+i 0)^\alpha\hat f$, with $(\sigma,\hat r,\omega)\mapsto h'(\frac{\sigma}{|\sigma|},|\sigma|\hat r,\omega)$ obeying uniform (for $|\sigma|\leq 1$) bounds in $\Hb^{\sfs,-\eta}([0,\bhm)_{\hat r}\times\Sph^2_\omega)$ while $(\sigma+i 0)^\alpha\hat f=0$ for $\hat r<\bhm$. Using that the indicial solutions at weight $0$ are $\ubar h_{\mu\nu}^\Ups$, this allows us to improve~\eqref{EqKELtfup} to
  \[
    h'(\sigma,\hat r,\omega) = e^{-i\sigma\cT_1(\hat r)}(1-i\sigma\hat r)^{-2}\breve a(\sigma)\ubar h^\Ups + \tilde h'(\sigma,\hat r,\omega),
  \]
  where
  \begin{equation}
  \label{EqKELtfSpaces}
    \breve a(\sigma)\in\C^{4\times 4}_{\rm sym},\qquad
    \tilde h'(|\sigma|\hat\sigma,|\sigma|^{-1}r',\omega)\in H_{\scop,\bop}^{\sfs,\sfr,\alpha}(\tface).
  \end{equation}
  We use the factor $(1-i\sigma\hat r)^{-2}$ as a weak localization to a neighborhood of $r'=0$ with holomorphic dependence on $\sigma$; and we insert the prefactor $e^{-i\sigma\cT_1(\hat r)}$ to enforce outgoing behavior. Thus, for fixed $\hat f$, the matrix $\breve a$ and the function $\tilde h'$ (regarded as a function of $(\hat r,\omega)$) depend holomorphically on $\sigma$; and $\breve a$, $\tilde h'$ are bounded in norm by $\hat f$ in the spaces in~\eqref{EqKELtfSpaces}.

  %%%%%%%%%%
  \pfstep{Grafting the tf-solution into the sc-b-transition single space.} Let $\psi\in\CI([0,\infty))$ be equal to $1$ on $[4\bhm,\infty)$ and equal to $0$ on $[0,3\bhm]$. Set then
  \[
    \hat a_{\mu\nu}(\sigma) := (\sigma+i 0)^{-\alpha}\breve a_{\mu\nu}(\sigma),\qquad
    \hat h_\tface(\sigma,\hat r,\omega) := (\sigma+i 0)^{-\alpha}\psi(\hat r)\tilde h'(\sigma,\hat r,\omega),
  \]
  Note that $\psi(\hat r)\tilde h'(\sigma,\hat r,\omega)$ lies in $H_{\scbtop,|\sigma|}^{\sfs,\sfr,0,\alpha}(\hat X_b)$ (cf.\ \eqref{EqKELtfSpaces} and \citeII{(\ref*{EqF3scbtNormzf})}), and therefore we have
  \begin{equation}
  \label{EqKELtfutf}
    \hat h_\tface\in H_{\scbtop,|\sigma|}^{\sfs,\sfr,-\alpha,0}(\hat X_b),
  \end{equation}
  with norm bounded by that of $\hat f$. We claim that~\eqref{EqKELtf} holds. Note that the norm of the third term on the left, now relative to unweighted b-densities, is the $H_{\scbtop,|\sigma|}^{\sfs-2,\sfr+1,-\alpha+3-\eta,0}(\hat X_b)$-norm. To verify~\eqref{EqKELtf}, we first note that
  \[
    |\sigma|^2 L_\tface(\hat\sigma)=\hat r^{-2}(-(\hat r\pa_{\hat r})^2+\hat r\pa_{\hat r}+\slDelta)-\sigma^2:=L_{\tface,0}+\sigma^2 L_{\tface,2},
  \]
  where
  \begin{equation}
  \label{EqKELtfDiff}
  \begin{split}
    L_{\tface,0}-\hat L(0)&\in\rho_\cD^3\Diffb^2(\hat X_b), \\
    \pa_\sigma\hat L(0)&\in\rho_\cD^2\Diffb^1(\hat X_b), \\
    L_{\tface,2}-\frac12\pa_\sigma^2\hat L(0)&\in\rho_\cD\Diffb^0(\hat X_b).
  \end{split}
  \end{equation}
  Using the notation $\chi:=(1-\hat\sigma r')^{-2}=(1-\sigma\hat r)^{-2}\in\cA^{1,0,0}((\hat X_b)_\scbtop)$, we then write
  \begin{equation}
  \label{EqKELtfRewrite}
  \begin{split}
    &\hat L(\sigma)\bigl( e^{-i\sigma\cT_1(\hat r)}\chi\hat a(\sigma)h^\Ups + \hat h_\tface(\sigma)\bigr) - \hat f(\sigma) \\
    &\quad = \hat a(\sigma)[\hat L(0),e^{-i\sigma\cT_1(\hat r)}\chi]h^\Ups + \hat a(\sigma)\sigma\pa_\sigma\hat L(0)(e^{-i\sigma\cT_1(\hat r)}\chi h^\Ups) + \hat a(\sigma)\frac{\sigma^2}{2}\pa_\sigma^2\hat L(0)(e^{-i\sigma\cT_1(\hat r)}\chi h^\Ups) \\
    &\quad \qquad + (\sigma+i 0)^{-\alpha}\Bigl(\hat L(0)(\psi(\hat r)\tilde h') + \sigma\pa_\sigma\hat L(0)(\psi(\hat r)\tilde h') + \frac{\sigma^2}{2}\pa_\sigma^2\hat L(0)(\psi(\hat r)\tilde h') \Bigr) - \hat f(\sigma).
  \end{split}
  \end{equation}
  The second term lies in
  \[
    |\sigma|^{-\alpha} \sigma\rho_\cD^2\Diffb^1(\hat X_b)\bigl(e^{-i\sigma\cT_1(\hat r)}\cA^{2,0,0}((\hat X_b)_\scbtop)\bigr)\subset e^{-i\sigma\cT_1(\hat r)}\cA^{4,3-\alpha,1-\alpha}((\hat X_b)_\scbtop).
  \]
  Since $1-\alpha>0$ and since $\sfr<1$ at the outgoing radial set, this implies, by Lemma~\ref{LemmaNMult} (with shifts of spatial weight by $\frac32$ due to our present usage of an unweighted b-density), a bound on the $H_{\scbtop,|\sigma|}^{\sfs,\sfr+3,3-\alpha,0}$-norm in terms of the norm of $\hat f(\sigma)$. The fifth term is of class $\sigma\rho_\cD^2\Diffb^1(\hat X_b) H_{\scbtop,|\sigma|}^{\sfs,\sfr,-\alpha,0}(\hat X_b)\subset H_{\scbtop,|\sigma|}^{\sfs-1,\sfr+1,3-\alpha,1}(\hat X_b)$ by~\eqref{EqKELtfutf}--\eqref{EqKELtfDiff} and $\Diffb^1(\hat X_b)\subset\Diffscbt^{1,1,0,0}(\hat X_b)$.

  In the remaining terms, we replace $\hat L(0)$ and $\frac12\pa_\sigma^2\hat L(0)$ by $L_{\tface,0}$ and $L_{\tface,2}$, respectively, and use~\eqref{EqKELtfDiff} to estimate the resulting error terms. Using Lemma~\ref{LemmaNMult} for the first and third term, these error terms are of class
  \begin{alignat*}{2}
    |\sigma|^{-\alpha}[\rho_\cD^3\Diffb^2,e^{-i\sigma\cT_1(\hat r)}\chi]h^\Ups &\subset e^{-i\sigma\cT_1(\hat r)}|\sigma|^{-\alpha}\cA^{5,3,1}((\hat X_b)_\scbtop) &&\subset H_{\scbtop,|\sigma|}^{\infty,\sfr+4,3-\alpha,0}, \\
    e^{-i\sigma\cT_1(\hat r)}|\sigma|^{-\alpha}\sigma^2\rho_\cD\Diffb^0\,\chi h^\Ups &\subset e^{-i\sigma\cT_1(\hat r)}|\sigma|^{-\alpha}\cA^{3,3,2}((\hat X_b)_\scbtop) &&\subset H_{\scbtop,|\sigma|}^{\infty,\sfr+2,3-\alpha,0}.
  \end{alignat*}
  Using~\eqref{EqKELtfutf} for the fourth and sixth term together with $\Diffb^2(\hat X_b)\subset\Diffscbt^{2,2,0,0}(\hat X_b)$, the error terms are
  \begin{align*}
    \rho_\cD^3\Diffb^2(\hat X_b) H_{\scbtop,|\sigma|}^{\sfs,\sfr,-\alpha,0} &\subset H_{\scbtop,|\sigma|}^{\sfs-2,\sfr+1,3-\alpha,0}, \\
    \rho_\cD\Diffb^0(\hat X_b) \sigma^2 H_{\scbtop,|\sigma|}^{\sfs,\sfr,-\alpha,0} &\subset H_{\scbtop,|\sigma|}^{\sfs,\sfr+1,3-\alpha,2}.
  \end{align*}
  In particular, all these terms are bounded in $H_{\scbtop,|\sigma|}^{\sfs-2,\sfr+1,3-\alpha,0}(\hat X_b)$ by the norm of $\hat f$. We further note that, by~\eqref{EqKELtfutf},
  \[
    |\sigma|^{-\alpha}\hat L(0)((1-\psi(\hat r))\tilde h'),\ \ |\sigma|^{-\alpha}\frac{\sigma^2}{2}\pa_\sigma^2\hat L(0)((1-\psi(\hat r))\tilde h')\in H_{\scbtop,|\sigma|}^{\sfs-2,\infty,\infty,0}(\hat X_b)
  \]
  with all seminorms bounded by that of $\hat f$, since $\supp(1-\psi(\hat r))$ is disjoint from $\scface\cup\tface\subset(\hat X_b)_\scbtop$; modulo these errors, we can thus drop $\psi$ from~\eqref{EqKELtfRewrite}.

  Additionally, we may replace $h^\Ups_{\mu\nu}$ in~\eqref{EqKELtfRewrite} by $\ubar h^\Ups_{\mu\nu}=\dd\hat z^\mu\otimes_s\dd\hat z^\nu$, since by~\eqref{EqKPureGaugeij}--\eqref{EqKPureGauge00} the difference of the two tensors lies in $\cA^{1-}(\hat X_b)$, and therefore the resulting error terms in the first and third term are of class
  \begin{alignat*}{2}
    |\sigma|^{-\alpha}[\rho_\cD^2\Diffb^2,e^{-i\sigma\cT_1(\hat r)}\chi]\cA^{1-} &\subset e^{-i\sigma\cT_1(\hat r)}|\sigma|^{-\alpha}\cA^{5-,3-,1}((\hat X_b)_\scbtop) &&\subset H_{\scbtop,|\sigma|}^{\infty,5-\eta,3-\alpha-\eta,0}, \\
    e^{-i\sigma\cT_1(\hat r)}|\sigma|^{-\alpha}\sigma^2\Diffb^0(\hat X_b)(\chi\cA^{1-}) &\subset e^{-i\sigma\cT_1(\hat r)}|\sigma|^{-\alpha}\cA^{2,3-,2}((\hat X_b)_\scbtop) &&\subset H_{\scbtop,|\sigma|}^{\infty,2-\eta,3-\alpha-\eta,0}
  \end{alignat*}
  for all $\eta>0$. Altogether, modulo terms whose $H_{\scbtop,|\sigma|}^{\sfs-2,\sfr+1,-\alpha+3-\eta,0}(\hat X_b)$-norms are bounded by the norm of $\hat f$, we now deduce
  \begin{align*}
    &\hat L(\sigma)\bigl(e^{-i\sigma\cT_1(\hat r)}\chi\hat a(\sigma)h^\Ups+\hat h_\tface(\sigma)\bigr) - \hat f(\sigma) \\
    &\qquad \equiv (\sigma+i 0)^{-\alpha}\Bigl[ \breve a(\sigma)\Bigl( [L_{\tface,0},e^{-i\sigma\cT_1(\hat r)}\chi]\ubar h^\Ups_{\mu\nu} + e^{-i\sigma\cT_1(\hat r)}\sigma^2 L_{\tface,2}\chi\ubar h^\Ups_{\mu\nu}\Bigr) \\
    &\qquad \quad\hspace{15em} + L_{\tface,0}(\tilde h') + \sigma^2 L_{\tface,2}(\tilde h') - |\sigma|^2 f' \Bigr].
  \end{align*}
  But since $\ubar h_{\mu\nu}^\Ups$ is an indicial solution of $L_\tface(\hat\sigma)$ at $r'=0$ (i.e.\ $L_{\tface,0}\ubar h_{\mu\nu}^\Ups=0$), this is equal to $(\sigma+i 0)^{-\alpha}|\sigma|^2$ times
  \[
    L_\tface(\hat\sigma)\bigl(e^{-i\sigma\cT_1(\hat r)}\chi\breve a(\sigma)\ubar h_{\mu\nu}^\Ups+\tilde h'\bigr)-f'=L_\tface(\hat\sigma)h'-f'=0.
  \]
  The proof is complete.
\end{proof}

The following result provides the required control for applicability to uniform estimates for the linearized gauge-fixed Einstein equations on glued spacetimes.

\begin{thm}[3b-estimates for forward solutions with weakly decaying forcing]
\label{ThmKEFwdW}
  Let $\hat{\ubar t}_1$, $\hat M_b^+$, $\phi_+$, $\chi_0$ be as in Theorem~\usref{ThmKEFwd}. Let $\alpha_\cD\in(-2,-\frac32)$, and let $\sfs\in\CI(\Stb^*\hat M_b)$ be such that $\sfs,\alpha_\cD$ and $\sfs-2,\alpha_\cD+1$ are strongly Kerr-admissible orders. Then there exists a continuous linear map
  \begin{equation}
  \label{EqKEFwdWMap}
  \begin{split}
    L_+^{-1} &\colon \dot H_\tbop^{\sfs,\alpha_\cD+2,0}(\hat M_b^+;S^2\,\Ttsc^*\hat M_b) \\
    &\quad \to \dot H_\tbop^{\sfs-1,\alpha_\cD,0}(\hat M_b^+;S^2\,\Ttsc^*\hat M_b) \oplus \hat t_1\dot H_{-;\bop}^{(\infty;1)}([\hat{\ubar t}_1,\infty];\R^4) \oplus \hat t_1^2\dot H_{-;\bop}^{(\infty;2)}([\hat{\ubar t}_1,\infty];\scalspace_1) \\
    &\quad\qquad \oplus |\sigma|^{\alpha_\cD+\frac32}L^2([-1,1];\C^{4\times 4}_{\rm sym})
  \end{split}
  \end{equation}
  with the following properties.
  \begin{enumerate}
  \item\label{ItKEFwdWSol}{\rm (Solution.)} Given $f\in\dot H_\tbop^{\sfs,\alpha_\cD+2,0}$ and $L_+^{-1}f=:\bigl(h_0,\dot b,\scal,\wh{\Ups_{(0)}}\bigr)$, define\footnote{The membership of $h_+^\Ups(\scal_{(0)})$ follows from Lemma~\ref{LemmaNL2FT} with $\alpha=-(\alpha_\cD+\frac32)$.}
    \begin{align}
    \label{EqKEFwdWSolhUps}
      h_+^\Ups\bigl(\wh{\Ups_{(0)}}\bigr) &:= \phi_+\frac{1}{2\pi}\int_{-1}^1 e^{-i\sigma\hat t_1} \chi_0(\sigma\hat r) \wh{\Ups_{(0)}}(\sigma)h^\Ups\,\dd\sigma \in \bigcap_{\eta>0}\dot H_\tbop^{\infty,\alpha_\cD,\alpha_\cD+\frac32-\eta}(\hat M_b^+), \\
    \label{EqKEFwdWSol}
      h &:= h_0 + \hat g^{\prime\Ups}_b\bigl(\dot b(\hat t_1)\bigr) + h_{\rms 1}\bigl(\scal(\hat t_1)\bigr) + \breve h_{1,\rms 1}\bigl(\pa_{\hat t_1}\scal(\hat t_1)\bigr) + h_+^\Ups\bigl(\wh{\Ups_{(0)}}\bigr),
    \end{align}
    as in~\eqref{EqKEFwdhUps}--\eqref{EqKEFwdSol}. Then $L h=f$.
  \item\label{ItKEFwdWIndep}{\rm (Independence of orders.)} $L_+^{-1}$ is independent of the choice of $\sfs,\alpha_\cD$ (satisfying the stated assumptions).
  \item\label{ItKEFwdWFwd}{\rm (Forward mapping properties.)} Given $\bigl(h_0,\dot b,\scal,\wh{\Ups_{(0)}}\bigr)$ in the space~\eqref{EqKEFwdWMap}, and defining $h$ by~\eqref{EqKEFwdWSol}, we have $L h\in\dot H_\tbop^{\sfs-3,\alpha_\cD+2,0}$.
  \end{enumerate}
\end{thm}

The fact that in part~\eqref{ItKEFwdWFwd} the weights of $L h$ match those of $f$ shows that the description of the solution via the function spaces in~\eqref{EqKEFwdWMap} is sharp as far as weights are concerned.

The description~\eqref{EqKEFwdWSolhUps} (which already appeared in~\eqref{EqKEFwdhUps}) directly encodes the fact that $h_+^\Ups\bigl(\wh{\Ups_{(0)}}\bigr)$ is supported in $\hat t\geq 2\hat r>0$ in light of the cutoff $\phi_+$. For a (physical space) variant of such a term, see Theorem~\ref{ThmKHi}\eqref{ItKHiSol}.

\begin{rmk}[Sharper range of weights]
\label{RmkKEFwdWWeights}
  An approach to the proof of Theorem~\ref{ThmKEFwdW} which circumvents distributional/holomorphicity considerations near $\sigma=0$ is to solve $L h=f$ near $\cD\subset\hat M_b$ using the Mellin transform in $\hat r^{-1}$ and the inversion of the Mellin-transformed normal operator family via a combination of \cite{BaskinVasyWunschRadMink} (for the microlocal analysis in $\cD^\circ$, which is $\pa\ol{\R^4}$ minus the `north' and `south' poles) and \cite[\S{5.3}]{HintzNonstat} (for the microlocal analysis near $\pa\cD$). A further advantage of this approach is that after solving away $f$ modulo $\dot H_\tbop^{\sfs-2,\alpha_\cD+3-,0}$, one can directly apply Theorem~\ref{ThmKEFwd}; and moreover, this allows one to cover $\alpha_\cD\in(-\frac52,-\frac32)$. (This range is optimal, since for $\alpha_\cD<-\frac52$, there are further contributions to the low energy resolvent expansion, and thus to the late time expansion of $h$, arising from the next indicial root---namely, $1$---of $L_\tface(\hat\sigma)$.) Since in the present work the choice $\alpha_\cD=-\frac32-\delta$ is sufficient for any fixed $\delta>0$, we do not pursue this further.
\end{rmk}

\begin{proof}[Proof of Theorem~\usref{ThmKEFwdW}]
  \pfstep{Parts~\eqref{ItKEFwdWSol}--\eqref{ItKEFwdWIndep}.} We consider $f\in\CIc((\hat M_b^+)^\circ;S^2\,\Ttsc^*\hat M_b)$. Fix $\eta>0$ so that $\alpha_\cD-\eta>-\frac52$, and so that $\sfs-2,\alpha_\cD+1-\eta$ are strongly Kerr-admissible. Proposition~\ref{PropKELtf}, applied to $\hat f(\sigma)$, produces $\hat a(\sigma)\in\C^{4\times 4}_{\rm sym}$ and $\hat h_\tface(\sigma)\in H_{\scbtop,|\sigma|}^{\sfs,\sfs+\alpha_\cD,\alpha_\cD,0}$, with holomorphic dependence on $\sigma$ in $\Im\sigma>0$ and a fortiori satisfying the bound
  \begin{equation}
  \label{EqKEFwdWBd}
    |\sigma|^{-(\alpha_\cD+\frac32)}|\hat a(\sigma)| + \|\hat h_\tface(\sigma)\|_{H_{\scbtop,|\sigma|}^{\sfs,\sfs+\alpha_\cD,\alpha_\cD,0}} \leq C\|\hat f(\sigma)\|_{H_{\scbtop,|\sigma|}^{\sfs,\sfs+\alpha_\cD+2,\alpha_\cD+2,0}} \leq C'\|f\|_{H_\tbop^{\sfs,\alpha_\cD+2,0}}.
  \end{equation}
  For $|\sigma|\leq 1$, set
  \[
    \hat f'(\sigma) := \hat f(\sigma) - \hat L(\sigma)\bigl(e^{-i\sigma\cT_1(\hat r)}(1-i\sigma\hat r)^{-2}\hat a(\sigma)h^\Ups + \hat h_\tface(\sigma)\bigr),
  \]
  which a fortiori lies in $H_{\scbtop,|\sigma|}^{\sfs-2,\sfs+\alpha_\cD+1-\eta,\alpha_\cD+3-\eta,0}$. Proposition~\ref{PropKELo}, applied with $\sfs-1,\alpha_\cD+1-\eta$ in place of $\sfs,\alpha_\cD$, produces
  \[
    \bigl(\wh{h_0}(\sigma),\dot b_0(\sigma),\scal_0(\sigma)\bigr) := \wt L(\sigma)^{-1}\bigl(\hat f'(\sigma),0\bigr),\qquad \Im\sigma\geq 0,\ |\sigma|\leq\sigma_0,
  \]
  obeying a uniform (in $\sigma$) $\tilde H_{\scbtop,|\sigma|}^{\sfs-1,\sfs+\alpha_\cD-\eta,\alpha_\cD+1-\eta,0}$-norm bound by $\|\hat f'\|_{H_{\scbtop,|\sigma|}^{\sfs-2,\sfs+\alpha_\cD+1-\eta,\alpha_\cD+3-\eta,0}}$ and thus by $\|\hat f\|_{\dot H_\tbop^{\sfs,\alpha_\cD+2,0}}$. Define $\wt{\dot b}_1$ and $\wt\scal_1$ by~\eqref{EqKEFwdb1S1}. In view of Proposition~\ref{PropKENon0}, we then have, analogously to~\eqref{EqKEFwdDecomp},
  \begin{equation}
  \label{EqKEFwdWh}
  \begin{split}
    \hat h(\sigma) &:= \hat L(\sigma)^{-1}\hat f(\sigma) \\
      &= \wh{h_0}(\sigma) + \hat h_\tface(\sigma) + e^{-i\sigma\cT_1(\hat r)}\Bigl[\hat g^{\prime\Ups}_b\bigl(\wt{\dot b}_1(\sigma)\bigr) + h_{\rms 1}\bigl(\wt\scal_1(\sigma)\bigr) + \breve h_{1,\rms 1}\bigl(-i\sigma\wt\scal_1(\sigma)\bigr) \\
      &\quad \hspace{12em} + \Bigl(\log\Bigl(\frac{\sigma\hat r}{\sigma\hat r+i}\Bigr)\Ups\bigl(-\sigma^2\wt\scal_1(\sigma)\bigr) + (1-i\sigma\hat r)^{-2}\hat a(\sigma)\Bigr)h^\Ups \\
      &\quad \hspace{12em} + \tilde h_\tface\bigl(-\sigma^2\wt\scal_1(\sigma),\sigma\hat r,\omega\bigr) \Bigr].
  \end{split}
  \end{equation}
  We then analyze the inverse Fourier transform~\eqref{EqKEFwdh} of $\hat h$ analogously to~\eqref{EqKEFwdhSum}. In the present setting, we encounter two additional terms, namely those involving $\hat h_\tface$ and $\hat a$: in the notation of~\eqref{EqKEFwdhSum} (in particular with $\psi\in\CIc([0,\sigma_0))$ equal to $1$ on $[0,\frac{\sigma_0}{2}]$), the first additional term is equal to
  \[
    \frac{1}{2\pi}\int_{\gamma_0} e^{-i\sigma\hat t}\psi(|\sigma|)\hat h_\tface(\sigma)\,\dd\sigma = \frac{1}{2\pi}\int_{-\sigma_0}^{\sigma_0} e^{-i\sigma\hat t}\psi(|\sigma|)\hat h_\tface(\sigma)\,\dd\sigma,
  \]
  which is thus bounded in $\Htb^{\sfs,\alpha_\cD,0}$ by~\eqref{EqNFT3b}. The second additional term can be absorbed into the definition of $h_{\rm COM,\tface}$ in~\eqref{EqKEFwdhCOMtf} upon adding $(1-i\sigma\hat r)^{-2}\hat a(\sigma)h^\Ups$ to the expression in square brackets there. One then further splits $h_{\rm COM,\tface}$ analogously to~\eqref{EqKEFwdtf1}--\eqref{EqKEFwdtf2}, where now Lemma~\ref{LemmaNL2FT} (with $\alpha=-(\alpha_\cD+\frac32)$) gives
  \begin{equation}
  \label{EqKEFwdWhtf1}
  \begin{split}
    h_{\rm COM,\tface,1} &:= \frac{1}{2\pi}\int_{-\sigma_0}^{\sigma_0} e^{-i\sigma\hat t_1}\psi(|\sigma|)\chi\,\Bigl(\log(\sigma+i 0)\Ups\bigl(-\sigma^2\wt\scal_1(\scal)\bigr) + \hat a(\sigma)\Bigr)h^\Ups \\
      &\in \bigcap_{\eta>0} H_\tbop^{\infty,\alpha_\cD,\alpha_\cD+\frac32-\eta}(\hat M_b),
  \end{split}
  \end{equation}
  whereas $h_{\rm COM,\tface,2}$ in~\eqref{EqKEFwdtf2} gets modified by an additional term
  \begin{equation}
  \label{EqKEFwdWhtf2Extra}
    \frac{1}{2\pi}\int_{-\sigma_0}^{\sigma_0} e^{-i\sigma\hat t_1}\psi(|\sigma|) \bigl((1-\chi) + ((1-i\sigma\hat r)^{-2}-1)\bigr)\hat a(\sigma)h^\Ups\,\dd\sigma
  \end{equation}
  which, being the inverse Fourier transform in $\hat t$ of an element of
  \[
    e^{-i\sigma\cT_1}|\sigma|^{-\alpha}L^2_\cp((-\sigma_0,\sigma_0))\cA^{0,0,1}((\hat X_b)_\scbtop)=L^2_\cp((-\sigma_0,\sigma_0))\cdot e^{-i\sigma\cT_1}\cA^{0,-\alpha,1-\alpha}((\hat X_b)_\scbtop),
  \]
  lies in $\Htb^{\sfs-1,\alpha_\cD,0}$ (since $\sfs-1+\alpha_\cD<-\frac32$ at the outgoing radial set).

  The remainder of the construction of $L_+^{-1}$ now follows that of Theorem~\ref{ThmKEFwd} with straightforward modifications; in particular, in~\eqref{EqKEFwdUps1}, one needs to add the term $e^{-i\sigma\hat t_1}\psi(|\sigma|)\hat a(\sigma)$ to the integrand.

  %%%%%%%%%%
  \pfstep{Part~\eqref{ItKEFwdWFwd}.} We consider the action of $L$ on
  \[
    h = h_0 + \hat g^{\prime\Ups}_b(\dot b(\hat t_1)) + \Bigl(h_{\rms 1}\bigl(\scal(\hat t_1)\bigr) + \breve h_{1,\rms 1}\bigl(\pa_{\hat t_1}\scal(\hat t_1)\bigr)\Bigr) + h_+^\Ups\bigl(\wh{\Ups_{(0)}}\bigr)
  \]
  term by term. Each term is supported in $\hat M_b^+$, and thus we only need to check membership in $\Htb^{\sfs-3,\alpha_\cD+2,0}$. Since $L\in\rho_\cD^2\Difftb^2(\hat M_b;S^2\,\Ttsc^*\hat M_b)$, we have $L h_0\in\dot H_\tbop^{\sfs-3,\alpha_\cD+2,0}$. Analogously to~\eqref{EqNMinkLdecomp}, we now write
  \begin{equation}
  \label{EqKEFwdWL}
  \begin{split}
    L=\sum_{j=0}^2 L_j\pa_{\hat t_1}^j, \quad &L_0\in\rho_\cD^2\Diffb^2(\hat X_b;S^2\,\Ttsc^*_{\hat X_b}\hat M_b), \\
    &L_1=\wh{[L,\hat t_1]}(0)\in\rho_\cD\Diffb^1,\quad L_2=-|\dd\hat t_1|_{\hat g_b^{-1}}^2\in\rho_\cD\Diffb^0.
  \end{split}
  \end{equation}
  We thus have, on a level set of $\hat t_1$,
  \[
    L \hat g^{\prime\Ups}_b(\dot b) = L_0\hat g^{\prime\Ups}_b(\dot b(\hat t_1)) + L_1\hat g^{\prime\Ups}_b(\pa_{\hat t_1}\dot b(\hat t_1)) + L_2\hat g^{\prime\Ups}_b(\pa_{\hat t_1}^2\dot b(\hat t_1)).
  \]
  The first term vanishes. The second term is the product of $\pa_{\hat t_1}\dot b\in H^\infty$ with a tensor $h\in\cA^{3-}(\hat X_b)$ (using~\eqref{EqKEComplComm}). The Fourier transform of this in $\hat t=\hat t_1-\cT_1(\hat r)$ is an element of
  \[
    \la\sigma\ra^{-\infty}L^2(\R_\sigma)\cdot e^{-i\sigma\cT_1(\hat r)}\cA^{3-}(\hat X_b),
  \]
  to which Lemmas~\ref{LemmaNMult} (with $\alpha=\beta=3-$, $\gamma=0$) and \ref{LemmaNMultHi} (with $\alpha=3-$) apply. By~\eqref{EqNFT3b}, we thus obtain $L_1\hat g^{\prime\Ups}_b(\pa_{\hat t_1}\dot b)\in\Htb^{\sfs,\alpha_\cD+2,0}$ since $\alpha_\cD+2<-\frac32+3$ and $\sfs+\alpha_\cD+2<-\frac32+3$ at the outgoing radial set. Finally, the third term is the product of $\pa_{\hat t_1}^2\dot b\in H^\infty(\R)$ with an element of $\hat r^{-1}\cA^1=\cA^2$, which is again contained in $\Htb^{\sfs-1,\alpha_\cD+2,0}$, now because of\footnote{This is one of the places where working with $\alpha_\cD<-\frac32$ makes the bookkeeping of weights at $\cD$ in the presence of the term $\breve h_{1,\rms 1}$ easier than in the setting of Theorem~\ref{ThmKEFwd}.} $\alpha_\cD+2<-\frac32+2$ and $(\sfs-1)+(\alpha_\cD+2)<-\frac32+2$ at the outgoing radial set.

  In a similar vein, we compute, at time $\hat t_1$,
  \begin{equation}
  \label{EqKEFwdWLhs1}
  \begin{split}
    &L\bigl( h_{\rms 1}(\scal) + \breve h_{1,\rms 1}(\pa_{\hat t_1}\scal) \bigr) \\
    &\qquad = L_0 h_{\rms 1}(\scal(\hat t_1)) + \bigl( L_1 h_{\rms 1}(\pa_{\hat t_1}\scal(\hat t_1)) + L_0 \breve h_{1,\rms 1}(\pa_{\hat t_1}\scal(\hat t_1)) \bigr) \\
    &\qquad \qquad + L_2 h_{\rms 1}(\pa_{\hat t_1}^2\scal(\hat t_1)) + L_1 \breve h_{1,\rms 1}(\pa_{\hat t_1}^2\scal(\hat t_1)) + L_2 \breve h_{1,\rms 1}(\pa_{\hat t_1}^3\scal(\hat t_1)).
  \end{split}
  \end{equation}
  Both summands in the first line vanish. Furthermore, by the same argument as above,
  \begin{alignat*}{2}
    L_2 h_{\rms 1}(\pa_{\hat t_1}^2\scal)&\in\rho_\cD^2\Diffb^0(\hat X_b)\cA^2(\hat X_b)H^\infty(\R_{\hat t_1}) &&\subset \Htb^{\sfs,\alpha_\cD+2,0}, \\
    L_1\breve h_{1,\rms 1}(\pa_{\hat t_1}^2\scal)&\in\rho_\cD\Diffb^1(\hat X_b)\cA^{1-}(\hat X_b)H^\infty(\R_{\hat t_1}) &&\subset \Htb^{\sfs,\alpha_\cD+2,0}, \\
    L_2\breve h_{1,\rms 1}(\pa_{\hat t_1}^3\scal)&\in\rho_\cD\Diffb^0(\hat X_b)\cA^{1-}(\hat X_b)H^\infty(\R_{\hat t_1}) &&\subset \Htb^{\sfs,\alpha_\cD+2,0},
  \end{alignat*}
  where we use that $\alpha_\cD+2<-\frac32+2$.

  Finally, recalling~\eqref{EqNMinkLdecompSpec} and the definition~\eqref{EqKEFwdWSolhUps}, we have
  \begin{equation}
  \label{EqKEFwdWMapFwd}
  \begin{split}
    &2\pi L h_+^\Ups\bigl(\wh{\Ups_{(0)}}\bigr) \\
    &\quad = [L,\phi_+] \int_{-1}^1 e^{-i\sigma\hat t_1}\chi_0(\sigma\hat r)\wh{\Ups_{(0)}}(\sigma)h^\Ups\,\dd\sigma \\
    &\quad\qquad + \phi_+\int_{-1}^1 e^{-i\sigma\hat t_1} \Bigl( \hat L(0) - i\sigma\wh{[L,\hat t_1]}(0) - \frac12\sigma^2[[L,\hat t_1],\hat t_1]\Bigr)\bigl( \chi_0(\sigma\hat r)\wh{\Ups_{(0)}}(\sigma)h^\Ups \bigr)\,\dd\sigma.
  \end{split}
  \end{equation}
  In view of~\eqref{EqKEFwdWSolhUps}, the first term on the right hand side defines an element of the space $\Difftb^{2,-2,\infty}(\hat M_b)\Htb^{\infty,\alpha_\cD,\alpha_\cD+\frac32-\eta}\subset\Htb^{\infty,\alpha_\cD+2,0}$. In the final line, we note that
  \begin{equation}
  \label{EqKEFwdWMapFwdUps}
  \begin{alignedat}{2}
    \hat L(0)\bigl(\chi_0(\sigma\hat r)h^\Ups\bigr) &= [\hat L(0),\chi_0]h^\Ups &&\in \cA^{\infty,2,1}((\hat X_b)_\scbtop), \\
    \sigma\wh{[L,\hat t_1]}(0)(\chi_0 h^\Ups) &\in \sigma\rho_\cD\Diffb^1(\hat X_b) \cA^{\infty,0,0} &&\subset \cA^{\infty,2,1}, \\
    \sigma^2[[L,\hat t_1],\hat t_1](\chi_0 h^\Ups) &\in \sigma^2\rho_\cD \cA^{\infty,0,0} &&\subset \cA^{\infty,3,2};
  \end{alignedat}
  \end{equation}
  multiplied by $\wh{\Ups_{(0)}}(\sigma)$, these lie in $\cA^{\infty,\alpha_\cD+\frac72,\alpha_\cD+\frac52}((\hat X_b)_\scbtop)\cdot L^2([-1,1]_\sigma)$. Again using Lemma~\ref{LemmaNMult} and the isomorphism~\eqref{EqNFT3b}, noting that $\alpha_\cD+\frac52>0$, and moreover using the rapid vanishing of $\chi_0(\sigma\hat r)$ as $|\sigma|\to\infty$ in conjunction with Lemma~\ref{LemmaNMultHi}, the final line of~\eqref{EqKEFwdWMapFwd} thus lies in $\Htb^{\sfs,\alpha_\cD+2,0}$. This finishes the proof.
\end{proof}

For later purposes, we record weighted $L^2$-bounds on the size of $h$ in~\eqref{EqKEFwdWSol}. For $\hat r_0>4\bhm$ and $\mu\geq 1$, set
\begin{equation}
\label{EqKEhatOmega}
  \hat\Omega_{\hat r_0,\mu} := \bigl\{ (\hat t,\hat x)\in\hat M_b \colon 0\leq\hat t\leq\mu,\ \bhm\leq|\hat x|\leq\mu\hat r_0+2(\mu-\hat t) \bigr\}.
\end{equation}
In view of the globally timelike nature of $\dd\hat t$ and the fact that $\hat g_b$ asymptotes to the Minkowski metric as $\hat r\to\infty$, we see that for all sufficiently large $\mu$, all four boundary hypersurfaces of the domain $\hat\Omega_{\hat r_0,\mu}$ are spacelike, and one of them (where $\hat t=0$) is initial (i.e.\ future timelike vector fields point into the domain there) while all others are final. We write
\[
  \ubar\mu(\hat r_0)\geq 1
\]
for the smallest $\mu$ for which this holds; note that this is a decreasing function of $\hat r_0$.

\begin{lemma}[Size of the forward solution]
\label{LemmaKEFwdWSize}
  In the notation of Theorem~\usref{ThmKEFwdW}, the forward solution of $L h=f\in\dot H_\tbop^{\sfs,\alpha_\cD+2,0}(\hat M_b^+;S^2\,\Ttsc^*\hat M_b)$ satisfies
  \begin{equation}
  \label{EqKEFwdWSizeMem}
    h\in\dot H_\tbop^{\sfs-1,\alpha_\cD,-2}(\hat M_b^+;S^2\,\Ttsc^*\hat M_b).
  \end{equation}
  Moreover, there exists a constant $C$ so that for all $\hat r_0>4\bhm$, $\mu\geq\ubar\mu(\hat r_0)$, and $\nu\geq 0$,
  \begin{equation}
  \label{EqKEFwdWSizenu}
    \|h\|_{H_\tbop^{\sfs-1,\alpha_\cD+\nu,-2+\nu}(\hat\Omega_{\hat r_0,\mu})^{\bullet,-}} \leq C\mu^\nu\|f\|_{H_\tbop^{\sfs,\alpha_\cD+2,0}(\hat\Omega_{\hat r_0,\mu})^{\bullet,-}}.
  \end{equation}
\end{lemma}

\begin{rmk}[Large spatial weight]
\label{RmkKEFwdWeight}
  We stress that the membership~\eqref{EqKEFwdWSizeMem} is \emph{false} if we were to work with $\alpha_\cD\in(-\frac32,-\frac12)$, the reason being that the expansion terms involving $\dot b$ and $\scal$ in~\eqref{EqKEFwdSol} only lie in $\dot H_\tbop^{*,-\frac32-\eta,-1}$ and $\dot H_\tbop^{*,-\frac32-\eta,-2}$ for all $\eta>0$. (This follows from Lemma~\ref{LemmaNExpTerm}.)
\end{rmk}

\begin{proof}[Proof of Lemma~\usref{LemmaKEFwdWSize}]
  To prove the membership~\eqref{EqKEFwdWSizeMem}, we only need to show this membership for the terms $\hat g_b^{\prime\Ups}(\dot b(t_1))$, $h_{\rms 1}(\scal(\hat t_1))$, and $\breve h_{1,\rms 1}(\pa_{\hat t_1}\scal(\hat t_1))$ in~\eqref{EqKEFwdWSol}; but this follows from Lemma~\ref{LemmaNExpTerm} for $(k,\alpha)=(1,1)$, $(2,2)$, and $(1,1)$, using that $\sfs-1+\alpha_\cD<-\frac32$ (by the strong Kerr-admissibility of $\sfs,\alpha_\cD$) is less than $-\frac32-k+\alpha=-\frac32$ indeed.

  We next prove the estimate~\eqref{EqKEFwdWSizenu}. Write $\hat z=(\hat t,\hat x)$. Fix $\chi\in\CIc(\R)$ to be equal to $1$ on a sufficiently large set so that $\chi(\frac{|\hat z|}{\mu})=1$ on $\hat\Omega_{\hat r,\mu}$ for all $\mu$. The function $(\frac{\la\hat z\ra}{\mu})^\nu\chi(\frac{|\hat z|}{\mu})$ is uniformly bounded in $\cA^0(\ol{\R^4})$ and thus, a fortiori, all of its derivatives along any finite number of 3b-vector fields obey uniform (in $\mu$) bounds in $L^\infty(\hat M_b)$. Therefore, multiplication by this function is uniformly (in $\mu$) bounded on all weighted 3b-Sobolev spaces. We can thus estimate
  \begin{align*}
    \|h\|_{\Htb^{\sfs-1,\alpha_\cD+\nu,-2+\nu}(\hat\Omega_{r_0,\mu})^{\bullet,-}} &\leq C\|\la\hat z\ra^\nu h\|_{\Htb^{\sfs-1,\alpha_\cD,-2}(\hat\Omega_{r_0,\mu})^{\bullet,-}} \\
      &\leq C\mu^\nu\Bigl\| \Bigl(\frac{\la\hat z\ra}{\mu}\Bigr)^\nu \chi\Bigl(\frac{|\hat z|}{\mu}\Bigr)h \Bigr\|_{\Htb^{\sfs-1,\alpha_\cD,-2}(\hat M_b^+)} \\
      &\leq C'\mu^\nu\|h\|_{\Htb^{\sfs-1,\alpha_\cD-2}(\hat M_b^+)} \\
      &\leq C''\mu^\nu\|f\|_{\Htb^{\sfs,\alpha_\cD+2,0}(\hat M_b^+)}.
  \end{align*}
  But by finite speed of propagation, $h|_{\hat\Omega_{\hat r_0,\mu}}$ only depends on $f|_{\hat\Omega_{\hat r_0,\mu}}$, so applying this estimate for $f$ being an extension of $f|_{\hat\Omega_{\hat r_0,\mu}}$ with minimal norm (see \citeII{\S\ref*{SsEstFn}} for details) yields~\eqref{EqKEFwdWSizenu}.
\end{proof}

\begin{rmk}[Improved se-regularity]
\label{RmkKEFwdImpr}
  A propagation of regularity argument would allow one to improve the regularity in~\eqref{EqKEFwdWSizeMem} from $\sfs-1$ to $\sfs$. We shall not implement this at the current 3b-level, but rather at the se-level in~\S\ref{SsEse}; see in particular~\eqref{EqEse}.
\end{rmk}

%%%%%%%%%%%%%%%%%%%%%%%%%%%%%%%%%%%%%%%%%%%%%%%%%%
\subsection{Solution operators preserving higher regularity}
\label{SsKHi}

To motivate the notion of higher regularity we shall now study, recall that the Kerr black hole we are studying in this section is described in coordinates $\hat t,\hat x$ which, in the gluing problem, are related to the coordinates on the ambient spacetime via $\hat t=\frac{t}{\eps}$, $\hat x=\frac{x}{\eps}$. We will be interested in s-regularity on $\wt M$ (see~\S\ref{EqNsewtM}), i.e.\ uniform (in $\eps$) bounds of functions upon differentiation along $\pa_t$, $\la\hat r\ra\pa_{\hat x}$ (cf.\ \eqref{EqNsVF}). In $\hat r=|\hat x|\geq\bhm>0$, this is equivalent to uniform bounds upon differentiation along
\begin{equation}
\label{EqKHisReg}
  \eps^{-1}\pa_{\hat t},\qquad
  \hat r(\pa_{\hat r}+\pa_{\hat t}),\qquad
  \pa_\omega\ \ \Bigl(\omega=\frac{\hat x}{|\hat x|}\Bigr).
\end{equation}
Note here that we can replace $\hat r\pa_{\hat r}$ by $\hat r(\pa_{\hat r}+\pa_{\hat t})$ since $\hat r\pa_{\hat r}=\hat r(\pa_{\hat r}+\pa_{\hat t})-\eps\hat r\cdot\eps^{-1}\pa_{\hat t}$, with $\eps\hat r$ smooth on $\wt M$; usage of~\eqref{EqKHisReg} is rather arbitrary at present, but it is particular convenient for us below since solutions of wave equations on Kerr naturally have regularity along the last two vector fields in~\eqref{EqKHisReg}, as shown below on the function spaces of relevance to us here. The Fourier transform in $\hat t$ intertwines the vector fields in~\eqref{EqKHisReg} with
\begin{equation}
\label{EqKHisRegFT}
  \eps^{-1}\sigma,\qquad
  \hat r(\pa_{\hat r}-i\sigma),\qquad
  \pa_\omega.
\end{equation}
See also~\eqref{ItNMultMod}. Without the presence of $\eps^{-1}\pa_{\hat t}$ in~\eqref{EqKHisReg}, it is natural to add the vector fields $\pa_{\hat t},\pa_{\hat x}$, or equivalently $\pa_{\hat t},\pa_{\hat r},\hat r^{-1}\pa_\omega$ to the lists~\eqref{EqKHisReg} and \eqref{EqKHisRegFT}; one should expect regularity with respect to these vector fields to be automatic. This motivates the following definition, in which `$\sO$' stands for `outgoing'.

\begin{definition}[3b-s-spaces]
\label{DefKHiSpaces}
  We use the notation from~\eqref{EqNSpcM} and~\eqref{EqNFT3b}, with function spaces on $\hat M_b$, $\hat X_b$ being defined as spaces of extendible distributions on $\hat M_b\subset\cM$, $\hat X_b\subset\cX$. Let $\Omega\subset\cV(\Sph^2)$ be a finite spanning set over $\CI(\Sph^2)$. We then define for time-translation-invariant $\sfs\in\CI(\Stb^*\hat M_b)$ and arbitrary $\alpha_\cD\in\R$, $k\in\N_0$ the normed spaces
  \begin{equation}
  \label{EqKHisSpace}
  \begin{split}
    H_{\tbop;\sO}^{(\sfs;k),\alpha_\cD,0}(\hat M_b) &:= \bigl\{ u\in\Htb^{\sfs,\alpha_\cD,0}(\hat M_b) \colon \pa_{\hat t}^l\pa_{\hat x}^\beta(\hat r(\pa_{\hat r}+\pa_{\hat t}))^j\Omega^\gamma u\in\Htb^{\sfs,\alpha_\cD,0}(\hat M_b), \\
      &\quad \hspace{18em} l+|\beta|+j+|\gamma|\leq k \bigr\}, \\
    H_{(\tbop;\sop),\eps}^{(\sfs;k),\alpha_\cD,0}(\hat M_b) &:= \bigl\{ u\in\Htb^{\sfs,\alpha_\cD,0}(\hat M_b) \colon (\eps^{-1}\pa_{\hat t})^j u\in H_{\tbop;\sO}^{(\sfs;k-j),\alpha_\cD,0}(\hat M_b),\ 0\leq j\leq k \bigr\}, \\
    H_{(\wh\tbop;\sO),\sigma}^{(\sfs;k),\alpha_\cD,0}(\hat X_b) &:= \bigl\{ u\in H_{\wh\tbop,\sigma}^{\sfs,\alpha_\cD,0}(\hat X_b) \colon \sigma^l\pa_{\hat x}^\beta(\hat r(\pa_{\hat r}-i\sigma))^j\Omega^\gamma u \in H_{\wh\tbop,\sigma}^{\sfs,\alpha_\cD,0}(\hat X_b), \\
      &\quad \hspace{18em} l+|\beta|+j+|\gamma|\leq k \}, \\
    H_{(\wh\tbop;\sop),(\sigma;\eps)}^{(\sfs;k),\alpha_\cD,0}(\hat X_b) &:= \bigl\{ u\in H_{\wh\tbop,\sigma}^{\sfs,\alpha_\cD,0}(\hat X_b) \colon (\eps^{-1}\sigma)^j u \in H_{(\wh\tbop;\sO),\sigma}^{(\sfs;k-j),\alpha_\cD,0}(\hat X_b),\ j=0,\ldots,k \bigr\};
  \end{split}
  \end{equation}
  the norm on the second space is given by
  \[
    \|u\|_{H_{(\tbop;\sop),\eps}^{(\sfs;k),\alpha_\cD,0}(\hat M_b)} = \sum_{j+l+|\beta|+j'+|\gamma|\leq k} \|(\eps^{-1}\pa_{\hat t})^j\pa_{\hat t}^l\pa_{\hat x}^\beta(\hat r(\pa_{\hat r}+\pa_{\hat t}))^{j'}\Omega^\gamma u\|_{\Htb^{\sfs,\alpha_\cD,0}(\hat M_b)},
  \]
  with the norms on the second and third spaces defined analogously.
\end{definition}

\begin{rmk}[Vector fields]
\label{RmkKHiSpacesVF}
  In the space $H_{(\wh\tbop;\sO),\sigma}$, one can omit testing with $\pa_{\hat x}^\beta$ when $\sigma$ remains in a compact set. However, as $|\sigma|\to\infty$, the presence of $\pa_{\hat x}^\beta$ becomes important, as in compact sets of $\hat X_b^\circ$ the intersection of the semiclassical (in $h=|\sigma|^{-1}$) characteristic sets of $i^{-1}|\sigma|^{-1}(\pa_{\hat r}-i\sigma)=h D_r-\hat\sigma$, $\hat\sigma=\frac{\sigma}{|\sigma|}$, and the elements of $|\sigma|^{-1}\Omega$ is nonempty. This would thus allow for nontrivial $e^{i\sigma\hat r}$ oscillations in compact subsets of $\hat X_b$, which we do not wish to allow.
\end{rmk}

\begin{rmk}[$\eps$-independent spaces]
\label{RmkKHiSpacesNoEps}
  If in this section one fixes $\eps:=1$, one obtains estimates for forward solutions of $L$ which encode stronger regularity on $\hat M_b$ of the sort that is useful if one is only interested in $L$ itself (rather than its role for the gluing problem).
\end{rmk}

The following result is an immediate consequence of (the variable order version of)~\eqref{EqNFT3b}.

\begin{lemma}[Fourier transform and s-regularity]
\label{LemmaKHiFT}
  The Fourier transform in $\hat t$ defines an isomorphism
  \[
    \cF \colon H_{(\tbop;\sop),\eps}^{(\sfs;k),\alpha_\cD,0}(\hat M_b) \to L^2\bigl(\R_\sigma;H_{(\wh\tbop;\sop),(\sigma;\eps)}^{(\sfs;k),\alpha_\cD,0}(\hat X_b)\bigr)
  \]
  with the property that $\cF$ and $\cF^{-1}$ have uniformly bounded operator norms for $\eps\in(0,1)$.
\end{lemma}

The spaces $H_{(\wh\tbop;\sO),\sigma}$ encode \emph{outgoing module regularity}. That this is compatible with estimates for the spectral family on an asymptotically flat space is well-known; the relevant statement for present purposes is as follows. (We only need estimates for real $\sigma$ here, as the only purpose of estimates for nonreal $\sigma$ for us is to \emph{justify} contour shifting arguments.)

\begin{prop}[Estimates for the spectral family: module regularity]
\label{PropKHiSpec}
  Let $\sfs,\alpha_\cD$ be strongly Kerr-admissible orders, and assume that also $\sfs-1,\alpha_\cD$ are strongly Kerr-admissible. Then for all $\sigma_0>0$ and $k\in\N_0$, there exists a constant $C$ so that we have the uniform estimate
  \begin{equation}
  \label{EqKHiSpec}
    \|u\|_{H_{(\wh\tbop;\sO),\sigma}^{(\sfs;k),\alpha_\cD,0}} \leq C\|\hat L(\sigma)u\|_{H_{(\wh\tbop;\sO),\sigma}^{(\sfs;k),\alpha_\cD+2,0}},\qquad \sigma\in\R,\ |\sigma|\geq\sigma_0.
  \end{equation}
\end{prop}

For bounded $\sigma$, one can replace the norm on $\hat L(\sigma)u$ by the $H_{(\wh\tbop;\sO),\sigma}^{(\sfs-1;k),\alpha_\cD+2,0}$-norm, which is the sum of the $H_\scop^{\sfs-1,\sfs+\alpha_\cD+1}$-norms of $\hat L(\sigma)u$ and its up to $k$-fold derivatives along $\hat r(\pa_{\hat r}-i\sigma)$, $\Omega$.

\begin{proof}[Proof of Proposition~\usref{PropKHiSpec}]
  The case $k=0$ is a consequence of \citeII{Proposition~\ref*{PropEstFThi}} and Proposition~\ref{PropKENon0}. For $k\geq 1$, this is an essentially standard module regularity statement in the sense of \cite{HassellMelroseVasySymbolicOrderZero,HaberVasyPropagation}; see \cite[\S\S{2--3}]{GellRedmanHassellShapiroZhangHelmholtz} for a detailed argument (in a closely related asymptotically Euclidean setting) for bounded real nonzero $\sigma$, and also \cite[\S{4}]{BaskinVasyWunschRadMink}, \cite[Lemma~3.20, Proposition~3.21]{HintzConicWave}.

  We only (briefly) discuss the case $k=1$. We define the test operators $A_1=\hat r(\pa_{\hat r}-i\sigma)$, $\{A_2,\cdots,A_N\}=\Omega$, and $B_0=\sigma$, $B_i=\pa_{\hat x^i}$, $i=1,2,3$. (Thus $A_i\in\Diff_{\scop,|\sigma|^{-1}}^{1,1,1}(\hat X_b)=h^{-1}\hat r\Diff_{\scop,|\sigma|^{-1}}^1$ where $h=|\sigma|^{-1}$, while $B_0\in\Diff_{\scop,|\sigma|^{-1}}^{0,0,1}=h^{-1}\Diff_{\scop,|\sigma|^{-1}}^0$, $B_i\in\Diff_{\scop,|\sigma|^{-1}}^{1,0,1}=h^{-1}\Diff_{\scop,|\sigma|^{-1}}^1$.) In the coordinates $\xi\,\dd\hat r+\eta\,\dd\omega$,\footnote{In particular, $\xi$ and $\hat r^{-1}\eta$ are non-degenerate fiber-linear coordinates on the scattering cotangent bundle.} their (semiclassical) scattering principal symbols are equal to $a_1=i\hat r(\xi-\sigma)$, resp.\ functions $a_j=a_j(\omega,\eta)$, $j=2,\ldots,N$, which are linear in $\eta$. Together with the principal symbols $b_i$ of the operators $B_i$, every classical symbol of class $\hat r S^1$ which vanishes at the outgoing radial set
  \[
    \{\xi-\sigma=0,\ \hat r^{-1}\eta=0,\ \hat r=\infty\}
  \]
  can be written as a linear combination (with symbolic coefficients) of the symbols $a_1,\ldots,a_N$, $b_1,b_2,b_3$. Controlling $B_0 u$ and $B_i u$, $i=1,2,3$, in $H_{\wh\tbop,\sigma}^{\sfs,\alpha_\cD,0}=H_{\scop,|\sigma|^{-1}}^{\sfs,\sfs+\alpha_\cD,\sfs}(\hat X_b)$ amounts to controlling $u$ in $H_{\scop,|\sigma|^{-1}}^{\sfs+1,\sfs+\alpha_\cD,\sfs+1}$ by $\hat L(\sigma)u$ in $H_{\scop,|\sigma|^{-1}}^{\sfs+1,\alpha_\cD+\sfs+3,\sfs+1}$, i.e.\ in a space with only the (semiclassical) regularity order and the semiclassical order increased by one. For the semiclassical order, this is trivial since $\hat L(\sigma)$ commutes with multiplication by $\sigma$. Consider thus the regularity order; we need to estimate $\|u\|_{H_{\scop,|\sigma|^{-1}}^{\sfs+1,\sfs+\alpha_\cD,\sfs}}\leq C\|\hat L(\sigma)\|_{H_{\scop,|\sigma|^{-1}}^{\sfs+1,\sfs+\alpha_\cD+3,\sfs}}$ (in fact, we have a bound by $\|\hat L(\sigma)u\|_{H_{\scop,|\sigma|^{-1}}^{\sfs,\sfs+\alpha_\cD+1,\sfs}}$). Near $\hat r=\infty$, control of $\sfs+1$ semiclassical scattering derivatives is immediate from elliptic regularity, while in compact subsets of $\hat X_b^\circ$, where this is only nontrivial at infinite (semiclassical) frequencies (i.e.\ it only matters for $\pa_{\hat t}$-orthogonal null-geodesics), it follows by applying the standard regularity theory (above-threshold radial point estimates over the event horizon and real principal type propagation) with regularity $\sfs+1$ in place of $\sfs$.

  It remains to control $A_i u$. We compute principal symbols of commutators of $\hat L(\sigma)$ with the $A_i$ in the (semiclassical) scattering calculus modulo symbols of class $|\sigma|^{-1}\hat r^{-1}S^2$: the principal symbol of $\hat L(\sigma)$  being $p\equiv\xi^2+\hat r^{-2}|\eta|^2-\sigma^2$ (modulo lower order symbols), we have $H_p\equiv 2\xi\pa_{\hat r}+2\hat r^{-3}|\eta|^2\pa_\xi+\hat r^{-2}H_{|\eta|^2}$ and thus
  \[
    \upsigma^2([\hat L(\sigma),A_1]) \equiv H_p(\hat r(\xi-\sigma)) \equiv 2\bigl(\xi(\xi-\sigma) + \hat r^{-2}|\eta|^2\bigr) \equiv p + (\xi-\sigma)^2 + \hat r^{-2}|\eta|^2,
  \]
  while $\upsigma^2([\hat L(\sigma),A_j])\equiv H_p a_j=\hat r^{-2}H_{|\eta|^2}a_j(\omega,\eta)$ vanishes quadratically at $\eta=0$. Therefore, we can write
  \[
    [\hat L(\sigma),A_i]=c_i\hat L(\sigma)+\sum_{j=1}^N Y_{i,j}A_j + \sum_{j=1}^3 Z_{i,j}B_j + R_i
  \]
  where $c_i\in\CI(\hat X_b)$, $Y_{i,j},Z_{i,j},R_i\in\Psi_{\scop,|\sigma|^{-1}}^{1,-1,1}$ (i.e.\ one order lower in all senses compared with $\hat L(\sigma)\in\Diff_{\scop,|\sigma|^{-1}}^{2,0,2}$), and with $Y_{i,j}$ having vanishing principal symbol at the outgoing radial set.

  Now, given $u\in H_{\wh\tbop,\sigma}^{\sfs,\alpha_\cD,0}$ with $\hat L(\sigma)u=:f\in H_{(\wh\tbop;\sO),\sigma}^{(\sfs;1),\alpha_\cD+2,0}$, one considers the equation for $U=(A_1 u,\ldots,A_N u)$, which reads
  \[
    \bigl( \delta_{i j}\hat L(\sigma) - Y_{i,j} \bigr)_{i,j=1,\ldots,N} U = F := (f_1,\ldots,f_N),\quad
    f_i := (A_i+c_i)f + \sum_{j=1}^3 Z_{i,j}B_j u + R_i u.
  \]
  Note that the threshold condition at the outgoing radial set for $(\delta_{i,j}\hat L(\sigma)-Y_{i,j})$ is the same as for $\hat L(\sigma)$ due to the vanishing of the principal symbols of $Y_{i,j}$ there. Since $U\in H_{\wh\tbop,\sigma}^{\sfs-1,\alpha_\cD}$, the assumption that also $\sfs-1,\alpha_\cD$ are strongly Kerr-admissible thus implies that we can estimate $\|U\|_{H_{\wh\tbop,\sigma}^{\sfs-1,\alpha_\cD,0}}\leq C\|F\|_{H_{\wh\tbop,\sigma}^{\sfs-1,\alpha_\cD+2,0}}$. This completes the proof of~\eqref{EqKHiSpec} for $k=1$. The case $k\geq 2$ can be treated by induction.
\end{proof}

\begin{rmk}[Low energies; transition face normal operators]
\label{RmkKHiSpec}
  If $\hat L(0)$ were invertible, the estimate~\eqref{EqKHiSpec} would be valid uniformly for \emph{all} $\sigma$ in the closed upper half plane; near zero energy, this follows from module regularity at the outgoing radial set in the scattering calculus on $[0,1)_{\rho'}\times\Sph^2$, $\rho'=(\hat r|\sigma|)^{-1}$ (or, more globally, in the sc-b-transition calculus), proved in a completely analogous manner. Furthermore, we have analogous module regularity estimates for the transition face normal operators in \citeII{Lemma~\ref*{LemmaEstMcInvft}}.
\end{rmk}

\bigskip

We are almost in a position to prove an analogue of Theorem~\ref{ThmKEFwdW} (and also of Theorem~\ref{ThmKEFwd}) which captures s-regularity, with uniform estimates in the parameter $\eps\in(0,1)$. We need the following spaces.

\begin{definition}[Function spaces on the real line]
\label{DefKHiFn}
  For $s\in\R$, $k\in\N_0$, we write $H_{-;\eps}^{s;k}(\R)$ for the space $H^{s+k}(\R)$ equipped with the $\eps$-dependent norm
  \[
    \|u\|_{H_{-;\eps}^{(s;k)}(\R)} := \sum_{j=0}^k \| (\eps^{-1}\pa_{\hat t})^j u \|_{H^s(\R)}.
  \]
  Recalling Definition~\usref{DefNHdashb}, we similarly define
  \[
    \tau^l\dot H_{-;\bop;\eps}^{(s;l;k)}([1,\infty]) := \bigl\{ u \in \tau^l\dot H_{-;\bop}^{(s;l)}([1,\infty]) \colon (\eps^{-1}\pa_\tau)^j\pa_\tau^l u\in L^2([1,\infty)),\ j=0,\ldots,k \bigr\},
  \]
  with its natural norm.
\end{definition}

In view of~\eqref{EqNHdashb}, we have the equivalent characterizations
\begin{alignat*}{2}
  H_{-;\eps}^{(s;k)}(\R) &= \bigl\{ u\in\sS'(\R) \colon& (1+\eps^{-1}|\sigma|)^k\hat u&\in\la\sigma\ra^{-s}L^2(\R_\sigma) \bigr\}, \\
  \tau^l\dot H_{-;\bop;\eps}^{(s;l;k)}([1,\infty]) &= \bigl\{ u \in \dot\sS'([1,\infty)) \colon& (1+\eps^{-1}|\sigma|)^k \sigma^l\hat u &\in \la\sigma\ra^{-s}L^2(\R_\sigma) \bigr\},
\end{alignat*}
with uniform equivalences of norms; the spaces on the right are equipped with the norms $\|\la\sigma\ra^s(1+\eps^{-1}|\sigma|)^k\hat u\|_{L^2}$ and $\|\la\sigma\ra^s(1+\eps^{-1}|\sigma|)^k\sigma^l\hat u\|_{L^2}$, respectively. This characterization of $\tau^l\dot H_{-;\bop;\eps}^{(s;l;k)}$ allows for the use of real orders $l\in[0,\frac12)$: for such $l$ we still have $|\sigma|^{-l}\la\sigma\ra^{-s}L^2(\R)\subset L^1(\R)\subset\sS'(\R)$, and we can thus define
\[
  \tau^l\dot H_{-;\bop;\eps}^{(s;l;k)}([1,\infty]) := \bigl\{ \cF^{-1}v \colon v \in (1+\eps^{-1}|\sigma|)^{-k}|\sigma|^{-l}\la\sigma\ra^{-s}L^2(\R_\sigma),\ \supp\cF^{-1}v\subset[1,\infty) \bigr\}
\]
with its natural norm.

Before we state the result for the linearized gauge-fixed Einstein equations $L h=f$ on Kerr, we illustrate a subtle technical point in the simple example of the operator $\pa_{\hat t}$ on $\R_{\hat t}$: if one wishes to prove bounds on spaces $H_{-;\eps}^{(s;k)}$ (or $H_{(\tbop;\sop),\eps}^{(\sfs;k),\alpha_\cD,0}$ for $L$) which are uniform in $\eps$, one cannot use cutoff functions in $\hat t$ when $k\geq 1$, as multiplication operators by such cutoffs are not uniformly bounded on these spaces; only cutoffs in $\eps\hat t$ are uniformly bounded. To make this concrete, consider the equation $\pa_{\hat t}u=f$ for $f\in\dot H_{-;\eps}^{(0;k)}([0,\infty)$, which we shall solve using the Fourier transform in $\hat t$; thus $\hat u(\sigma)=i\sigma^{-1}\hat f(\sigma)$, $\Im\sigma>0$. In analogy with (and using the notation of)~\eqref{EqKEFwdhSum}, we then have
\begin{equation}
\label{EqKHiuhilo}
\begin{split}
  &u=u_{\rm lo}+u_{\rm hi}, \\
  &\quad u_{\rm hi} := \frac{1}{2\pi}\int_\R e^{-i\sigma\hat t}(1-\psi(|\sigma|))\hat u(\sigma)\,\dd\sigma, \quad
  u_{\rm lo} := \frac{1}{2\pi}\int_{\gamma_0} e^{-i\sigma\hat t}\psi(|\sigma|) i\sigma^{-1}\hat f(\sigma)\,\dd\sigma.
\end{split}
\end{equation}
The bounds on $\hat f$ imply corresponding bounds for $\hat u$ away from $\sigma=0$, and therefore $u_{\rm hi}\in H_{-;\eps}^{(0;k)}(\R)$, with a uniform (in $\eps$) bound $\|u_{\rm hi}\|_{H_{-;\eps}^{(0;k)}(\R)}\leq C\|f\|_{\dot H_{-;\eps}^{(0;k)}([0,\infty))}$. Since $u=0$ for $\hat t<0$, we have the same bound for $\chi_-(\eps\hat t)u_{\rm lo}=-\chi_-(\eps\hat t)u_{\rm hi}$ where $\chi_-\in\CI(\R)$ equals $1$ on $(-\infty,-\delta]$ and $0$ on $[-\frac12\delta,\infty)$, where $\delta>0$ is any fixed number. We can thus write
\begin{equation}
\label{EqKHiu12}
  u = u_1 + u_2,\qquad u_1 := u_{\rm hi} + \chi_-(\eps\hat t)u_{\rm lo},\quad u_2 := (1-\chi_-(\eps\hat t))u_{\rm lo},
\end{equation}
with $u_1\in\dot H_{-;\eps}^{(0;k)}([-\delta\eps^{-1},\infty))$ and $u_2\in(\hat t+\delta\eps^{-1}+1)\dot H_{-;\bop;\eps}^{(0;1;k)}([-\delta\eps^{-1},\infty])$ uniformly bounded (in $\eps$) in norm by $f$. That is, if, starting from the decomposition~\eqref{EqKHiuhilo}, one wishes to decompose $u$ into a bounded piece $u_1$ and a linearly growing piece $u_2$, one is forced to enlarge the $\hat t$-support by an (arbitrarily small) multiple of $\eps^{-1}$ towards negative times as compared to $\supp f$.\footnote{One can also see that this is unavoidable by observing that, writing $f(\hat s)=f_0(\eps\hat s)$ with $f_0\in H^k$, we have $u(\hat t)=\int_0^{\hat t} f_0(\eps s)\,\dd s=\eps^{-1}F_0(\eps\hat t)$ where $F_0(s)=\int_0^s f_0(s')\,\dd s'$: this means that $u$, at times $\delta'\eps^{-1}$ for any fixed $\delta'>0$, is typically of size $\gtrsim\eps^{-1}$. Thus, the linear growth from $u_{\rm lo}$ is already dominant at such times; that is, in a decomposition as in~\eqref{EqKHiu12}, but with $\chi_-$ now equal to $1$ on $\negreal$, the piece $u_1$ no longer lies in any $\dot H_{-;\eps}^{(0;k)}$ space due to lack of (uniform) integrability.} The bounds for the norms of $u_1,u_2$ are \emph{not} uniform in $\delta$: choosing $\chi_{-,\delta}(t):=\chi_{-,1}(\delta^{-1}t)$ where $\chi_{-,1}\in\CI(\R)$ equals $1$ on $(-\infty,-1]$ and $0$ on $[-\frac12,\infty)$, multiplication by $\chi_{-,\delta}$ is a bounded map on $H_{-;\eps}^{(0;k)}(\R)$ with norm $\sim\delta^{-k}$.

We now return to the equation $L h=f$ and prove the following variant of Theorem~\ref{ThmKEFwdW}.

\begin{thm}[Forward solutions with weakly decaying forcing: higher regularity solution operators]
\label{ThmKHi}
  Let $\hat t_1$ be as in~\eqref{EqKEt1}. Fix $\check\chi_0\in\CIc(\R)$ with $\int\check\chi_0(\tau)\,\dd\tau=1$ and $\supp\check\chi_0\in(1,2)$. Define
  \[
    \hat M_{b,\ubar t}^+ := \ol{\{\hat t\geq\ubar t\}} \subset \hat M_b.
  \]
  Let $\alpha_\cD\in(-2,-\frac32)$, and let $\sfs\in\CI(\Stb^*\hat M_b)$ be such that $\sfs,\alpha_\cD$ as well as $\sfs-1,\alpha_\cD$ and $\sfs-3,\alpha_\cD+3$ are strongly Kerr-admissible orders. Let $\delta>0$ and $k\in\N_0$. Then there exists a continuous linear map
  \begin{equation}
  \label{EqKHiMap}
  \begin{split}
    &L_{\delta,+}^{-1} \colon \dot H_{(\tbop;\sop),\eps}^{(\sfs;k),\alpha_\cD+2,0}(\hat M_{b,0}^+;S^2\,\Ttsc^*\hat M_b) \\
    &\quad \to \dot H_{(\tbop;\sop),\eps}^{(\sfs-1;k),\alpha_\cD,0}(\hat M_{b,-\delta\eps^{-1}}^+;S^2\,\Ttsc^*\hat M_b) \\
    &\quad \qquad \oplus (\hat t_1+\delta\eps^{-1}+1)\dot H_{-;\bop;\eps}^{(N;1;k)}([-\delta\eps^{-1},\infty];\R^4) \\
    &\quad \qquad  \oplus (\hat t_1+\delta\eps^{-1}+1)^2\dot H_{-;\bop;\eps}^{(N;2;k)}([-\delta\eps^{-1},\infty];\scalspace_1) \\
    &\quad \qquad \oplus (\hat t_1+\delta\eps^{-1}+1)^{-(\alpha_\cD+\frac32)}\dot H_{-;\bop;\eps}^{(N;-(\alpha_\cD+\frac32);k)}([-\delta\eps^{-1},\infty];\C^{4\times 4}_{\rm sym})
  \end{split}
  \end{equation}
  which is uniformly bounded (in $\eps\in(0,1)$) for all $N$ and satisfies the following properties.
  \begin{enumerate}
  \item\label{ItKHiSol}{\rm (Solution.)} Given $f\in\dot H_{(\tbop;\sop),\eps}^{(\sfs;k),\alpha_\cD+2,0}$ and $L_{\delta,+}^{-1}f=:(h_0,\dot b,\scal,\Ups_{(0)})$, define
    \begin{equation}
    \label{EqKHiSolhUps}
      h_+^\Ups(\Ups_{(0)})(\hat t,\hat r,\omega) := \int \hat r^{-1}\check\chi_0\Bigl(\frac{\hat s}{\hat r}\Bigr)\Ups_{(0)}(\hat t-\hat s)h^\Ups\,\dd\hat s
    \end{equation}
    (where $\Ups_{(0)}h^\Ups=\sum_{0\leq\mu\leq\nu\leq 3}(\Ups_{(0)})_{\mu\nu}h^\Ups_{\mu\nu}$ in the notation~\eqref{EqKPureGaugeij}--\eqref{EqKPureGauge00}) and
    \begin{equation}
    \label{EqKHiSol}
      h := h_0 + \hat g_b^{\prime\Ups}(\dot b(\hat t_1)) + h_{\rms 1}\bigl(\scal(\hat t_1)\bigr) + \breve h_{1,\rms 1}\bigl(\pa_{\hat t_1}\scal(\hat t_1)\bigr) + h_+^\Ups\bigl(\Ups_{(0)}\bigr).
    \end{equation}
    Then $L h=f$.
  \item\label{ItKHiNorm}{\rm (Operator norm bound.)} As a map~\eqref{EqKHiMap}, we have the operator norm bound
    \begin{equation}
    \label{EqKHiNorm}
      \|L_{\delta,+}^{-1}\| \leq C_{N,k}\delta^{-k}
    \end{equation}
    where $C_k$ only depends on $\sfs,\alpha_\cD$, but not on $\eps,\delta$.
  \item\label{ItKHiIndep}{\rm (Independence of orders.)} $L_{\delta,+}^{-1}$ is independent of the choices of $\sfs,\alpha_\cD$ (satisfying the stated assumptions) and $k,N$.
  \end{enumerate}
\end{thm}

Besides the enlargements of $\hat t$-supports, note that also the term~\eqref{EqKHiSolhUps} in the description~\eqref{EqKHiSol} of the solution differs from~\eqref{EqKEFwdhUps}: it inherits regularity upon differentiation in $\eps^{-1}\pa_{\hat t}$ from $\wh{\Ups_{(0)}}$, which would not be the case for~\eqref{EqKEFwdhUps} due to the localization factor $\phi_+$ which is not uniformly bounded upon differentiation along $\eps^{-1}\pa_{\hat t}$.

\begin{rmk}[The case $\eps=1$]
\label{RmkKHiFwdSolNoEps}
  When $\eps=1$, and thus uniformity in $\eps$ is not an issue, one can ensure that all terms in the description of $h$ are supported in $\hat t\geq 0$ via the introduction of appropriate cutoff functions.
\end{rmk}

\begin{proof}[Proof of Theorem~\usref{ThmKHi}]
  Recall from Definition~\ref{DefKHiFn} that $f\in\dot H_{(\tbop;\sop),\eps}^{(\sfs;k),\alpha_\cD+2,0}(\hat M_{b,0}^+)$ is equivalent to
  \[
    (\eps^{-1}\pa_{\hat t})^j f \in \dot H_{\tbop;\sO}^{(\sfs;k-j),\alpha_\cD+2,0}(\hat M_{b,0}^+),\qquad j=0,\ldots,k.
  \]
  On the Fourier transform side, and for real $\sigma$, this amounts to $L^2_\sigma$-membership for $(\eps^{-1}\sigma)^j\hat f$ in the module regularity space $H_{(\wh\tbop;\sO),\sigma}^{(\sfs;k-j),\alpha_\cD+2,0}$ appearing in~\eqref{EqKHiSpec} (with $k-j$ in place of $k$). The proof strategy is thus to exploit that $\eps^{-1}\sigma$ is a parameter which commutes with all operators and thus with all estimates, while Proposition~\ref{PropKHiSpec} and Remark~\ref{RmkKHiSpec} will be used to deal with the module regularity. We shall repeat previous arguments rather explicitly so that the reader may more easily verify that all steps of our argument are compatible with higher s-regularity, and to demonstrate the origin of the bound~\eqref{EqKHiNorm}.

  %%%%%%%%%%
  \pfstep{Decomposition of the inverse Fourier transform of $\hat L(\sigma)^{-1}\hat f(\sigma)$.} Consider a source term $f\in\CIc((\hat M_{b,0}^+)^\circ;S^2\,\Ttsc^*\hat M_b)$. For real frequencies away from $0$, we can directly apply Proposition~\ref{PropKHiSpec} to obtain, for
  \[
    \hat h(\sigma) := \hat L(\sigma)^{-1}\hat f(\sigma),
  \]
  the uniform bounds
  \[
    (\eps^{-1}|\sigma|)^j \|\hat h(\sigma)\|_{H_{(\wh\tbop;\sO),\sigma}^{(\sfs;k-j),\alpha_\cD,0}} \leq C(\eps^{-1}|\sigma|)^j \|\hat f(\sigma)\|_{H_{(\wh\tbop;\sO),\sigma}^{(\sfs;k-j),\alpha_\cD+2,0}},
  \]
  for $j=0,\ldots,k$, so
  \[
    \|\hat h(\sigma)\|_{H_{(\wh\tbop;\sop),(\sigma;\eps)}^{(\sfs;k),\alpha_\cD,0}} \leq C\|\hat f(\sigma)\|_{H_{(\wh\tbop;\sop),(\sigma;\eps)}^{(\sfs;k),\alpha_\cD+2,0}},\qquad \sigma\in\R,\ |\sigma|\geq 1.
  \]

  For low frequencies, we proceed as in the proof of Theorem~\ref{ThmKEFwdW}. The module regularity version of Proposition~\ref{PropKELtf} produces $\hat a$, $\hat h_\tface$ satisfying the uniform bound
  \[
    (\eps^{-1}|\sigma|)^j \Bigl( |\sigma|^{-(\alpha_\cD+\frac32)}|\hat a(\sigma)| + \|\hat h_\tface(\sigma)\|_{H_{(\wh\tbop;\sO),|\sigma|}^{(\sfs;j),\alpha_\cD}} \Bigr) \leq C(\eps^{-1}|\sigma|)^j\|\hat f(\sigma)\|_{H_{(\wh\tbop;\sO),|\sigma|}^{(\sfs;k-j),\alpha_\cD+2}}
  \]
  for $j=0,\ldots,k$, so
  \begin{equation}
  \label{EqKHitfTerm}
    (1+\eps^{-1}|\sigma|)^k |\sigma|^{-(\alpha_\cD+\frac32)} |\hat a(\sigma)| + \|\hat h_\tface(\sigma)\|_{H_{(\wh\tbop;\sop),(\sigma;\eps)}^{(\sfs;k),\alpha_\cD}} \leq C\|\hat f(\sigma)\|_{H_{(\wh\tbop;\sop),(\sigma;\eps)}^{(\sfs;k),\alpha_\cD+2}}.
  \end{equation}
  For the remainder term $\hat f'(\sigma):=\hat f(\sigma)-\hat L(\sigma)(e^{-i\sigma\cT_1(\hat r)}(1-i\sigma\hat r)^{-2}\hat a(\sigma)h^\Ups+\hat h_\tface(\sigma))$ we then have the uniform (in $|\sigma|\leq 1$, and now also in $\eps\in(0,1)$) bound
  \[
    \|\hat f'(\sigma)\|_{H_{(\wh\tbop;\sop),(\sigma;\eps)}^{(\sfs-2;k),\alpha_\cD+3-\eta,0}} \leq C\|\hat f(\sigma)\|_{H_{(\wh\tbop;\sop),(\sigma;\eps)}^{(\sfs;k),\alpha_\cD+2}}
  \]
  for any fixed $\eta>0$. We next wish to apply a module regularity version of Proposition~\ref{PropKELo} (with $\sfs-1,\alpha_\cD+1-\eta$ in place of $\sfs,\alpha_\cD$), i.e.\ the estimate~\eqref{EqKELo}, to bound
  \[
    \bigl(\wh{h_0}(\sigma),\dot b_0(\sigma),\scal_0(\sigma)\bigr) := \wt L(\sigma)^{-1}\bigl(\hat f'(\sigma),0\bigr).
  \]
  Such a module regularity version is readily proved by combining Remark~\ref{RmkKHiSpec} with the observation that $f_{\rm Kerr}$ and $f_{\rm COM}$ in~\eqref{EqKECompl}, and thus also $f_{\rm mod}$ in~\eqref{EqKELomod}, have infinite module regularity (as follows immediately from~\eqref{EqKECompl2}). Thus,
  \begin{equation}
  \label{EqKHiSingTerms}
    \|\wh{h_0}(\sigma)\|_{H_{(\wh\tbop;\sop),(\sigma;\eps)}^{(\sfs-1;k),\alpha_\cD+1-\eta,0}} + (1+\eps^{-1}|\sigma|)^k\bigl(|\dot b_0(\sigma)| + |\scal_0(\sigma)|\bigr) \leq C\|\hat f'(\sigma)\|_{H_{(\wh\tbop;\sop),(\sigma;\eps)}^{(\sfs-2;k),\alpha_\cD+3-\eta,0}}.
  \end{equation}

  Write $\wt{\dot b}_1(\sigma)=i\sigma^{-1}\dot b_0(\sigma)$ and $\wt\scal_1(\sigma):=-\sigma^{-2}\scal_0(\sigma)$ as in~\eqref{EqKEFwdb1S1}, and define $\dot b_1,\scal_1$ by the contour integrals
  \begin{equation}
  \label{EqKHiSingTermsNew}
    \dot b_1(\hat t_1) := \frac{1}{2\pi} \int_{\gamma_0} e^{-i\sigma\hat t_1}\psi(|\sigma|)\wt{\dot b}_1(\sigma)\,\dd\sigma,\qquad
    \scal_1(\hat t_1) := \frac{1}{2\pi} \int_{\gamma_0} e^{-i\sigma\hat t_1}\psi(|\sigma|)\wt\scal_1(\sigma)\,\dd\sigma;
  \end{equation}
  as in~\eqref{EqKEFwdb1scal1}, where $\psi\in\CIc((-\sigma_0,\sigma_0))$ (with $\sigma_0>0$ fixed) is equal to $1$ near $0$. We then have the pointwise (in $\sigma$) identity~\eqref{EqKEFwdWh}. We compute the inverse Fourier transform of $\hat h$, analogously to~\eqref{EqKEFwdhSum} and the subsequent discussion, to be
  \[
    h = h_{\rm reg} + h_{\rm mod} + h_{\tface,1} + h_{\tface,2},
  \]
  where, for $\chi\in\CIc((\hat X_b)_\scbtop)$ equal to $1$ near $\zface$ and $0$ near $\scface$,
  \begin{subequations}
  \begin{align}
  \label{EqKHihReg}
    h_{\rm reg} &:= \frac{1}{2\pi}\int_\R e^{-i\sigma\hat t}\Bigl( \bigl(1-\psi(|\sigma|)\bigr)\hat h(\sigma) + \psi(|\sigma|)\wh{h_0}(\sigma) + \psi(|\sigma|)\hat h_\tface(\sigma) \Bigr)\,\dd\sigma, \\
  \label{EqKHihExp}
    h_{\rm mod} &:= \hat g_b^{\prime\Ups}(\dot b_1(\hat t_1)) + h_{\rms 1}(\scal_1(\hat t_1)) + \breve h_{1,\rms 1}\bigl(\pa_{\hat t_1}\scal_1(\hat t_1)\bigr), \\
  \label{EqKHihtf1}
    h_{{\rm tf},1} &:= \frac{1}{2\pi} \int_{-\sigma_0}^{\sigma_0} e^{-i\sigma\hat t_1}\psi(|\sigma|)\chi\Bigl[\log(\sigma+i 0)\Ups\bigl(-\sigma^2\wt\scal_1(\sigma)\bigr) + \hat a(\sigma) \Bigr]h^\Ups\,\dd\sigma, \\
  \label{EqKHihtf2}
    \begin{split}
      h_{{\rm tf},2}&:= \frac{1}{2\pi}\int_{-\sigma_0}^{\sigma_0} e^{-i\sigma\hat t_1}\psi(|\sigma|) \Bigl[ \chi\log\Bigl(\frac{\hat r}{\sigma\hat r+i}\Bigr)\Ups\bigl(-\sigma^2\wt\scal_1(\sigma)\bigr)h^\Ups \\
      &\quad \hspace{6em} + (1-\chi)\log\Bigl(\frac{\sigma\hat r}{\sigma\hat r+i}\Bigr)\Ups\bigl(-\sigma^2\wt\scal_1(\sigma)\bigr)h^\Ups + \tilde h_\tface\bigl(\sigma^2\wt\scal_1(\sigma),\sigma\hat r,\omega\bigr) \\
      &\quad \hspace{6em} + \bigl((1-i\sigma\hat r)^{-2}-\chi\bigr)\hat a(\sigma)h^\Ups \Bigr]\,\dd\sigma.
    \end{split}
  \end{align}
  \end{subequations}
  (Here,~\eqref{EqKHihtf1} is the same as~\eqref{EqKEFwdWhtf1}, while~\eqref{EqKHihtf2} is the same as~\eqref{EqKEFwdtf2} plus the additional term~\eqref{EqKEFwdWhtf2Extra}.) By Lemma~\ref{LemmaKHiFT}, $h_{\rm reg}$ is uniformly bounded in $H_{(\tbop;\sop),\eps}^{(\sfs-1;k),\alpha_\cD,0}$, while $h_{{\rm tf},2}$ is uniformly bounded in $H_{(\tbop;\sop),\eps}^{(\sfs;k),-\frac32-\eta,0}$ for any fixed $\eta>0$.

  %%%%%%%%%%
  \pfstep{Estimates for negative times.} To analyze $h_{\rm mod}$ and $h_{{\rm tf},1}$, fix
  \[
    \chi_{-,1}\in\CI(\R),\quad \chi_{-,1}|_{(-\infty,-1]}=1,\ \ \chi_{-,1}|_{[-\frac12,\infty)}=0,\qquad
    \chi_{-,\delta} := \chi_{-,1}\Bigl(\frac{\cdot}{\delta}\Bigr),
  \]
  and note that $\chi_{-,\delta}h=0$, so
  \begin{equation}
  \label{EqKHiFwdSupp}
    \chi_{-,\delta}(\eps\hat t)(h_{\rm mod} + h_{{\rm tf},1}) = -\chi_{-,\delta}(\eps\hat t)(h_{\rm reg} + h_{{\rm tf},2}).
  \end{equation}
  Let $\phi\in\CIc(\hat X_b^\circ)$ be equal to $1$ on a large compact set so that $\phi\hat g_b^{\prime\Ups}(\dot b)$ (where $\dot b$ runs over a basis of $\R^4$), $\phi h_{\rms 1}(\scal)$ (where $\scal$ runs over a basis of $\scalspace_1$), and $\phi h_{\mu\nu}^\Ups$ are linearly independent. Multiply~\eqref{EqKHiFwdSupp} by $\phi$. Since the function $(\hat t,\hat x)\mapsto\phi(\hat x)\chi_{-,\delta}(\eps\hat t)$ remains uniformly bounded upon differentiation along the basic 3b-vector fields $\hat r\pa_{\hat t}$, $\hat r\pa_{\hat r}$, $\pa_\omega$, and
  \[
    (\eps^{-1}\pa_{\hat t})^j(\phi(\hat x)\chi_{-,\delta}(\eps\hat t))=\phi(\hat x)(\pa_t^j\chi_{-,\delta})(\eps\hat t)=\delta^{-j}\phi(\hat x)(\pa_t^j\chi_{-,1})(\delta^{-1}\eps\hat t)
  \]
  for all $j\in\N_0$, we conclude that we have a uniform (in $\eps,\delta\in(0,1)$) bound
  \begin{align*}
    \|\phi(\hat x)\chi_{-,\delta}(\eps\hat t)(h_{\rm mod}+h_{\tface,1})\|_{H_{(\tbop;\sop),\eps}^{(\sfs-1;k),0,0}} &\leq C_k\delta^{-k}\|h_{\rm reg}+h_{\tface,2}\|_{H_{(\tbop;\sop),\eps}^{(\sfs-1;k),\alpha_\cD,0}} \\
      &\leq C_k'\delta^{-k}\|f\|_{H_{(\tbop;\sop),\eps}^{(\sfs;k),\alpha_\cD+2,0}}.
  \end{align*}
  (The weight at $\cD$ on the left is irrelevant since $\phi\in\CIc(\hat X_b^\circ)$.)

  We now differentiate~\eqref{EqKHiFwdSupp} along $\pa_{\hat t}$ and use the uniform bound
  \[
    \|\pa_{\hat t_1}\dot b_1\|_{H_{-;\eps}^{(N;k)}(\R)},\ \|\pa_{\hat t_1}^2\scal_1\|_{H_{-;\eps}^{(N;k)}(\R)} \leq C_{N,k}\|f\|_{H_{(\tbop;\sop),\eps}^{(\sfs;k),\alpha_\cD+2,0}},
  \]
  which follows from~\eqref{EqKHiSingTerms}--\eqref{EqKHiSingTermsNew}, as well as the uniform $(1+\eps^{-1}|\sigma|)^{-k}L^2([-1,1])$ bound on $\sigma(-\log(\sigma+i 0)\sigma^2\wt\scal_1(\sigma)+\hat a(\sigma))$ which in addition uses~\eqref{EqKHitfTerm}. We then conclude, analogously to~\eqref{EqKEFwdBd3Pf2} and using that $\supp\chi_{-,2\delta}\subset\{\chi_{-,\delta}=1\}$, that
  \[
    \|\chi_{-,2\delta}(\eps(\hat t_1+\fT))\pa_{\hat t_1}\scal_1\|_{H_{-;\eps}^{(N;k)}} \leq C_{N,k}\delta^{-k}\|f\|_{H_{(\tbop;\sop),\eps}^{(\sfs;k),\alpha_\cD+2,0}}
  \]
  where $\fT$ is a constant, depending only on the choices of $\hat t_1,\phi$, so that $\hat t\geq 0$ implies $\hat t_1\geq-\fT$ on $\supp\phi$. Upon shifting $\hat t_1$ by $\fT$, we may assume that $\fT=0$.

  We now plug this information into~\eqref{EqKHiFwdSupp} for $\chi$ which on $\supp\phi$ is a function of $\sigma$ only which moreover equals $1$ on $\supp\psi$. Replacing $\delta$ by $\delta/2$ in the above arguments, we thus obtain
  \begin{equation}
  \label{EqKHiNeg}
  \begin{split}
    &\|\chi_{-,\delta}(\eps\hat t_1)\dot b\|_{H_{-;\eps}^{(N;k)}(\R)} + \|\chi_{-,\delta}(\eps\hat t_1)\scal_1\|_{H_{-;\eps}^{(N;k)}(\R)} + \|\chi_{-,\delta}(\eps\hat t_1)\Ups_1\|_{H_{-;\eps}^{(N;k)}(\R)} \\
    &\qquad\quad \leq C_{N,k}\delta^{-k}\|f\|_{H_{(\tbop;\sop),\eps}^{(\sfs;k),\alpha_\cD+2,0}(\hat M_b)}, \\
    &\quad \Ups_1(\hat t_1) := \frac{1}{2\pi}\int_{-\sigma_0}^{\sigma_0} e^{-i\sigma\hat t_1}\psi(|\sigma|)\Bigl(\Ups\bigl(-\sigma^2\log(\sigma+i 0)\wt\scal_1(\sigma)\bigr) + \hat a(\sigma)\Bigr)\,\dd\sigma.
  \end{split}
  \end{equation}
  (This is the analogue of~\eqref{EqKEFwdBdNeg}.)

  %%%%%%%%%%
  \pfstep{Reshuffling of the terms.} We now fix $\chi=\chi(\sigma\hat r):=\int e^{-i\sigma\hat r\zeta}\check\chi_0(\zeta)\,\dd\zeta$ in~\eqref{EqKHihReg}--\eqref{EqKHihtf2}, with $\check\chi_0\in\CIc((1,2))$ as in the statement of the Theorem. Note that
  \[
    h_{\tface,1} = \hat r^{-1}\check\chi_0\Bigl(\frac{\cdot}{\hat r}\Bigr) * \Ups_1,
  \]
  with convolution in $\hat t_1$. Split $\Ups_1=\chi_{-,\delta}(\eps\hat t_1)\Ups_1+(1-\chi_{-,\delta}(\eps\hat t_1))\Ups_1$ and correspondingly $h_{\tface,1}=h_{\tface,1,-}+h_{\tface,1,+}$. In view of~\eqref{EqKHiNeg} and using Lemma~\ref{LemmaNL2FT} (upon passing through the Fourier transform), we then have a uniform bounds
  \[
    \|h_{\tface,1,-}\|_{H_{(\tbop;\sop),\eps}^{(N;k),-\frac32-\eta,0}} + \|h_{\tface,1,+}\|_{H_{(\tbop;\sop),\eps}^{(N;k),\alpha_\cD,\alpha_\cD+\frac32-\eta}} \leq C_{N,\eta,k}\delta^{-k}\|f\|_{H_{(\tbop;\sop),\eps}^{(\sfs;k),\alpha_\cD+2,0}},
  \]
  cf.\ \eqref{EqKEFwdhComtfm}. (Note here that module regularity is weaker than full 3b-regularity, so only bounds for derivatives along $\eps^{-1}\pa_{\hat t}$ require an argument; but on the Fourier transform side this vector field becomes $-i\eps^{-1}\sigma$ and thus does not affect the arguments in Lemma~\ref{LemmaNL2FT}.) Furthermore, on $\supp h_{\tface,1,+}$ we have $\hat t\geq -\eps^{-1}\delta + \hat r$. Write then
  \begin{equation}
  \label{EqKHihDecomp}
    h = (h_{\rm reg} + h_{{\rm tf},2} + h_{\rm mod,-}+h_{\tface,1,-}) + (h_{\rm mod,+} + h_{\tface,1,+})
  \end{equation}
  where
  \[
    h_{\rm mod,-} := \hat g_b^{\prime\Ups}\bigl(\chi_{-,\delta}(\eps\hat t_1)\dot b_1(\hat t_1)\bigr) + h_{\rms 1}\bigl(\chi_{-,\delta}(\eps\hat t_1)\scal_1(\hat t_1)\bigr) + \breve h_{1,\rms 1}\bigl(\pa_{\hat t_1}\bigl(\chi_{-,\delta}(\eps\hat t_1)\scal_1(\hat t_1)\bigr)\bigr)
  \]
  (with $H_{(\tbop;\sop),\eps}^{(\sfs-1;k),\alpha_\cD,0}$-norm uniformly bounded by $\delta^{-k}$) and
  \begin{align*}
    h_{\rm mod,+} &:= \hat g_b^{\prime\Ups}\bigl((1-\chi_{-,\delta}(\eps\hat t_1))\dot b_1(\hat t_1)\bigr) + h_{\rms 1}\bigl((1-\chi_{-,\delta}(\eps\hat t_1))\scal_1(\hat t_1)\bigr) \\
      &\quad\qquad + \breve h_{1,\rms 1}\bigl(\pa_{\hat t_1}\bigl((1-\chi_{-,\delta}(\eps\hat t_1))\scal_1(\hat t_1)\bigr)\bigr).
  \end{align*}
  On the support of the second parenthesis of~\eqref{EqKHihDecomp}, we have $\hat t\geq-\eps^{-1}\delta$, and since $\hat t\geq 0$ on $\supp h$, also the first parenthesis is supported in $\hat t\geq-\eps^{-1}\delta$; and its $H_{(\tbop;\sop),\eps}^{(\sfs-1;k),\alpha_\cD,0}$-norm is uniformly bounded by $\delta^{-k}\|f\|_{H_{(\tbop;\sop),\eps}^{(\sfs;k),\alpha_\cD+2,0}}$. We may thus set
  \begin{align*}
    L_{\delta,+}^{-1}f&=\Bigl( h_{\rm reg} + h_{{\rm tf},2} + h_{\rm mod,-}+h_{\tface,1,-}, \\
      &\quad \hspace{3em} (1-\chi_{-,\delta}(\eps\hat t_1))\dot b_1,\ (1-\chi_{-,\delta}(\eps\hat t_1))\scal_1,\ (1-\chi_{-,\delta}(\eps\hat t_1))\Ups_1 \Bigr).
  \end{align*}
  This completes the proof.
\end{proof}

%%%%%%%%%%%%%%%%%%%%%%%%%%%%%%%%%%%%%%%%%%%%%%%%%%%%%%%%%%%%%%%%%%%%%%
\section{Linearized gauge-fixed Einstein equations on glued spacetimes}
\label{SE}

In this section, we combine the uniform se-regularity theory from \citeII{\S\ref*{SEst}} with estimates for forward solutions of the Kerr model operator $L$ from~\S\ref{SK} to prove uniform estimates for forward solutions of the linearized gauge-fixed Einstein vacuum equations on standard domains in glued spacetimes on se-Sobolev spaces (\S\ref{SsEse}). These are then upgraded to (tame) estimates on spaces encoding higher s-regularity (\S\ref{SsEs}) by adapting the arguments from \citeII{\S\S\ref*{SsScS}, \ref*{SsNTame}}.

Using the notation introduced in~\S\ref{SssNse}, we shall work on glued spacetimes (see \citeII{Definition~\ref*{DefGl}}) of the following type.

\begin{definition}[Kerr-mod-$\cO(\hat\rho^2)$ glued spacetime]
\label{DefEGlue}
  Fix a smooth $4$-dimensional globally hyperbolic Lorentzian manifold $(M,g)$, an inextendible timelike geodesic $\cC\subset M$, parameterized by arc length by $c\colon I_\cC\subseteq\R\to M$, Fermi normal coordinates $z=(z^0,z^1,z^2,z^3)=(t,x)\in U_{\rm Fermi}\subset I_\cC\times\R^3$ around $\cC$, with $\pa_t$ future timelike, and subextremal Kerr parameters $b=(\bhm,\bha)$ where $\bhm>0$, $\bha\in\R^3$. A \emph{Kerr-mod-$\cO(\hat\rho^2)$ glued spacetime} is then a pair $(\wt M,\wt g)$ where $\wt M=[[0,1)_\eps\times M;\{0\}\times\cC]$, and $\wt g$ is a section of $S^2\wt T^*\wt M$ over $\wt M\setminus\wt K^\circ$, where $\wt K:=\{|\hat x|\leq\bhm\}$ with $\hat x:=\frac{x}{\eps}$, with the following properties.
  \begin{enumerate}
  \item\label{ItEGlueReg}{\rm (Regularity.)} The tensor $\wt g$ has regularity $\CI+\cC_\seop^{\infty,1,2}=\CI+\rho_\circ\hat\rho^2\CI_\seop$ (as a section of $S^2\wt T^*\wt M$ over $\wt M\setminus\wt K^\circ$).
  \item\label{ItEGlueModFar}{\rm (Far field behavior.)} $\wt g|_{M_\circ}=\upbeta_\circ^*g$, i.e.\ in local coordinates $w\in\R^4$ on a chart on $M\setminus\cC$, we have $\wt g_{\mu\nu}(\eps,w)|_{\eps=0}=g_{\mu\nu}(w)$.
  \item\label{ItEGlueMod2}{\rm (Near field behavior.)} In the coordinates $\eps,t,\hat x$, we have
    \begin{equation}
    \label{EqEGlueMod}
    \begin{split}
      (\hat h_{(2)})_{\hat\mu\hat\nu}(\eps,t,\hat x) := \eps^{-2}\bigl(\wt g_{\mu\nu}(\eps,t,\hat x) - (\hat g_b)_{\hat\mu\hat\nu}(\hat x)\bigr) &\in \eps^{-2}(\hat\rho^2\CI + \cC_\seop^{\infty,1,2}) \\
        &= \rho_\circ^{-2}\CI + \rho_\circ^{-1}\cC_\seop^\infty.
    \end{split}
    \end{equation}
    Here, the indices $\mu,\nu$ refer to the coordinates $t,x$ on $M$ near $\cC$, and the indices $\hat\mu,\hat\nu$ refer to the corresponding coordinates $\hat t=\frac{t}{\eps}$, $\hat x=\frac{x}{\eps}$ on the Kerr spacetime manifold $\hat M_b^\circ$.
  \end{enumerate}
  More generally, if $\sF\subset\cC_\seop^{2,1,2}$ is a linear space of sections of $S^2\wt T^*\wt M$ over $\wt M\setminus\wt K^\circ$, and the space $\cC_\seop^{\infty,1,2}$ in \eqref{ItEGlueReg} and \eqref{ItEGlueMod2} is replaced by $\sF$, we say that $\wt g$ has regularity $\CI+\sF$.
\end{definition}

This is stronger than \citeII{Definition~\ref*{DefGl}} in that we require in~\eqref{EqEGlueMod} the equality of $\wt g$ and the Kerr metric $\hat g_b$ not just at $\hat M$, but in fact modulo terms which vanish quadratically at $\hat M$; hence the terminology \emph{Kerr-mod-$\cO(\hat\rho^2)$}. The main result of \cite{HintzGlueLocI} does produce a Kerr-mod-$\cO(\hat\rho^2)$ glued spacetime; we recall the details in Theorem~\ref{ThmTrGlueLocI} below. (In this case, the tensor $\hat h_{(2)}$ has a restriction to $\eps=0$ of class $\CI(I_\cC;\cA^{(-2,0),-1}(\hat X_b;S^2\,\Ttsc^*_{\hat X_b}\hat M_b))$. In Definition~\ref{DefEGlue}, we are content with the minimal assumptions, as far as the regularity of $\wt g$ near $\hat M$ is concerned, for which the uniform analysis of the linearized gauge-fixed Einstein operator below goes through.)

\begin{rmk}[Nondegeneracy of $\wt g$]
  Due to the non-degeneracy of $\wt g|_{M_\circ}=\upbeta_\circ^*g$ and $\wt g|_{\hat M_p}$ (which is the Kerr metric by~\eqref{EqEGlueMod}) for $p\in\cC$, there exists, for any precompact open set $\Omega\subset M$, a number $\eps_0>0$, depending only on $\Omega$ and the $\cC_\seop^{0,1,2}$-norm of the correction term in part~\eqref{ItEGlueReg}, so that $\wt g$ is Lorentzian on $\wt\upbeta^{-1}([0,\eps_0]\times\bar\Omega)$. For simplicity of notation, we shall often not make the open neighborhood of $M_\circ\cup(\hat M\setminus\wt K^\circ)$ on which $\wt g$ is Lorentzian explicit. See also the discussion preceding \citeI{Notation~\ref*{NotGLDef}}.
\end{rmk}

We shall analyze forward solutions of wave equations on $(\wt M,\wt g)$ uniformly (for all sufficiently small $\eps>0$) on \emph{standard domains}. We first recall from \citeII{Definition~\ref*{DefGlDynStdM}}:

\begin{definition}[Standard domains in $(M,g)$]
\label{DefEStdM}
  Fix a future timelike vector field $T\in\cV(M)$. A standard domain $\Omega\subset M$ is a compact submanifold with corners all of whose boundary hypersurfaces are spacelike, and so that one boundary hypersurface $X\subset\Omega$ is initial (i.e.\ $T$ points into $\Omega$ at $X$) and all others are final. If $\Omega\cap\cC\neq\emptyset$, we moreover require $\Omega$ to be of the form
  \[
    \Omega = \Omega_{t_0,t_1,r_0} := \{(t,x) \colon t_0\leq t\leq t_1,\ |x|\leq r_0+2(t_1-t) \}
  \]
  for $t_0,t_1\in I_\cC$, $t_0<t_1$, $r_0>0$, so that $\Omega_{t_0,t_1,r_0}\subset U_{\rm Fermi}$; furthermore $\Omega$ is non-refocusing, i.e.\ null-geodesics in $\Omega$ intersect $\cC\cap\Omega$ at most once; and, writing covectors over $r>0$ as $-\sigma\,\dd t+\xi\,\dd r+\eta$, $\eta\in T^*\Sph^2$, the quantity $\sigma^{-1}r H_G\frac{\xi}{\sigma}$ has a positive lower bound on $T^*\Omega\cap\{G(\zeta):=|\zeta|_{g^{-1}}^2=0,\ \zeta\neq 0\}$ for $\frac{\xi}{\sigma}\in[-\frac34,\frac34]$.\footnote{In the Minkowskian case $G=-\sigma^2+\xi^2+r^{-2}|\eta|^2$, we have $\frac12 H_G=\sigma\pa_t+\xi\pa_r+r^{-3}|\eta|^2\pa_\xi+r^{-2}H_{|\eta|^2}$ and thus $\sigma^{-1}r H_G(\frac{\xi}{\sigma})=\frac{r^{-2}|\eta|^2}{\sigma^2}=1-\frac{\xi^2}{\sigma^2}$ on $G^{-1}(0)$.}
\end{definition}

It follows from \citeII{Lemma~\ref*{LemmaGlDynStdEx}} that for any $t_0\in I_\cC$ and $t'_0<t'_1\in\R$, the domain $\Omega_{t_0+\lambda t'_0,t_0+\lambda t'_1,\lambda}$ is a standard domain for all sufficiently small $\lambda>0$. We next recall from \citeII{Definition~\ref*{DefGlDynStd}} that the standard domain $\wt\Omega\subset\wt M\setminus\wt K^\circ$ associated with $\Omega$ is equal to $[0,1)_\eps\times\Omega$ when $\Omega\cap\cC=\emptyset$, and
\[
  \wt\Omega = \wt\upbeta^{-1}([0,1)_\eps\times\Omega)\setminus\wt K^\circ.
\]
We denote the $\eps$-level sets of such a standard domain by
\[
  \wt\Omega \cap M_\eps =: \Omega_\eps.
\]
We shall only work on standard domains of this latter type, as the analysis on standard domains not intersecting $\cC$ (or $\hat M$) follows from standard (non)linear hyperbolic theory for wave equations (with mild parameter dependence).

For the remainder of this section, we fix the following data.
\begin{itemize}
\item{\rm (Parameters for constraint damping and gauge modification.)} We fix the choices
  \begin{equation}
  \label{EqEcd1formsFix}
    \cd_C,\cd_\Ups\in\CIc(\hat X_b^\circ;T^*_{\hat X_b^\circ}\hat M_b^\circ),\qquad \gamma_C,\gamma_\Ups>0
  \end{equation}
  from Propositions~\ref{PropKCD} and \ref{PropKUps}, and the operator $L$ from~\eqref{EqKEOp}--\eqref{EqKECUps}.
\item{\rm (Background metric.)} We fix a Kerr-mod-$\cO(\hat\rho^2)$ metric
  \begin{equation}
  \label{EqEBgMetric}
    \wt g_0=(g_{0,\eps})_{\eps\in(0,1)}
  \end{equation}
  of regularity $\CI+\cC_\sop^{\infty,1,2}$. This plays the role of a \emph{background metric} for the wave coordinate gauge, and in our eventual application will be equal to a formal solution of the Einstein vacuum equations as constructed in \cite{HintzGlueLocI}.
\end{itemize}

Write $\hat z=(\hat t,\hat x)$. Since $\cd_C=\cd_{C,\mu}\dd\hat z^\mu$ for some $\cd_{C,\mu}\in\CIc(\R^3_{\hat x})$, we have
\begin{equation}
\label{EqEcd1forms}
  \cd_C = \eps^{-1}\wt\cd_C,\qquad \wt\cd_C(\eps,t,\hat x) := \cd_{C,\mu}(\hat x)\dd z^\mu \in \CI(\wt M;\wt T^*\wt M),\qquad z=(t,x),
\end{equation}
and similarly
\[
  \cd_\Ups = \eps^{-1}\wt\cd_\Ups,\qquad \wt\cd_\Ups \in \CI(\wt M;\wt T^*\wt M),
\]
with both $\wt\cd_C$ and $\wt\cd_\Ups$ having supports disjoint from $M_\circ$.

\begin{definition}[Gauge-fixed Einstein operator on a glued spacetime]
\label{DefEOp}
  Define a modified symmetric gradient and modified divergence analogously to~\eqref{EqKEOp} by\footnote{By definition, $\delta_{\wt g_0}^*$ acts on 1-forms on an $\eps$-level set of $\wt M$ by $\delta_{g_{0,\eps}}^*$; similarly for $\wt E_{\rm CD}$. Similar comments apply to all other definitions below.}
  \begin{alignat*}{2}
    \delta_{\wt g_0,\gamma_C}^* &:= \delta_{\wt g_0}^* + \eps^{-1}\gamma_C\wt E_{\rm CD},&\qquad \wt E_{\rm CD}\wt\omega &:= 2\wt\cd_C\otimes_s\wt\omega - \la\wt\cd_C,\wt\omega\ra_{\wt g_0^{-1}}\wt g_0, \\
    \delta_{\wt g_0,\gamma_\Ups} &:= \delta_{\wt g_0} + \eps^{-1}\gamma_\Ups\wt E_\Ups,&\qquad \wt E_\Ups\wt h&:=2\iota_{\wt\cd_\Ups^\sharp}\wt h - \wt\cd_\Ups\tr_{\wt g_0}\wt h,
  \end{alignat*}
  where $\wt\cd_\Ups^\sharp:=\wt g_0^{-1}(\wt\cd_\Ups,\cdot)$. For a nondegenerate section $\wt g=(g_\eps)_{\eps\in(0,1)}$ of $S^2\wt T^*\wt M$, define moreover the gauge 1-form
  \begin{equation}
  \label{EqEOpUps}
    \Ups(\wt g;\wt g_0) := \wt g(\wt g_0)^{-1}\delta_{\wt g}\sfG_{\wt g}\wt g_0 - (\delta_{\wt g_0,\gamma_\Ups}-\delta_{\wt g_0})\sfG_{\wt g_0}(\wt g-\wt g_0).
  \end{equation}
  The gauge-fixed Einstein operator is then
  \begin{equation}
  \label{EqEOpNonlin}
    \wt P(\wt g;\wt g_0) := 2\bigl( \Ric(\wt g) - \Lambda\wt g - \delta_{\wt g_0,\gamma_C}^*\Ups(\wt g;\wt g_0)\bigr).
  \end{equation}
  Its restriction to $M_\eps$ depends only on $g_\eps,g_{0,\eps}$, and is denoted $P_\eps(g_\eps;g_{0,\eps})$.
\end{definition}

\begin{rmk}[Gauge 1-form]
\label{RmkEGauge1form}
  In local coordinates, the first summand of~\eqref{EqEOpUps} is the 1-form
  \[
    \wt g_{\mu\nu}\wt g^{\rho\lambda}\bigl(\Gamma(\wt g)_{\rho\lambda}^\nu-\Gamma(\wt g_0)_{\rho\lambda}^\nu\bigr),
  \]
  defined on an $\eps$-level set by this expression with $\wt g$, $\wt g_0$ replaced by $g_\eps,g_{0,\eps}$. The second summand of~\eqref{EqEOpUps} amounts to a modification of the standard wave coordinate gauge (relative to $\wt g_0$), leading to a generalized wave coordinate gauge; see \cite[Definition~3.27]{HintzMink4Gauge} for an analogous definition.
\end{rmk}

We shall linearize $\wt P(\cdot;\wt g_0)$ around a Kerr-mod-$\cO(\hat\rho^2)$-metric $\wt g=(g_\eps)_{\eps\in(0,1)}$ associated to the same data $(M,g)$, $\cC$, $b=(\bhm,\bha)$ as $\wt g_0$. Thus, we define
\begin{equation}
\label{EqEOp}
  \wt L_{\wt g;\wt g_0}\wt h := D_{\wt g}\wt P(\wt h;\wt g_0) = \frac{\dd}{\dd s}\wt P(\wt g+s\wt h;\wt g_0)\Big|_{s=0}.
\end{equation}
That is, $\wt L_{\wt g;\wt g_0}$ is a family of linear operators $L_{g_\eps,g_{0,\eps}}$, acting on symmetric 2-tensors on the $\eps$-level sets of $\wt M\setminus\wt K^\circ$, defined by linearizing $g\mapsto 2(\Ric(g)-\Lambda g-\delta_{g_{0,\eps},\gamma_C}^*\Ups(g;g_{0,\eps}))$ around $g=g_\eps$. From~\cite[\S3]{GrahamLeeConformalEinstein}, we see that
\begin{equation}
\label{EqEOpExpl}
\begin{split}
  &\wt L_{\wt g;\wt g_0}\wt h = \Box_{\wt g}\wt h - 2\Lambda\wt h - 2(\delta_{\wt g}^*-\delta_{\wt g_0,\gamma_C}^*)\delta_{\wt g}\sfG_{\wt g}\wt h + 2\sR_{\wt g} + 2\delta_{\wt g_0,\gamma_C}^*\bigl(\sE_{\wt g;\wt g_0}\wt h + (\delta_{\wt g_0,\gamma_\Ups}-\delta_{\wt g_0})\sfG_{\wt g_0}\wt h\bigr), \\
  &\qquad (\sR_{\wt g}\wt h)_{\mu\nu} := \Riem(\wt g)_{\kappa\mu\nu\lambda}\wt h^{\kappa\lambda} + \frac12\bigl(\Ric(\wt g)_{\mu\lambda}\wt h_\nu{}^\lambda + \Ric(\wt g)_{\nu\lambda}\wt h_\mu{}^\lambda\bigr), \\
  &\qquad (\sE_{\wt g;\wt g_0}\wt h)_\mu := (\Gamma(\wt g)_{\kappa\nu}^\lambda-\Gamma(\wt g_0)_{\kappa\nu}^\lambda)(\wt g_{\mu\lambda}\wt h^{\kappa\nu}-\wt h_{\mu\lambda}\wt g^{\kappa\nu}),
\end{split}
\end{equation}
where indices are raised and lowered using $\wt g$; and we recall $(\Box_{\wt g}\wt h)_{\mu\nu}=-\wt h_{\mu\nu;\lambda}{}^\lambda$. Our analysis of $\wt L_{\wt g;\wt g_0}$ rests on the following structural properties.

\begin{lemma}[Properties of $\wt L_{\wt g;\wt g_0}$]
\label{LemmaEOp}
  We have
  \[
    \wt L_{\wt g;\wt g_0}\in\hat\rho^{-2}(\CI+\cC_\seop^{\infty,1,2})\Diffse^2(\wt M\setminus\wt K^\circ;S^2\wt T^*\wt M).
  \]
  In fact, we have the following more precise description.
  \begin{enumerate}
  \item\label{ItEOpNorm} Identify $S^2\wt T^*\wt M|_{M_\circ}=\upbeta^*S^2 T^*M$, and also $S^2\wt T^*\wt M|_{\hat M_p}\cong S^2\,\Ttsc_{\hat X_b}\hat M_b$ (for all $p\in\cC$) via
    \begin{equation}
    \label{EqEOpId}
      \dd z^\mu\otimes_s\dd z^\nu\mapsto\dd\hat z^\mu\otimes_s\dd\hat z^\nu.
    \end{equation}
    The normal operators of $\wt L_{\wt g;\wt g_0}$ are then
    \begin{alignat}{2}
    \label{EqEOpNormLcirc}
      L_\circ &:= N_{M_\circ}(\wt L_{\wt g;\wt g_0}) = D_g P(\cdot;g) &&= \Box_g - 2\Lambda + 2\sR_g \\
        &&&\in \upbeta_\circ^*\Diff^2(M;S^2 T^*M), \nonumber\\
    \label{EqEOpNormL}
      L &:= N_{\hat M}(\eps^2\wt L_{\wt g;\wt g_0}) &&= 2(D_{\hat g_b}\Ric+\delta_{\hat g_b,\gamma_C}^*\delta_{\hat g_b,\gamma_\Ups}\sfG_{\hat g_b}) \\
        &&&\in \rho_\cD^2\Diff_{\tbop,\rm I}^2(\hat M_b;S^2\,\Ttsc^*_{\hat X_b}\hat M_b). \nonumber
    \end{alignat}
  \item\label{ItEOpReg} If $\wt g$ has regularity $\CI+\cC_\seop^{d_0+2,1,2}$, $d_0\in\N_0$, then $\wt L_{\wt g;\wt g_0}\in\hat\rho^{-2}(\CI+\cC_\seop^{d_0,1,2})\Diffse^2$.
  \item\label{ItEOpTame} Define an operator $\wt L_{\wt g;\wt g_0,1}$ by
    \[
      \wt L_{\wt g;\wt g_0} = \wt L_{\wt g_0;\wt g_0} + \wt L_{\wt g;\wt g_0,1}.
    \]
    If $\wt g$ has regularity $\CI+\cC_{\seop;\sop}^{(2;k),1,2}$ (so $\wt g-\wt g_0\in\cC_{\seop;\sop}^{(2;k),1,2}(\wt M\setminus\wt K^\circ;S^2\wt T^*\wt M)$) where $k\in\N_0$, then
    \begin{equation}
    \label{EqEOpL1}
      \wt L_{\wt g;\wt g_0,1}\in\hat\rho^{-2}\cC_\sop^{k,1,2}\Diffse^2.
    \end{equation}
    Let $\wt\Omega\subset\wt M\setminus\wt K^\circ$ be a standard domain. Then the maximum $\|\wt L_{\wt g;\wt g_0,1}\|_{\hat\rho^{-2}\cC_{\sop,\eps}^{k,1,2}(\Omega_\eps)}$ of the $\hat\rho^{-2}\cC_{\sop,\eps}^{k,1,2}(\Omega_\eps)$-norms of the coefficients of~\eqref{EqEOpL1} (when expressed in terms of $(\hat\rho\pa_t)^i(\hat\rho\pa_x)^\beta$, $i+|\beta|\leq 2$) obeys the tame estimate
    \begin{equation}
    \label{EqEOpTame}
      \|\wt L_{\wt g;\wt g_0,1}\|_{\hat\rho^{-2}\cC_{\sop,\eps}^{k,1,2}(\Omega_\eps)} \leq C_k\bigl(1+\|(\wt g-\wt g_0)|_{M_\eps}\|_{\cC_{(\seop;\sop),\eps}^{(2;k),1,2}(\Omega_\eps)}\bigr)
    \end{equation}
    for all $0<\eps\leq\eps_0$. The constant $C_k$ (which is allowed to depend on $\wt g_0$) can be taken to be uniform for perturbations of $\wt g-\wt g_0$ which are of size $\leq 1$ in $\cC_\sop^{2,1,2}$.
  \end{enumerate}
\end{lemma}

The notation $L$ for the $\hat M$-normal operator (which is independent of the base point $p\in\cC$ of the fiber of $\hat M$) in~\eqref{EqEOpNormL} is consistent with the notation~\eqref{EqKEOp}. This consistency is the reason for defining the gauge 1-form and the gauge-fixed Einstein operator in the manner~\eqref{EqEOpUps}--\eqref{EqEOpNonlin}.

\begin{proof}[Proof of Lemma~\usref{LemmaEOp}]
  This is a variant of \citeI{Lemma~\ref*{LemmaGLNabla}, Corollary~\ref*{CorGLCurvature}}. In the notation of part~\eqref{ItEOpTame}, we write
  \[
    \wt g = \wt g_0 + \wt h.
  \]
  Since the vector fields $\pa_t,\pa_{x^i}\in\CI(M;T M)$, regarded as $\eps$-independent vector fields, satisfy $\pa_t,\pa_{x^i}\in\hat\rho^{-1}\Vse(\wt M)$, we have, for $\wt h\in|x|^2\CI+\cC_\seop^{d_0+2,1,2}\subset\cC_\seop^{d_0+2,1,2}$,
  \begin{equation}
  \label{EqEOpGammas}
    \Gamma(\wt g)_{\kappa\mu\nu} = \Gamma(\wt g_0)_{\kappa\mu\nu} + \frac12(\pa_\mu\wt h_{\nu\kappa} + \pa_\nu\wt h_{\mu\kappa} - \pa_\kappa\wt h_{\mu\nu}) \in \hat\rho^{-1}(\CI + \cC_\seop^{d_0+1,1,2})(\wt M\setminus\wt K^\circ).
  \end{equation}
  The restriction of $\Gamma(\wt g)_{\kappa\mu\nu}$ to a point $(t,x)$ in $M\setminus\cC=M_\circ\setminus\pa M_\circ$ is equal to $\Gamma(g)_{\kappa\mu\nu}$, and the restriction of $\eps\Gamma(\wt g)_{\kappa\mu\nu}$ to a point $(t,\hat x)$ on $\hat M$ is equal to $\frac12(\pa_{\hat\mu}(\hat g_b)_{\hat\nu\hat\kappa}+\pa_{\hat\nu}(\hat g_b)_{\hat\mu\hat\kappa}-\pa_{\hat\kappa}(\hat g_b)_{\hat\mu\hat\nu})=\Gamma(\hat g_b)_{\hat\kappa\hat\mu\hat\nu}$ since $\eps\pa_\mu=\pa_{\hat\mu}$. We have an analogous statement for $\Gamma(\wt g)_{\mu\nu}^\kappa$; this now involves the coefficients of the inverse metric
  \[
    \wt g^{-1} = (\wt g_0+\wt h)^{-1} = \wt g_0^{-1} - \wt g_0^{-1}\wt h(\wt g_0+\wt h)^{-1}.
  \]
  Since the coefficients of $\wt g$ and $\wt g^{-1}$ are uniformly bounded on $\wt\Omega\cap\{\eps\leq\eps_0\}$ for sufficiently small $\eps_0$ (depending on $\Omega$ and $\|\wt h\|_{\cC_\seop^{0,1,2}}$), we immediately obtain that $\wt g^{-1}-\wt g_0^{-1}$ lies in $\cC_\seop^{0,1,2}$, and by direct differentiation along $\pa_{\hat t},\pa_{\hat x^i}$ in fact in $\cC_\seop^{d_0+2,1,2}$. It is now straightforward, by expressing $\wt L_{\wt g;\wt g_0}$ in $t,x$ coordinates, that for $\wt h\in\cC_{\seop;\sop}^{(d_0+2;k),1,2}$ we have $\wt L_{\wt g;\wt g_0,1}\in\hat\rho^{-2}\cC_{\seop;\sop}^{(d_0;k),1,2}$. Since $\wt L_{\wt g_0;\wt g_0}\in\hat\rho^{-2}(\CI+\cC_\sop^{\infty,1,2})\Diffse^2$, this proves (for $k=0$) part~\eqref{ItEOpReg} and also (for $d_0=0$) the membership~\eqref{EqEOpL1}.

  The identification of the $M_\circ$-normal operator in~\eqref{EqEOpNormLcirc} uses that $\supp\wt\cd_C$ and $\supp\wt\cd_\Ups$ are disjoint from $M_\circ$, so the $M_\circ$-normal operators of $\delta_{\wt g}^*$ and $\delta_{\wt g_0,\gamma_C}^*$ are both equal to $\delta_g^*$ simply; similarly for the other terms of $\wt L_{\wt g;\wt g_0}$. For the proof of~\eqref{EqEOpNormL}, one multiplies~\eqref{EqEOpExpl} by $\eps^2$, uses $\eps\pa_\mu=\pa_{\hat\mu}$ and $\eps\Gamma(\wt g)_{\mu\nu}^\lambda|_{(t,\hat x)}=\Gamma(\hat g_b)_{\hat\mu\hat\nu}^{\hat\lambda}|_{\hat x}$; this implies that the $\hat M$-normal operator of $\eps\delta_{\wt g}^*$ is $\delta_{\hat g_b}^*$ (upon identifying $S^2\wt T^*\wt M|_{\hat M_p}$ with $\Ttsc^*_{\hat X_b}\hat M_b$ by~\eqref{EqEOpId}), likewise for $\eps\delta_{\wt g_0}^*$. Furthermore, for the zeroth order term
  \[
    \eps(\delta_{\wt g_0,\gamma_C}^*-\delta_{\wt g}^*) = \gamma_C\bigl(2\wt\cd_C\otimes_s(\cdot)-\la\wt\cd_C,\cdot\ra_{\wt g_0^{-1}}\wt g_0\bigr) + \eps(\delta_{\wt g_0}^*-\delta_{\wt g}^*),
  \]
  which is an operator acting between sections of the bundles $\wt T^*\wt M$ and $S^2\wt T^*\wt M$, the $\hat M$-normal operator is obtained by replacing $\wt\cd_C=\cd_{C,\mu}\,\dd z^\mu$ by $\cd_{C,\mu}\,\dd\hat z^\mu$ and replacing $\wt g_0,\wt g$ by $\hat g_b$; thus, it equals $\delta_{\hat g_b,\gamma_C}^*-\delta_{\hat g_b}^*$. Dealing with the remaining terms of $\eps^2\wt L_{\wt g;\wt g_0}$ in the same fashion, its $\hat M$-normal operator is thus found to be
  \begin{align*}
    &\Box_{\hat g_b} - 2(\delta_{\hat g_b}^*-\delta_{\hat g_b,\gamma_C}^*)\delta_{\hat g_b}\sfG_{\hat g_b} + 2\sR_{\hat g_b} + \delta_{\hat g_b,\gamma_C}^*(\delta_{\hat g_b,\gamma_\Ups}-\delta_{\hat g_b})\sfG_{\hat g_b} \\
    &\qquad = 2\bigl( D_{\hat g_b}\Ric + \delta_{\hat g_b,\gamma_C}^*\delta_{\hat g_b,\gamma_\Ups}\sfG_{\hat g_b}\bigr),
  \end{align*}
  where we used $\sE_{\hat g_b,\hat g_b}=0$.

  The tame estimate~\eqref{EqEOpTame} is an easy consequence of Lemma~\ref{LemmaNTame} (in particular, using Cramer's formula for the components of $\wt g^{-1}$).
\end{proof}

%%%%%%%%%%%%%%%%%%%%%%%%%%%%%%%%%%%%%%%%%%%%%%%%%%
\subsection{Uniform estimates on se-Sobolev spaces}
\label{SsEse}

By combining the uniform se-regularity estimate \citeII{Theorem~\ref*{ThmEstStd}} with the estimates for the Kerr model operator $L$ established in Lemma~\ref{LemmaKEFwdWSize}, we are now able to prove uniform bounds for forward solutions of $\wt L\wt h=\wt f$ on sufficiently small domains by a minor adaptation of the arguments in~\citeII{\S\ref*{SssScUnifSm}}. As in the reference, we shall work on a fixed domain to which we pull back $\wt L$ using a scaling map.

Concretely, define
\[
  M' := \R_{t'} \times \R^3_{x'},\qquad
  \cC' := \R_{t'} \times \{0\},\qquad
  \wt M' := [[0,1)_\eps\times M';\{0\}\times\cC'],
\]
and $M'_\eps := \{\eps\}\times M'\subset\wt M'$ for $\eps>0$; introduce further $\hat x':=\frac{x'}{\eps}$ and $\wt K':=\{|\hat x'|\leq\bhm\}\subset\wt M'$, and equip $\wt M'\setminus(\wt K')^\circ$ with the family $\wt g_b\in\CI(\wt M'\setminus(\wt K')^\circ;S^2\wt T^*\wt M')$ of Kerr metrics (recalling $b=(\bhm,\bha)$)
\begin{equation}
\label{EqEseKerr}
  (\wt g_b)_{\mu\nu}(\eps,t',x') := (\hat g_{\eps\bhm,\eps\bha})_{\hat\mu\hat\nu}(x') = (\hat g_b)_{\hat\mu\hat\nu}(x'/\eps).
\end{equation}
Thus, $(\wt M',\wt g_b)$ is a smooth Kerr-mod-$\cO(\hat\rho^2)$ glued spacetime associated with $(M',-\dd t'{}^2+\dd x'{}^2)$, $\cC'$, and $b$. Consider the standard domain
\begin{subequations}
\begin{equation}
\label{EqEseDomain}
  \Omega' = \Omega_{0,1,1} = \{ (t',x') \colon 0\leq t'\leq 1,\ |x'|\leq 1+2(1-t') \}
\end{equation}
and the associated standard domain
\begin{equation}
\label{EqEseDomain2}
  \wt\Omega' = \wt\upbeta'{}^{-1}([0,1)\times\Omega') \setminus (\wt K')^\circ.
\end{equation}
\end{subequations}
Fix $t_0\in I_\cC$ (in the notation of Definition~\ref{DefEGlue}). Using the coordinates $\eps>0,t',x'$ on $\wt M'$ and $\eps>0,t,x$ on $\wt M$, we map a neighborhood of $\wt\Omega'\subset\wt M'$ into a neighborhood of $\hat M_{t_0}\subset\wt M$ via
\begin{equation}
\label{EqEseScale}
  \wt S_\lambda \colon (\eps,t',x') \mapsto (\lambda\eps, t_0+\lambda t', \lambda x'),\qquad \lambda>0.
\end{equation}
(In terms of $\hat t'=\frac{t'}{\eps}$, $\hat x'=\frac{x'}{\eps}$ and $\hat t=\frac{t-t_0}{\eps}$, $\hat x=\frac{x}{\eps}$, this map takes $(\eps,\hat t',\hat x')\mapsto(\lambda\eps,\hat t',\hat x')$.) If $M=M'$, so $\wt S_\lambda$ maps $\wt M'\to\wt M'$, then by~\eqref{EqEseKerr} and $\wt S_\lambda^*\dd z'=\lambda\,\dd z'$, $z'=(t',x')$, we have the invariance
\begin{equation}
\label{EqEseGbScale}
  \wt g_b=\lambda^{-2}\wt S_\lambda^*\wt g_b.
\end{equation}
We now study the pullback of a general glued spacetime metric and of the corresponding linearized gauge-fixed Einstein operator.

\begin{lemma}[Rescalings of metrics and operators]
\label{LemmaEseScale}
  Let $\wt g_0$, $\wt g$ denote Kerr-mod-$\cO(\hat\rho^2)$ metrics associated with $(M,g)$, $\cC$, $b$, with $\wt g_0$ of regularity $\CI+\cC_\sop^{\infty,1,2}$ and $\wt g$ of regularity $\CI+\cC_\seop^{d_0+2,1,2}$, $d_0\in\N_0$. With $t_0\in I_\cC$ and $\wt S_\lambda$ given by~\eqref{EqEseScale}, define
  \[
    \wt g_{0,(\lambda)} := \lambda^{-2}\wt S_\lambda^*\wt g_0, \qquad
    \wt g_{(\lambda)} := \lambda^{-2}\wt S_\lambda^*\wt g, \qquad
    \wt L_{(\lambda)} := \lambda^2\wt S_\lambda^*\wt L_{\wt g;\wt g_0}
  \]
  on $\wt\Omega'$; these are well-defined for $0<\lambda\leq\lambda_0$ when $\lambda_0<1$ is sufficiently small.
  \begin{enumerate}
  \item\label{ItEseScaleMetric}{\rm (Metrics.)} We have
    \begin{equation}
    \label{EqEseScaleMetric}
    \begin{split}
      \wt g_{(\lambda)}-\wt g_b &\in \lambda^2 L^\infty\bigl( (0,\lambda_0]; \cC_\seop^{d_0+2,0,2}(\wt\Omega';S^2\wt T^*\wt M') \bigr), \\
      \wt g_{0,(\lambda)}-\wt g_b &\in \lambda^2 L^\infty\bigl( (0,\lambda_0]; \cC_\sop^{\infty,0,2}(\wt\Omega';S^2\wt T^*\wt M') \bigr).
    \end{split}
    \end{equation}
  \item\label{ItEseScaleOp}{\rm (Operators.)} We have $\wt L_{(\lambda)}=\wt L_{\wt g_{(\lambda)},\wt g_{0,(\lambda)}}$, and
    \begin{equation}
    \label{EqEseScaleOp}
      \wt L_{(\lambda)} - \wt L_{\wt g_b;\wt g_b} \in \lambda^2 L^\infty\bigl( (0,\lambda_0]; \cC_\seop^{d_0}\Diffse^2(\wt\Omega';S^2\wt T^*\wt M')\bigr).
    \end{equation}
    Finally,
    \begin{equation}
    \label{EqEseScaleKerrModel}
      \eps^2\wt L_{\wt g_b;\wt g_b}=L
    \end{equation}
    where $L$ is given by~\eqref{EqEOpNormL}; here we regard $L$ as an operator on $\wt M'$ via $\eps$-independent extension in the coordinates $\hat t',\hat x'$, and we use the bundle identification~\eqref{EqEOpId} with $z'=(t',x')$, $\hat z'=(\hat t',\hat x')$ in place of $z$, $\hat z$.
  \end{enumerate}
\end{lemma}
\begin{proof}
  %%%%%%%%%%
  \pfstep{Part~\eqref{ItEseScaleMetric}.} While one can argue using explicit computations as in the proofs of \citeII{Lemmas~\ref*{LemmaScUnifSmMetric} and \ref*{LemmaScUnifSmOp}}, we give a more conceptual proof here instead.

  We identify $\wt g$ with $\wt S_1^*\wt g$ and can thus work entirely on $\wt M'$ and with $t_0=0$. We fix $\hat\rho'=(\eps^2+|x'|^2)^{\frac12}$ as a defining function of the front face $\hat M'$ of $\wt M'$. Since both $\wt g$ and $\wt g_b$ are Kerr-mod-$\cO(\hat\rho^2)$ spacetime metrics, we have
  \[
    \wt g-\wt g_b\in\hat\rho'{}^2\CI+\cC_\seop^{d_0+2,1,2}\subset\cC_\seop^{d_0+2,0,2} = \hat\rho'{}^2\cC_\seop^{d_0+2}(\wt\Omega';S^2\wt T^*\wt M').
  \]
  In view of~\eqref{EqEseGbScale}, part~\eqref{ItEseScaleMetric} now follows from the following observations. First, $\wt S_\lambda^*\hat\rho'{}^2=\lambda^2\hat\rho'{}^2$, which justifies the weight in $\lambda$. Second, to justify the uniform control of the $\cC_\seop$-norms, consider $\wt V:=\hat\rho'\pa_{t'}=(\eps^2+|x'|^2)^{\frac12}\pa_{t'}\in\Vse(\wt M')$ and its pullback: the vector field $(\wt S_\lambda^*\wt V)|_{M'_\eps}$ is the pullback of $\wt V|_{M'_{\lambda\eps}}=(\lambda^2\eps^2+|x'|^2)^{\frac12}\pa_{t'}$ under the map $(t',x')\mapsto(\lambda t',\lambda x')$, which is equal to $(\lambda^2\eps^2+|\lambda x'|^2)^{\frac12}\lambda^{-1}\pa_{t'}=\wt V|_{M'_\eps}$; that is, $\wt S_\lambda^*\wt V=\wt V$, similarly for $\wt V=\hat\rho'\pa_{x'}$.

  The arguments for $\wt g_0$ are similar, the only difference being that the pullback of the s-vector field $\pa_{t'}$ under $\wt S_\lambda^*$ is $\lambda^{-1}\pa_{t'}$. Since this is stronger than $\pa_{t'}$, in the sense that it is a multiple of $\pa_{t'}$ by a scalar which is $<1$, this ensures uniform $L^\infty$-bounds on s-derivatives of $\wt g_{0,(\lambda)}-\wt g_b$, as desired.

  %%%%%%%%%%
  \pfstep{Part~\eqref{ItEseScaleOp}.} The definitions of $\wt\cd_C,\wt\cd_\Ups$ in~\eqref{EqEcd1forms} ensure that
  \begin{equation}
  \label{EqEseScalecd}
    \wt S_\lambda^*\wt\cd_C=\lambda\wt\cd_C,\qquad \wt S_\lambda^*\wt\cd_\Ups=\lambda\wt\cd_\Ups.
  \end{equation}
  Consider first the operator $\wt E$ given on $M'_\eps$ by $\omega\mapsto\eps^{-1}\wt\cd_C|_{M'_\eps}\otimes_s\omega=\eps^{-1}\cd_{C,\mu}(\hat x')\dd z'{}^\mu\otimes_s\omega$. (This is part of $\wt E_{\rm CD}$ in Definition~\ref{DefEOp}). For $\omega=\omega_\nu\,\dd z'{}^\nu$, we then compute using $(\wt S_\lambda)^*\dd z'=\lambda\,\dd z'$ and $(\wt S_\lambda)_*\dd z'=\lambda^{-1}\dd z'$ that
  \begin{align*}
    (\wt S_\lambda^*\wt E)|_{(\eps,t',\hat x')}(\omega) &= \wt S_\lambda^*\bigl(\wt E|_{(\lambda\eps,\lambda t',\hat x')}(\omega_\nu\lambda^{-1}\dd z'{}^\nu)\bigr) \\
      &= \wt S_\lambda^*\bigl((\lambda\eps)^{-1}\cd_{C,\mu}(\hat x')\dd z'{}^\mu\otimes_s \omega_\nu\lambda^{-1}\dd z'{}^\nu\bigr) \\
      & = \eps^{-1}\cd_{C,\mu}(\hat x') \dd z'{}^\mu\otimes_s\omega \\
      &= \wt E|_{(\eps,t',\hat x')}(\omega).
  \end{align*}
  Similarly, one shows that $\wt S_\lambda^*\wt E_\bullet=\wt E_\bullet$ for $\bullet={\rm CD},\Ups$. Writing $g_{0,(\lambda),\eps}:=\wt g_{0,(\lambda)}|_{M'_\eps}$, so $(g_{0,(\lambda),\eps})_{\mu\nu}(t',x')=(g_{0,\lambda\eps})_{\mu\nu}(\lambda t',\lambda x')$, we furthermore have $\Gamma(g_{0,(\lambda),\eps})_{\mu\nu}^\kappa=\lambda\Gamma(g_{0,\lambda\eps})_{\mu\nu}^\kappa$, and therefore
  \begin{align*}
    (\wt S_\lambda^*\delta^*_{\wt g_0})(\omega)|_{(\eps,z')} &= \wt S_\lambda^*\Bigl( \delta_{g_{0,\lambda\eps}}^*\bigl(\omega_\mu(z'/\lambda)\lambda^{-1}\dd z'{}^\mu \bigr)\Bigr) \\
      &= \wt S_\lambda^*\Bigl( \lambda^{-1}\bigl(\lambda^{-1}(\pa_\nu\omega_\mu)(z'/\lambda)-\Gamma(g_{0,\lambda\eps})_{\mu\nu}^\kappa \omega_\kappa(z'/\lambda)\bigr) \dd z'{}^\nu\otimes_s\dd z'{}^\mu \Bigr) \\
      &= \wt S_\lambda^*\Bigl( \bigl( (\pa_\nu\omega_\mu)(z'/\lambda) - \Gamma(g_{0,(\lambda),\eps})_{\mu\nu}^\kappa\omega_\kappa(z'/\lambda) \bigr) \lambda^{-1}\dd z'{}^\nu\otimes_s\lambda^{-1}\dd z'{}^\mu\Bigr) \\
      &= (\delta_{\wt g_{0,(\lambda)}}^*\omega)|_{(\eps,z')}.
  \end{align*}
  We have thus shown that $\wt S_\lambda^*\delta_{\wt g_0,\gamma_C}^*=\delta_{\wt g_{0,(\lambda)},\gamma_C}^*$. Similar arguments yield
  \[
    \wt S_\lambda^*\delta_{\wt g,\gamma_\Ups} = \lambda^{-2}\delta_{\wt g_{(\lambda)},\gamma_\Ups},\qquad
    \wt S_\lambda^*\sfG_{\wt g} = \sfG_{\wt g_{(\lambda)}},
  \]
  and then indeed $\wt S_\lambda^*\wt L_{\wt g;\wt g_0}=\lambda^{-2}\wt L_{\wt g_{(\lambda)};\wt g_{0,(\lambda)}}$, which proves the first statement. The equality~\eqref{EqEseScaleKerrModel} is a consequence of this: in view of the invariance~\eqref{EqEseGbScale}, we have $\lambda^2\wt S_\lambda^*(\eps^2\wt L_{\wt g_b;\wt g_b})=\eps^2\wt L_{\wt g_b;\wt g_b}$; but as $\lambda\to 0$, this converges to the $\hat M'_0$-normal operator $L$.

  Finally, the membership~\eqref{EqEseScaleOp} is a simple consequence of part~\eqref{ItEseScaleMetric} and the invariance~\eqref{EqEseScaleKerrModel}. As an illustration of the relevant computations, we note that $\Gamma(\wt g_{(\lambda)})_{\kappa\mu\nu}-\Gamma(\wt g_b)_{\kappa\mu\nu}\in\lambda^2 L^\infty((0,\lambda_0];\cC_\seop^{d_0+1,0,1}(\wt\Omega'))$: the prefactor $\lambda^2$ arises from~\eqref{EqEseScaleMetric}, while the change in orders from $(d_0+2,0,2)$ to $(d_0+1,0,1)$ arising exactly as in~\eqref{EqEOpGammas}.
\end{proof}

We are now ready to prove our first uniform estimate for solutions of $\wt L\wt h=\wt f$.

\begin{thm}[Uniform se-estimates on small domains]
\label{ThmEse}
  Define $\wt\Omega'$, $\wt S_\lambda$ by~\eqref{EqEseDomain2}, \eqref{EqEseScale}. There exists $d_0\in\N$ so that the following holds for all Kerr-mod-$\cO(\hat\rho^2)$ glued spacetime metrics $\wt g$ of regularity $\CI+\cC_\seop^{d_0+2,1,2}$. Let $\alpha_\circ,\hat\alpha\in\R$ with $\alpha_\cD:=\alpha_\circ-\hat\alpha\in(-2,-\frac32)$. Fix an order function $\sfs\in\CI(\Sse^*\wt M)$ of the form \citeII{(\ref*{EqEstAdmIndEx})} (in particular $\hat\sfs:=\sfs|_{\hat M_t}\in\CI(\Sse^*_{\hat M_p}\wt M)=\CI(\Stb^*_\cX\cM)$ is independent of $p\in\cC$) and so that $\hat\sfs-3,\alpha_\cD$ and $\hat\sfs-5,\alpha_\cD+1$ as well as $\hat\sfs-1,\alpha_\cD+2$ are strongly Kerr-admissible. Then there exists $\lambda_0>0$ so that for all $\lambda\in(0,\lambda_0]$ and with $\wt\Omega_{(\lambda)}=\wt S_\lambda(\wt\Omega')$, the unique forward solution of $L_{g_\eps;g_{0,\eps}}h=f$ on $\Omega_{(\lambda),\eps}=\wt\Omega_{(\lambda)}\cap M_\eps$ satisfies the estimate
  \begin{equation}
  \label{EqEse}
    \|h\|_{H_{\seop,\eps}^{\sfs,\alpha_\circ,\hat\alpha-2}(\Omega_{(\lambda),\eps})^{\bullet,-}} \leq C\|f\|_{H_{\seop,\eps}^{\sfs,\alpha_\circ,\hat\alpha-2}(\Omega_{(\lambda),\eps})^{\bullet,-}}
  \end{equation}
  uniformly for all $\eps<\lambda$; here $C=C(\sfs,\alpha_\circ,\hat\alpha,\lambda)$. One can moreover choose $\lambda_0$ and $C$ so that, for fixed $\lambda$, this holds uniformly for all sufficiently small (depending on $\wt g_0,\lambda$ but not on $\eps$) perturbations of $g_\eps-g_{0,\eps}$ as measured in the $\cC_{\seop,\eps}^{d_0+2,1,2}(\Omega_{(\lambda),\eps})$-norm. (That is, one can replace $g_\eps$ by $g_\eps+h_\eps$ provided $\|h_\eps\|_{\cC_{\seop,\eps}^{d_0+2,1,2}(\Omega_{(\lambda),\eps})}\leq\eta$ for an $\eps$-independent quantity $\eta>0$.)
\end{thm}

We shall ultimately apply this for $\wt g$ which are perturbations of $\wt g_0$ in the stated sense. Let us compare Theorem~\ref{ThmEse} to \citeII{Theorem~\ref*{ThmScUnifSm}}: the weight with which the solution $h$ is estimated at $\hat M$ is $\hat\alpha-2$, compared to $\hat\alpha$ in \citeII{Theorem~\ref*{ThmScUnifSm}}. Since $\wt L_{\wt g;\wt g_0}=(L_{g_\eps;g_{0,\eps}})_{\eps\in(0,1)}$ itself lies in $\hat\rho^{-2}\Diffse^2$ (with se-regular coefficients), this amounts to a loss of two orders of $\hat M$-decay. This will be seen to be due to the loss of two orders of time decay for the Kerr model $L$, cf.\ \eqref{EqKEFwdWSizeMem}, since $\hat t^2=\eps^{-2}(\eps\hat t)^2$, and $|\eps\hat t|\lesssim\lambda<1$ on $\Omega_{(\lambda),\eps}$.

\begin{rmk}[Large relative weight $\alpha_\circ-\hat\alpha$]
\label{RmkEseWeight}
  Mirroring Remark~\ref{RmkKEFwdWeight}, we emphasize that our proof of~\eqref{EqEse} crucially uses that $\alpha_\cD=\alpha_\circ-\hat\alpha$ is less than $-\frac32$. Indeed, only the precise relationship of weights of $f$ and $h$ in Lemma~\ref{LemmaKEFwdWSize} allows one to absorb the difference
  \begin{equation}
  \label{EqEseWeightDiff}
    \wt L_{\wt g;\wt g_0}-\wt L_{\wt g_b;\wt g_b}
  \end{equation}
  (which is of weight $0,0$ as an se-operator) between the operator of interest and the Kerr model family in our uniform estimate, cf.\ \eqref{EqEseEst4}--\eqref{EqEseEst5} below. In our application, with $\wt g$ a $\cO(\eps^N)$-perturbation of $\wt g_0$ and $\wt g_0$ itself being the formal solution from \citeI{Theorem~\ref*{ThmM}}, the leading order term of the operator~\eqref{EqEseWeightDiff} at $\hat M$ arises from the $\eps^2$ correction term denoted $h_{(2)}$ in \citeII{Proposition~\ref*{PropFhM2} and Remark~\ref*{RmkFhM2LinKerr}}.
\end{rmk}

\begin{proof}[Proof of Theorem~\usref{ThmEse}]
  It suffices to work with $d_0=\infty$, since the constant in~\eqref{EqEse} only depends on finitely many derivatives of the coefficients of $L_{g_\eps;g_{0,\eps}}$ and thus automatically continues to hold for some large but finite value of $d_0$.

  In the notation of Lemma~\ref{LemmaEseScale}, we shall study $\wt L_{(\lambda)}=(L_{(\lambda),\eps})_{\eps\in(0,1)}$ on the (fixed) domain $\wt\Omega'$; we write $\Omega'_\eps=\wt\Omega'\cap M'_\eps$. Since $\sfs-1$ and $\alpha_\circ-(\hat\alpha-2)$ are strongly Kerr-admissible, \citeII{Theorem~\ref*{ThmEstStd}} gives the uniform (in $\lambda$ and $\eps$) estimate
  \begin{equation}
  \label{EqEseSymbEst}
    \|h\|_{H_{\seop,\eps}^{\sfs-1,\alpha_\circ,\hat\alpha-2}(\Omega'_\eps)^{\bullet,-}} \leq C\Bigl(\|L_{(\lambda),\eps}h\|_{H_{\seop,\eps}^{\sfs-1,\alpha_\circ,\hat\alpha-4}(\Omega'_\eps)^{\bullet,-}} + \|h\|_{H_{\seop,\eps}^{\sfs_0-1,\alpha_\circ,\hat\alpha-2}(\Omega'_\eps)^{\bullet,-}}\Bigr);
  \end{equation}
  here we fix an order $\sfs_0$ of the form \citeII{(\ref*{EqEstAdmIndEx})} so that $\sfs_0+\eta<\sfs-3$ for some (small) $\eta>0$ and so that $\sfs_0,\alpha_\cD$ and $\sfs_0-2,\alpha_\cD+1$ are strongly Kerr-admissible.

  Set now
  \[
    \hat\Omega_\eps := \bigl\{ (\hat t,\hat x)\in\hat M_b^\circ \colon 0\leq\hat t\leq\eps^{-1},\ \bhm\leq|\hat x|\leq\eps^{-1}-2(\eps^{-1}-\hat t) \bigr\}.
  \]
  Using Proposition~\ref{PropNse3b} (for the adaptation for extendible/supported spaces, see \citeII{Proposition~\ref*{PropEstFnSobRel}}), we estimate the final, error, term in~\eqref{EqEseSymbEst} using Lemma~\ref{LemmaKEFwdWSize} (with $\hat r_0=\eps^{-1}$, $\mu=\eps^{-1}$, and $\nu=2$, and recalling that $\ubar\mu(\hat r_0)=\ubar\mu(\eps^{-1})$ decreases with $\eps$) by
  \begin{align*}
    C_1\eps^{-\hat\alpha+4}\|(\Psi_\eps)_*h\|_{\Htb^{\hat\sfs_0-1+\eta,\alpha_\cD+2,0}(\hat\Omega_\eps)^{\bullet,-}} &\leq C_2\eps^{-\hat\alpha+4} \eps^{-2} \|L((\Psi_\eps)_*h)\|_{\Htb^{\hat\sfs_0+\eta,\alpha_\cD+2,0}(\hat\Omega_\eps)^{\bullet,-}} \\
      &\leq C_3\bigl\|\Psi_\eps^*\bigl(\eps^{-2}L((\Psi_\eps)_*h)\bigr)\bigr\|_{H_{\seop,\eps}^{\sfs_0+\eta',\alpha_\circ,\hat\alpha-2}(\Omega'_\eps)^{\bullet,-}}
  \end{align*}
  where $\eta'>\eta$ is chosen so that $\sfs_0+\eta'<\sfs-2$ still. We now replace the operator $\Psi_\eps^*(\eps^{-2}L(\Psi_\eps)_*)=\wt L_{\wt g_b;\wt g_b}|_{M'_\eps}=:L_{\wt g_b,\eps}$ by $L_{(\lambda),\eps}$; in view of~\eqref{EqEseScaleOp}, we finally obtain the estimate
  \begin{equation}
  \label{EqEseEst4}
    \|h\|_{H_{\seop,\eps}^{\sfs_0-1,\alpha_\circ,\hat\alpha-2}(\Omega'_\eps)^{\bullet,-}} \leq C_3\Bigl( \|L_{(\lambda),\eps}h\|_{H_{\seop,\eps}^{\sfs_0+\eta',\alpha_\circ,\hat\alpha-2}(\Omega'_\eps)^{\bullet,-}} + C_4\lambda^2\|h\|_{H_{\seop,\eps}^{\sfs_0+2+\eta',\alpha_\circ,\hat\alpha-2}(\Omega'_\eps)^{\bullet,-}}\Bigr)
  \end{equation}
  for the error term in~\eqref{EqEseSymbEst}. But when $\lambda\leq\lambda_0$, with $\lambda_0>0$ small enough, we can absorb
  \begin{equation}
  \label{EqEseEst5}
  \begin{split}
    C C_3 C_4\lambda^2\|h\|_{H_{\seop,\eps}^{\sfs_0+2+\eta',\alpha_\circ,\hat\alpha-2}(\Omega'_\eps)^{\bullet,-}} &\leq C C_3 C_4 C_5\lambda^2\|h\|_{H_{\seop,\eps}^{\sfs-1,\alpha_\circ,\hat\alpha-2}(\Omega'_\eps)^{\bullet,-}} \\
      &\leq \frac12\|h\|_{H_{\seop,\eps}^{\sfs-1,\alpha_\circ,\hat\alpha-2}(\Omega'_\eps)^{\bullet,-}}
  \end{split}
  \end{equation}
  into the left hand side of~\eqref{EqEseSymbEst}. (If $L_{(\lambda),\eps}$ is modified by a term of size $\leq\eta$ in $\cC_{\seop,\eps}^{d_0,1,2}\hat\rho^{-2}\Diffse^2$, one obtains a further error term $C\eta\|h\|_{H_{\seop,\eps}^{\sfs_0+2+\eta',\alpha_\circ,\hat\alpha-2}}(\Omega'_\eps)^{\bullet,-}$ in~\eqref{EqEseEst4} which can similarly be absorbed when $\eta>0$ is small enough.) Since
  \[
    \|L_{(\lambda),\eps}h\|_{H_{\seop,\eps}^{\sfs-1,\alpha_\circ,\hat\alpha-4}(\Omega'_\eps)^{\bullet,-}} + \|L_{(\lambda),\eps}h\|_{H_{\seop,\eps}^{\sfs_0+\eta',\alpha_\circ,\hat\alpha-2}(\Omega'_\eps)^{\bullet,-}} \leq C\|L_{(\lambda),\eps}h\|_{H_{\seop,\eps}^{\sfs-1,\alpha_\circ,\hat\alpha-2}(\Omega'_\eps)^{\bullet,-}},
  \]
  this proves the uniform bound
  \[
    \|h\|_{H_{\seop,\eps}^{\sfs-1,\alpha_\circ,\hat\alpha-2}(\Omega'_\eps)^{\bullet,-}} \leq C\|L_{(\lambda),\eps}h\|_{H_{\seop,\eps}^{\sfs-1,\alpha_\circ,\hat\alpha-2}(\Omega'_\eps)^{\bullet,-}}.
  \]
  For any fixed $\lambda>0$, this is equivalent to the estimate
  \[
    \|h\|_{H_{\seop,\eps}^{\sfs-1,\alpha_\circ,\hat\alpha-2}(\Omega_{(\lambda),\eps})^{\bullet,-}} \leq C\|f\|_{H_{\seop,\eps}^{\sfs,\alpha_\circ,\hat\alpha-2}(\Omega_{(\lambda),\eps})^{\bullet,-}},\qquad f=L_{g_\eps;g_{0,\eps}}h.
  \]

  Finally, we can recover\footnote{This is only a cosmetic improvement. For our application, any fixed finite loss of se-regularity of $h$ compared to $f$ can be dealt with with purely notational modifications.} one degree of se-regularity by appealing to \citeII{Theorem~\ref*{ThmEstStd}} (with $\sfs_0=\sfs-1$ and $\hat\alpha-2$ in place of $\sfs$ and $\hat\alpha$): this gives
  \[
    \|h\|_{H_{\seop,\eps}^{\sfs,\alpha_\circ,\hat\alpha-2}(\Omega_{(\lambda),\eps})^{\bullet,-}} \leq C\Bigl(\|f\|_{H_{\seop,\eps}^{\sfs,\alpha_\circ,\hat\alpha-2}(\Omega_{(\lambda),\eps})^{\bullet,-}} + \|h\|_{H_{\seop,\eps}^{\sfs-1,\alpha_\circ,\hat\alpha-2}(\Omega_{(\lambda),\eps})^{\bullet,-}}\Bigr)
  \]
  and thus finishes the proof of the desired estimate~\eqref{EqEsEst}.
\end{proof}

%%%%%%%%%%%%%%%%%%%%%%%%%%%%%%%%%%%%%%%%%%%%%%%%%%
\subsection{Higher s-regularity; tame estimates}
\label{SsEs}

By a simple adaptation of the arguments in~\citeII{\S\ref*{SsScS}}, we now prove a generalization of Theorem~\ref{ThmEse} which gives uniform control also on s-derivatives of $h$.

\begin{prop}[Higher s-regularity]
\label{PropEs}
  We use the notation and assumptions of Theorem~\usref{ThmEse}, except we now assume that $\wt g$ has regularity $\CI+\cC_{\seop;\sop}^{(d_0+2;k),1,2}$ for some $k\in\N_0$; and we require $\hat\sfs-5,\alpha_\cD$ and $\hat\sfs-7,\alpha_\cD+1$ as well as $\hat\sfs-3,\alpha_\cD+2$ are strongly Kerr-admissible. Then there exists $\lambda_0>0$ so that for all $\lambda\in(0,\lambda_0]$, the unique forward solution of $L_{g_\eps;g_{0,\eps}}h=f$ satisfies the estimate
  \begin{equation}
  \label{EqEsEst}
    \|h\|_{H_{(\seop;\sop),\eps}^{(\sfs;k),\alpha_\circ,\hat\alpha-2}(\Omega_{(\lambda),\eps})^{\bullet,-}} \leq C\|f\|_{H_{(\seop;\sop),\eps}^{(\sfs;k),\alpha_\circ,\hat\alpha-2}(\Omega_{(\lambda),\eps})^{\bullet,-}}
  \end{equation}
  for all $\eps<\lambda$; here $C=C(\sfs,k,\alpha_\circ,\hat\alpha,\lambda)$. The constants $\lambda_0,C$ can be chosen uniformly for all $\wt g$ which are $\cC_\seop^{d_0+2,1,2}(\Omega_{(\lambda),\eps})$-perturbations of $\wt g_0$ (in the sense described in Theorem~\usref{ThmEse}).
\end{prop}
\begin{proof}
  The proof is very similar to that of \citeII{Theorem~\ref*{ThmScS}}; we carry it out for completeness and to show where the Kerr-mod-$\cO(\hat\rho^2)$-nature of $\wt g$ is used (see~\eqref{EqEsEstKerrMod2}).

  We take $\lambda_0$ to be half the value produced by Theorem~\ref{ThmEse}. For $k=0$, the estimate~\eqref{EqEsEst} is the content of Theorem~\ref{ThmEse}. Suppose we have established~\eqref{EqEsEst} for $k-1$ in place of $k$. Let $\Omega^\sharp:=\Omega_{0,2,2}$, write $\wt\Omega^\sharp\subset\wt M'\setminus(\wt K')^\circ$ for the associated standard domain in $\wt M'$, and let $\wt\Omega^\sharp_{(\lambda)}=\wt S_\lambda(\wt\Omega^\sharp)$; set $\Omega^\sharp_{(\lambda),\eps}=\wt\Omega^\sharp_{(\lambda)}\cap M'_\eps$. (Thus $\wt\Omega^\sharp_{(\lambda)}\subset\wt\Omega_{(2\lambda)}$, hence the halving of $\lambda_0$.) Extend $f\in H_{(\seop;\sop),\eps}^{(\sfs;k),\alpha_\circ,\hat\alpha-2}(\Omega_{(\lambda),\eps})^{\bullet,-}$ to $f^\sharp\in H_{(\seop;\sop),\eps}^{(\sfs;k),\alpha_\circ,\hat\alpha-2}(\Omega^\sharp_{(\lambda),\eps})^{\bullet,-}$ (with norm bounded by that of $f$). The inductive hypothesis gives for the forward solution $h^\sharp$ of
  \begin{equation}
  \label{EqEsEstPf}
    L_{g_\eps;g_{0,\eps}}h^\sharp=f^\sharp
  \end{equation}
  uniform (in $\eps$) control of
  \begin{equation}
  \label{EqEsEsthsharp}
    h^\sharp \in H_{(\seop;\sop),\eps}^{(\sfs;k-1),\alpha_\circ,\hat\alpha-2}(\Omega^\sharp_{(\lambda),\eps})^{\bullet,-}
  \end{equation}
  by the norm of $f^\sharp$.

  We now differentiate~\eqref{EqEsEstPf} along $\pa_t^k$, obtaining
  \begin{equation}
  \label{EqEsEstDiff}
    L_{g_\eps;g_{0,\eps}}(\pa_t^k h^\sharp) = \pa_t^k f^\sharp + [L_{g_\eps;g_{0,\eps}},\pa_t^k]h^\sharp.
  \end{equation}
  The first term on the right lies in $H_{\seop,\eps}^{\sfs,\alpha_\circ,\hat\alpha-2}(\Omega^\sharp_{(\lambda),\eps})^{\bullet,-}$. For the second term, we use that
  \begin{equation}
  \label{EqEsEstKerrMod2}
    \wt L_{\wt g;\wt g_0}-\wt L_{\wt g_b;\wt g_b}\in\cC_{\seop;\sop}^{(d_0;k)}\Diffse^2
  \end{equation}
  (which is a simple variant of~\eqref{EqEseScaleOp} for $\lambda=1$) and also note that $[\wt L_{\wt g_b;\wt g_b},\pa_t]=0$ since Kerr and the 1-forms $\wt\cd_C,\wt\cd_\Ups$ are stationary. Therefore, $[\wt L_{\wt g;\wt g_0},\pa_t^k]\in\sum_{j=0}^{k-1}\hat\rho^{-2}\cC_{\seop;\sop}^{(d_0;j)}\Diffse^2\circ\pa_t^j$ maps $h^\sharp$ from~\eqref{EqEsEsthsharp} into $H_{\seop,\eps}^{\sfs-2,\alpha_\circ,\hat\alpha-2}(\Omega^\sharp_{(\lambda),\eps})^{\bullet,-}$ (with uniform bounds). Applying the case $k=0$ to the equation~\eqref{EqEsEstDiff} (with $\sfs$ reduced by $2$) gives the uniform estimate
  \begin{equation}
  \label{EqEsEstUnif}
  \begin{split}
    &\|h^\sharp\|_{H_{(\seop;\sop),\eps}^{(\sfs-2;k),\alpha_\circ,\hat\alpha-2}(\Omega^\sharp_{(\lambda),\eps})^{\bullet,-}} \\
    &\qquad \leq C\Bigl( \|h^\sharp\|_{H_{(\seop;\sop),\eps}^{(\sfs-1;k-1),\alpha_\circ,\hat\alpha-2}(\Omega^\sharp_{(\lambda),\eps})^{\bullet,-}} + \|\pa_t^k h^\sharp\|_{H_{\seop,\eps}^{\sfs-2,\alpha_\circ,\hat\alpha-2}(\Omega^\sharp_{(\lambda),\eps})^{\bullet,-}} \Bigr) \\
    &\qquad \leq C\|f^\sharp\|_{H_{(\seop;\sop),\eps}^{(\sfs;k),\alpha_\circ,\hat\alpha-2}(\Omega^\sharp_{(\lambda),\eps})^{\bullet,-}}.
  \end{split}
  \end{equation}
  (The first summand in the second line is controlled by~\eqref{EqEsEsthsharp}.)

  Finally, the microlocal elliptic and propagation estimates on (se;s)-Sobolev spaces established in \cite{HintzGlueLocII} imply, as at the end of the proof of \citeII{Theorem~\ref*{ThmScS}}, the estimate
  \[
    \|h\|_{H_{(\seop;\sop),\eps}^{(\sfs;k),\alpha_\circ,\hat\alpha-2}(\Omega_{(\lambda),\eps})} \leq C\Bigl( \|f^\sharp\|_{H_{(\seop;\sop),\eps}^{(\sfs;k),\alpha_\circ,\hat\alpha-2}(\Omega_{(\lambda),\eps})^{\bullet,-}} + \|h^\sharp\|_{H_{(\seop;\sop),\eps}^{(\sfs-2;k),\alpha_\circ,\hat\alpha-2}(\Omega^\sharp_{(\lambda),\eps})^{\bullet,-}}\Bigr).
  \]
  Together with~\eqref{EqEsEstUnif} and the norm bound of $f^\sharp$ by $f$, this concludes the inductive step and thus finishes the proof of~\eqref{EqEsEst}.
\end{proof}

We next upgrade this to estimates which are tame in the s-regularity order. We recall that we fixed $(M,g),\cC,b$, the choices of~\eqref{EqEcd1formsFix} as well as a background metric $\wt g_0$ in~\eqref{EqEBgMetric}.

\begin{prop}[Tame estimates]
\label{PropEsTame}
  There exists an integer $d\in\N$ so that the following holds. Fix the following data.
  \begin{itemize}
  \item A time $t_0\in I_\cC$.
  \item For $\lambda>0$, define $\wt\Omega_{(\lambda)}\subset\wt M\setminus\wt K^\circ$ to be the standard domain associated with $\Omega_{(\lambda)}:=\Omega_{t_0,t_0+\lambda,\lambda}$; let $\lambda_1>0$ be such that $[t_0,t_0+\lambda_1]\subset I_\cC$.
  \item Let $\alpha_\circ,\hat\alpha\in\R$ with $\alpha_\cD:=\alpha_\circ-\hat\alpha\in(-2,-\frac32)$. Fix an order function $\sfs\in\CI(\Sse^*\wt M)$ of the form \citeII{(\ref*{EqEstAdmIndEx})} (in particular $\hat\sfs:=\sfs|_{\hat M_t}\in\CI(\Sse^*_{\hat M_p}\wt M)=\CI(\Stb^*_\cX\cM)$ is independent of $p\in\cC$) and so that $\hat\sfs-5,\alpha_\cD$ and $\hat\sfs-7,\alpha_\cD+1$ as well as $\hat\sfs-3,\alpha_\cD+2$ are strongly Kerr-admissible. 
  \end{itemize}
  Then there exists $\lambda_0\in(0,\lambda_1]$ depending only on $t_0$ and $\wt g,\wt g_0$ on $\Omega_{(\lambda_1)}$ so that for all $\lambda\in(0,\lambda_0]$ and all $\N_0\ni k\geq d$, there exists a constant $C=C(\lambda,k)$ so that for forward solutions $h$ of $L_{g_\eps;g_{0,\eps}}h=f$, the tame estimate
  \begin{equation}
  \label{EqEsTame}
  \begin{split}
    &\|h\|_{H_{\sop,\eps}^{k,\alpha_\circ,\hat\alpha-2}(\Omega_{(\lambda),\eps})^{\bullet,-}} \\
    &\qquad \leq C_k\Bigl(\|f\|_{H_{\sop,\eps}^{k+d,\alpha_\circ,\hat\alpha-2}(\Omega_{(\lambda),\eps})^{\bullet,-}} + \|g_\eps-g_{0,\eps}\|_{\cC_{\sop,\eps}^{k+d,1,2}(\Omega_{(\lambda),\eps})} \|f\|_{H_{\sop,\eps}^{d,\alpha_\circ,\hat\alpha-2}(\Omega_{(\lambda),\eps})^{\bullet,-}} \Bigr)
  \end{split}
  \end{equation}
  holds uniformly for $\eps\in(0,\lambda)$. The constants $\lambda_0$ and $C=C(\lambda,k)$ can be chosen to be uniform for sufficiently small (independently of $k$ and $\eps$) perturbations of $\wt g-\wt g_0$ as measured in the $\cC_{\sop,\eps}^{d,1,2}(\Omega_{(\lambda),\eps})$-norm (in an analogous sense to the statement of Theorem~\usref{ThmEse}).
\end{prop}
\begin{proof}
  The proof of this proceeds in exactly the same fashion as the proof of \citeII{Theorem~\ref*{ThmNTame}}. In the estimate~\eqref{EqEsTame}, the difference of metrics $\wt g-\wt g_0$ appears instead of the $\hat\rho^{-2}\cC_{\sop,\eps}^{k+d,1,2}(\Omega_{(\lambda),\eps})$-norm of the coefficients of the operator $\wt L_{\wt g;\wt g_0,1}=\wt L_{\wt g;\wt g_0}-\wt L_{\wt g_0;\wt g_0}$ since the latter norm obeys a tame estimate in terms of the former norm, as demonstrated in~\eqref{EqEOpTame}.

  In short, for $k$ below any fixed number, the estimate~\eqref{EqEsTame} is a consequence of Proposition~\ref{PropEs}. For large $k$, one applies an extension operator to the coefficients of $\wt L_{\wt g;\wt g_0,1}|_{\wt\Omega_{(\lambda)}}$ to a slightly larger domain (so that the coefficients on the enlarged domain are uniformly controlled by their restriction to $\wt\Omega_{(\lambda)}$ in all $\cC_\sop^{k,1,2}$-spaces) and controls solutions of the thus extended operator. The inductive step in $k$ (for an estimate that is tame in the s-regularity order while still being precise in the se-order, see~\citeII{(\ref*{EqNTameInd})}) is performed by considering the equation satisfied by $\pa_t^{k'}h$, $k'\leq k$ on an enlarged domain (which shrinks down to the original domain $\Omega_{(\lambda),\eps}$ when $k'$ reaches $k$). The fact that the commutator of $\pa_t$ with $\wt L_{\wt g;\wt g_0}$ (and thus also higher order commutators) produce se-operators with weights $0,0$ at $M_\circ,\hat M$ (as opposed to $0,-2$ for $\wt L_{\wt g;\wt g_0}$ itself) allows one to close the iterative step much as in the proof of Proposition~\ref{PropEs}. Similarly, the 2 se-derivatives lost upon controlling one more s-derivative are recovered by appealing to microlocal (se;s)-regularity results, namely those which are tame in the s-regularity order; these are proved in~\citeII{\S\ref*{SssNTameMl}}.
\end{proof}

\begin{rmk}[Other domains]
\label{RmkEsTameOther}
  Mirroring~\citeII{Remark~\ref*{RmkNTameOther}}, we note that Proposition~\ref{PropEsTame} remains valid for the standard domains associated with $\Omega_{t_0+c_1\lambda,t_0+c_2\lambda,r_0\lambda}$ for any $c_1<c_2$ and $r_0>0$, with $\lambda_0>0$ depending on $c_1,c_2,r_0$. One way of seeing this for $r_0=1$ is that the passage to small domains can be accomplished with $\wt S_\lambda\colon(\eps,t',x')\mapsto\bigl(\lambda\eps,t_0+(c_1+(c_2-c_1)t')\lambda,\lambda(c_2-c_1)x'\bigr)$: the key Lemma~\ref{LemmaEseScale} remains valid, and thus so do the subsequent estimates.
\end{rmk}

%%%%%%%%%%%%%%%%%%%%%%%%%%%%%%%%%%%%%%%%%%%%%%%%%%%%%%%%%%%%%%%%%%%%%%
\section{Correction of formal solutions to true solutions}
\label{STr}

In this section, we will prove our main result, Theorem~\ref{ThmIPrec}, by correcting a formal solution of the gluing problem to a true solution. We recall:

\begin{thm}[Formal solution of the black hole gluing problem]
\label{ThmTrGlueLocI}
  (See \citeI{Theorem~\ref*{ThmM} and \S\ref*{SX}}.) Let $\Lambda\in\R$. Suppose we are given:
  \begin{itemize}
  \item a globally hyperbolic spacetime $(M,g)$ satisfying $\Ric(g)-\Lambda g=0$,
  \item an inextendible timelike geodesic $\cC\subset M$. Fix Fermi normal coordinates $(t,x)$ around $\cC$;
  \item subextremal Kerr parameters $\bhm>0$ and $\bha\in\R^3$, $|\bha|<\bhm$;
  \item a Cauchy hypersurface $X\subset M$ which is orthogonal to $\cC$ at the unique point of intersection.
  \end{itemize}
  Assume that
  \begin{enumerate}
  \myitem{ItTrGlueLocIKIDs}{I} $(M,g)$ does not have any nontrivial Killing vector fields in the domain of dependence of a precompact connected open neighborhood $\cU\subset X$ of the point $\cC\cap X$; \emph{or}
  \myitem{ItTrGlueLocIKdS}{II} $(M,g)$ is isometric to a neighborhood of the domain of outer communications of a extremal Kerr or Kerr--(anti) de~Sitter black hole, in which case we take $\cU\subset X$ to be a precompact connected open set containing $\cC\cap X$ and a point in the black hole interior; \emph{or}
  \myitem{ItTrGlueLocINoncompact}{III} the Cauchy hypersurface of $M$ is noncompact.
  \end{enumerate}
  Then there exists a Kerr-mod-$\cO(\hat\rho^2)$ glued spacetime $(\wt M,\wt g_0)$ (see Definition~\usref{DefEGlue}) with these data so that the following properties hold.
  \begin{enumerate}
  \item $\wt g_0$ is a formal solution of
    \[
      \Ric(\wt g_0)-\Lambda\wt g_0 = 0,
    \]
    in the sense that as a section of $S^2\wt T^*\wt M$ over $\wt M\setminus\wt K^\circ$ which is defined in a neighborhood of $(\hat M\setminus\wt K)\cup M_\circ$, the tensor $\Ric(\wt g_0)-\Lambda\wt g_0$ is smooth and vanishes to infinite order at $\hat M$, $M_\circ$. Moreover, in the settings~\eqref{ItTrGlueLocIKIDs} and \eqref{ItTrGlueLocIKdS}, it vanishes to infinite order at the lift $\wt X$ of $[0,1)\times X$ to $\wt M$, whereas in the setting~\eqref{ItTrGlueLocINoncompact} it vanishes to infinite order at the lift of $[0,1)\times V$ where $V\subset X$ is any fixed precompact subset.
  \item $\wt g_0$ is polyhomogeneous: $\wt g_0=\wt\upbeta^*g+\wt g_{(1)}$ where $\wt g_{(1)}\in\cA_\phg^{\cE,\N_0\cup\hat\cE}(\wt M\setminus\wt K^\circ;S^2\wt T^*\wt M)$ for some index sets\footnote{In fact, $\cE\subset(1,0)_+\cup((3+\N_0)\times\N_0)$ where $(1,0)_+=\{(1+j,l)\colon j,l\in\N_0,\ l\leq j\}\subset\C\times\N_0$.} $\cE\subset(1+\N_0)\times\N_0$ and $\hat\cE\subset(3+\N_0)\times\N_0$.
  \item\label{ItTrGlueLocIDom} In the settings~\eqref{ItTrGlueLocIKIDs} and \eqref{ItTrGlueLocIKdS}, $\wt g_0$ is equal to $g$ outside the domain of influence of $\bar\cU$.
  \end{enumerate}
\end{thm}

We shall show the existence of $\wt h\in\eps^\infty\CI(\wt M\setminus\wt K^\circ;S^2\wt T^*\wt M)$ so that $\wt g_0+\wt h$ solves the Einstein equations over any fixed compact subset of $M$; see Theorem~\ref{ThmTrSemi} for the precise result. The strategy is to work with the gauge-fixed version of the Einstein equations from Definition~\ref{DefEOp} and use a Nash--Moser iteration scheme, based on the tame estimates from Proposition~\ref{PropEsTame}, uniformly in $\eps$ for small $\eps>0$ to construct $\wt h=(h_\eps)$. To do this, we first produce a more precise description of the forward solution of the linearized problem from Proposition~\ref{PropEsTame} in~\S\ref{SsTrAsy} (see Theorem~\ref{ThmTrAsyTame}). The resulting linear solution operators have operator norms which blow up as the quantity $\delta>0$ controlling the enlargement of supports tends to $0$ (cf.\ already~\eqref{EqKHiNorm}); to close the nonlinear iteration, we thus need to develop a custom-made Nash--Moser scheme (Theorem~\ref{ThmNM}), which is a variant of that in \cite{SaintRaymondNashMoser}.

\bigskip

\emph{For the remainder of this section, we fix $\wt g_0$ from Theorem~\ref{ThmTrGlueLocI}, the 1-forms and quantities $\cd_C,\cd_\Ups,\gamma_C,\gamma_\Ups$ and the gauge-fixed Einstein operator $\wt g\mapsto\wt P(\wt g;\wt g_0)$ as in~\eqref{EqEcd1formsFix}--\eqref{EqEBgMetric} and Definition~\usref{DefEOp}.}

%%%%%%%%%%%%%%%%%%%%%%%%%%%%%%%%%%%%%%%%%%%%%%%%%%
\subsection{Extracting asymptotics at \texorpdfstring{$\hat M$}{the front face}}
\label{SsTrAsy}

The control on $h$ in Proposition~\ref{PropEsTame} is not yet sufficient for a nonlinear application. The reason is that the forward solution operator loses $2$ orders of decay at $\hat M$: while $\wt L_{\wt g;\wt g_0}^{-1}$ maps $H_{\sop,\eps}^{*,\alpha_\circ,\hat\alpha-2}$ into $H_{\sop,\eps}^{*,\alpha_\circ,\hat\alpha-2}$, the forward map $\wt L_{\wt g;\wt g_0}$ maps $H_{\sop,\eps}^{*,\alpha_\circ,\hat\alpha-2}$ into $H_{\sop,\eps}^{*,\alpha_\circ,\hat\alpha-4}$ (similarly for the nonlinear gauge-fixed Einstein operator when $\alpha_\circ,\hat\alpha$ are sufficiently large). This is analogous to how the description~\eqref{EqKEFwdWSizeMem} of the forward solution of the Kerr model problem $L h=f$ loses two powers of decay at $\cT\subset\hat M_b$ (with $L^{-1}$ in Lemma~\ref{LemmaKEFwdWSize} mapping $\dot H_\tbop^{*,\alpha_\cD+2,0}\to\dot H_\tbop^{*,\alpha_\cD,-2}$, which then $L$ only maps into $\dot H_\tbop^{*,\alpha_\cD+2,-2}$). The remedy is to use the more precise description of the forward solution operator $L^{-1}$ given by~\eqref{EqKEFwdWSol} in Theorem~\ref{ThmKEFwdW}, as by Theorem~\ref{ThmKEFwdW}\eqref{ItKEFwdWFwd} this description is precise enough to get mapped back into the original space $\dot H_\tbop^{*,\alpha_\cD+2,0}$ under $L$. More precisely, for compatibility with s-regularity on glued spacetimes, we must use Theorem~\ref{ThmKHi}. We proceed to implement this on the level of the linearized gauge-fixed Einstein equations on a glued spacetime and for s-Sobolev spaces: we shall precisely capture the nature of those pieces of the forward solution whose $\hat M$-order is weaker than $\hat\alpha$; these pieces are essentially (modulated) (large) zero energy states, i.e.\ $t$-dependent elements of the space $\cK$ in~\eqref{EqKE0Ker}.

\begin{rmk}[Regularity losses]
\label{RmkTrAsy}
  At this stage, we can afford to be rather imprecise in our bookkeeping of differentiability orders: we shall only keep track of s-regularity, and we will freely give up fixed finite amounts of it: the Nash--Moser theorem will handle any such losses.
\end{rmk}

\begin{definition}[Function space for forward solutions]
\label{DefTrAsy}
  Let $\Omega_{t_0,t_1,r_0}\subset M$ denote a standard domain, let $\wt\Omega_{t_0,t_1,r_0}\subset\wt M\setminus\wt K^\circ$ denote the associated standard domain, and write $\Omega_{t_0,t_1,r_0,\eps}:=\wt\Omega_{t_0,t_1,r_0}\cap M_\eps$. Let $k\in\N_0$, $\alpha_\circ,\hat\alpha\in\R$, with $\alpha_\circ-\hat\alpha\in(-2,-\frac32)$. We then set\footnote{This is $H^k(\Omega_{t_0,t_1,r_0,\eps};S^2 T^*M)^{\bullet,-}\oplus H^k([t_0,t_1];\C^4)^{\bullet,-}\oplus H^k([t_0,t_1];\scalspace_1)^{\bullet,-}\oplus H^k([t_0,t_1];\C^{4\times 4}_{\rm sym})^{\bullet,-}$ as a vector space, but with $\eps$-dependent norm~\eqref{EqTrAsyNorm}.}
  \begin{equation}
  \label{EqTrAsySpace}
  \begin{split}
    &\cD_\eps^{k,\alpha_\circ,\hat\alpha}(\Omega_{t_0,t_1,r_0,\eps}) \\
    &\qquad := H_{\sop,\eps}^{k,\alpha_\circ,\hat\alpha}(\Omega_{t_0,t_1,r_0,\eps};S^2\wt T^*\wt M)^{\bullet,-} \oplus \eps^{\hat\alpha-\frac52}H^k([t_0,t_1];\C^4)^{\bullet,-} \oplus \eps^{\hat\alpha-\frac72}H^k([t_0,t_1];\scalspace_1)^{\bullet,-} \\
    &\qquad\quad\hspace{13.42em} \oplus \eps^{\alpha_\circ}H^k([t_0,t_1];\C^{4\times 4}_{\rm sym})^{\bullet,-},
  \end{split}
  \end{equation}
  with norm
  \begin{equation}
  \label{EqTrAsyNorm}
  \begin{split}
    &\bigl\|\bigl(h_0,\dot b,\scal,\Ups_{(0)}\bigr)\bigr\|_{\cD_\eps^{k,\alpha_\circ,\hat\alpha}(\Omega_{t_0,t_1,r_0,\eps})} \\
    &\qquad := \|h_0\|_{H_{\sop,\eps}^{k,\alpha_\circ,\hat\alpha}(\Omega_{t_0,t_1,r_0,\eps})^{\bullet,-}} + \|\eps^{-\hat\alpha+\frac52}\dot b\|_{H^k([t_0,t_1])^{\bullet,-}} + \|\eps^{-\hat\alpha+\frac72}\scal\|_{H^k([t_0,t_1])^{\bullet,-}} \\
    &\qquad \quad \hspace{2em} + \|\eps^{-\alpha_\circ}\Ups_{(0)}\|_{H^k([t_0,t_1])^{\bullet,-}}.
  \end{split}
  \end{equation}
\end{definition}

\begin{lemma}[Weights]
\label{LemmaTrAsyWeights}
  Let $\alpha,\beta\in\R$, and $k\in\N_0$. Suppose that
  \[
    u \in \eps^{\alpha-\frac32}H^k(\R),\qquad
    h \in \cA^\beta(\hat X_b).
  \]
  Let $K\subset M$ be compact. Then for any fixed $\eta>0$, we have uniform bounds
  \[
    \|u h\|_{H_{\sop,\eps}^{k,-\frac32+\alpha+\beta-\eta,\alpha}(K)} \leq C\|\eps^{-\alpha+\frac32}u\|_{H^k}
  \]
  where $C=C(\alpha,\beta,K,\eta)$.
\end{lemma}

For $(\alpha,\beta)=(\hat\alpha-1,1)$ and $(\hat\alpha-2,2)$, this elucidates the $\eps$-weights in the first line on the right of~\eqref{EqTrAsySpace}. The full justification for Definition~\ref{DefTrAsy} is given in the (proof of) Theorem~\ref{ThmTrAsyTame} below.

\begin{proof}[Proof of Lemma~\usref{LemmaTrAsyWeights}]
  It suffices to prove this for $k=0$. Furthermore, we can reduce to $\alpha=0$ by multiplying $u$ by $\eps^{-\alpha}$; and since $h=\la\hat x\ra^{-\beta}h_0$, $h_0\in\cA^0(\hat X_b)$, with $\la\hat x\ra^{-1}$ a defining function of $M_\circ$, we may assume $\beta=0$. Working in the region $0\leq t\leq 1$, $\eps\lesssim|x|\lesssim 1$, and writing $u(t)=\eps^{-\frac32}u_0(t)$ with $u_0\in L^2(\R;|\dd t|)$, we can use $\hat\rho=|x|$ and $\rho_\circ=\frac{\eps}{|x|}$ as defining functions for $\hat M$ and $M_\circ$. For $\beta_\circ,\hat\beta\in\R$, we then estimate
  \begin{align*}
    &\int_0^1 \int_{\eps\lesssim|x|\lesssim 1} |x|^{-2\hat\beta}\Bigl(\frac{\eps}{|x|}\Bigr)^{-2\beta_\circ} |\eps^{-\frac32}u_0(t)\phi(\hat x)|^2\,\dd x\,\dd t \\
    &\qquad \lesssim \int_0^1 \int_{1\lesssim|\hat x|\lesssim\eps^{-1}} \eps^{-2\hat\beta-3} |\hat x|^{2(\beta_\circ-\hat\beta)} |u_0(t)|^2\,\eps^3\dd\hat x\,\dd t \\
    &\qquad \sim \eps^{-2\hat\beta}\|u_0\|_{L^2([0,1];|\dd t|)}^2 \int_1^{\eps^{-1}} \hat r^{2(\beta_\circ-\hat\beta)+2}\,\dd\hat r.
  \end{align*}
  For $\hat\beta\leq 0$ and $2(\beta_\circ-\hat\beta)+2<-1$, this is uniformly bounded; we thus take $\hat\beta=0$ and $\beta_\circ=-\frac32-\eta$.
\end{proof}

Recall that the individual terms produced by the s-regularity preserving solution operator from Theorem~\ref{ThmKHi} have support extending to small (as measured in terms of $\eps\hat t$) negative times. Our tame solution operators for $\wt L_{\wt g;\wt g_0}$ which describe forward solutions sufficiently well for a nonlinear iteration similarly enlarge supports. The quantity $\delta>0$ controlling the enlargement in the $t$-support will, in the Nash--Moser iteration, be chosen as a geometric sequence with sum less than any fixed small number $\eta>0$, and will in any case be ultimately irrelevant by domain of dependence considerations for the quasilinear gauge-fixed Einstein equation.

\begin{thm}[Tame estimates for the linearized gauge-fixed Einstein operator and its forward solutions]
\label{ThmTrAsyTame}
  We use the notation and assumptions of Proposition~\usref{PropEsTame}. There exists a constant $d\in\N$ so that the following holds. Fix $\lambda\in(0,\lambda_0]$ and write $t_1=t_0+\lambda$, $r_0=\lambda$. Fix $\delta_0<t_1-t_0$ with $t_0-\delta_0\in I_\cC$. Then for every $\delta\leq\delta_0$, there exists a linear map $L_{g_\eps;g_{0,\eps},\delta,+}^{-1}$ which for all $\N_0\ni k\geq d$ maps
  \[
    L_{g_\eps;g_{0,\eps},\delta,+}^{-1} \colon H_{\sop,\eps}^{k+d,\alpha_\circ,\hat\alpha-2}(\Omega_{t_0+\delta,t_1,r_0,\eps};S^2\wt T^*\wt M)^{\bullet,-} \to \cD_\eps^{k,\alpha_\circ,\hat\alpha}(\Omega_{t_0,t_1,r_0,\eps})
  \]
  continuously, and which has the following properties.
  \begin{enumerate}
  \item{\rm (Solution.)} Fix $\check\chi_0\in\CIc((1,2))$ with $\int\check\chi_0(s)\,\dd s=1$, and define
    \begin{equation}
    \label{EqTrAsyTamehplus}
      h_+^\Ups(\Ups_{(0)})(t,r,\omega) := \int_\R r^{-1}\check\chi_0\Bigl(\frac{s}{r}\Bigr)\Ups_{(0)}(t-s)h^\Ups\,\dd s,\qquad t\in[t_0,t_1]
    \end{equation}
    for $\Ups_{(0)}\colon(-\infty,t_1]\to\C^{4\times 4}_{\rm sym}$ with support in $t\geq t_0$; here, in the notation~\eqref{EqKPureGaugeij}--\eqref{EqKPureGauge00}, we write $\Ups_{(0)}h^\Ups=\sum_{0\leq\mu\leq\nu\leq 3}(\Ups_{(0)})_{\mu\nu}h^\Ups_{\mu\nu}$. Set $\ft_1:=(t-t_0)-r$. Given $\fh=\bigl(h_0,\dot b,\scal,\Ups_{(0)}\bigr)\in\cD_\eps^{k,\alpha_\circ,\hat\alpha}(\Omega_{t_0,t_1,r_0,\eps})$. 
    \begin{equation}
    \label{EqTrAsyTameSol}
      \Xi_\eps(\fh) := h_0 + \hat g_b^{\prime\Ups}(\dot b(\ft_1)) + h_{\rms 1}(\scal(\ft_1)) + \eps\breve h_{1,\rms 1}\bigl(\pa_{\ft_1}\scal(\ft_1)\bigr) + h_+^\Ups(\Ups_{(0)}).
    \end{equation}
    Then
    \begin{equation}
    \label{EqTrAsyTameSolMem}
      \Xi_\eps(\fh) \in H_{\sop,\eps}^{k-d,\alpha_\circ,\hat\alpha-2}(\Omega_{t_0,t_1,r_0,\eps})^{\bullet,-}.
    \end{equation}
    Moreover, given $f\in H_{\sop,\eps}^{k+d,\alpha_\circ,\hat\alpha-2}(\Omega_{t_0+\delta,t_1,r_0,\eps})^{\bullet,-}$ and $L_{g_\eps;g_{0,\eps},\delta,+}^{-1}f=:\fh$, we have
    \[
      L_{g_\eps;g_{0,\eps}}(\Xi_\eps\fh) = f
    \]
    on $\Omega_{t_0,t_1,r_0,\eps}$.
  \item{\rm (Tame estimate).} We have the (uniform for $\eps\in(0,\lambda)$ and $\delta\in(0,\delta_0]$) tame estimate
    \begin{equation}
    \label{EqTrAsyTameEst}
    \begin{split}
      &\|\fh\|_{\cD_\eps^{k,\alpha_\circ,\hat\alpha}(\Omega_{t_0,t_1,r_0,\eps})} \\
      &\qquad \leq C_k\delta^{-k-d}\Bigl( \|f\|_{H_{\sop,\eps}^{k+d,\alpha_\circ,\hat\alpha-2}(\Omega_{t_0+\delta,t_1,r_0,\eps})^{\bullet,-}} \\
      &\qquad \quad \hspace{5em} + \|g_\eps-g_{0,\eps}\|_{\cC_{\sop,\eps}^{k+d,1,2}(\Omega_{t_0+\delta,t_1,r_0,\eps})} \|f\|_{H_{\sop,\eps}^{d,\alpha_\circ,\hat\alpha-2}(\Omega_{t_0+\delta,t_1,r_0,\eps})^{\bullet,-}} \Bigr).
    \end{split}
    \end{equation}
  \end{enumerate}
  The constants $\lambda_0$ and $C_k$ (for fixed $\lambda$) can be chosen to be uniform (in $\eps,\delta$) for sufficiently small (independently of $k,\eps,\delta$) perturbations of $\wt g-\wt g_0$ as measured in the $\cC_{\sop,\eps}^{d,1,2}(\Omega_{t_0,t_1,r_0,\eps})$-norm.
\end{thm}

The membership~\eqref{EqTrAsyTameSolMem} is the glued spacetime analogue of Lemma~\ref{LemmaKEFwdWSize}.

\begin{proof}[Proof of Theorem~\usref{ThmTrAsyTame}]
  \pfstep{Reduction to the Kerr model problem.} Given a forcing term $f\in H_{\sop,\eps}^{k+d,\alpha_\circ,\hat\alpha-2}(\Omega_{t_0+\delta,t_1,r_0,\eps};S^2\wt T^*\wt M)^{\bullet,-}$, we apply Proposition~\ref{PropEsTame} to obtain tame estimates for the forward solution of $L h=f$ in the space
  \[
    h \in H_{\sop,\eps}^{k,\alpha_\circ,\hat\alpha-2}(\Omega_{t_0+\delta,t_1,r_0,\eps};S^2\wt T^*\wt M)^{\bullet,-}.
  \]
  We next compare the operator $\wt L_{\wt g;\wt g_0}$ with the Kerr model $\wt L_{\wt g_b;\wt g_b}$ (see~\eqref{EqEseKerr} and \eqref{EqEseScaleKerrModel}): using Lemma~\ref{LemmaEseScale}\eqref{ItEseScaleOp} and writing $g_{b,\eps}=\wt g_b|_{M_\eps}$, we have
  \begin{equation}
  \label{EqTrAsyTameEq}
    \eps^{-2}L h = L_{g_{b,\eps};g_{b,\eps}}h = -\bigl( L_{g_\eps;g_{0,\eps}} - L_{g_{b,\eps};g_{b,\eps}} \bigr) h =: f' \in H_{\sop,\eps}^{k-2,\alpha_\circ,\hat\alpha-2}(\Omega_{t_0+\delta,t_1,r_0,\eps})^{\bullet,-},
  \end{equation}
  with the norm of $f'$ satisfying the tame estimate
  \begin{equation}
  \label{EqTrAsyTamefp}
    \|f'\|_{H_{\sop,\eps}^{k-2,\alpha_\circ,\hat\alpha-2}} \leq C_k\Bigl( \|h\|_{H_{\sop,\eps}^{k,\alpha_\circ,\hat\alpha-2}} + \|g_\eps-g_{0,\eps}\|_{\cC_{\sop,\eps}^{k,1,2}(\Omega_{t_0+\delta,t_1,r_0,\eps})}\|h\|_{H_{\sop,\eps}^{d,\alpha_\circ,\hat\alpha-2}}\Bigr)
  \end{equation}
  as follows from Lemma~\ref{LemmaNTame}. Passing to 3b-spaces with $\hat t=\frac{t-t_0}{\eps}$, $\hat x=\frac{x}{\eps}$ (and volume density $|\dd\hat t\,\dd\hat x|$), and setting $\alpha_\cD:=\alpha_\circ-\hat\alpha$, equation~\eqref{EqTrAsyTameEq} reads
  \begin{equation}
  \label{EqTrAsyTameEqh}
  \begin{split}
    &L h = \eps^2 f' \in \eps^{\hat\alpha-2}H_{(\tbop;\sop),\eps}^{(0;k-2),\alpha_\cD+2,0}(\hat\Omega_{\delta,\eps})^{\bullet,-}, \\
    &\qquad \hat\Omega_{\delta,\eps} := \bigl\{ (\hat t,\hat x) \colon \delta\eps^{-1}\leq\hat t\leq(t_1-t_0)\eps^{-1},\ \bhm\leq|\hat x|\leq\eps^{-1}r_0 + 2((t_1-t_0)\eps^{-1}-\hat t) \bigr\}.
  \end{split}
  \end{equation}
  Results on forward solutions for $L$ require the source term to lie in a 3b-space with suitable 3b-regularity orders $\sfs$; we put $\eps^2 f'$ into such a space by noting that the inclusion map $H_{(\tbop;\sop),\eps}^{(0;k-2),\alpha_\cD+2,0}\to H_{(\tbop;\sop),\eps}^{(\sfs;k-d_0),\alpha_\cD+2,0}$ is uniformly (in $\eps$) bounded as soon as $d_0-2\geq\max(0,\sup\sfs)$.

  \pfstep{Solution of the Kerr model problem.} We now extract the leading order behavior of $h$ at $\hat M$ by applying the solution operator $L_{\delta,+}^{-1}$ from Theorem~\ref{ThmKHi} to equation~\eqref{EqTrAsyTameEqh}, with $\hat t$ shifted\footnote{This shift can be avoided if we define $\hat t=\frac{t-(t_0+\delta)}{\eps}$.} by $\delta\eps^{-1}$, to an extension of $\eps^2 f'$ to an element of $\dot H_{(\tbop;\sop),\eps}^{(\sfs;k-d_0),\alpha_\cD+2,0}(\hat M_{b,\delta\eps^{-1}}^+)$ with uniformly controlled norm. Concretely, we choose $\hat t_1=\hat t-\hat r$, and we again replace 3b-regularity by s-regularity and giving up $d_0\geq\max(0,-\inf\sfs)$ s-derivatives in the process. This produces, upon restriction to the domain $\hat\Omega_{0,\eps}$ and the time interval $\hat t\geq 0$,
  \begin{align}
    h_0 &\in \eps^{\hat\alpha-2} H_{(\tbop;\sop),\eps}^{(0;k-2 d_0),\alpha_\cD,0}(\hat\Omega_{0,\eps};S^2\,\Ttsc^*\hat M_b)^{\bullet,-}, \nonumber\\
  \label{EqTrAsyTameExpdotb}
    \dot b &\in \eps^{\hat\alpha-2}(\hat t_1+1) H_{-;\bop,\eps}^{(0;1;k-2 d_0)}([0,(t_1-t_0)\eps^{-1}];\C^4)^{\bullet,-}, \\
    \scal &\in \eps^{\hat\alpha-2}(\hat t_1+1)^2 H_{-;\bop,\eps}^{(0;2;k-2 d_0)}([0,(t_1-t_0)\eps^{-1}];\scalspace_1)^{\bullet,-}, \nonumber\\
    \Ups'_{(0)} &\in \eps^{\hat\alpha-2}(\hat t_1+1)^{-(\alpha_\cD+\frac32)} H_{-;\bop;\eps}^{(0;-(\alpha_\cD+\frac32);k-2 d_0)}([0,(t_1-t_0)\eps^{-1}];\C^{4\times 4}_{\rm sym})^{\bullet,-}, \nonumber
  \end{align}
  with norms bounded by $C_k\delta^{-(k-d_0)}$ times the norm of $\eps^2 f'$ in $\eps^{\hat\alpha-2}H_{(\tbop;\sop),\eps}^{(0;k-2),\alpha_\cD+2,0}(\hat\Omega_{\delta,\eps})^{\bullet,-}$, so that
  \begin{equation}
  \label{EqTrAsyTameExp}
    h=h_0+\hat g_b^{\prime\Ups}(\dot b(\hat t_1))+h_{\rms 1}(\scal(\hat t_1))+\breve h_{1,\rms 1}(\pa_{\hat t_1}\scal(\hat t_1))+h_+^\Ups(\Ups'_{(0)}).
  \end{equation}
  Note that since $f'$ itself is controlled tamely by $h$ by~\eqref{EqTrAsyTamefp}, and $h$ tamely by $f$ by Proposition~\ref{PropEsTame}, we thus have tame estimates for $h_0,\dot b,\scal,\Ups'_{(0)}$ by $f$, with constants $C_k\delta^{-(k-d_0)}$.

  \pfstep{Passing back to the glued spacetime.} In the final step of the proof, we interpret the terms in the expansion~\eqref{EqTrAsyTameExp} from the perspective of the glued spacetime. (This includes passing from $S^2\,\Ttsc^*\hat M_b$ back to $S^2\wt T^*\wt M$ via $\dd\hat z^\mu\otimes_s\dd\hat z^\nu\mapsto\dd z^\mu\otimes_s\dd z^\nu$.) First, we note that $h_0\in H_{\sop,\eps}^{k-2 d_0,\alpha_\circ,\hat\alpha}(\Omega_{t_0,t_1,r_0,\eps})^{\bullet,-}$. Next, since $\eps(\hat t_1+1)$ is bounded for $0\leq\hat t_1\leq(t_1-t_0)\eps^{-1}$, we can bound $\eps\dot b$ in $\eps^{\hat\alpha-2}H_{-;\bop;\eps}^{(0;0;k-2 d_0)}([0,(t_1-t_0)\eps^{-1}];\C^4)^{\bullet,-}$ (giving up one b-derivative) by $\dot b$ measured in the space~\eqref{EqTrAsyTameExpdotb}, which in view of $L^2(\R_{\hat t_1};|\dd\hat t_1|)=\eps^{\frac12}L^2(\R_{t_1};|\dd t_1|)$ amounts to bounding $\eps\dot b\in\eps^{\hat\alpha-\frac32}H^{k-2 d_0}([0,t_1-t_0];\C^4)^{\bullet,-}$, so
  \[
    \dot b \in \eps^{\hat\alpha-\frac52}H^{k-2 d_0}([t_0,t_1];\C^4)^{\bullet,-}.
  \]
  Similarly, using the uniform boundedness of $\eps^2(\hat t_1+1)^2$, we obtain bounds for
  \[
    \scal \in \eps^{\hat\alpha-\frac72}H^{k-2 d_0}([t_0,t_1];\scalspace_1)^{\bullet,-}.
  \]
  Note also that $\breve h_{1,\rms 1}(\pa_{\hat t_1}\scal)=\eps\breve h_{1,\rms 1}(\pa_{t_1}\scal)$.

  For $\Ups'_{(0)}$, we consider a minimal norm extension $\Ups^\sharp_{(0)}$ to $[0,\infty)$ for which we thus have bounds for $\sigma\mapsto\int e^{i\sigma\hat t}\eps^{-\hat\alpha+2}\Ups^\sharp_{(0)}(\hat t)\,\dd\hat t$ in the space $|\sigma|^{\alpha_\cD+\frac32}(1+\eps^{-1}|\sigma|)^{-(k-2 d_0)}L^2(\R_\sigma;|\dd\sigma|)$, which when written in terms of $\hat t=\frac{t-t_0}{\eps}$ and $\zeta=\frac{\sigma}{\eps}$ amounts to bounds for the function $\zeta\mapsto\eps^{-1}\int e^{i\zeta t}\eps^{-\hat\alpha+2}\Ups_{(0)}^\sharp(\frac{t-t_0}{\eps})\,\dd t$ in $\eps^{\alpha_\cD+\frac32} |\zeta|^{\alpha_\cD+\frac32}(1+|\zeta|)^{-(k-2 d_0)}\eps^{-\frac12}L^2(\R_\zeta;|\dd\zeta|)$, i.e.\ bounds for $\Ups_{(0)}^\flat(t):=\Ups^\sharp_{(0)}(\frac{t-t_0}{\eps})$ in
  \[
    \eps^{\hat\alpha+\alpha_\cD}\cF_{t\to\zeta}^{-1}\bigl( (1+|\zeta|)^{-(k-2 d_0)}|\zeta|^{\alpha_\cD+\frac32}L^2(\R_\zeta) \bigr).
  \]
  But if $\psi\in\CIc(\R)$ is a fixed cutoff equal to $1$ on $[t_0,t_1]$, we can replace $\Ups_{(0)}^\flat$ by $\psi\Ups_{(0)}^\flat$ without changing its restriction to $[t_0,t_1]$. On the Fourier transform side, this amounts to convolving $\wh{\Ups_{(0)}^\flat}$ with $\hat\psi\in\sS(\R)$. Since this smooths out the $|\zeta|^{\alpha_\cD+\frac32}$-singularity of $\wh{\Ups_{(0)}^\flat}$ at $\zeta=0$, we thus have uniform bounds for $\psi\Ups_{(0)}^\flat$ in
  \[
    \eps^{\hat\alpha+\alpha_\cD}\cF^{-1}\bigl(\la\zeta\ra^{-(k-2 d_0)+\alpha_\cD+\frac32}L^2(\R_\zeta)\bigr) \subset \eps^{\hat\alpha+\alpha_\cD}H^{k-2 d_0}(\R_t).
  \]
  Therefore, $\Ups_{(0)}(t):=\Ups'_{(0)}(\frac{t-t_0}{\eps})$ is bounded in $\eps^{\hat\alpha+\alpha_\cD}H^{k-2 d_0}([t_0,t_1])^{\bullet,-}$ via the same tame bounds as described after~\eqref{EqTrAsyTameExp}. Note, finally, that in the notation~\eqref{EqKHiSolhUps} but with $\Ups'_{(0)}$ in place of $\Ups_{(0)}$, and passing to the coordinates $t,r,\omega$,
  \begin{align*}
    h_+^\Ups(\Ups'_{(0)}) \colon (t,r,\omega) &\mapsto \int (r/\eps)^{-1}\check\chi_0\Bigl(\frac{\hat s}{r/\eps}\Bigr)\Ups'_{(0)}\Bigl(\frac{t-t_0}{\eps}-\hat s\Bigr)h^\Ups\,\dd\hat s \\
      &= \int r^{-1}\check\chi_0\Bigl(\frac{s}{r}\Bigr)\Ups_{(0)}(t-s)h^\Ups\,\dd s,
  \end{align*}
  which matches $h_+^\Ups(\Ups_{(0)})$ in the notation~\eqref{EqTrAsyTamehplus}.

  We have thus shown that
  \[
    L_{g_\eps;g_{0,\eps},\delta,+}^{-1}(f) := (h_0,\dot b,\scal,\Ups_{(0)}) \in D^{k-2 d_0,\alpha_\circ,\hat\alpha}_\eps(\Omega_{t_0,t_1,r_0,\eps}),
  \]
  is well-defined, together with the estimate~\eqref{EqTrAsyTameEst} (except for a shift in $k$, which we can undo simply by starting the argument with $k+2 d_0$ in place of $k$, and increasing $d$).

  %%%%%%%%%%
  \pfstep{Bounds on $\Xi_\eps(\fh)$.} It remains to establish the membership~\eqref{EqTrAsyTameSolMem}. The term $h_0$ was already addressed above, and uniform bounds for $h_+^\Ups(\Ups_{(0)})\in H_{\sop,\eps}^{k,\alpha_\circ,\alpha_\circ+\frac32-\eta}\subset H_{\sop,\eps}^{k,\alpha_\circ,\hat\alpha-2}$ follow from Lemma~\ref{LemmaNL2FT}; note here that if in Lemma~\ref{LemmaNL2FT}, the function $\hat a$ remains uniformly bounded upon multiplication by $(1+\eps^{-1}|\sigma|)^k$, then the conclusion~\eqref{EqNL2FTMem} can be strengthened to uniform bounds in $H_{(\tbop;\sop),\eps}^{(\infty;k),-\frac32-\alpha,-\alpha-\eta}$ (applied here with $\alpha=-(\alpha_\cD+\frac32)$). Bounds for the terms involving $\dot b,\pa_{t_1}\scal$, resp.\ $\scal$ in the spaces $H_{\sop,\eps}^{k-1,\alpha_\circ,\hat\alpha-1}$, resp.\ $H_{\sop,\eps}^{k,\alpha_\circ,\hat\alpha-2}$ follow from Lemma~\ref{LemmaTrAsyWeights} with $(\alpha,\beta)=(\hat\alpha-1,1)$, resp.\ $(\hat\alpha-2,2)$; we use here that for each of these pairs $(\alpha,\beta)$ we have $\alpha_\circ<-\frac32+\alpha+\beta=-\frac32+\hat\alpha$.
\end{proof}

Even though the output of $\Xi_\eps$ is only of class $H_{\sop,\eps}^{*,\alpha_\circ,\hat\alpha-2}$, we claim that it nonetheless gets mapped by linearized gauge-fixed Einstein operators into $H_{\sop,\eps}^{*,\alpha_\circ,\hat\alpha-2}$, i.e.\ \emph{with no loss} in the $\hat M$-decay order relative to the input $f$; this reflects the fact that the terms other than $h_0$ in~\eqref{EqTrAsyTameSol}, which have $\hat M$-weights worse than $\hat\alpha$, are annihilated to the appropriate order by $L_{\wt g;\wt g_0}$, being (large, generalized) zero energy states of the Kerr model $L$. (This is thus the glued spacetime analogue of Theorem~\ref{ThmKEFwdW}\eqref{ItKEFwdWFwd}.)

\begin{lemma}[$\Xi_\eps$ and the linearized gauge-fixed Einstein operator]
\label{LemmaTrAsyLin}
  In the notation~\eqref{EqTrAsyTameSol}, the map $L_{g_{0,\eps};g_{0,\eps}}\circ\Xi_\eps$ is uniformly bounded as a map
  \[
    L_{g_{0,\eps};g_{0,\eps}}\circ\Xi_\eps \colon \cD_\eps^{k,\alpha_\circ,\hat\alpha}(\Omega_{t_0,t_1,r_0,\eps}) \to H_{\sop,\eps}^{k-d-2,\alpha_\circ,\hat\alpha-2}(\Omega_{t_0,t_1,r_0,\eps})^{\bullet,-}.
  \]
\end{lemma}
\begin{proof}
  We use~\eqref{EqEseScaleOp} for $\wt g=\wt g_0$ with $\lambda=1$, which gives $\wt L_{\wt g_0;\wt g_0}-\wt L_{\wt g_b;\wt g_b}\in\cC_\sop^\infty\Diffse^2$; recalling~\eqref{EqTrAsyTameSolMem}, this maps $\Xi_\eps(\cD_\eps^{k,\alpha_\circ,\hat\alpha}(\Omega_{t_0,t_1,r_0,\eps}))\subset H_{\sop,\eps}^{k-d,\alpha_\circ,\hat\alpha-2}(\Omega_{t_0,t_1,r_0,\eps})^{\bullet,-}$ into the space $H_{\sop,\eps}^{k-d-2,\alpha_\circ,\hat\alpha-2}(\Omega_{t_0,t_1,r_0,\eps})^{\bullet,-}$.

  It remains to study $\wt L_{\wt g_b;\wt g_b}$ acting on $\Xi_\eps(\fh)$. Acting on the piece $h_0\in H_{\sop,\eps}^{k,\alpha_\circ,\hat\alpha}$ in~\eqref{EqTrAsyTameSol}, it produces $H_{\sop,\eps}^{k-2,\alpha_\circ,\hat\alpha-2}$. For the remaining terms, we pass to the rescaled coordinates $\hat t,\hat x$ and the bundle $S^2\,\Ttsc^*\hat M_b$; under this identification, $\wt L_{\wt g_b;\wt g_b}=\eps^{-2}L$. We then proceed similarly to the proof of Theorem~\ref{ThmKEFwdW}\eqref{ItKEFwdWFwd} using the expansion~\eqref{EqKEFwdWL} of $L$. Since $t_1=t_0+\eps\hat t_1$, we have
  \[
    L \hat g_b^{\prime\Ups}(\dot b(t_1)) = L_0\hat g_b^{\prime\Ups}(\dot b(t_1)) + \eps L_1\hat g_b^{\prime\Ups}(\pa_{t_1}\dot b(t_1)) + \eps^2 L_2\hat g_b^{\prime\Ups}(\pa_{t_1}^2\dot b(t_1)).
  \]
  The first term vanishes; the second term is the product of $\eps\pa_{t_1}\dot b\in\eps^{\hat\alpha-\frac32}H^{k-1}$ with an element of $\cA^2(\hat X_b)$ and thus lies in $H_{\sop,\eps}^{k-1,\alpha_\circ+2,\hat\alpha}$ by Lemma~\ref{LemmaTrAsyWeights} with $(\alpha,\beta)=(\hat\alpha,2)$ and using that $-\frac32+\hat\alpha+2>\alpha_\circ+2$. The third term is the product of an element of $\eps^{\hat\alpha-\frac12}H^{k-2}$ with an element of $\cA^2(\hat X_b)$ and thus also lies in $H_{\sop,\eps}^{k-2,\alpha_\circ,\hat\alpha}$.

  Similarly, we can bound
  \[
    L\bigl(h_{\rms 1}(\scal)+\eps h_{1,\rms 1}(\pa_{t_1}\scal)\bigr) = \eps^2 L_2 h_{\rms 1}(\pa_{t_1}^2\scal) + \eps^2 L_1\breve h_{1,\rms 1}(\pa_{t_1}^2\scal) + \eps^3 L_2\breve h_{1,\rms 1}(\pa_{t_1}^3\scal)
  \]
  (cf.\ \eqref{EqKEFwdWLhs1}) uniformly in $H_{\sop,\eps}^{k-3,\alpha_\circ+2,\hat\alpha}$ by means of Lemma~\ref{LemmaTrAsyWeights} with $(\alpha,\beta)=(\hat\alpha,2)$.

  Finally, consider the term $L h_+^\Ups(\Ups_{(0)})$. Write $\alpha_\cD=\alpha_\circ-\hat\alpha\in(-2,-\frac32)$. We need to prove uniform $H_{\sop,\eps}^{k-2,\alpha_\circ+2,\hat\alpha}$-bounds for $L h_+^\Ups(\Ups_{(0)})$, which are equivalent to uniform $H_{(\tbop;\sop),\eps}^{(0;k-2),\alpha_\cD+2,0}$-bounds for $L h_+^\Ups(\Ups^\natural_{(0)})$ where $\Ups^\natural_{(0)}=\eps^{-\hat\alpha+2}\Ups_{(0)}\in\eps^{\alpha_\cD+2}H^k$. Write now $L=L_0+L_1\pa_{\hat t}+L_2\pa_{\hat t}^2$ where $L_0=\hat L(0)\in\rho_\cD^2\Diffb^2$, $L_1\in\rho_\cD\Diffb^1$, and $L_2\in\Diffb^0$. From~\eqref{EqTrAsyTamehplus}, we then have
  \[
    L h_+^\Ups(\Ups_{(0)}^\natural)(\hat t,\hat r,\omega) = \int (L_0+L_1\pa_{\hat t}+L_2\pa_{\hat t}^2)\Bigl( \hat r^{-1}\check\chi_0\Bigl(\frac{\hat s}{\hat r}\Bigr)\Ups_{(0)}^\natural(\eps\hat t-\eps\hat s) h^\Ups\Bigr) \,\dd\hat s \\
  \]
  Choose an extension of $\Ups_{(0)}^\natural$ to an element of $H^k(\R)$ with minimal norm, and write $\Ups_{(0)}^\sharp(\hat t)=\Ups_{(0)}^\natural(t_0+\eps\hat t)$, which thus has uniform $\eps^{\alpha_\cD+\frac32}H_{-;\eps}^{(0;k)}$ bounds (the shift by $\frac12$ arising from $L^2(\R_t;|\dd t|)=\eps^{-\frac12}L^2(\R_{\hat t};|\dd\hat t|)$). The Fourier transform in $\hat t$ of $L h_+^\Ups(\Ups_{(0)}^\natural)(\hat t,\hat r,\omega)$ is
  \[
    [L_0,\chi_0(\sigma\hat r)]\Bigl(\wh{\Ups_{(0)}^\sharp}(\sigma)h^\Ups\Bigr) - i\sigma L_1\Bigl( \chi_0(\sigma\hat r)\wh{\Ups_{(0)}^\sharp}(\sigma)h^\Ups\Bigr) - \sigma^2 L_2\Bigl( \chi_0(\sigma\hat r)\wh{\Ups_{(0)}^\sharp}(\sigma)h^\Ups\Bigr).
  \]
  Since $[L_0,\chi_0(\sigma\hat r)]\in\cA^{\infty,2,1}((\hat X_b)_\scbtop)\Diffb(\hat X_b)$ (with rapid decay as $|\sigma|\to\infty$), the first term is uniformly bounded in
  \begin{equation}
  \label{EqTrAsyL0}
    \la\sigma\ra^{-N}\cA^{\infty,2,1}\cdot\eps^{\alpha_\cD+\frac32}(1+\eps^{-1}|\sigma|)^{-k}L^2
  \end{equation}
  for any $N$. While $\eps^{\alpha_\cD+\frac32}$ itself is not uniformly bounded as $\eps\searrow 0$, the crucial gain here is that we can trade vanishing at $\sigma=0$ for positive powers of $\eps$, to wit,
  \[
    \eps^{\alpha_\cD+\frac32} = (\eps^{-1}|\sigma|)^{-(\alpha_\cD+\frac32)} |\sigma|^{\alpha_\cD+\frac32}
  \]
  implies that~\eqref{EqTrAsyL0} is a subspace of (with the inclusion map uniformly bounded)
  \[
    \la\sigma\ra^{-N}\cA^{\infty,\alpha_\cD+\frac72,\alpha_\cD+\frac52}((\hat X_b)_\scbtop)\cdot(1+\eps^{-1}|\sigma|)^{-k-(\alpha_\cD+\frac32)}L^2.
  \]
  Using Lemmas~\ref{LemmaNMult} and \ref{LemmaKHiFT}, and in view of $k+(\alpha_\cD+\frac32)\geq k-1$, this gives uniform $H_{(\tbop;\sop),\eps}^{(0;k-1),\alpha_\cD+2,0}$-bounds. The analysis of the terms involving $L_1$ and $L_2$ is similar (cf.~\eqref{EqKEFwdWMapFwdUps}).
\end{proof}

%%%%%%%%%%%%%%%%%%%%%%%%%%%%%%%%%%%%%%%%%%%%%%%%%%
\subsection{Nash--Moser iteration; local solution}
\label{SsTrLoc}

The solution operator $L_{g_\eps;g_{0,\eps},\delta,+}^{-1}$ of Theorem~\ref{ThmTrAsyTame} enlarges $t$-supports by an amount $\delta$, and its operator norm blows up as $\delta\searrow 0$. But since we need to stay within a fixed compact time interval, we must, in a nonlinear iteration scheme, apply this solution operator for $\delta=\delta_1,\delta_2,\ldots$ where $\delta_j\searrow 0$ (with $\sum_{j=1}^\infty\delta_j$ less than some fixed small positive number). Thus, some care is needed to ensure the convergence of such an iteration scheme. The point is that $L_{g_\eps,g_{0,\eps},\delta_j,+}^{-1}$ will act on an input which, due to the very fast convergence of a Newton-type iteration scheme, has norm much smaller than a large positive power of $\delta_j$. The following is an abstract Nash--Moser type result which suits our needs; it also allows for an increase of supports for the smoothing operators, as in \citeII{Theorem~\ref*{ThmNTameNM}}.

\begin{thm}[Nash--Moser with large solution operators]
\label{ThmNM}
  Let $(B^s,|\cdot|_s)$ and $(\bfB^s,\|\cdot\|_s)$ be Banach spaces for $s\in\N_0$. Suppose that $B^s\subset B^t$ with $|\cdot|_t\leq|\cdot|_s$ for $s\geq t$. For $\eta\in[0,1]$, let $B_\eta^s\subset B^s$ be a linear subspace (with the induced norm), with $B_\eta^s\subseteq B_{\eta'}^s$ whenever $\eta\leq\eta'$, and with $B^s_0=B^s$. Set $B_{(\eta)}^\infty=\bigcap_{s=0}^\infty B_{(\eta)}^s$. Make the analogous definitions and assumptions for $\bfB^s$. Let $d\in\N$, and suppose we are given $c>0$ and the following data:
  \begin{enumerate}
  \item a (nonlinear) $\cC^2$ map $\phi\colon\{u\in B^\infty\colon\|u\|_{3 d}<c\}\to\bfB^\infty$ mapping $u\in B_\eta^\infty$ into $\bfB_\eta^\infty$ which satisfies the following bounds for $|u|_{3 d}<c$:
    \begin{align*}
      \| \phi(u) \|_s &\leq C_s(1+|u|_{s+d}), \qquad s\geq d, \\
      \| \phi'(u)v \|_{2 d} &\leq C_1|v|_{3 d}, \\
      \bigl\| \phi''(u)(v,w) \bigr\|_{2 d} &\leq C_2|v|_{3 d}|w|_{3 d};
    \end{align*}
  \item for every $u\in B^{3 d}$ with $|u|_{3 d}<c$ a collection of linear maps
    \[
      \psi_{\eta,\delta}(u) \colon\bfB_\eta^\infty \to B_{\eta-\delta}^\infty,\qquad \phi'(u)(\psi_{\eta,\delta}(u)f)\bigr) = f\ \forall f\in\bfB_\eta^\infty,
    \]
    for $0<\delta<\eta\leq 1$, satisfying the tame estimate
    \begin{equation}
    \label{EqNMpsi}
      |\psi_{\eta,\delta}(u)f|_s \leq C_s\delta^{-s-d}\bigl( \|f\|_{s+d} + |u|_{s+d}\|f\|_{2 d}\bigr),\qquad s\geq d;
    \end{equation}
  \item linear maps $S_\theta\colon B^\infty_{\theta^{-1/2}}\to B_0^\infty$ which for $\theta\geq\eta^{-2}$ map $S_\theta\colon B_\eta^\infty\to B_{\eta-\theta^{-1/2}}^\infty$, and which moreover satisfy
    \begin{equation}
    \label{EqNMsmooth}
      |S_\theta u|_s \leq C_{s,t}\theta^{s-t}|u|_t\quad\text{for}\ s\geq t, \qquad
      |u-S_\theta u|_s \leq C_{s,t}\theta^{s-t}|u|_t\quad\text{for}\ s\leq t.
    \end{equation}
  \end{enumerate}
  Suppose $\phi(0)\in\bfB_\eta^\infty$. Then if $\|\phi(0)\|_{2 d}$ is sufficiently small depending on $\eta,c,d$, and the constants $C_s,C_{s,t}$ for $s,t\leq 64 d^2+133 d+52$, there exists $u\in B^{3 d+3}$ with $|u|_{3 d}<c$ and $\phi(u)=0$.
\end{thm}

\begin{rmk}[Limited regularity]
\label{RmkNMLimitedReg}
  Theorem~\ref{ThmNM} is far from satisfactory from the point of view of regularity, as we do not directly get an infinitely regular solution. But in our application, we may simply increase $d$ (and start at smaller values of $\eps$ where our formal solution is more accurate) and use uniqueness for forward solutions of quasilinear wave equations to get arbitrarily regular solutions; see the proof of Theorem~\ref{ThmTrLoc} for details.
\end{rmk}

\begin{proof}[Proof of Theorem~\usref{ThmNM}]
  For $\theta_0>1$ (large) and $\zeta>0$ (small, depending only on $d$) to be chosen in the course of the proof, we fix the quantities
  \begin{equation}
  \label{EqNMzeta}
    \theta_k = \theta_0^{(\frac54)^k},\qquad
    \delta_k := \theta_k^{-\zeta}.
  \end{equation}
  Starting with $u_0=0$ and $\eta_0=\eta$, we shall iteratively define
  \[
    v_k := -\psi_{\eta_k,\delta_k}(u_k)\phi(u_k),\qquad
    u_{k+1} := u_k + S_{\theta_k}v_k,\quad \eta_{k+1}:=\eta_k-(\delta_k+\theta_k^{-\frac12}).
  \]
  Note that when $\theta_0$ is sufficiently large (depending on the choice of $\zeta$), we have $\sum_{k=0}^\infty(\delta_k+\theta_k^{-\frac12})<\eta$.

  Much as in \cite{SaintRaymondNashMoser}, we will show that if $\|\phi(0)\|_{2 d}\leq\theta_0^{-4}$ (with $\theta_0>1$ sufficiently large), the sequences $v_k,u_k$ are well-defined, and there exist constants $(U_t)_{t\geq d}$ and $V$ (independent of $k$) so that
  \begin{align}
  \label{EqNM1}\tag{i$_k$}
    |u_k|_{3 d}&<c,\quad \|\phi(u_k)\|_{2 d}\leq\theta_k^{-4}, \\
  \label{EqNM2}\tag{ii$_k$}
    |v_k|_{3 d+3} &\leq V\theta_k^{-3}, \\
  \label{EqNM3}\tag{iii$_k$}
    (1+|u_{k+1}|_{t+2 d}) &\leq U_t\theta_k^{2 d+(t+d)\zeta}(1+|u_k|_{t+2 d}),\quad t\geq d.
  \end{align}
  Note that~\eqref{EqNM1} is valid for $k=0$. We assume the validity of \eqref{EqNM1}, (\hyperref[EqNM2]{ii$_{k-1}$}), and (\hyperref[EqNM3]{iii$_{k-1}$}). 

  We record that for $N=8(2 d+1)$ and a suitable choice of $T$ (made below, depending only on $d$), we have
  \begin{equation}
  \label{EqNM4}\tag{iv$_k$}
    (1+|u_k|_{T+2 d}) \leq \theta_k^N.
  \end{equation}
  For $k=0$ this is clear; and assuming the validity of (\hyperref[EqNM3]{iii$_{k-1}$}) and (\hyperref[EqNM4]{iv$_{k-1}$}), then
  \[
    (1+|u_k|_{T+2 d}) \leq U_T\theta_{k-1}^{2 d+(T+d)\zeta}(1+|u_{k-1}|_{T+2 d}) = (U_T\theta_{k-1}^{-1}) \theta_{k-1}^{N+2 d+1+(T+d)\zeta} \leq \theta_k^N
  \]
  provided $\theta_0\geq U_T$ (thus $U_T\theta_{k-1}^{-1}\leq 1$) and $\theta_k\geq\theta_{k-1}^{1+\frac{2 d+1+(T+d)\zeta}{N}}$. This inequality holds for the choice~\eqref{EqNMzeta} of $\theta_k$ provided $\zeta,T$ satisfy
  \begin{equation}
  \label{EqNMTzeta1}
    (T+d)\zeta < \frac{N}{8},
  \end{equation}
  since then $1+\frac{2 d+1+(T+d)\zeta}{N}<\frac54$.

  %%%%%%%%%%
  \pfstep{Proof of~\eqref{EqNM2}.} We first note that, for any $t\geq d$,
  \begin{align}
    |v_k|_t &\leq C_t\delta_k^{-t-d}\bigl( \|\phi(u_k)\|_{t+d} + |u_k|_{t+d}\|\phi(u_k)\|_{2 d} \bigr) \nonumber\\
      &\leq C_t\delta_k^{-t-d}\bigl( C_{t+d}(1+|u_k|_{t+2 d}) + |u_k|_{t+d} C_{2 d}(1+|u_k|_{3 d}) \bigr) \nonumber\\
  \label{EqNMvkt}
      &\leq C_t\delta_k^{-t-d}(C_{t+d} + (1+c)C_{2 d}) (1 + |u_k|_{t+2 d}).
  \end{align}
  For $t=d$, the first line together with (\hyperref[EqNM1]{i$_k$}) gives
  \begin{equation}
  \label{EqNMvkd}
    |v_k|_d \leq C_d\delta_k^{-2 d} (1+|u_k|_{2 d})\theta_k^{-4} \leq C_d(1+c)\theta_k^{-4+2 d\zeta}.
  \end{equation}
  For $t=T$ on the other hand, we use~\eqref{EqNM4} to conclude that
  \begin{equation}
  \label{EqNMvkT}
    |v_k|_T \leq C_T \delta_k^{-T-d}(C_{T+d}+(1+c)C_{2 d}\bigr)\theta_k^N = C_T(C_{T+d}+(1+c)C_{2 d}) \theta_k^{N+(T+d)\zeta}.
  \end{equation}
  Setting $\alpha=\frac{1}{2(2 d+3)}$, we interpolate between the bounds~\eqref{EqNMvkd} and \eqref{EqNMvkT} using~\eqref{EqNMsmooth}, so
  \begin{align*}
    |v_k|_{3 d+3} &\leq |S_{\theta_k^\alpha}v_k|_{3 d+3} + |v_k-S_{\theta_k^\alpha}v_k|_{3 d+3} \\
      &\leq C_{3d+3,d}\theta_k^{(2 d+3)\alpha}|v_k|_d + C_{3 d+3,T}\theta_k^{-(T-3 d-3)\alpha}|v_k|_T \\
      &\leq C_{c,d,T}\bigl(\theta_k^{-\frac72+2 d\zeta} + \theta_k^{-(T-3 d-3)\alpha+N+(T+d)\zeta}\bigr).
  \end{align*}
  But $-\frac72+2 d\zeta\leq -3$ provided $\zeta\leq\frac{1}{4 d}$, and similarly $-(T-3 d-3)\alpha+N+(T+d)\zeta\leq -3$ provided $\zeta<\alpha$ and
  \begin{equation}
  \label{EqNMTzeta2}
    T \geq \frac{N+3+(3 d+3)\alpha+d\zeta}{\alpha-\zeta}.
  \end{equation}
  This proves~\eqref{EqNM2}. As $\zeta\searrow 0$, the right hand side of~\eqref{EqNMTzeta2} converges to the finite limit $2(2 d+3)(N+3)+(3 d+3)$. Let us thus fix $T=2(2 d+3)(N+3)+(3 d+3)+1$, and fix $\zeta\in(0,\frac{1}{2(2 d+3)})$ so small that~\eqref{EqNMTzeta2} and also~\eqref{EqNMTzeta1} are valid. (Then $T+2 d=64 d^2+133 d+52$.)

  %%%%%%%%%%
  \pfstep{Proof of~\eqref{EqNM3}.} This follows from~\eqref{EqNMvkt} via
  \begin{align*}
    1+|u_{k+1}|_{t+2 d} &\leq 1 + |u_k|_{t+2 d} + |S_{\theta_k}v_k|_{t+2 d} \\
      &\leq 1 + |u_k|_{t+2 d} + C_{t+2 d,t}\theta_k^{2 d}|v_k|_t \\
      &\leq (1 + |u_k|_{t+2 d}) ( 1 + C'_{t,d}\theta_k^{2 d}\theta_k^{(t+d)\zeta} ).
  \end{align*}

  %%%%%%%%%%
  \pfstep{Proof of~\textnormal{(\hyperref[EqNM1]{i$_{k+1}$})}.} Since $u_k=\sum_{j=0}^{k-1}S_{\theta_j}v_j$, we have, for $t\in[0,1]$,
  \[
    |u_k+t S_{\theta_k}v_k|_{3 d} \leq C_{3 d,3 d}\sum_{j=0}^{k-1}|v_j|_{3 d} \leq C_{3 d,3 d}V\sum_{j=0}^{k-1} \theta_j^{-3} < c
  \]
  if we fix $\theta_0$ sufficiently large (depending on $C_{3 d,3 d},V,c$). A second order Taylor expansion thus shows that $\phi(u_{k+1})=\Phi_1+\Phi_2$ where the terms
  \begin{align*}
    \Phi_1 &= \phi(u_k) + \phi'(u_k)(S_{\theta_k}v_k) = -\phi'(u_k)(v_k-S_{\theta_k}v_k), \\
    \Phi_2 &= \int_0^1 (1-t)\phi''(u_k+t S_{\theta_k}v_k)(S_{\theta_k}v_k,S_{\theta_k}v_k)\,\dd t
  \end{align*}
  satisfy the bounds
  \begin{alignat*}{3}
    \|\Phi_1\|_{2 d} &\leq C_1|v_k-S_{\theta_k}v_k|_{3 d} &&\leq C_1 C_{3 d,3 d+3}\theta_k^{-3}|v_k|_{3 d+3} &&\leq C_1 C_{3 d,3 d+3} V\theta_k^{-6}, \\
    \|\Phi_2\|_{2 d} &\leq C_2|S_{\theta_k}v_k|_{3 d}^2 &&\leq C_2 C_{3 d,3 d+3}|v_k|_{3 d+3}^2 &&\leq C_2 C_{3 d,3 d+3}V^2\theta_k^{-6}.
  \end{alignat*}
  Increasing $\theta_0$ further so that $\theta_0\geq C_0:=C_1 C_{3 d,3 d+3}V+C_2 C_{3 d,3 d+3}V^2$, we thus obtain
  \[
    \|\phi(u_{k+1})\|_{2 d} \leq (C_0\theta_k^{-1}) \theta_k^{-5} \leq \theta_k^{-5} = \theta_{k+1}^{-4}.
  \]

  To complete the proof, it remains to observe that in view of~\eqref{EqNM2}, the sequence $u_K=\sum_{k=0}^K S_{\theta_k}v_k\in B^\infty_{\eta_K}$ converges in $B^{3 d+3}$ to some limit $u$ which (upon taking limits in the space $B^{2 d}$ in the estimate~\eqref{EqNM1}) satisfies $\phi(u)=0$.
\end{proof}

We can now solve the gauge-fixed Einstein vacuum equations in small standard domains. Recall that $\wt P(\wt g_0;\wt g_0)\in\CIdot(\wt M\setminus\wt K^\circ;S^2\wt T^*\wt M)$, and thus for any standard domain $\wt\Omega^\sharp\subset\wt M\setminus\wt K^\circ$ associated with a standard domain in $M$ (see Definition~\usref{DefEStdM}), we have
\begin{equation}
\label{EqNMLocFormal}
  \wt P(\wt g_0;\wt g_0)\in \CIdot(\wt\Omega^\sharp) = \eps^\infty\CI(\wt\Omega^\sharp) \subset \eps^\infty\cC_\sop^\infty(\wt\Omega^\sharp).
\end{equation}
(That is, this vanishes to infinite order at $\eps=0$.)

\begin{thm}[Nonlinear solution in small domains]
\label{ThmTrLoc}
  There exists $d\in\N$ so that the following holds. Let $t_0\in I_\cC$. Then there exist $\eta>0$, $\lambda_0>0$, $\eps_0>0$ so that the following holds for any fixed $\lambda\in(0,\lambda_0]$. Set
  \[
    \Omega:=\Omega_{t_0-\lambda,t_0+\lambda,\lambda},\qquad
    \wt\Omega:=\wt\upbeta^{-1}([0,\eps_0]\times\Omega)\setminus\wt K^\circ,\qquad
    \Omega_\eps:=\wt\Omega\cap M_\eps.
  \]
  Suppose $\eps_1\in(0,\eps_0]$ and $\wt h_{\rm in}=(h_{{\rm in},\eps})_{\eps\in(0,\eps_1]}\in\eps^\infty\cC_\sop^\infty(\wt\Omega)$ are such that
  \begin{gather}
  \label{EqNMLocHinSize}
    \sup_{\eps\in(0,\eps_1]}\|\wt h_{\rm in}\|_{\cC_{\sop,\eps}^{d,1,2}(\Omega_\eps)}<\eta, \\
  \label{EqNMLocHinSol}
    \wt P(\wt g_0+\wt h_{\rm in};\wt g_0)=0\ \text{on}\ \wt\Omega\cap\{t\leq t_0\}.
  \end{gather}
  Then there exists $\eps_2=\eps_2(\|\wt h_{\rm in}\|_{\cC_\sop^{T,T,T}})>0$ where $T=64 d^2+134 d+52$ so that the equation
  \begin{equation}
  \label{EqNMLocSol}
    P_\eps(g_{0,\eps}+h_{{\rm in},\eps} + h_\eps;g_{0,\eps}) = 0
  \end{equation}
  has a unique solution $h_\eps$ for $\eps\in(0,\eps_2]$ vanishing in $t\leq t_0$; and it satisfies
  \begin{equation}
  \label{EqNMLocwth}
    \wt h=(h_\eps)_{\eps\in(0,\eps_2]} \in \eps^\infty\cC_\sop^\infty(\wt\Omega\cap\{\eps\leq\eps_2\}).
  \end{equation}
  If instead of~\eqref{EqNMLocHinSol} we only assume $\wt P(\wt g_0+\wt h_{\rm in};\wt g_0)$ to vanish to infinite order at $t=t_0$, then we can solve~\eqref{EqNMLocSol} in $t\geq t_0$ with $h_\eps$ vanishing to infinite order at $t=t_0$ (and still satisfying~\eqref{EqNMLocwth}). Finally, if $\wt h_{\rm in}\in\CIdot(\wt\Omega)$ (i.e.\ $(\eps\pa_\eps)^i\wt h_{\rm in}\in\eps^\infty\CI_\sop(\wt\Omega)$ for all $i\in\N_0$), then also $\wt h\in\CIdot(\wt\Omega\cap\{\eps\leq\eps_2\})$.
\end{thm}

In other words, in the first, resp.\ second part of Theorem~\ref{ThmTrLoc}, we can extend an initial solution $g_{0,\eps}+h_{{\rm in},\eps}$ of the equation $P_\eps(g_\eps;g_{0,\eps})=0$ in $t\leq t_0$, resp.\ a formal solution at $\eps=0$ and $t=t_0$, for a \emph{fixed} (i.e.\ $\eps$-independent) amount of time $\lambda>0$ to the larger domain $\Omega$.

\begin{rmk}[Domains far from the small Kerr black hole]
\label{RmkTrLocTriv}
  The existence, smoothness, and boundedness by $C_N\eps^N$ for all $N$, of solutions of the equation~\eqref{EqNMLocSol} on standard domains which are \emph{disjoint from $\cC$} is a consequence of standard (finite-time) quasilinear hyperbolic theory. We thus do not address this explicitly here.
\end{rmk}

\begin{proof}[Proof of Theorem~\usref{ThmTrLoc}]
  %%%%%%%%%%
  \pfstep{Nash--Moser setup.} Fix $\alpha_\cD\in(-2,-\frac32)$, and let
  \[
    \alpha_\circ = \hat\alpha + \alpha_\cD,\qquad \hat\alpha>100.
  \]
  We fix $\lambda$ to be half the value produced by Theorem~\ref{ThmTrAsyTame} for $\wt g$ in place of $\wt g_0$ and work with $t_0-\lambda,t_0+\lambda,\lambda$ in place of $t_0,t_1,r_0$; and we set $\eps_0=\lambda$. Note that the same value of $\lambda$ works also for
  \[
    \wt g=(g_\eps)_{\eps\in(0,\eps_0]}:=\wt g_0+\wt h_{\rm in}
  \]
  under the smallness condition~\eqref{EqNMLocHinSize}, as well as for small $\cC_\sop^{d,1,2}$ perturbations of $\wt g$. Here $d$ is at least equal to the value produced by Theorem~\ref{ThmTrAsyTame}, though we permit ourselves to increase it further in the course of the proof; we shall in particular require $d$ to be large enough so that $T\geq\hat\alpha+d$.

  Define the domains
  \[
    \Omega_\eta:=\Omega_{t_0-\lambda+\eta\lambda,t_0+\lambda,\lambda},\qquad
    \wt\Omega_\eta:=\wt\upbeta^{-1}([0,\eps_0]\times\Omega_\eta),\qquad
    \Omega_{\eta,\eps}:=\wt\Omega_\eta \cap M_\eps.
  \]
  The plan is to apply Nash--Moser iteration to the family of maps
  \[
    \phi_\eps \colon B_\eps^\infty=B_{\eps,0}^\infty \to \bfB_\eps^\infty=\bfB_{\eps,0}^\infty,
    \qquad B_{\eps,\eta}^k := D_\eps^{k,\alpha_\circ,\hat\alpha}(\Omega_{\eta,\eps}),
    \quad \bfB_{\eps,\eta}^k := H_{\sop,\eps}^{k-d,\alpha_\circ,\hat\alpha-2}(\Omega_{\eta,\eps})^{\bullet,-},
  \]
  defined by
  \[
    \phi_\eps(\fh) := P_\eps(g_\eps + \Xi_\eps(\fh);g_{0,\eps}).
  \]
  We proceed to check the hypotheses of Theorem~\ref{ThmNM}.

  \begin{enumerate}
  \item{\rm (Mapping properties of $\phi_\eps$.)} We claim that for $\fh\in B_\eps^\infty$ with $|\fh|_{3 d}<c:=1$, we have
    \begin{equation}
    \label{EqNMLocphiEst}
      \|\phi_\eps(\fh)\|_s\leq C_s(1+|\fh|_{s+d}),\qquad s\geq d,
    \end{equation}
    where the constants $C_s$ depend only on the norm $\|h_{{\rm in},\eps}\|_{\cC_{\sop,\eps}^{s+d,\alpha_\circ+d,\hat\alpha+d}(\Omega_\eps)}$. We note that Sobolev embedding~\eqref{EqNSobEmbed} gives a uniform upper bound
    \[
      \eps^{-96}\|\Xi_\eps(\fh)\|_{\cC_{\sop,\eps}^{2 d-3}} \leq \|\Xi_\eps(\fh)\|_{\cC_{\sop,\eps}^{2 d-3,\alpha_\circ,\hat\alpha-2-\frac32}} \leq C\|\Xi_\eps(\fh)\|_{H_{\sop,\eps}^{2 d,\alpha_\circ,\hat\alpha-2}} \leq C'
    \]
    (by~\eqref{EqTrAsyTameSolMem}) for some universal constants $C,C'$, so $\|\Xi_\eps(\fh)\|_{\cC_{\sop,\eps}^{2 d-3}}\leq C\eps^{96}$ is less than any fixed small positive number when $\eps\leq\eps_1$, for $\eps_1>0$ small enough. We may thus expand
    \begin{equation}
    \label{EqNMLocphi}
    \begin{split}
      \phi_\eps(\fh)&=P_\eps(g_{0,\eps};g_{0,\eps})+\int_0^1 L_{g_{0,\eps}+z h_{{\rm in},\eps}}(h_{{\rm in},\eps})\,\dd z \\
        &\quad\hspace{5.45em} +\int_0^1 L_{g_{0,\eps}+h_{{\rm in},\eps}+z\Xi_\eps(\fh);g_{0,\eps}}(\Xi_\eps(\fh))\,\dd z.
    \end{split}
    \end{equation}
    The $\|\cdot\|_s$-norm of the first term only depends on the fixed background metric $g_{0,\eps}$. For the second term, we recall from Lemma~\ref{LemmaEOp}\eqref{ItEOpTame} that the coefficients of $L_{g_{0,\eps}+h;g_{0,\eps}}-L_{g_{0,\eps};g_{0,\eps}}$ (as an operator of class $\hat\rho^{-2}\Diffse^2$) obey (tame) estimates in terms of the metric perturbation $h$; this term will thus contribute to the constant $C_s$ in~\eqref{EqNMLocphiEst}. In the third term, we use Lemma~\ref{LemmaTrAsyLin} which in present notation implies that $L_{g_{0,\eps};g_{0,\eps}}\circ\Xi_\eps$ is uniformly bounded as a map $B_{\eps,\eta}^k\to\bfB_{\eps,\eta}^k$. The operator $L_{g_{0,\eps}+h_{{\rm in},\eps}+z\Xi_\eps(\fh);g_{0,\eps}}-L_{g_{0,\eps};g_{0,\eps}}$ has coefficients which vanish to high order as $\eps\searrow 0$ (since $h_{{\rm in},\eps}$ and $\Xi_\eps(\fh)$ do) and obey tame estimates in $\fh$, and it therefore maps $\Xi_\eps(\fh)$ uniformly into the space $\bfB_{\eps,\eta}^s$, with tame estimates for the output (whose constants even scale like a large positive power of $\eps$).
  \item{\rm (Mapping properties of $\phi_\eps',\phi_\eps''$.)} Consider again $\fh\in B_\eps^\infty$ with $|\fh|_{3 d}<1$. The linearization of $\phi_\eps$ at $\fh$ is given by
    \begin{align*}
      \phi_\eps'(\fh)(\fh') &= L_{g_\eps+\Xi_\eps(\fh);g_{0,\eps}}(\Xi_\eps(\fh')) \\
        &= L_{g_{0,\eps};g_{0,\eps}}(\Xi_\eps(\fh')) + \bigl( L_{g_{0,\eps}+h_{{\rm in},\eps}+\Xi_\eps(\fh);g_{0,\eps}}-L_{g_{0,\eps};g_{0,\eps}}\bigr)(\Xi_\eps(\fh')).
    \end{align*}
    By Lemma~\ref{LemmaTrAsyLin}, the first term is uniformly bounded in $\bfB_\eps^{2 d}$ by $|\fh'|_{3 d}$. The coefficients of the operator acting on $\Xi_\eps(\fh')$ in the second term vanish to high order as $\eps\searrow 0$, and thus uniform bounds in $\bfB_\eps^{2 d}$ for it in terms of $|\fh'|_{3 d}$ are straightforward. This establishes the bound
    \[
      \|\phi_\eps'(\fh)(\fh')\|_{2 d}\leq C_1|\fh'|_{3 d},
    \]
    where $C_1$ depends on $\|h_{{\rm in},\eps}\|_{\cC_{\sop,\eps}^{3 d,\alpha_\circ+d,\hat\alpha+d}}$ but not on $\eps$. The bound
    \[
      \|\phi_\eps'(\fh)(\fh',\fh'')\|_{2 d}\leq C_2|\fh'|_{3 d}|\fh''|_{3 d}
    \]
    is a simple consequence of the high order of vanishing of $\fh'$ and $\fh''$ as $\eps\searrow 0$.
  \item{\rm (Linear right inverses.)} Given $\fh\in B_{\eps,\eta}^\infty$ with $|\fh|_{3 d}<1$, we define a right inverse of $\phi_\eps'(\fh)\colon B_\eta^\infty\to\bfB_\eta^\infty$ by $\psi_{\eps,\eta,\delta}(\fh):=L_{g_\eps+\Xi_\eps(\fh);g_{0,\eps},\delta,+}^{-1}\colon\bfB_\eta^\infty\to B_{\eta-\delta}^\infty$; see Theorem~\ref{ThmTrAsyTame}. Moreover, the estimates~\eqref{EqTrAsyTameEst} are tame estimates for $\psi_{\eps,\eta,\delta}(\fh)$ whose constants are uniform in $\eps,\delta$.
  \item{\rm (Smoothing operators.)} We define smoothing operators $S_\theta\colon B_{\eps,\eta}^\infty\to B_{\eps,\eta-\theta^{-1/2}}^\infty$ by separately smoothing each component. Smoothing operators on $H_{\sop,\eps}^{k,\alpha_\circ,\hat\alpha}$ were constructed in \citeII{Corollary~\ref*{CorNTameSmooths}}, and those on $H^k([t_0,t_1])^{\bullet,-}$ in \citeII{Corollary~\ref*{CorNTameSmooth}}.
  \item{\rm (Accuracy $\phi_\eps(0)$.)} We have $\phi_\eps(0)=P_\eps(g_{0,\eps}+h_{{\rm in},\eps};g_{0,\eps})$, which we need to be smaller (in $\bfB^{2 d}$) than a constant depending on $C_s$ for $s\leq 64 d^2+133 d+52$. The above arguments show that $C_s$ only depends on $\|h_{{\rm in},\eps}\|_{\cC_{\sop,\eps}^{s+d,\alpha_\circ+d,\hat\alpha+d}(\Omega_\eps)}\leq\|h_{{\rm in},\eps}\|_{\cC_{\sop,\eps}^{T,T,T}(\Omega_\eps)}\leq\|\wt h_{\rm in}\|_{\cC_\sop^{T,T,T}(\wt\omega)}$.
  \end{enumerate}
  Theorem~\ref{ThmNM} (with $\eta=\lambda$) thus produces, for all sufficiently small $\eps>0$, a solution $\fh_\eps\in D_\eps^{3 d+3}$ of $\phi_\eps(\fh_\eps)=0$, and thus a solution $h_\eps:=\Xi_\eps(\fh_\eps)$ of $P_\eps(g_\eps+h_\eps;g_{0,\eps})=0$. Considering the function spaces used, $h_\eps$ vanishes for $t\leq t_0-\lambda$. Uniqueness of solutions of quasilinear wave equations implies that $h_\eps$ must in fact vanish for $t\leq 0$.

  %%%%%%%%%%
  \pfstep{Higher regularity, higher order of vanishing.} For any fixed $\eps\in(0,\eps_2]$, standard (finite time) continuation criteria for quasilinear wave equations imply that $h_\eps$ is, in fact, smooth on $\Omega_\eps$ since $g_{0,\eps}$, $h_{{\rm in},\eps}$, and $P_\eps(g_{0,\eps}+h_{{\rm in},\eps};g_{0,\eps})$ are. In order to show $\wt h\in\eps^\infty\cC_\sop^\infty(\wt\Omega\cap\{\eps\leq\eps_2\})$, it thus suffices to show that for all $D\in\N$ there exists $\eps_D$ so that $\wt h\in\eps^D\cC_\sop^D(\wt\Omega\cap\{\eps\leq\eps_D\})$. But this follows by applying Theorem~\ref{ThmNM} with increasingly large values of $d$ (tending to $+\infty$ as $D\nearrow+\infty$). (This is possible since $\wt h_{\rm in}\in\cC_\sop^{T,T,T}$ for \emph{all} $T$.) By uniqueness for quasilinear wave equations, the solutions thus produced for any fixed value of $\eps$ are independent of the value of $D$ (with $\eps_D\geq\eps$).

  %%%%%%%%%%
  \pfstep{Solutions in a half space.} To prove the final statement, write $H$ for the Heaviside function, and set
  \[
    \tilde f_-:=H(t_0-t)\wt P(\wt g_0+\wt h_{\rm in};\wt g_0).
  \]
  Due to the infinite order vanishing of $\wt P(\wt g_0+\wt h_{\rm in};\wt g_0)$ at $t=t_0$ and using~\eqref{EqNMLocFormal} and $\wt h_{\rm in}\in\eps^\infty\cC_\sop^\infty(\wt\Omega)$, we have $\tilde f_-\in\eps^\infty\cC_\sop^\infty(\wt\Omega)$. Write $\tilde f_-=(f_{-,\eps})_{\eps\in(0,\eps_1]}$. Solving~\eqref{EqNMLocSol} on the domain $\Omega_\eps\cap\{t\geq t_0\}$ with trivial Cauchy data at $t=t_0$ is the same as solving
  \[
    P_\eps(g_\eps+h_{{\rm in},\eps}+h_\eps;g_{0,\eps}) - f_{-,\eps} = 0
  \]
  on $\Omega_\eps$ with $h_\eps|_{t<t_0}=0$. The solution of this equation can thus be accomplished by applying the above arguments to the map
  \[
    \phi_\eps(\fh) := P_\eps(g_\eps+\Xi_\eps(\fh);g_{0,\eps}) - f_{-,\eps}
  \]
  and setting $h_\eps=\Xi_\eps(\fh_\eps)$ for the solution $\fh_\eps$ of $\phi_\eps(\fh_\eps)=0$ as before. Indeed, the term $f_{-,\eps}$ only enters in the estimate~\eqref{EqNMLocphiEst}, and specifically affects the constant $C_s$ which, as before, can be chosen to be a function of $\|h_{{\rm in},\eps}\|_{\cC_{\sop,\eps}^{s+d,\alpha_\circ+d,\hat\alpha+d}(\Omega_\eps)}$.

  %%%%%%%%%%
  \pfstep{Regularity in $\eps$.} We consider the setting~\eqref{EqNMLocHinSol}; the proof when $\wt P(\wt g_0+\wt h_{\rm in};\wt g_0)$ vanishes to infinite order at $t=t_0$ is analogous. Differentiating $P_\eps(g_\eps+h_\eps;g_{0,\eps})=0$, with $g_\eps=g_{0,\eps}+h_{{\rm in},\eps}$ and $h_\eps\in\eps^\infty\CI_\sop(\wt\Omega)$ along $\eps\pa_\eps$ gives
  \begin{align*}
    &D_{g_\eps+h_\eps}\Ric(\eps\pa_\eps g_\eps + \eps\pa_\eps h_\eps) - \Lambda(\eps\pa_\eps g_\eps+\eps\pa_\eps h_\eps) - [\eps\pa_\eps,\delta_{g_{0,\eps},\gamma_C}^*]\Ups(g_\eps+h_\eps;g_{0,\eps}) \\
    &\qquad\qquad - \delta_{g_{0,\eps},\gamma_C}^*\Bigl( D_1\Ups|_{g_\eps+h_\eps}(\eps\pa_\eps g_\eps+\eps\pa_\eps h_\eps;g_{0,\eps}) + D_2\Ups|_{g_{0,\eps}}(g_\eps+h_\eps;\eps\pa_\eps g_{0,\eps})\Bigr) = 0.
  \end{align*}
  If in this expression we set $h_\eps$ and $\eps\pa_\eps h_\eps$ to $0$, the resulting tensor is of class $\CIdot\subset\eps^\infty\cC_\sop^\infty$. Thus, the terms involving $h_\eps$ without $\eps$-derivatives contribute $\eps^\infty\cC_\sop^\infty$, whereas those terms involving $\eps\pa_\eps h_\eps$ combine to give $L_{g_\eps+h_\eps;g_{0,\eps}}(\eps\pa_\eps h_\eps)$. In summary, we obtain a linear equation
  \[
    L_{g_\eps+h_\eps;g_{0,\eps}}(\eps\pa_\eps h_\eps) = f_\eps \in \eps^\infty\cC_\sop^\infty(\wt\Omega),
  \]
  for $\eps\pa_\eps h_\eps$, with $f_\eps=0$ for $t\leq t_0$. Proposition~\ref{PropEs} implies that $\eps\pa_\eps h_\eps\in\eps^\infty\cC_\sop^\infty$. (This argument can be made rigorous by considering, for any fixed $\eps>0$, finite difference quotients---for which one obtains uniform estimates as the step size tends to $0$---and taking limits as the step size tends to $0$.) Higher $\eps$-regularity is proved similarly.
\end{proof}

Recall now that if $X$ is a spacelike hypersurface inside of a spacetime $(M_\eps,g_\eps)$, then the initial data of $g_\eps$ at $X$ are the first, resp.\ second fundamental form $\gamma_\eps$, resp.\ $k_\eps$ of $X\subset(M_\eps,g_\eps)$. If $g_\eps$ satisfies the Einstein vacuum equations $\Ric(g_\eps)-\Lambda g_\eps=0$ at $X$, then $\gamma_\eps,k_\eps$ satisfy the \emph{constraint equations}, see \citeI{(\ref*{EqIConstraints})}. A standard argument involving the second Bianchi identity conversely shows that if $g_\eps$ solves a gauge-fixed version of the Einstein equations, and its initial data satisfy the constraint equations and the gauge condition, then $g_\eps$ solves the Einstein equations (and the gauge condition holds) in the domain of dependence of the Cauchy hypersurface $X$ \cite[\S{VI.8}]{ChoquetBruhatGR}. For present purposes, we can outsource the discussion of the constraints to \cite{HintzGlueLocI}, cf.\ the infinite order vanishing of $\Ric(\wt g_0)-\Lambda\wt g_0$ at $\wt X$ in Theorem~\ref{ThmTrGlueLocI}, and simply note the following.

\begin{cor}[Nonlinear solution of the Einstein vacuum equations in small domains]
\label{CorTrLocEin}
  In the notation of Theorem~\usref{ThmTrLoc}, and assuming that $P_\eps(g_{0,\eps}+h_{{\rm in},\eps};g_{0,\eps})$ vanishes to infinite order at $t=t_0$, suppose that the gauge 1-form
  \[
    \Ups(g_{0,\eps}+h_{{\rm in},\eps};g_{0,\eps})
  \]
  (defined in~\eqref{EqEOpNonlin}) vanishes to second order at $t=t_0$. Write $g_\eps=g_{0,\eps}+h_{{\rm in},\eps}+h_\eps$ where $h_\eps$ is the solution of~\eqref{EqNMLocSol} in $t\geq t_0$ which vanishes to infinite order at $t=t_0$. Then
  \begin{equation}
  \label{EqTrLocEin}
    \Ric(g_\eps) - \Lambda g_\eps = 0\ \text{on}\ \Omega_\eps\cap\{t\geq t_0\}.
  \end{equation}
\end{cor}
\begin{proof}
  Applying the second Bianchi identity $\delta_{g_\eps}\sfG_{g_\eps}\Ric(g_\eps)=0$ to the equation (cf.\ \eqref{EqEOpNonlin})
  \[
    \Ric(g_\eps) - \Lambda g_\eps - \delta_{g_{0,\eps},\gamma_C}^*\Ups(g_\eps;g_{0,\eps}) = 0
  \]
  implies the equation
  \[
    2\delta_{g_\eps}\sfG_{g_\eps}\delta_{g_{0,\eps},\gamma_C}^* \Ups = 0,\qquad \Ups := \Ups(g_\eps;g_{0,\eps}).
  \]
  This is a linear homogeneous wave equation for $\Ups$. Since $h_\eps$ vanishes to third order at $t=t_0$, the gauge 1-form $\Ups$ vanishes to second order at $t=t_0$, and thus has trivial Cauchy data. Therefore, $\Ups=0$ in $\Omega_\eps\cap\{t\geq t_0\}$, which implies the claim.
\end{proof}

This is the starting point of the semiglobal existence result, to which we turn next.

%%%%%%%%%%%%%%%%%%%%%%%%%%%%%%%%%%%%%%%%%%%%%%%%%%
\subsection{Semiglobal solution}
\label{SsTrSemi}

Concatenating Theorem~\ref{ThmTrLoc} (and uniform bounds away from the small black hole as in Remark~\ref{RmkTrLocTriv}) a finite number of times, we can now solve the gluing problem over any compact subset of the original spacetime $(M,g)$.

\begin{thm}[Semiglobal solution of the Einstein vacuum equations]
\label{ThmTrSemi}
  Suppose $\wt g_0=(g_{0,\eps})_{\eps\in(0,1)}$ is a (polyhomogeneous) Kerr-mod-$\cO(\hat\rho^2)$ metric on $\wt M$ associated with $(M,g)$, $\cC$, $(\bhm,\bha)$, as produced by Theorem~\usref{ThmTrGlueLocI}. Let $K\subset M$ be any compact set; in the setting~\eqref{ItTrGlueLocIKIDs} of Theorem~\usref{ThmTrLoc} we require $K$ to lie in the future of the Cauchy hypersurface $X$. Then for a sufficiently small $\eps_0>0$, there exists a metric perturbation
  \[
    \wt h = (h_\eps)_{\eps\in(0,\eps_0]} \in \CIdot\bigl(\wt\upbeta^{-1}([0,\eps_0]\times K)\setminus\{|\hat x|<\bhm\};S^2\wt T^*\wt M\bigr)
  \]
  (i.e.\ $(\eps\pa_\eps)^i\wt h\in\eps^\infty\CI_\sop$ for all $i\in\N_0$) so that for each $\eps\in(0,\eps_0]$, the metric $g_\eps=g_{0,\eps}+h_\eps$ solves the Einstein vacuum equations
  \[
    \Ric(g_\eps) - \Lambda g_\eps = 0
  \]
  on $K\setminus\{|\hat x|<\bhm\}$. Moreover, in settings~\eqref{ItTrGlueLocIKIDs} and \eqref{ItTrGlueLocIKdS} of Theorem~\usref{ThmTrGlueLocI} (and using the notation $\cU\subset X$ from there), we have $g_\eps=g$ outside the domain of influence of $\bar\cU$.
\end{thm}
\begin{proof}
  The proof is very similar to that of \citeII{Theorem~\ref*{ThmScSG}} (see also the final paragraph of \citeII{\S\ref*{SssNToyT}}); only the fact that we need to keep track of the gauge condition (to ensure that we end up solving the Einstein equations) requires attention. Thus, we fix a metric splitting $\R\times X_0$ of $(M,g)$ \cite{BernalSanchezTimeFn} so that, for $X_\cT:=\{\cT\}\times X_0$, we have $J^-(K)\cap X\subset J^+(X_1)$. Let $X'_0\subset(-1,1)\times X_0$ be a modification of $X_0$ near $\{\fp\}=X_0\cap\cC$ to a Cauchy hypersurface which is equal to a level set of the Fermi normal coordinate $t$ near $\fp$, and equal to $X_0$ outside of a neighborhood of $\fp$ (see \citeI{Lemma~\ref*{LemmaGKMod}}). Using \citeII{Theorem~\ref*{ThmIVP}}, we can modify $\wt g_0$ by adding an element of $\eps^\infty\CI(\wt M;S^2\wt T^*\wt M)$ so that the new $\wt g_0$ solves $\Ric(\wt g_0)-\Lambda\wt g_0$ not only to infinite order as $\eps\searrow 0$, but also to infinite order at the lift of $[0,1)\times K'_0$; here we take $K'_0=X'_0$ in the settings~\eqref{ItTrGlueLocIKIDs} and \eqref{ItTrGlueLocIKdS} of Theorem~\ref{ThmTrGlueLocI}, whereas in the setting~\eqref{ItTrGlueLocINoncompact} we fix $K'_0\subset X'_0$ to be a sufficiently large compact set so that $K$ is contained in the domain of dependence of an open subset $V'_0\subset\ol{V'_0}\subset(K'_0)^\circ$ (and thus the potential violation of the constraints in $X'_0\setminus K'_0$ will not propagate, for sufficiently small $\eps$, into the set $K$).

  We parameterize $\cC=c(I_\cC)$ by its arc length parameter $t$. Let $\cT^+=\max_K\cT$. For each $t$ in the set $I_K\subset I_\cC$ of times with $\cT(c(t))\in[0,\cT^+]$, fix the value $\lambda=\lambda(t)>0$ from Theorem~\ref{ThmTrLoc}. Since $I_K$ is compact, we can pick a finite subset $\{t_1,\ldots,t_N\}\subset I_K$ and a small number $\eta>0$ so that $\eta<\frac12\lambda(t_i)$ for all $i=1,\ldots,N$, and so that the union of intervals $(t_i-\eta,t_i+\eta)$ covers $I_K$.

  Apply \citeII{Proposition~\ref*{PropGlDynCover}} with the value of $\eta$ fixed above. This produces a finite collection $\Omega_0,\ldots,\Omega_J$ of standard domains in $M$ which satisfy the following properties: \begin{enumerate*} \item $K\subseteq\bigcup_{j=0}^J\Omega_j$; \item the initial Cauchy surface of $\Omega_j$ is contained either in $X'_0$ or in $\bigcup_{k\leq j-1}\Omega_k^\circ$; \item all $\Omega_j$ which intersect $\cC$ are standard domains of the form $\Omega_{t_0,t_1,r_0}$ for some values $t_0,t_1\in I_\cC$, $r_0>0$ depending on $j$ but always satisfying $|t_1-t_0|,r_0<\eta$.\end{enumerate*} Write $\wt\Omega_j=\wt\upbeta^{-1}([0,1)\times\Omega_j)\setminus\wt K^\circ$ for the standard domain in $\wt M$ associated to $\Omega_j$. 

  We construct the solution $\wt h_{\rm in}$ of
  \begin{equation}
  \label{EqTrSemiSol}
    \wt P(\wt g_0+\wt h_{\rm in};\wt g_0)=0,\quad \wt\Ups(\wt g_0+\wt h_{\rm in};\wt g_0)=0,\qquad \wt h_{\rm in}\ \text{vanishes to infinite order at}\ \wt X.
  \end{equation}
  on (the lift to $\wt M\setminus\wt K^\circ$ of) $\bigcup_{k\leq J}\Omega_k$ for all sufficiently small $\eps>0$ via an iterative construction in $J_0$. Suppose we have a solution on the lift of $\bigcup_{k\leq J_0-1}\Omega_k$ for some $J_0\in\{0,\ldots,J\}$, defined for $\eps\in(0,\eps_{J_0-1}]$ (with $\eps_{-1}:=\frac12$, say).

  If the initial boundary hypersurface $X_{J_0}$ of $\Omega_{J_0}$ is contained in $X'_0$, we consider two cases. First, suppose that $X_{J_0}\cap\cC=\emptyset$. In this case, standard hyperbolic theory (cf.\ Remark~\ref{RmkTrLocTriv}) applies to the equation $P_\eps(g_{0,\eps}+h_{J_0,\eps};g_{0,\eps})=0$ and produces, for all $\eps<\eps_{J_0}$ with $\eps_{J_0}>0$ sufficiently small, a local solution $\wt h_{J_0}=(h_{J_0,\eps})_{\eps\in(0,\eps_{J_0})}\in\CIdot(\wt\Omega_{J_0})$ which vanishes to infinite order at the initial boundary hypersurface. (Note that by uniqueness of solutions of quasilinear wave equations, $\wt h_{J_0}$ equals the already constructed part of $\wt h_{\rm in}$ on the common domain of definition). Second, suppose that $X_{J_0}\cap\cC\neq\emptyset$; then Theorem~\ref{ThmTrLoc} applies and produces, in view of Corollary~\ref{CorTrLocEin}, the desired solution of~\eqref{EqTrSemiSol} on $\wt\Omega_{J_0}$.

  If $X_{J_0}$ is not contained in $X'_0$, pick $\chi\in\CI(\Omega_{J_0})$ to be equal to $0$ near the initial boundary hypersurface $X_{J_0}$ of $\Omega_{J_0}$ and equal to $1$ outside of $\bigcup_{k\leq J_0-1}\Omega_k^\circ$. Set $\wt h_{J_0,{\rm in}}:=\chi\wt h_{\rm in}$; then
  \[
    \wt P(\wt g_0+\chi\wt h_{\rm in};\wt g_0) = 0\ \text{for $\eps\in(0,\eps_{J_0-1}]$ near the lift of $X_{J_0}$ to $\wt M\setminus\wt K^\circ$.}
  \]
  Theorem~\ref{ThmTrLoc} thus applies and produces a number $\eps_{J_0}\in(0,\eps_{J_0-1}]$ and a solution $\wt h_{J_0}=(h_{J_0,\eps})_{\eps\in(0,\eps_{J_0}]}\in\eps^\infty\cC_\sop^\infty(\wt\Omega_{J_0})$, vanishing near $\wt X_{J_0}$, of the equation
  \[
    \wt P(\wt g_0+\chi\wt h_{\rm in}+\wt h_{J_0};\wt g_0) = 0\ \text{on $\wt\Omega_{J_0}$}
  \]
  Again by uniqueness, we have $\chi\wt h_{\rm in}+\wt h_{J_0}=\wt h_{\rm in}$ on (the lift to $\wt M\setminus\wt K^\circ$ of) $\bigcup_{k\leq J_0-1}\Omega_k$. Since $\wt\Ups(\wt g_0+\chi\wt h_{\rm in};\wt g_0)=0$ near $X_{J_0}$, the gauge propagation argument in the proof of Corollary~\ref{CorTrLocEin} implies that we must have $\wt\Ups(\wt g_0+\chi\wt h_{\rm in}+\wt h_{J_0};\wt g_0)$ on all of $\wt\Omega_{J_0}$.

  We can thus extend $\wt h_{\rm in}$ to $\bigcup_{k\leq J_0}\Omega_k$ by defining it to be equal to $\chi\wt h_{\rm in}+\wt h_{J_0}$ on $\wt\Omega_{J_0}$. This produces a solution of~\eqref{EqTrSemiSol} on $\bigcup_{k\leq J_0}\Omega_k$ and thus completes the inductive step.

  The equations~\eqref{EqTrSemiSol} imply $\Ric(\wt g_0+\wt h_{\rm in})-\Lambda(\wt g_0+\wt h_{\rm in})=0$. Once we reach $J_0=J$, we thus simply set $\wt h:=\wt h_{\rm in}$ to conclude.

  Finally, note that since we construct $h_\eps$ as the solution of a quasilinear wave equation near $(M_\eps,g_{0,\eps})$, we automatically get a finite speed of propagation statement. Concretely, recalling part~\eqref{ItTrGlueLocIDom} of Theorem~\ref{ThmTrGlueLocI}, we have $\Ric(\wt g_0)-\Lambda\wt g_0=0$ outside of the domain of influence of $\cU$; and therefore $\wt h$ vanishes there as well.
\end{proof}

%%%%%%%%%%%%%%%%%%%%%%%%%%%%%%%%%%%%%%%%%%%%%%%%%%%%%%%%%%%%%%%%%%%%%%
\section{Applications to black hole mergers}
\label{SApp}

Due to the general applicability of Theorems~\ref{ThmTrGlueLocI} and \ref{ThmTrSemi}, we can now freely engineer spacetimes describing mergers of black holes with extreme mass ratios. We describe one particularly interesting scenario:

\begin{thm}[Extreme mass ratio mergers]
\label{ThmAppMerger}
  Let $(M,g^{\rm KdS}_{M_\bullet,a_\bullet})$ be a slowly rotating Kerr--de~Sitter spacetime, extended beyond the event and cosmological horizon (see for example~\cite[(3.9), \S{3.2}]{HintzVasyKdSStability}), so $M=\R_{t_{\rm KdS}}\times X$ where $X=\{p\in\R^3\colon|p|\in(r_-,r_+)\}$ where $r_-$ is less than the radius $r_{\cH^+}$ of the event horizon $\cH^+$ (but larger than the radius of the Cauchy horizon), and $r_+<\infty$ is larger than the radius $r_{\bar\cH^+}$ of the cosmological horizon $\bar\cH^+$; and the level sets of $t_{\rm KdS}$ are transversal to the future event and cosmological horizons. Let $\cC$ be a timelike geodesic in $(M,g)$ passing through $t_{\rm KdS}=0$ which is future inextendible and crosses into $r<r_{\cH^+}$ or $r>r_{\bar\cH^+}$ in finite time. Denote by $t,x$ Fermi normal coordinates around $\cC$; fix subextremal Kerr parameters $\bhm>0$, $|\bha|<\bhm$. Then there exists a family $g_\eps$, $0<\eps\leq\eps_0\ll 1$, of solutions of
  \begin{equation}
  \label{EqAppMergerSol}
    \Ric(g_\eps) - \Lambda g_\eps = 0
  \end{equation}
  on $M_\eps\cap\{t\geq 0\}$ where $M_\eps:=M\setminus\{|x|\leq\eps\bhm\}$, which satisfies the following properties.
  \begin{enumerate}
  \item\label{ItAppMerger1}{\rm (Small black hole traveling along $\cC$.)} $\wt g=(g_\eps)_{\eps\in(0,1)}$ is a Kerr-mod-$\cO(\hat\rho^2)$ glued spacetime metric on $\wt M=[[0,1)_\eps\times M;\{0\}\times\cC]$ of regularity $\CI+\cA_\phg^{\cE,\hat\cE}$ where $\cE\subset(\N_0+1)\times\N_0$ and $\hat\cE\subset(\N_0+3)\times\N_0$.
  \item\label{ItAppMerger2}{\rm (Late time behavior.)} There exist Kerr--de~Sitter parameters
    \[
      (M_{\bullet,\eps},a_{\bullet,\eps}) = (M_\bullet,a_\bullet) + \cO(\eps)
    \]
    so that, for some ($\eps$-independent) $\alpha>0$,
    \begin{equation}
    \label{EqAppMergerLate}
      g_\eps = g^{\rm KdS}_{M_{\bullet,\eps},a_{\bullet,\eps}} + \tilde g_\eps,\qquad 
    |\pa_{t_{\rm KdS}}^i\pa_x^\beta(\tilde g_\eps)_{\mu\nu}|\leq C_{i\alpha}\eps e^{-\alpha t_{\rm KdS}}\ \forall\,i,\beta.
    \end{equation}
  \end{enumerate}
\end{thm}

When $\cC\cap\cH^+\neq\emptyset$ (resp.\ $\cC\cap\bar\cH^+\neq\emptyset$), the spacetime $(M_\eps,g_\eps)$ thus describes the merger (resp.\ the escape) of a mass $\eps$ Kerr black hole (in the precise sense of Definition~\ref{DefEGlue}) with (resp.\ from) the unit mass $M_\bullet$ Kerr--de~Sitter black hole which, after the merger (resp.\ after the escape of the small black hole through the cosmological horizon), settles down an exponential rate to a nearby Kerr--de~Sitter black hole. Tracking the evolution of the small black hole deep inside the unit mass KdS black hole would be very challenging near the Cauchy horizon of the latter in view of the Strong Cosmic Censorship conjecture \cite{DafermosLukKerrCauchyHorI}. If the small black hole crosses the cosmological horizon, its further evolution is, conjecturally, described via a combination of \cite{FournodavlosSchlueExpanding} and dynamical versions of the many-black-hole spacetimes constructed in \cite{HintzGluedS}: one expects the small black hole to settle down to a KdS black hole as it approaches the conformal boundary in the cosmological region of the unit mass KdS black hole.

\begin{proof}[Proof of Theorem~\usref{ThmAppMerger}]
  Let us extend $(M,g^{\rm KdS}_{M_\bullet,a_\bullet})$ slightly to $\R\times X_\delta$, $X_\delta=\{r_--\delta<r<r_++\delta\}$, where $0<\delta\ll 1$, and correspondingly extend $\cC$ towards the future as a timelike geodesic. Since $\cC$ crosses $r=r_{\cH^+}$ or $r=r_{\bar\cH^+}$ in finite time, it then continues to reach $r=r_--\delta$ or $r=r_++\delta$ in finite time as well, say by time $t_{\rm KdS}=T$. Apply then Theorem~\ref{ThmTrGlueLocI} with data $(M_\delta\cap\{-2<t<T+2\},g^{\rm KdS}_{M_\bullet,a_\bullet})$, $\cC$, $\bhm,\bha$ to obtain a formal solution $\wt g_0$, and then Theorem~\ref{ThmTrSemi} with $K=[-1,T+1]\times X$ to find a correction $\wt h\in\eps^\infty\cC^\infty$ so that
  \[
    (g_\eps)_{\eps\in(0,\eps_0]} = \wt g := \wt g_0 + \wt h
  \]
  satisfies~\eqref{EqAppMergerSol}.

  For all $\eps>0$ with $\eps\bhm<\frac{\delta}{2}$, consider the initial data $(\gamma_\eps,k_\eps)$ of $g_\eps$ at $\{t_{\rm KdS}=T+1\}\subset M$. Note that the small black hole has already left the set $X_{\delta/2}$ at this point. Thus, Definition~\ref{DefEGlue}\eqref{ItEGlueModFar} implies that $g_\eps-g=\cO(\eps)$ near $M\cap\{t_{\rm KdS}=T+1\}$ as a smooth symmetric 2-tensor, and therefore $(\gamma_\eps,k_\eps)$ is $\eps$-close (in $\cC^k$ for every $k$) to the initial data $(\gamma_{M_\bullet},k_{M_\bullet})$ of $g^{\rm KdS}_{M_\bullet,a_\bullet}$. Therefore, the black hole stability result \cite[Theorem~1.4]{HintzVasyKdSStability} applies and describes (in a suitable generalized harmonic gauge) the full future evolution of $(\gamma_\eps,k_\eps)$, leading to~\eqref{EqAppMergerLate}.
\end{proof}

\begin{rmk}[Kerr black hole mergers]
\label{RmkAppMergerKerr}
  One can instead use the nonlinear stability result for slowly rotating Kerr black holes proved in \cite{KlainermanSzeftelKerr,GiorgiKlainermanSzeftelStability,ShenGCMKerr} to obtain a spacetime describing the merger of a light subextremal Kerr black hole with a unit mass slowly rotating Kerr black hole, following by the relaxation of the resulting single Kerr black hole. (One now requires that $\cC$ cross the future event horizon in finite time.) Here, it is important that the metrics produced by Theorems~\ref{ThmTrGlueLocI} and \ref{ThmTrSemi} obey a finite speed of propagation property: this ensures, in the notation of the above proof, that the initial data $(\gamma_\eps,k_\eps)$ at time $T+1$ (which we take to be a spacelike hypersurface transversal to $\cH^+$ and asymptotic to spacelike infinity $i^0$) are \emph{compactly supported} perturbations of the initial data of the original Kerr black hole $g^{\rm Kerr}_{M_\bullet,a_\bullet}$.
\end{rmk}

Two more variants are as follows.
\begin{enumerate}
\item{\rm (Gluing in several black holes.)} There is a natural analogue of Theorem~\ref{ThmTrGlueLocI} when instead of a single curve $\cC$ we have a finite collection $\cC_1,\ldots,\cC_N$, $N\in\N$ of mutually disjoint inextendible timelike geodesics. Also Theorem~\ref{ThmTrSemi} extends to such a scenario (seeing as the construction of the metric correction is accomplished by patching together local solutions in sets which one can thus simply each choose to contain points of at most a single $\cC_i$). If one glues $N$ small mass black holes into a subextremal Kerr--de~Sitter or Kerr spacetime, one obtains an analogue of Theorem~\ref{ThmAppMerger} or Remark~\ref{RmkAppMergerKerr} which describes the merger of $N$ subextremal Kerr black holes of small mass with a single unit mass black hole.
\item{\rm (Iterative application of the gluing construction: merger cascades.)} We can apply the black hole gluing Theorems~\ref{ThmTrGlueLocI} and \ref{ThmTrSemi} (including the version involving several non-intersecting timelike geodesics) to any fixed spacetime $(M_\eps,g_\eps)$ which arose from just such a gluing construction, and iterate such a procedure any finite number of times. In this fashion, we can, for example, construct spacetimes in which a mass $\eps_1\eps_2\eps_3$ black hole merges with a mass $\eps_1\eps_2$ black hole which in turn merges, together with another mass $\eps_1\eps_2$ black hole, with a mass $\eps_1$ black hole which finally falls into a unit mass black hole. See Figure~\ref{FigAppCascade}.
\end{enumerate}

\begin{figure}[!ht]
\centering
\includegraphics{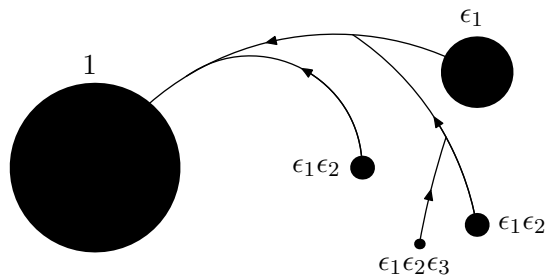}
\caption{Illustration of a cascade of extreme mass ratio mergers with a unit mass Kerr--de~Sitter black hole. The arrows indicate the timelike geodesics along which the black holes are traveling. The black discs indicate the location of each black hole at $t_{\rm KdS}=0$; the numbers are the masses.}
\label{FigAppCascade}
\end{figure}

%%%%%%%%%%%%%%%%%%%%%%%%%%%%%%%%%%%%%%%%%%%%%%%%%%%%%%%%%%%%%%%%%%%%%%
\bibliographystyle{alphaurl}
\bibliography{
/scratch/users/hintzp/ownCloud/research/bib/math,
/scratch/users/hintzp/ownCloud/research/bib/mathcheck,
/scratch/users/hintzp/ownCloud/research/bib/phys
}

\end{document}